\documentclass[11pt,titlepage]{amsart}

 \usepackage{amsaddr}
\usepackage{hyperref}
\usepackage{cleveref}
\usepackage{latexsym,amsmath, bm}
\usepackage{enumerate}
\usepackage{amsfonts}
\usepackage{amssymb}
\usepackage{geometry}
\usepackage{latexsym}
\usepackage{fixmath}
\usepackage{faktor}
\usepackage{mathtools}
\usepackage[T1]{fontenc}
\usepackage{fourier}
\usepackage{bbm}
\usepackage{xcolor}
\usepackage{tikz}
\usetikzlibrary{shapes,decorations,arrows,calc,arrows.meta,fit,positioning}
\usepackage{graphicx}
\usepackage{fullpage}

\numberwithin{equation}{section}

\renewcommand\ln{\log}
\newcommand\disteq{\sim}

\newcommand{\vgamma}{\vec\gamma}
\newcommand{\vGamma}{\vec\Gamma}
\newcommand{\vDelta}{\vec\Delta}

\newcommand{\va}{\vec a}
\newcommand{\vb}{\vec b}

\newcommand{\vd}{\vec d}

\newcommand{\vg}{\vec g}
\newcommand{\vh}{\vec h}
\newcommand{\vi}{\vec i}

\newcommand{\vk}{\vec k}

\newcommand{\vm}{\vec m}

\newcommand{\vt}{\vec t}

\newcommand{\vv}{\vec v}

\newcommand{\vx}{\vec x}
\newcommand{\vy}{\vec y}

\newcommand{\vC}{\vec C}

\newcommand{\vG}{\vec G}
\newcommand{\vH}{\vec H}

\newcommand{\vJ}{\vec J}

\newcommand{\vM}{\vec M}

\newcommand{\vX}{\vec X}
\newcommand{\vY}{\vec Y}

\newcommand\FF{\mathbb{F}}

\newcommand\fD{\mathfrak{D}}

\newcommand\fM{\mathfrak{M}}
\newcommand\fP{\mathfrak{P}}

\newcommand\cA{\mathcal{A}}
\newcommand\cB{\mathcal{B}}
\newcommand\cC{\mathcal{C}}
\newcommand\cD{\mathcal{D}}
\newcommand\cE{\mathcal{E}}

\newcommand\cG{\mathcal{G}}

\newcommand\cK{\mathcal{K}}
\newcommand\cL{\mathcal{L}}
\newcommand\cM{\mathcal{M}}
\newcommand\cN{\mathcal{N}}
\newcommand\cO{\mathcal{O}}
\newcommand\cP{\mathcal{P}}
\newcommand\cQ{\mathcal{Q}}
\newcommand\cR{\mathcal{R}}
\newcommand\cS{\mathcal{S}}
\newcommand\cT{\mathcal{T}}

\newcommand\cV{\mathcal{V}}

\newcommand\cX{\mathcal{X}}
\newcommand\cY{\mathcal{Y}}

\newcommand\MU{\vec\mu}

\newcommand\PSI{\vec\psi}
\newcommand\RHO{{\vec\rho}}

\newcommand\SIGMA{\vec\sigma}

\newcommand\THETA{\vec\theta}

\newcommand\G{{\vec G}}
\newcommand\vmu{\vec\mu}
\newcommand\vpsi{\vec\psi}
\newcommand\vomega{\vec\omega}
\newcommand\vsigma{\vec\sigma}
\newcommand\vtheta{\vec\theta}

\newcommand\dTV{d_{\mathrm{TV}}}

\newcommand{\Po}{{\rm Po}}
\newcommand{\Bin}{{\rm Bin}}
\newcommand{\Mult}{{\rm Mult}}

\newcommand\nix{\,\cdot\,}
\newcommand\dd{{\mathrm d}}
\renewcommand{\vec}[1]{\boldsymbol{#1}}
\newcommand{\vecone}{\vec{1}}

\newcommand\KL[2]{D_{\mathrm{KL}}\bc{{{#1}\|{#2}}}}

\newcommand\eul{\mathrm{e}}
\newcommand\eps{\varepsilon}
\newcommand\ZZ{\mathbb{Z}}
\newcommand\NN{\mathbb{N}}
\newcommand\pr{\mathbb{P}} 
\renewcommand\Pr{\pr}

\newcommand\Var{\mathrm{Var}}
\newcommand\Erw{\mathbb{E}}

\newcommand\RR{\mathbb{R}}

\newcommand{\Whp}{W.h.p.}
\newcommand{\whp}{w.h.p.}

\newtheorem{definition}{Definition}[section]
\newtheorem{claim}[definition]{Claim}
\newtheorem{example}[definition]{Example}
\newtheorem{remark}[definition]{Remark}
\newtheorem{theorem}[definition]{Theorem}
\newtheorem{lemma}[definition]{Lemma}
\newtheorem{proposition}[definition]{Proposition}
\newtheorem{corollary}[definition]{Corollary}

\newtheorem{fact}[definition]{Fact}

\newcommand\Lem{Lemma}
\newcommand\Prop{Proposition}
\newcommand\Thm{Theorem}

\newcommand\Sec{Section}

\newcommand\bc[1]{\left({#1}\right)}
\newcommand\cbc[1]{\left\{{#1}\right\}}

\newcommand{\bck}[1]{\left\langle{#1}\right\rangle}
\newcommand\brk[1]{\left\lbrack{#1}\right\rbrack}

\newcommand\norm[1]{\left\|{#1}\right\|}
\newcommand\abs[1]{\left|{#1}\right|}

\newcommand{\Erdos}{Erd\H{o}s}
\newcommand{\Renyi}{R\'enyi}

\newcommand{\SYM}{\textbf{SYM}}
\newcommand{\BAL}{\textbf{BAL}}
\newcommand{\POS}{\textbf{POS}}

\newcommand{\matTr}{{\mathrm t}}
\newcommand{\intD}{{\mathrm d}}
\newcommand{\cpxI}{{\mathrm i}}
\DeclareMathOperator{\Cov}{Cov}
\DeclareMathOperator{\im}{im}
\newcommand{\m}{\bar m}

\newcommand{\vmIndP}{P}
\newcommand{\vmIndPSp}{\cP_{\mathrm{L}}}
\newcommand{\vmIndPD}{\Delta_{\mathrm{L}}}
\newcommand{\vmIndSp}{\cL}
\newcommand{\vmInd}{\ell}
\newcommand{\svmInd}{{\underline\vmInd}}
\newcommand{\vmRInd}{\vec\ell}
\newcommand{\vmCol}{\omega}
\newcommand{\vmColSp}{\cO}
\newcommand{\vmPrior}{{\mathfrak p}}
\newcommand{\vmRPrior}{\vec\pi}
\newcommand{\vmAssSp}{\cX}
\newcommand{\vmAss}{\chi}
\newcommand{\vmRAss}{\vec\chi}
\newcommand{\lltLat}{\mathfrak{L}}
\newcommand{\lltLatA}{\mathfrak{C}}
\newcommand{\lltLatR}{\mathfrak{R}}
\newcommand{\lltLatP}{\mathfrak{P}}
\newcommand{\lltVL}{i}
\newcommand{\epsGen}{\varepsilon_{\mathrm{gen}}}
\newcommand{\lltEpsP}{\varepsilon_{\mathrm{p}}}
\newcommand{\secB}{E^{(2)}}
\newcommand{\thiB}{E^{(3)}}
\newcommand{\lltBox}{\cQ}
\newcommand{\lltHEA}{{\underline\vmAss}}
\newcommand{\lltRHEA}{{\underline\vmRAss}}
\newcommand{\lltCFA}{\gamma}
\newcommand{\lltCFR}{\rho}
\newcommand{\lltCFT}{\tau}
\newcommand{\lltRVL}{\vx}
\newcommand{\lltRCFA}{\vec\lltCFA}
\newcommand{\lltRCFR}{\vec\lltCFR}

\newcommand{\lltECFA}{\bar\lltCFA}
\newcommand{\lltECFR}{\bar\lltCFR}

\newcommand{\lltCFAng}{\varphi}
\newcommand{\lltCFInt}{f}
\newcommand{\lltCFLog}{g}
\newcommand{\acSupp}{\cA}
\newcommand{\acAD}{\alpha}
\newcommand{\acRAD}{\vec\alpha}
\newcommand{\acEAD}{\bar\acAD}
\newcommand{\acADSp}{\cP_{\mathrm{A}}}
\newcommand{\acADD}{\Delta_{\mathrm{A}}}
\newcommand{\mcCol}{\vmCol}
\newcommand{\mcColSp}{\Omega}
\newcommand{\mcFPsi}{\Psi}
\newcommand{\mcFPsiP}{P}
\newcommand{\mcFEZ}{Z}
\newcommand{\mcFExi}{\xi}
\newcommand{\mcFdeg}{k}
\newcommand{\mcRFdeg}{\vec\mcFdeg}
\newcommand{\mcFInd}{\vmInd}
\newcommand{\mcRFInd}{\vec\mcFInd}
\newcommand{\mcFSInd}{{\underline\vmInd}}
\newcommand{\mcFIndSp}{{\vmIndSp_{\mathrm{F}}}}
\newcommand{\mcFED}{\mu}
\newcommand{\mcVdeg}{d}
\newcommand{\mcRVdeg}{\vec\mcVdeg}
\newcommand{\mcVInd}{\lambda}
\newcommand{\mcRVInd}{\vec\mcVInd}
\newcommand{\mcVSInd}{{\underline\mcVInd}}
\newcommand{\mcVIndSp}{{\vmIndSp_{\mathrm{V}}}}
\newcommand{\mcVED}{\nu}

\newcommand{\mcFL}{a}
\newcommand{\mcFLSp}{F}
\newcommand{\mcFHESp}{\cA}
\newcommand{\mcFHEA}{y}

\newcommand{\mcFSHEA}{\underline y}
\newcommand{\mcRFSHEA}{{\underline{\vec y}}}
\newcommand{\mcVL}{x}
\newcommand{\mcVLSp}{V}
\newcommand{\mcVA}{\sigma}
\newcommand{\mcRVA}{\vec\mcVA}
\newcommand{\mcRVAN}{\hat{\vsigma}}
\newcommand{\mcRVAPost}{\vsigma}
\newcommand{\mcVHESp}{\cX}
\newcommand{\mcVHEA}{\chi}
\newcommand{\mcVSHEA}{\underline\chi}
\newcommand{\mcRVSHEA}{{\underline{\vec\chi}}}
\newcommand{\mcRVSHEAN}{\hat{\mcRVSHEA}}

\newcommand{\mcT}{t}
\newcommand{\mcSeq}{T}
\newcommand{\mcSeqSp}{\mathfrak T}
\newcommand{\mcDegDens}{\alpha}

\newcommand{\mcFpsi}{\psi}
\newcommand{\mcFRpsi}{\vec\mcFpsi}
\newcommand{\mcFEpsi}{\bar\mcFpsi}
\newcommand{\mcZ}{Z}
\newcommand{\mcZE}{\bar Z}
\newcommand{\mcBMD}{\mu}
\newcommand{\mcBij}{g}
\newcommand{\mcRBij}{\vec\mcBij}
\newcommand{\mcG}{G}
\newcommand{\mcRG}{\vec\mcG}
\newcommand{\mcRGN}{\hat{\mcRG}}

\newcommand{\mcCFA}{\gamma}
\newcommand{\mcRCFA}{\vec\mcCFA}
\newcommand{\mcRCFAN}{\hat{\mcRCFA}}
\newcommand{\mcCFR}{\rho}
\newcommand{\mcRCFR}{\vec\mcCFR}
\newcommand{\mcRCFRN}{\hat{\mcRCFR}}
\newcommand{\mcPrior}{\vmPrior}
\newcommand{\mcCov}[1]{\Sigma_{#1}}
\newcommand{\mcVCov}[1]{\Sigma_{\mathrm{V},#1}}
\newcommand{\mcFCov}[1]{\Sigma_{\mathrm{F},#1}}
\newcommand{\mcECov}[1]{\Sigma_{\mathrm{E},#1}}
\newcommand{\mcVDGCD}{h}

\newcommand{\mcVDTot}{D}
\newcommand{\mcFDTot}{K}
\newcommand{\mcDTot}{D^*}
\newcommand{\mcDVDelta}{\Delta_{\mathrm{D}}}
\newcommand{\mcDFDelta}{\Delta_{\mathrm{K}}}
\newcommand{\mcDVDeltaSp}{\cD_{\mathrm{D}}}
\newcommand{\mcDFDeltaSp}{\cD_{\mathrm{K}}}
\newcommand{\mcRDSC}{\vec t}
\newcommand{\acVdegP}{p_{\mathrm{d}}}
\newcommand{\acFdegP}{p_{\mathrm{k}}}
\newcommand{\acVAD}{\alpha_{\mathrm{V}}}
\newcommand{\acFAD}{\alpha_{\mathrm{F}}}
\newcommand{\acJAD}{\alpha}
\newcommand{\acSuppT}{\acSupp^\circ}
\newcommand{\acSSpV}{\cS}
\newcommand{\acSSpT}{\acSSpV^\circ}
\newcommand{\acS}{s}
\newcommand{\acRS}{\vec\acS}
\newcommand{\acRSN}{\hat{\acRS}}
\newcommand{\acADPD}{\acADD}
\newcommand{\acVAF}[1]{\alpha_{\mathrm{V},#1}}
\newcommand{\acFAF}[1]{\alpha_{\mathrm{F},#1}}
\newcommand{\acJAF}[1]{\acAD_{#1}}
\newcommand\CE[2]{H\bc{{{#1}\|{#2}}}}
\DeclareMathOperator{\supp}{supp}

\renewcommand{\d}{\bar d}
\renewcommand{\k}{\bar k}
\newcommand{\sd}{\underline d}
\newcommand{\sP}{\underline P}
\newcommand{\sk}{\underline k}

\newcommand{\PP}{\fP}

\crefname{theorem}{Theorem}{Theorems}

\newcommand{\hGi}{\hat \G_{t, \eps, \pi,T}}

\newcommand{\Gmm}{\G_{t, \eps, \pi,T}(\vm_\eps(t), \vm'_\eps(t))}

\newcommand{\sGmm}{\G^*_{t, \eps, \pi,T}(\vm_\eps(t), \vm'_\eps(t))}
\newcommand{\sGmpm}{\G^*_{t, \eps, \pi,T}(\vm_\eps(t)+1, \vm'_\eps(t))}
\newcommand{\sGmmp}{\G^*_{t, \eps, \pi,T}(\vm_\eps(t), \vm'_\eps(t)+1)}

\begin{document}

\title{Inference and mutual information on random factor graphs} 

\author{Amin Coja-Oghlan, Max Hahn-Klimroth, Philipp Loick, No\"ela M\"uller, Konstantinos~Panagiotou, Matija Pasch}
\thanks{Amin Coja-Oghlan and Philipp Loick are supported by DFG CO 646/3. Max Hahn-Klimroth is supported by Stiftung Polytechnische Gesellschaft. Konstantinos Panagiotou, Matija Pasch: The research leading to these results has received funding from the European Research Council, ERC Grant Agreement 772606-PTRCSP}

\address{{\tt \{acoghlan,hahnklim,loick,nmueller\}@math.uni-frankfurt.de}, Goethe University, Mathematics Institute, 10 Robert Mayer St, Frankfurt 60325, Germany.}
\address{{\tt \{kpanagio, pasch\}@math.lmu.de}, Mathematisches Institut der Universit\"at M\"unchen, Theresienstr.~39, 80333 M\"unchen, Germany}

\begin{abstract}
\noindent
Random factor graphs provide a powerful framework for the study of inference problems such as decoding problems or the stochastic block model.
Information-theoretically the key quantity of interest is the mutual information between the observed factor graph and the underlying ground truth around which the factor graph was created; in the stochastic block model, this would be the planted partition.
The mutual information gauges whether and how well the ground truth can be inferred from the observable data.
For a very general model of random factor graphs we verify a formula for the mutual information predicted by physics techniques.
As an application we prove a conjecture about low-density generator matrix codes from [Montanari: IEEE Transactions on Information Theory 2005].
Further applications include phase transitions of the stochastic block model and the mixed $\vk$-spin model from physics.
\end{abstract}

\maketitle

\newpage

\section{Introduction}\label{Sec_intro}

\subsection{Background and motivation}\label{Sec_background}
Since the 1990s there has been an immense interest in inference and learning problems on random graphs.
One motivation has been to seize upon random graphs as benchmarks for inference algorithms of all creeds and denominations.
An excellent example of this is the stochastic block model; the impressive literature on this model alone is surveyed in~\cite{Abbe}.
A second, no less salient motivation has been the use of random graphs in probabilistic constructions.
Concrete examples include powerful error correcting codes such as low density generator matrix or low density parity check codes, which have since found their way into modern communications standards~\cite{KRU,RichardsonUrbanke}.
Further prominent recent applications include compressed sensing and group testing~\cite{Aldridge_2019,Donoho_2006,Donoho_2013}.
It appears hardly a stretch to claim that  in terms of real world impact these constructions occupy top ranks among applications of the probabilistic method and, indeed, modern combinatorics generally.

Yet many applications of {the probabilistic method} to inference problems still lack a satisfactory rigorous justification.
Some are supported primarily by empirical evidence, i.e., not much more than a bunch of computer experiments.
Quite a few others have been inspired by a versatile but non-rigorous approach from physics known as the `cavity method.'
But while there has been progress in recent years, vast gaps between the physics predictions and their rigorous vindications remain.
One important reason for this is that the random graph models used in practical inference tend to be significantly more intricate than, say, a classical binomial random graph.
For instance, a highly popular breed of low-density parity check codes use delicately tailored degree distributions for both the variable nodes and the check nodes of the Tanner graph~\cite{RichardsonUrbanke}.

In this paper we significantly advance the rigorous state of the art by corroborating important cavity method predictions wholesale for a rich class of inference problems that accommodates the very general choices of degree distributions of interest in high-dimensional Bayesian inference problems and coding theory.
Generally, the objective in such inference problems is to recover the ground truth from the observable data.
Think, for instance, of retrieving the hidden communities in the stochastic block model or of reconstructing the original message from a noisy codeword.
For this broad class of models we rigorously establish the formulas that the cavity method predicts for the mutual information, which is the key information-theoretic potential that gauges precisely how much it is possible in principle to learn about the ground truth.
Technically we build upon and extend the methods developed in~\cite{CKPZ} for random graph models of \Erdos-\Renyi\ type.
While we follow a similar general proof strategy, the greater generality of the present results necessitates significant upgrades to virtually all of the moving parts.
For example,
due to the more rigid combinatorial structure of graphs with given degrees many of the manoeuvres that are straightforward for binomial random graphs now require delicate coupling arguments.

We proceed to highlight applications of our main results to three specific problems that have each received a great deal of attention in their own right:
low-density generator matrix codes, the stochastic block model and the mixed $\vk$-spin model, which hails from mathematical physics. 
Then in \Sec~\ref{Sec_main} we state the main results concerning the general class of random factor graph models.
\Sec~\ref{Sec_strategy} contains an overview of the proof strategy and a detailed comparison with prior work.

\subsection{Low-density generator matrix codes}\label{Sec_LDGM}
A powerful and instructive class of error-correcting codes, low-density generator matrix (`ldgm') codes are based on random bipartite graphs with given degree distributions. 
Specifically, let $\vd,\vk\geq0$ be bounded integer-valued random variables, let $n$ be an integer and let $\vm\disteq\Po(n\Erw[\vd]/\Erw[\vk])$ be a Poisson variable.
One vertex class $V=\{x_1,\ldots,x_n\}$ of the graph represents the bits of the original message.
The other class $F=\{a_1,\ldots,a_{\vm}\}$ represents the rows of the code's generator matrix.
To obtain the random graph $\G$ create for each variable node $x_i$ an independent copy $\vd_i$ of $\vd$.
Similarly,  create an independent copy $\vk_i$ of $\vk$ for each check node $a_i$.
Then given the event
\begin{align}\label{eqDeg}
\sum_{i=1}^{n}\vd_i=\sum_{i=1}^{\vm}\vk_i
\end{align}
that the total degrees on both sides match let $\vG$ be a  random bipartite graph where every $x_i$ has degree $\vd_i$ and every $a_i$ has degree $\vk_i$.
We tacitly restrict to $n$ such that the event (\ref{eqDeg}) has positive probability.

The generator matrix of the ldgm code is now precisely the $\vm\times n$ biadjacency matrix $A(\vG)$ of $\vG$, viewed as a matrix over $\FF_2$.
Thus, the rows of $A(\vG)$ correspond to the check nodes $a_1,\ldots,a_{\vm}$, the columns correspond to $x_1,\ldots,x_n$ and the $(i,j)$-entry equals one iff $a_i$ and $x_j$ are adjacent.
For a given message $\vx\in\FF_2^n$ the corresponding codeword reads $\vy=A(\vG)\vx\in\FF_2^{\vm}$.
The receiver on the other end of a noisy channel observes a scrambled version $\vy^*$ of $\vy$.
Specifically, $\vy^*$ is obtained from $\vy$ by flipping every bit with probability $\eta\in(0,1/2)$ independently.
To gauge the potential of the code, the key question is how much information about the original $\vx$ the receiver can possibly extract from $\vy^*$.
Naturally, the receiver also knows $\vG$.
Hence, we aim to work out the conditional mutual information
\begin{align*}
I(\vx,\vy^*\mid\vG)=\sum_{x\in\FF_2^n,y\in\FF_2^{\vm}}\pr\brk{\vx=x,\vy^*=y\mid\vG}\log\frac{\pr\brk{\vx=x,\vy^*=y\mid\vG}}{2^n\pr\brk{\vy^*=y\mid\vG}}.
\end{align*}

A precise prediction as to its asymptotical value was put forward on the basis of the physicists' cavity method.
As most such predictions, the formula comes as a variational problem that asks to optimise a functional called the Bethe free entropy over a space of probability measures.
Specifically, let $\fP_*([-1,1])$ be the space of all probability measures $\rho$ on the interval $[-1,1]$ with mean zero.
Let $(\vtheta_{i,\rho})_{i\geq1} \subseteq[-1,1]$ be a family of samples from $\rho$.
Further, let $(\vJ_i)_{i\geq1}$ be Rademacher variables, i.e., $\pr\brk{\vJ_i=1}=\pr\brk{\vJ_i=-1}=1/2$.
In addition, let $(\hat\vk_{i,j})_{i,j\geq1}$ be random variables with distribution
\begin{align}\label{eqhatk}
\pr\brk{\hat\vk_i=\ell}&=\frac{\ell\pr\brk{\vk=\ell}}{\Erw[\vk]}&(\ell\geq0).
\end{align}
All of these are independent. 
Finally, let $\Lambda(z)=z\log(z)$.
Then the Bethe free entropy reads
\begin{align*}
\cB_{\mathrm{ldgm}}(\rho,\eta)&=\Erw\brk{\frac{1}{2}\Lambda\bc{\sum_{\sigma\in\{0,1\}}\prod_{i=1}^{\vd}1+(-1)^\sigma\vJ_i(1-2\eta)\prod_{j=1}^{\hat\vk_i-1}\vtheta_{i,j}}-\frac{\Erw[\vd](\vk-1)}{\Erw[\vk]}\Lambda\bc{1+\vJ_1(1-2\eta)\prod_{j=1}^{\vk}\vtheta_{1,j}}}.
\end{align*}

\begin{theorem}\label{Thm_ldgm}
For any $\vd,\vk$ and for all $\eta\in(0,1)$ we have
\begin{align*}
{\lim_{n\to\infty}\frac{1}{n}I(\vx,\vy^*\mid\vG)}&=\bc{1+\frac{\Erw[\vd]}{\Erw[\vk]}}\ln(2)+\eta\ln(\eta)+(1-\eta)\ln(1-\eta)-\sup_{\pi\in\fP_*([-1,1])}\cB_{\mathrm{ldgm}}(\pi,\eta)&\mbox{in probability}.
\end{align*}
\end{theorem}
\noindent
\Thm~\ref{Thm_ldgm} completely solves a well known conjecture~\cite[Conjecture~1]{MontanariBounds} and significantly extends the results from~\cite{NorJan, CKPZ},
which required the restrictive assumption that the check degree $\vk$ be constant.

A possible objection to a result such as \Thm~\ref{Thm_ldgm} might be that the resulting formula appears exceedingly complicated as it leaves us with a potentially difficult variational problem.
Yet two points are to be made in defense.
First, by vindicating the precise formula predicted by the cavity method, the theorem and its proof show that this technique and the ideas behind it do indeed get to the bottom of the problem.
Second, since the formula involves a supremum, any $\pi\in\PP_*([-1,1])$ yields an upper bound on the mutual information.
Hence, the heuristic population dynamics algorithm deemed to produce good candidate maximisers and beloved of physicists, can be harnessed to get rigorous bounds in one direction.
Finally, in some cases it is possible to precisely identify the maximiser analytically~\cite{Cond,COEGHR}.

\subsection{The stochastic block model}\label{Sec_sbm}
An instructive model of graph clustering, the stochastic block model presumes that a random graph is created in two steps.
First each of the $n$ vertices $\{x_1,\ldots,x_n\}$ receives one of $q\geq2$ possible colours $\vsigma^*_{x_i}\in[q]$ uniformly and independently.
Then a sparse random graph is created where vertices with the same colour are either more likely to be connected by an edge (assortative case), or less likely (disassortative).
Different versions of this model have been proposed.
While in the simplest one edges are inserted independently, here we consider a model from~\cite{Mossel_2015} that produces a $d$-regular graph.
Hence, let $d\geq3$ be an integer and let $\G=\G(n,d)$ be a random $d$-regular graph.
Further, given a parameter $\beta>0$ let $\vG^*=\G^*(n,d,\vsigma^*)$ be a random graph drawn from the distribution
\begin{align}\label{eqsbm}
\pr\brk{\G^*=G\mid\vsigma^*}&\propto\exp\brk{-\beta\sum_{vw\in E(G)}\vecone\cbc{\vsigma^*_v=\vsigma^*_w}},
\end{align}
with the $\propto$-symbol hiding the normalisation required to obtain a probability distribution.
Thus, the parameter $\beta$ tunes the penalty that we impose on monochromatic edges by comparison to the null model $\G$.
At $\beta=0$ there is no such penalty and $\G^*$ and $\G$ are identical.
But even for positive $\beta$ the random graphs $\G,\G^*$ may still be indistinguishable and in effect recovering $\vsigma^*$ may be impossible.
Hence, a fundamental question is for what $q,d,\beta$ it is possible to discriminate between $\G,\G^*$.
Formally, we recall that the {\em Kullback-Leibler divergence} of $\G^*,\G$ is defined as
\begin{align*}
\KL{\G^*}{\G}=\sum_G\pr\brk{\G^*=G}\ln\frac{\pr\brk{\G^*=G}}{\pr\brk{\G=G}}.
\end{align*}
The Kullback-Leibler divergence is an information-theoretic potential that gauges the similarity of two random graph models.
In particular, if $\KL{\G^*}{\G}=\Omega(n)$, then $\G,\G^*$ can be told apart because natural observables will take vastly different values on the two models.

Whether $\KL{\G^*}{\G}=\Omega(n)$ depends on the value of the Bethe free entropy for the stochastic block model.
To be precise, let $\cP([q])$ be the set of all probability distributions $(\mu(1),\ldots,\mu(q))$ on $[q]$.
We identify $\cP([q])$ with the standard simplex in $\RR^q$.
Further, let $\PP_*([q])$ be the set of all probability measures $\pi$ on $\cP([q])$ such that $\int\mu(\sigma)\dd\pi(\mu)=1/q$ for every $\sigma\in[q]$.
In other words, the mean of $\pi$ is the barycenter of the simplex.
Let $(\vmu_{i,\pi})_{i\geq1}$ be a family of independent samples from $\pi$ and let
\begin{align*}
\cB_{\mathtt{sbm}}(\pi,\beta)&=\Erw\brk{\frac{\Lambda\bc{\sum_{\sigma=1}^{q}\prod_{i=1}^{d}1-(1-\eul^{-\beta})\vmu_{i,\pi}(\sigma)}}{q \bc{1-(1-\eul^{-\beta})/q}^{d}}
-\frac{d\Lambda\bc{1-(1-\eul^{-\beta})\sum_{\sigma=1}^q\vmu_{1,\pi}(\sigma)\vmu_{2,\pi}(\sigma)}}{2\bc{1-(1-\eul^{-\beta})/q}}}.
\end{align*} 

\begin{theorem}\label{Thm_sbm}
Let
\begin{align*}
\beta^*&=\inf\cbc{\beta>0:\sup_{\pi\in\PP_*([q])}\cB_{\mathtt{sbm}}(\pi,\beta)>\ln(q)+\frac{d}{2}\ln\bc{1-(1-\eul^{-\beta})/q}}.
\end{align*}
\begin{enumerate}[(i)]
\item If $\beta<\beta^*$, then $\lim_{n\to\infty}\frac1n\KL{\G^*}{\G}=0$.
\item If $\beta>\beta^*$, then $\lim_{n\to\infty}\frac1n\KL{\G^*}{\G}>0$.
\end{enumerate}
\end{theorem}

\Thm~\ref{Thm_sbm} easily implies that for $\beta>\beta^*$ it is information-theoretically possible to recover a non-trivial approximation to $\vsigma^*$ from $\G^*$.
In other words, there exists an exponential time algorithm that likely outputs a colouring $\tau$ of the vertices that has a significantly greater overlap with the ground truth $\vsigma^*$ than a random guess.
An open question is whether for $\beta>\beta^*$ this problem can even be solved by a polynomial time algorithm.
The going conjecture is that in general the answer is `no' and that efficient recoverability kicks in only at a second threshold $\beta^{**}>\beta^*$ for many interesting choices of $q,d$~\cite{Decelle}.

\subsection{The mixed $k$-spin model}\label{Sec_mixed_k_spin}
Not only do the main results of this paper facilitate rigorous proofs of physics predictions for problems in computer science,
but also, conversely, do we obtain new theorems on problems of keen interest in statistical physics.
For example, the mixed $\vk$-spin model is an important spin glass model~\cite{PanchenkoBook}; its purpose is to describe the magnetic interactions in metallic alloys.
To define the model let $\vk\geq2$ be an integer-valued random variable such that $\Erw[\vk^{2+\eps}]<\infty$ for some $\eps>0$ and $\Pr\brk{\vk=2}>0$.
Let $(\vk_i)_{i\geq1}$ be a sequence of independent copies of $\vk$.
Moreover, let $d>0$ and let $\vH=\vH_{\vk}(n,\vm)$ be a (non-uniform) random hypergraph on $V_n=\{x_1,\ldots,x_n\}$ with $\vm=\Po(dn/\Erw[\vk])$ independent hyperedges $a_1,\ldots,a_{\vm}$ such that $a_i$ comprises $\vk_i$ vertices, drawn uniformly without replacement.
Thus, in the special case that $\vk$ is constant we obtain the classical binomial random hypergraph.
To turn this random hypergraph into a spin glass model we draw for each of its edges $a_i$ an independent standard Gaussian $\vJ_i$.
Additionally, let $\beta>0$ be a parameter, commonly coined the inverse temperature.
Then the {\em Boltzmann distribution} of the model is the probability distribution on $\{\pm1\}^{V_n}$ defined by
\begin{align*}
\mu_{\vH,\vJ,\beta}(\sigma)&=\frac{\exp\bc{\beta\sum_{i=1}^{\vm}\vJ_i\prod_{x\in a_i}\sigma_x}}{Z(\vH,\vJ,\beta)}
\quad(\sigma\in\{\pm1\}^{V_n}),&\mbox{ where }&&
Z(\vH,\vJ,\beta)&=\sum_{\tau\in\{\pm1\}^{V_n}}\exp\bc{\beta\sum_{i=1}^{\vm}\vJ_i\prod_{x\in a_i}\tau_x}.
\end{align*}
The normalising term $Z(\vH,\vJ,\beta)$ is known as the {\em partition function}.

A key question is whether for given $d,\beta,\vk$ there occur long-range correlations between the magnetic `spins' observed at $x_1,\ldots,x_n$.
Formally, let $\vsigma\in\{\pm1\}^{V_n}$ signify a sample from the Boltzmann distribution.
Then we say that {\em long-range correlations are absent} if
\begin{align}\label{eqLRC}
\lim_{n\to\infty}\frac1{n^2}\sum_{x, y \in V_n} \Erw\abs{\mu_{\vH,\vJ,\beta}(\{\vsigma_{x}=\vsigma_{y}=1\})-\mu_{\vH,\vJ,\beta}(\{\vsigma_{x}=1\})\mu_{\vH,\vJ,\beta}(\{\vsigma_{y}=1\})}=0.
\end{align}
In words, \eqref{eqLRC} expresses that for most pairs $x,y$ of vertices the spins $\vsigma_{x},\vsigma_{y}$ are essentially independent.
If \eqref{eqLRC} is violated, we say that long-range correlations are present.

According to physics predictions for a given $\beta>0$ long-range correlations emerge at a critical value $d_{\beta,\vk}$ that can be determined in terms of the Bethe free entropy~\cite{pnas,MM}.
The methods developed in this paper enable us to corroborate this formula rigorously.
Specifically, let $\PP_*([-1,1])$ be the space of all probability measures on $[-1,1]$ with mean zero.
Given $\pi\in\PP_*([-1,1])$ let $(\vmu_{\pi,i,j})_{i,j\geq1}$ be a family of independent samples from $\pi$.
Additionally, let $(\hat\vk_i)_{i\geq1}$ be a family of independent copies of $\hat\vk$ from \eqref{eqhatk} and let $\vd=\Po(d)$.
Then the Bethe free entropy of the $\vk$-spin model reads
\begin{align*}
\cB_{\vk-\mathrm{spin}}(\pi)&=
\frac 12 \Erw \brk{\Lambda \bc{\sum_{\sigma_1 \in \cbc{\pm 1}} \prod_{i=1}^{\vd} \bc{1+ \sum_{\sigma_2,\ldots,\sigma_{\hat\vk_i}\in \cbc{\pm 1}} 
		\tanh \bc{\beta \vJ_j \prod_{j \in [\hat \vk_i]} \sigma_j} \prod_{j=2}^{\hat \vk_i}\frac{1+\sigma_j\vmu_{\pi,i,j}}{2} }}} \\
& \qquad - \frac{d}{\Erw[\vk]}
	\Erw \brk{(\vk-1)\Lambda \bc{1+ \sum_{\sigma \in \cbc{\pm 1}^{\vk}} \tanh \bc{\beta \vJ_1 \prod_{i=1}^{\vk} \sigma_j} \prod_{i=1}^{\vk}\frac{1+\sigma_i\vmu_{\pi,1,i}}2}}.
\end{align*}
\begin{theorem} \label{thm_mixed_k}
Let $d_{\beta,\vk}= \inf \cbc{d >0: \sup_{\pi \in \PP_*([-1,1])} \cB_{\vk-\mathrm{spin}}(\pi) > \log 2}$.
\begin{enumerate}[(i)]
\item Long-range correlations are absent for $d<d_{\beta,\vk}$.
\item For any $\eps>0$ there exists $d_{\beta,\vk}<d<d_{\beta,\vk}+\eps$ where long-range correlations are present. 
\end{enumerate}
\end{theorem}

\noindent
Thus, the point $d_{\beta,\vk}$, characterised by the Bethe variational principle, marks the onset of  complex magnetic interactions in the mixed $\vk$-spin model.
This critical value is known as the {\em replica symmetry breaking} phase transition in physics jargon.
As a further application of the main results we can pinpoint the so-called condensation phase transition of the Potts antiferromagnet on random $d$-regular graphs, another problem of interest in mathematical physics.
The details can be found in \Sec~\ref{Sec_applications}.

\section{The mutual information of random factor graphs}\label{Sec_main}
\noindent
The theorems quoted in \Sec~\ref{Sec_intro} are easy consequences of results on general random factor graph models.
These more general theorems, which we present next, constitute the main results of the paper.

\subsection{Random factor graph models}\label{Sec_rfg}
Remarkably many classical problems from combinatorics, statistics and physics can be expressed conveniently in the language of factor graph models \cite{MM, Pearl_2014, LF}.
A factor graph $G$ is a bipartite graph whose vertex classes are variable nodes $V(G)$ and factor nodes $F(G)$.
The former represent the variables of the combinatorial problem in question, such as the individual bits of a codeword.
Generally we assume that these variables range over a domain $\Omega\neq\emptyset$ of size $q=|\Omega| \geq 2$.
Moreover, the factor nodes encode the interactions between the variables, such as the linear relations imposed by the check matrix of a code.
Each factor node $a\in F(G)$ comes with a function $\psi_a:\Omega^{\partial a}\to(0,\infty)$ that assigns a positive weight to value combinations of the adjacent variables $\partial a$.
The factor graph gives rise to a probability distribution
\begin{align}
\mu_G(\sigma)=\frac{\psi_G(\sigma)}{Z_G}\mbox{, where }
\psi_G(\sigma)=\prod_{a\in F(G)}\psi_a(\sigma_{\partial a})\mbox{ and }
Z_G=\sum_{\tau\in\Omega^{V(G)}}\psi_G(\tau) \qquad(\sigma\in\Omega^{V(G)}). \label{eqZ}
\end{align}

To describe problems such as the ones from \Sec~\ref{Sec_intro} we introduce models where the factor graph itself is random.
Specifically, let $\vd,\vk\geq0$ be integer-valued random variables and let $(\vd_i)_{i\geq1}$, $(\vk_i)_{i\geq1}$ be independent copies of $\vd,\vk$.
Further, for each $k$ in the support of $\vk$ let $\Psi_k$ be a finite set of $k$-ary functions $\psi:\Omega^{k}\to(0,\infty)$.
Let $P_k$ be a probability distribution on $\Psi_k$ and let us write $\vpsi_k$ for a sample from $P_k$.
Further, let $\vpsi$ be a random variable distributed as $\vpsi_{\vk}$, let $P$ be the distribution of $\vpsi_{\vk}$ and let $k_\psi$ denote the arity of $\psi$.

Now, to construct a factor graph let $V_n=\{x_1,\ldots,x_n\}$ be a set of variable nodes and let $F_{\vm}=\{a_1,\ldots,a_{\vm}\}$ be a set of $\vm\disteq\Po(n\Erw[\vd]/\Erw[\vk])$ factor nodes.
We obtain the random factor graph $\vG$ as follows.
\begin{description}
\item[G1] given the event $\sum_{i=1}^{n}\vd_i=\sum_{i=1}^{\vm}\vk_i$, choose a bipartite graph on variable and factor nodes such that every $x_i$ has degree $\vd_i$ and every $a_j$ has degree $\vk_j$ uniformly at random.
\item[G2] choose for every factor node $a_i$ a weight function $\psi_{a_i}$ from the distribution $\vpsi_{\vk_i}$. 
\end{description}
In the language of inference problems the random factor graph $\G$ is going to provide a null model because the weight functions in \textbf{G2} are independent of the graph structure from \textbf{G1}.
For instance, in the context of the stochastic block model from \Sec~\ref{Sec_sbm}, this model plays the role of the purely random graph without a particular underlying colouring.

\subsection{The teacher-student scheme}\label{Sec_teacher}
The teacher-student scheme organically turns the null model into an inference problem.
A helpful metaphor might be to imagine a teacher who attempts to convey a ground truth $\vsigma^*$ to a student by presenting examples.
The ground truth itself is a random vector chosen uniformly from the space $\Omega^{V_n}$.
The set of examples corresponds to a factor graph $\G^*$.

To be precise, let $\fD$ be the $\sigma$-algebra generated by the degrees and the total number of factor nodes of the null model $\G$.
Then the factor graph $\G^*$ is chosen from the distribution
\begin{align}\label{eqG*}
\pr\brk{\G^*=G\mid\fD,\vsigma^*}=\frac{\pr\brk{\G=G\mid\fD}\psi_G(\vsigma^*)}{\Erw[\psi_{\vG}(\vsigma^*)\mid\fD,\vsigma^*]}.
\end{align}
Hence, we reweigh the  null model {\bf G1}--{\bf G2} according to the ground truth $\vsigma^*$, rewarding graphs under which $\vsigma^*$ receives a higher weight.
In the case of the stochastic block model, $\G^*$ matches the reweighing~\eqref{eqsbm}  that prefers bichromatic edges. 
The obvious question is how much of an imprint $\vsigma^*$ leaves on the resulting factor graph $\G^*$?
Before we answer this question in general let us illustrate how the examples from \Sec~\ref{Sec_intro} fit into the general framework.

\begin{example}[ldgm codes]\label{Ex_ldgm}
Let $\Omega=\{+1,-1\}$ with $+1=(-1)^0$ representing $0\in\FF_2$ and $-1$ representing $1\in\FF_2$.
For every degree $k\geq1$ there are two $k$-ary weight functions $\psi_{\eta,k,\pm1}$ defined by
\begin{align*}
\psi_{\eta,k,J}(\sigma)=1-(1-2\eta)J\prod_{i=1}^{k}\sigma_i&&(\sigma\in\Omega^k).
\end{align*}
The probability distribution $P_k$ is defined by $P(\psi_{\eta,k,J})=1/2$.
With this setup the bipartite graph structure of the null model $\G$ coincides with the bipartite graph introduced in \Sec~\ref{Sec_LDGM}.
Moreover, the $\pm1$-labels of the weight functions (i.e., value of $J$ such that $\psi_{a_i}=\psi_{\eta,\vk_i,J}$) represent the entries of the vector $\vy^*$.
Thus, while in the null model $\G$ these vector entries are purely random, in the reweighted model $\G^*$ the labels are distributed precisely as the entries of the vector $\vy^*$ from the ldgm model.
\end{example}

\begin{example}[stochastic block model]\label{Ex_sbm}
Let $\Omega=[q]$ be a set of $q$ colours.
We introduce a single binary weight function $\psi_{\beta,q}(\sigma_1,\sigma_2)=\exp(-\beta\vecone\{\sigma_1=\sigma_2\})$ and we let $\vd$ be the constant random variable $d$.
With this weight function the construction \eqref{eqG*} coincides with the definition \eqref{eqsbm} of the stochastic block model.  
\end{example}

The main theorem is going to provide a formula for the mutual information of $\G^*$ and the ground truth $\vsigma^*$, provided that the distribution $P$ on weight functions satisfies a number of easy-to-check conditions.
To state these conditions let us denote by $\cP(\Omega)$ the set of all probability distributions on $\Omega$, endowed with the topology inherited from Euclidean space.
Moreover, let $\PP_*(\Omega)$
signify the space of all probability measures $\pi$ on $\cP(\Omega)$ such that $\int_{\cP(\Omega)}\mu(\omega)\dd\pi(\mu)=1/q$ for all $\omega\in\Omega$.
Finally, for a given $\pi\in\PP_*(\Omega)$ let $(\vmu_{i,j,\pi})_{i,j\geq1}$ be independent samples from $\pi$ and recall $\Lambda(x)=x\log x$.
The assumptions read as follows.
\begin{description}
\item[DEG] there exists $\eps>0$ such that $\Erw[\vd^{2+\eps}],\Erw[\vk^{2+\eps}]<\infty$. 
\item[SYM] there exist reals $\eps,\xi>0$ such that for all $k\in\supp\vk$, $\psi\in\Psi_k$, $j\in[k]$, $\omega\in\Omega$ we have
\begin{align*}
\sum_{\sigma\in\Omega^{k}}\vecone\cbc{\sigma_j=\omega}\psi(\sigma)&=q^{k-1}\xi, \hspace{2 cm} \eps < \psi(\sigma)<1/\eps\qquad{(\sigma\in\Omega^{k})}.
\end{align*}
\item[BAL] for every $k\in\supp\vk$ the function $\mu\in\cP(\Omega)\mapsto\sum_{\sigma\in\Omega^k}\Erw\brk{\vpsi_k(\sigma)}\prod_{i=1}^{k}\mu(\sigma_i)$ is concave and attains its maximum at the uniform distribution on $\Omega$.
\item[POS] for any two probability distributions $\pi,\pi'\in\PP_*(\Omega)$ and any $k\in\supp\vk$ we have
\begin{align*}
\Erw&\brk{\Lambda\bc{\sum_{\tau\in\Omega^{k}}\vpsi_k(\tau)\prod_{i=1}^{k}\vmu_{i,1,\rho}(\tau_i)}}
+(k-1)\Erw\brk{\Lambda\bc{\sum_{\tau\in\Omega^{k}}\vpsi_k(\tau)\prod_{i=1}^{k}\vmu_{i,1,\rho'}(\tau_i)}}\\
&\geq\sum_{j=1}^{k}\Erw\brk{\Lambda\bc{\sum_{\tau\in\Omega^{k}}\vpsi_k(\tau)\vmu_{j,1,\rho}(\tau_j)\prod_{i\neq j}\vmu_{i,1,\rho'}(\tau_i)}}.
\end{align*}
\end{description}

The first assumption {\bf DEG} ensures that the factor graphs are
`sparse' or, formally, locally finite.
Yet {\bf DEG} allows for very general degree distributions, including Poisson and power law distributions.
Moreover, conditions {\bf SYM} and {\bf BAL} are symmetry conditions.
Roughly speaking, they provide that all the values $\omega\in\Omega$ are on the same footing, i.e., there is no semantic preference for any value.
Finally condition {\bf POS} can be viewed as a convexity requirement.
This assumption is needed for the technical reason of facilitating the interpolation method, a proof technique that we borrow from mathematical physics.
The conditions are easily seen to be satisfied in many models of interest including, of course, the stochastic block model and ldgm codes; see \Sec~\ref{Sec_applications}.
Crucially, the assumptions can be checked solely in terms of the weight functions; no random graphs considerations are required.
\footnote{We point out that {\bf POS} fails to hold in the case of the {\em assortative} stochastic block model.}

\subsection{The mutual information}\label{Sec_mutual_inf}
The main result of the paper vindicates the physicists' hunch that the mutual information between the teacher's ground truth $\vsigma^*$ and the data $\vG^*$ presented to the student is determined by the Bethe free entropy.
To state the result we introduce the following generic version of the Bethe functional.
Let $(\vpsi_{k,i})_{k,i}$ be a family of independent random weight functions such that $\vpsi_{k,i}$ is distributed as $\vpsi_k$.
Further, let $(\vh_{k,i})_{k,i}$ with $\vh_{k,i}\in[k]$ be a family of independent uniformly distributed indices.
Recalling that $(\hat\vk_i)_{i\geq1}$ are independent copies of $\hat\vk$ from \eqref{eqhatk}, we define
\begin{align}\label{eqBFE}
\cB(\pi)&=\frac{1}{q}\Erw\brk{\xi^{-\vd}\Lambda\bc{\sum_{\sigma\in\Omega}\prod_{i=1}^{\vd} \sum_{\tau\in\Omega^{\hat\vk_i}}
\vecone\cbc{\tau_{\vh_{\hat{\vk}_i,i}}=\sigma}\vpsi_{\hat\vk_i, i}(\tau) \prod_{j\in[\hat{\vk}_i]\setminus\{\vh_{\hat{\vk}_i,i}\}}\vmu_{i,j,\pi}(\tau_j)}}\\
&\qquad-\frac{\Erw[\vd]}{\xi\Erw[\vk]}\Erw\brk{(\vk-1)\Lambda\bc{\sum_{\tau\in\Omega^{\vk}}\vpsi_{\vk}(\tau)\prod_{j=1}^{\vk}\vmu_{1,j,\pi}(\tau_j)}}.\nonumber
\end{align}
The following theorem expresses the mutual information of $\G^*$ and $\vsigma^*$ given the degrees and the total number of factor nodes as the variational problem of maximising the Bethe functional.

\begin{theorem}\label{Thm_main}
For any random factor graph model that satisfies the conditions {\bf DEG}, {\bf SYM}, {\bf BAL} and {\bf POS},
\begin{align}\label{eqThm_main}
\lim_{n\to\infty}\frac{1}{n}I(\vsigma^*,\G^*\mid\fD)&=
\ln q+\frac{\Erw[\vd]}{\xi\Erw[\vk]} \Erw\brk{q^{-k_{\vpsi}}\sum_{\tau\in\Omega^{k_{\vpsi}}}\Lambda(\vpsi(\tau))}
-\sup_{\pi\in\PP_*(\Omega)}\cB(\pi)&&\mbox{in probability}.
\end{align}
\end{theorem}
\noindent
The formula~\eqref{eqThm_main} is in line with predictions from~\cite{LenkaFlorent}.
Moreover, the results quoted in \Sec~\ref{Sec_intro} are immediate consequences of \Thm~\ref{Thm_main}.

\section{Proof strategy}\label{Sec_strategy}
\noindent
In this section we survey the proof of \Thm~\ref{Thm_main}.
Subsequently we discuss how the strategy compares to prior work, particularly~\cite{CKPZ}. 
Throughout we tacitly assume that {\bf DEG}, {\bf SYM}, {\bf BAL} and {\bf POS} are satisfied. 

\subsection{The partition function}\label{Sec_Z}
The starting point for computing the mutual information is to observe that this quantity is closely connected to the partition function of $\G^*$.
\begin{proposition}\label{Prop_IZ}
\Whp\ we have
\begin{align*}
I(\vsigma^*,\G^*\mid\fD)/n=\ln q+\frac{\Erw[\vd]}{\xi\Erw[\vk]} \Erw\brk{q^{-k_{\vpsi}} \sum_{\tau\in\Omega^{k_{\vpsi}}}\Lambda(\vpsi(\tau))}-\Erw[\log Z(\G^*)]/n+o(1).
\end{align*}
\end{proposition}

\noindent
Hence, \Prop~\ref{Prop_IZ} reduces our task to computing $\Erw[\log Z(\G^*)]$.
This is still a formidable challenge because the logarithm sits inside the expectation; hence, routine techniques such as moment calculations do not bite.
Instead we will combine two separate techniques.
The first is a coupling argument known as the Aizenman-Sims-Starr scheme.
This argument will show that $\Erw[\log Z(\G^*)]$ is upper bounded by $\sup_{\pi}\cB(\pi)$.
The second component, the interpolation method, will supply the matching lower bound.

What these techniques have in common is that they both boil down to `local' calculations.
That is, we need to assess the impact on the partition function $Z(\G^*)$ of a small number of local changes such as addition of a few factor or variable nodes to $\G^*$.
We will perform these computations by way of a probabilistic argument, namely by tracing how they affect the average weight of a sample from the Boltzmann distribution of $\G^*$.
The key is a simple but powerful fact that trades as the Nishimori identity.

\subsection{The Nishimori identity}\label{Sec_Nishi}
To formulate this identity we need to introduce a slightly modified version of the random factor graph model $\G^*$.
Recall from \eqref{eqG*} that $\G^*$ was obtained by first drawing $\vsigma^*$ uniformly at random and then reweighting the null model $\G$ according to the weight of $\vsigma^*$.
If we combine these two steps the net effect should be, at least roughly, that a specific $G$ comes up with probability proportional to $Z(G)$, as every $\sigma\in\Omega^{V_n}$ provides $G$ with a $\psi_G(\sigma)$
chance of being sampled.
Thus, $\G^*$ should be roughly equivalent to the random factor graph model $\hat\G$ defined by
\begin{align}\label{eqhatG}
\pr\brk{\hat\G=G\mid\fD}&\propto Z_G\pr\brk{\G=G\mid\fD}.
\end{align} 
Indeed, this equivalence turns out to be exact if we make one minimal change.
Namely, instead of drawing the ground truth $\vsigma^*$ uniformly at random, we draw a sample from the distribution
\begin{align}\label{eqhatsigma}
\pr\brk{\hat\vsigma=\sigma\mid\fD}\propto \Erw\brk{\psi_{\G}(\sigma)\mid\fD}&&(\sigma\in\Omega^{V_n}).
\end{align}
The following is an extension of~\cite[\Prop~3.10]{CKPZ} to the present, more general class of factor graph models with given degrees.

\begin{proposition}\label{Prop_Nishi}
We have
\begin{align}\label{eqNishi}
\pr\brk{\hat\G=G\mid\fD}\mu_G(\sigma)=\pr\brk{\hat\vsigma=\sigma\mid\fD}\pr\brk{\G^*=G\mid\fD,\vsigma^*=\sigma}.
\end{align}
Furthermore, $\hat\vsigma$ and $\vsigma^*$ as well as $\G^*,\hat\G$ are mutually contiguous and
$\Erw[\ln Z_{\G^*}]=\Erw[\ln Z_{\hat\G}]+o(n).$
\end{proposition}
\noindent
The proof of \Prop~\ref{Prop_Nishi} relies on Bayes' formula combined with a somewhat subtle application of local limit theorems and other probabilistic tools.
The details can be found in \Sec~\ref{Sec_Nishiredux}.

\subsection{Degree pruning}\label{Sec_prune}
A further preparation is degree pruning.
Specifically, while in the random factor graph models $\G^*$ and $\hat\G$ may possess degrees as large as $n^{1/2-\eps}$, the following proposition shows that it suffices to prove the main result \eqref{eqThm_main} for bounded degree sequences.

\begin{proposition}\label{Prop_prune}
Assume that for any integer $L>0$ and for any $\vd,\vk$ such that $\vd,\vk\leq L$ the statement \eqref{eqThm_main} is true.
Then \eqref{eqThm_main} holds for all $\vd,\vk$ that satisfy {\bf DEG} and for which $\Erw \brk{\vd}, \Erw \brk{\vk} > 0$.
\end{proposition}
\noindent
The proof of \Prop~\ref{Prop_prune} is based on concentration inequalities and coupling arguments for bipartite graphs with given degree sequences. Hence, we may assume from here on that $\vd,\vk$ are bounded.

\subsection{Cavities and couplings}\label{Sec_cavs}
Two of the main steps towards the proof of \Thm~\ref{Thm_main}, the Aizenman-Sims-Starr scheme and the interpolation method, hinge on comparing random factor graphs with slightly different parameters.
For example, we will need to compare a random factor graph $\G^*$ with $n$ variable and $\Po(\Erw[\vd]n/\Erw[\vk])$ factor nodes and a factor graph with $n+1$ variable and the commensurate number of $\Po(\Erw[\vd](n+1)/\Erw[\vk])$ factor nodes.
In the classical case of binomial factor graphs as treated in~\cite{CKPZ} where factor nodes are drawn independently this coupling would be relatively straightforward.
Indeed, we could just add a variable node and a few extra factor nodes to the graph with $n$ variables.
However, in the present setting of given degrees matters are much more delicate.
For instance, how would you set up such a coupling for the $d$-regular stochastic block model from \Sec~\ref{Sec_sbm}?
Due to the given degrees the graph structure is too rigid to accommodate the necessary local changes.

To cope with this issue we first create a bit of wiggling room for ourselves by slightly reducing the number of factor nodes.
This idea has been used in prior work on factor graphs with rigid degree distributions such as~\cite{COEGHR}.
However, matters turn out to be rather more delicate here because we do not just work with purely random factor graphs, but with graphs drawn from the teacher-student model.
Thus, we need to take care to meticulously implement the weight shifts in accordance with \eqref{eqG*}.
Hence, for a small but fixed $\eps>0$ let $\vm_\eps=\Po((1-\eps)\Erw[\vd]n/\Erw[\vk])$ be a Poisson variable with a slightly smaller mean than $\vm$.
Because we assume that all degrees are bounded, with probability $1-\exp(-\Omega(n))$ we have $\sum_{i=1}^n\vd_i\geq\sum_{i=1}^{\vm_\eps}\vk_i$.
In fact, \whp\ the total variable degree exceeds the total degree of the first $\vm_\eps$ factor nodes by $\Omega(n)$.
Let $\G(n,\vm_\eps)$ be a random factor graph with variable nodes $x_1,\ldots,x_n$ and factor nodes $a_1,\ldots,a_{\vm_\eps}$ of degrees $\vk_1,\ldots,\vk_{\vm_{\eps}}$ drawn uniformly at random subject to the condition that the degree of each $x_i$ remains bounded by $\vd_i$.
Thus, some of the variable nodes will likely have a degree strictly smaller than their `target degree' $\vd_i$.
We refer to these variable degrees as {\em cavities}.
Further, given $\sigma\in\Omega^{V_n}$ let $\G^*(n,\vm_\eps,\sigma)$ be the random factor graph obtained as in \eqref{eqG*}, i.e., with $\fD_\varepsilon$ denoting the $\sigma$-algebra generated by the degrees and the total number of factors nodes of $\G(n,\vm_\eps)$ we let
\begin{align*}
\pr\brk{\G^*(n,\vm_\eps,\sigma)=G\mid\fD_\varepsilon}&\propto \pr\brk{\G(n,\vm_\eps)=G\mid\fD_\varepsilon}\psi_G(\sigma).
\end{align*}
The following proposition establishes that we can indeed think of $\G^*(n,\vm_{\eps}+1,\sigma)$ as being obtained from $\G^*(n,\vm_\eps,\sigma)$ by adding one extra factor node $a_{\vm_\eps+1}$.
Further, for two factor graphs $G,G'$ on the same set of nodes let $G\triangle G'$ be the symmetric difference of their edge sets.
\begin{proposition}\label{Prop_NoelaRough}
Assume that $|\sigma^{-1}(\omega)|=n/q+O(\sqrt n\log n)$ for all $\omega\in\Omega$.
Then there exists a coupling of $\G^*(n,\vm_{\eps},\sigma)$ and $\G^*(n,\vm_{\eps}+1,\sigma)$ such that
\begin{align*}
\pr\brk{\G^*(n,\vm_{\eps},\sigma)=\G^*(n,\vm_{\eps}+1,\sigma)-a_{\vm_\eps+1}\mid\fD_\varepsilon}&=1-\tilde O(1/n),\\
\pr\brk{\abs{\G^*(n,\vm_{\eps},\sigma)\triangle\G^*(n,\vm_{\eps}+1,\sigma)-a_{\vm_\eps+1}}>n^{2/3} \mid\fD_\varepsilon}&=1-\tilde O(1/n^2).
\end{align*}
\end{proposition}

There is a similar coupling that accommodates the addition of an extra variable node.
\begin{proposition}\label{Prop_NoelaRoughVar}
Assume that $|\sigma^{-1}(\omega)|=n/q+O(\sqrt n\log n)$ for all $\omega\in\Omega$.
Given the degree $\vgamma$ of $x_{n+1}$ in $\G^*(n+1,\vm_{\eps}+\vgamma,\sigma)$
then there exists a coupling of $\G^*(n,\vm_{\eps},\sigma)$ and $\G^*(n+1,\vm_{\eps}+\vgamma,\sigma)$ such that
\begin{align*}
\pr\brk{\G^*(n,\vm_{\eps},\sigma)=\G^*(n+1,\vm_{\eps}+\vgamma,\sigma)-x_{n+1}-\partial x_{n+1}\mid\fD_\varepsilon}&=1-\tilde O(1/n),\\
\pr\brk{\abs{\G^*(n,\vm_{\eps},\sigma)=\G^*(n+1,\vm_{\eps}+\vgamma,\sigma)-x_{n+1}-\partial x_{n+1}}> n^{2/3}\mid\fD_\varepsilon}&=1-\tilde O(1/n^2).
\end{align*}
\end{proposition}
The orders $\tilde O(1/n),\tilde O(1/n^2)$ of the error terms in \Prop s~\ref{Prop_NoelaRough} and~\ref{Prop_NoelaRoughVar} are vital to facilitate the computation of the partition function.
On a technical level, the tools that we develop for proving these propositions, and particularly for dealing with the fragile combinatorics of the factor graph models with given degrees, constitute the main novelty of the paper.
This is where we most visibly add to and improve over the machinery developed in prior work.
The details can be found in \Sec~\ref{Sec_Noela}.

\subsection{Aizenman-Sims-Starr and interpolation}\label{Sec_ASS}

\begin{figure}
\begin{minipage}{0.99 \linewidth}
\begin{tikzpicture}[]
    \node[draw=black, shape=circle, fill=none, minimum width=0.6cm, minimum height = 0.6cm] (A) at (11, -0.25) {};
    \node[draw=black, shape=circle, fill=none, minimum width=0.6cm, minimum height = 0.6cm] (B) at (13, -0.25) {};
    \node[draw=black, shape=circle, fill=none, minimum width=0.6cm, minimum height = 0.6cm] (C) at (15, -0.25) {};
    
    \node[draw=black, fill=black!30, shape=rectangle, minimum width=0.6cm, minimum height = 0.6cm] (D) at (11, -2) {};
    \node[draw=black, fill=black!30, shape=rectangle, minimum width=0.6cm, minimum height = 0.6cm] (E) at (13, -2) {};
    \node[draw=black, fill=black!30, shape=rectangle, minimum width=0.6cm, minimum height = 0.6cm] (F) at (15, -2) {};
    \node[draw = none] at (13, -3.1)  {$t = 1$};
    
    \draw[-] (A)--(D);
    \draw[-] (A)--(E);
    \draw[-] (B)--(D);
    \draw[-] (B)--(F);
    \draw[-] (C)--(D);
    \draw[-] (C)--(E);
    \draw[-] (C)--(F);
    
    \node[draw=black, shape=circle, fill=none, minimum width=0.6cm, minimum height = 0.6cm] (A2) at (2, 0) {};
    \node[draw=black, shape=circle, fill=none, minimum width=0.6cm, minimum height = 0.6cm] (B2) at (4.5, 0) {};
    \node[draw=black, shape=circle, fill=none, minimum width=0.6cm, minimum height = 0.6cm] (C2) at (7, 0) {};
    
    \node[draw=black, fill=blue!30, shape=rectangle, minimum width=0.6cm, minimum height = 0.6cm] (D2) at (1, -1.25) {};
    \node[draw=black, fill=blue!30, shape=rectangle, minimum width=0.6cm, minimum height = 0.6cm] (E2) at (2.166, -1.25) {};
    \node[draw=black, fill=blue!30, shape=rectangle, minimum width=0.6cm, minimum height = 0.6cm] (F2) at (3.333, -1.25) {};
    \node[draw=black, fill=blue!30, shape=rectangle, minimum width=0.6cm, minimum height = 0.6cm] (G2) at (4.5, -1.25) {};
    \node[draw=black, fill=blue!30, shape=rectangle, minimum width=0.6cm, minimum height = 0.6cm] (H2) at (5.666, -1.25) {};
    \node[draw=black, fill=blue!30, shape=rectangle, minimum width=0.6cm, minimum height = 0.6cm] (I2) at (6.833, -1.25) {};
    \node[draw=black, fill=blue!30, shape=rectangle, minimum width=0.6cm, minimum height = 0.6cm] (J2) at (8, -1.25) {};
    \node[draw = none] at (4.5, -3.1)  {$t = 0$};

    \node[draw=black, dashed, shape=circle, fill=orange!60, minimum width=0.6cm, minimum height = 0.6cm] (B3) at (0.6, -2.5) {};
    \node[draw=black, dashed, shape=circle, fill=orange!60, minimum width=0.6cm, minimum height = 0.6cm] (C3) at (1.3, -2.5) {};
    
    \draw[dashed] (D2) -- (B3);
    \draw[dashed] (D2) -- (C3);
    
    \node[draw=black, dashed, shape=circle, fill=orange!60, minimum width=0.6cm, minimum height = 0.6cm] (D3) at (2.166, -2.5) {};
    
    \draw[dashed] (E2) -- (D3);

    \node[draw=black, dashed, shape=circle, fill=orange!60, minimum width=0.6cm, minimum height = 0.6cm] (E3) at (2.9, -2.5) {};
    \node[draw=black, dashed, shape=circle, fill=orange!60, minimum width=0.6cm, minimum height = 0.6cm] (F3) at (3.6, -2.5) {};
    
    \draw[dashed] (F2) -- (E3);
    \draw[dashed] (F2) -- (F3);
    
    \node[draw=black, dashed, shape=circle, fill=orange!60, minimum width=0.6cm, minimum height = 0.6cm] (G3) at (4.5, -2.5) {};
    
    \draw[dashed] (G2) -- (G3);

    \node[draw=black, dashed, shape=circle, fill=orange!60, minimum width=0.6cm, minimum height = 0.6cm] (H3) at (5.3, -2.5) {};
    \node[draw=black, dashed, shape=circle, fill=orange!60, minimum width=0.6cm, minimum height = 0.6cm] (I3) at (6.0, -2.5) {};
    
    \draw[dashed] (H2) -- (H3);
    \draw[dashed] (H2) -- (I3);
    
     \node[draw=black, dashed, shape=circle, fill=orange!60, minimum width=0.6cm, minimum height = 0.6cm] (J3) at (6.833, -2.5) {};
    
    \draw[dashed] (I2) -- (J3);
    
    \node[draw=black, dashed, shape=circle, fill=orange!60, minimum width=0.6cm, minimum height = 0.6cm] (K3) at (8, -2.5) {};
    
    \draw[dashed] (J2) -- (K3);

    
    \draw[-] (A2)--(D2);
    \draw[-] (A2)--(E2);
    
    \draw[-] (B2)--(F2);
    \draw[-] (B2)--(G2);
    
    \draw[-] (C2)--(H2);
    \draw[-] (C2)--(I2);
    \draw[-] (C2)--(J2);
    \end{tikzpicture}
\end{minipage}
    \caption{Illustration of the interpolation method at 'times' $t=0$ and $t=1$. 
    }
    \label{fig:interpolaizenman}
    
\end{figure}
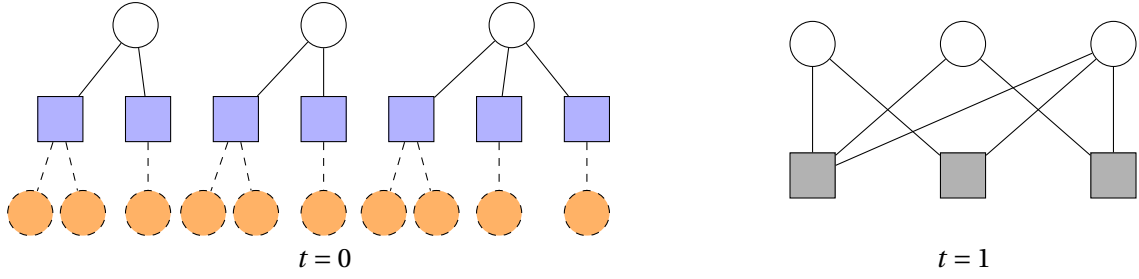

\Prop s~\ref{Prop_NoelaRough} and~\ref{Prop_NoelaRoughVar} in combination with a trick known as the Aizenman-Sims-Starr scheme yield the desired upper bound on the partition function.

\begin{proposition}\label{Prop_ASS}
We have $\Erw[\log Z(\G^*)]\leq n\sup_{\pi\in\PP_*(\Omega)}\cB(\pi)+o(n)$.
\end{proposition}

To prove \Prop~\ref{Prop_ASS} it suffices to establish the corresponding upper bound for $\G^*(n,\vm_\eps,\vsigma^*)$.
This is because similar but simpler arguments as in the proof of \Prop~\ref{Prop_NoelaRough} show that $\Erw[\log Z(\G^*)]=\Erw[\log Z(\G^*(n,\vm_\eps,\vsigma^*)]+O(\eps n)$.
Its proof can be found in \Sec~\ref{prune}.
Now, the Aizenman-Sims-Starr scheme 
for calculating the latter quantity is to write a telescoping sum
\begin{align*}
\Erw[\log Z(\G^*(n,\vm_\eps,\vsigma^*))]&=\sum_{N=1}^n\Erw[\log Z(\G^*(N+1,\vm_\eps(N+1),\vsigma^*_{N+1}))] -\Erw[\log Z(\G^*(N,\vm_\eps(N),\vsigma^*_N))].
\end{align*}
Hence, it suffices to bound the individual summands on the r.h.s., i.e., the differences
\begin{align}\label{eqProp_ASS_expl}
\Erw[\log Z(\G^*(n+1,\vm_\eps(n+1),\vsigma^*_{n+1}))] -\Erw[\log Z(\G^*(n,\vm_\eps(n),\vsigma^*_n))].
\end{align}
To this end we couple these two random factor graphs.
This is where \Prop s~\ref{Prop_NoelaRough} and~\ref{Prop_NoelaRoughVar} enter the fray.
Specifically, we think of both these factor graphs as being obtained from a smaller factor graph $\G_0^*$ that with variables nodes $x_1,\ldots,x_n$ and slightly fewer factor nodes than either of the two target factor graphs.
Then we obtain $\G^*(n,\vm_\eps(n),\vsigma_n^*)$ by adding a few random factors to $\G_0^*$.
Similarly, we obtain $\G^*(n+1,\vm_\eps(n+1),\vsigma^*_{n+1})$ from $\G_0^*$ by adding a few new random factor nodes as well as a new variable node $x_{n+1}$ along with a number of adjacent factor nodes.
Crucially, \Prop s~\ref{Prop_NoelaRough} and~\ref{Prop_NoelaRoughVar} provide the necessary accuracy to trace the impact of these manipulations on the partition function, and the Bethe functional emerges organically as an upper bound on \eqref{eqProp_ASS_expl}.

To obtain the matching lower bound we seize upon the interpolation method.
The basic idea is to set up a family of random factor graph models parametrised by time $t\in[0,1]$ such that the model at time $t=1$ coincides with $\G^*(n,\vm_\eps,\vsigma^*)$ while the model at time $t=0$ is so simple that its partition function can be read off easily.
In fact, the partition function of the $t=0$ model turns out to be $\sup_\pi\cB(\pi)$.
To derive the desired lower bound we prove that the derivative of the log-partition function remains non-negative as we increase $t$.
As in the Aizenman-Sims-Starr scheme, the computation of the derivative can be reduced to tracing the impact of local changes.
Hence, once more we bring \Prop~\ref{Prop_NoelaRough} to bear, this time in combination with the convexity assumption {\bf POS}, to prove the following.

\begin{proposition}\label{Prop_int}
We have $\Erw[\log Z(\G^*)]\geq n\sup_{\pi\in\PP_*(\Omega)}\cB(\pi)+o(n)$.
\end{proposition}

\noindent
Finally, combining \Prop~\ref{Prop_IZ}--\ref{Prop_int}, we obtain \Thm~\ref{Thm_main}.

\subsection{Discussion}\label{Sec_discussion}
There has been a great deal of interest in inference problems on random factor graphs recently.
The substantial literature on the stochastic block model alone, much of it devoted to corroborating the predictions from~\cite{Decelle}, is surveyed in~\cite{Abbe,MooreSurvey}.
The literature on applications to modern coding theory until about 2008 is surveyed in~\cite{RichardsonUrbanke}; important newer contributions include~\cite{KRU,KYMP}.
Further recent applications include compressed sensing~\cite{Donoho_2006, Donoho_2013}, group testing~\cite{Aldridge_2019, Coja_2020}, code-division multiple access~\cite{Guo_2008, Raymond_2007} and the patient zero problem~\cite{Altarelli_2014}.
Apart and beyond this rigorous literature, there is a vast body of work based on either physics techniques such as the cavity method or computer experiments.

The great variety of concrete problems studied individually underscores the potential of generic proof techniques or, even better, general theorems that rigorise these predictions wholesale.
A first contribution has been made by Coja-Oghlan, Krzalaka, Perkins and Zdeborov\'a~\cite{CKPZ}, who studied the teacher-student model on binomial random factor graph models.
While the general proof strategy that we pursue here is guided by that paper, the present factor graph models are more general by allowing prescribed degree sequences for both the variable and factor nodes.
From an application viewpoint this generality is highly desirable because, for example, the quality of an error correcting code or a group testing scheme can be boosted by optimising the degree distribution~\cite{RichardsonUrbanke}.
However, from a technical viewpoint this generality comes at the cost of losing (conditional) independence among the factor nodes.
This issue is well known in random graph theory, where random graphs with given degrees require far more intricate proofs than, e.g., the \Erdos--\Renyi\ model~\cite{Janson_2011}.
Here, these difficulties are exacerbated by the fact that we study not just the plain random graph, which serves as a our null model, but the reweighted random graph distribution induced by the teacher-student scheme.
In effect, many of the steps that were straightforwards in~\cite{CKPZ} become rather delicate due to stochastic dependencies.
The key tool that allows us to cope with these dependencies is \Prop~\ref{Prop_NoelaRough}.
Thus, while we follow the strategy from~\cite{CKPZ} of combining the Aizenman-Sims-Starr scheme with the interpolation method and although we adopt some of the technical ingredients from that work such as the `pinning lemma', the greater generality of the model leads us to crystallise and improve over the previous approach.

What are alternatives to the present strategy of combining the Aizenman-Sims-Starr scheme with the interpolation method?
A classical approach to inference problems on random graphs is the second moment method~\cite{Banks}.
Unfortunately, this approach does not generally allow for tight information-theoretic results.
The reason is that the precise formula for the mutual information or the information-theoretic threshold in, e.g., the stochastic block model comes in terms of the optimiser of the Bethe free entropy functional.
The distribution $\pi$ where the maximum is obtained mirrors the outcome of a complicated message passing process.
Intuitively, $\pi$ is an idealised version of the empirical distribution of Belief Propagation messages that whiz around the factor graph upon convergence when launched from either a uniform initialisation or from the completely polarised initialisation corresponding to the ground truth.
In some examples this fixed point can be characterised precisely and, unsurprisingly, turns out to be anything but trivial~\cite{Cond}.
But we cannot expect the expressiveness required for such a complicated object from a plain second moment computation.
A second conceptually elementary approach is to actually compute the message passing fixed point by hand, e.g., via the contraction method.
But due to the intricacy of the calculations this method has been pushed through in only a few special cases~\cite{Mossel_2015}.

Further powerful techniques include spatial coupling~\cite{GMU} and the adaptive interpolation method~\cite{BarbierChanMacris}.
Both potentially allow for precise results.
The basic idea behind spatial coupling is to convert the given model into a factor graph model with a superimposed geometric structure.
A plus of spatial coupling is that it sometimes allows for better inference algorithms. 
A disadvantage is that the construction has to be carried out case-by-case.
By comparison, the adaptive interpolation method has the advantage of being technically relatively clean.
However, at least on sparse models its combinatorial nuts and bolts appear to be roughly equivalent to the combination of Aizenman-Sims-Starr and the interpolation argument used here.
Furthermore, the latter approach has the merit of being closer in spirit to the physicists' cavity calculation.
In addition, at this time the adaptive interpolation method has not been extended to models with given general degree sequences. 

Further, there has been quite some work on dense random factor graph models where each variable appears in a constant fraction of factor nodes.
Examples are spiked matrix/tensor models~\cite{Barbier_2019} or models of neural networks such as the Hopfield model~\cite{Amit_1985, Mezard_2017}.
These methods are closer in nature to the classical Sherrington-Kirkpatrick model~\cite{PanchenkoBook}. 
It seems fair to say that more is known about dense models than sparse ones because certain central limit theorem-like simplifications arise.
In some cases, the Bethe variational principle reduces to a finite-dimensional or even scalar optimisation problem~\cite{Dia_2016, Lelarge}.


To conclude we note that the study of inference problems typically comes in two instalments: an information-theoretic view that asks for thresholds beyond which in principle sufficient information is available to form a non-trivial estimate of the ground truth and an algorithmic view interested in polynomial-time algorithms. While the two perspectives might appear disparate at first glance, infor\-mation-theoretic results on inference problems like in this paper in combination with tools such as spatial coupling have in the past led to efficient algorithms capable of attaining the information-theoretic thresholds \cite{Coja_2020, Donoho_2013}. We view this as an exciting avenue for future research.

\subsection{Organisation}\label{Sec_org}
In \Sec~\ref{Sec_groundwork} we introduce an extension of the random factor graph model from \Sec~\ref{Sec_main} that incorporates the bells and whistles required to facilitate the proofs of \Prop s~\ref{Prop_ASS} and~\ref{Prop_int}.
The section also contains the proofs of \Prop s~\ref{Prop_Nishi} and ~\ref{Prop_NoelaRough}.
\Sec~\ref{vom}--\ref{typ_ass} lay the foundation to prove \Prop~\ref{Prop_IZ} in \Sec~\ref{mi}.
Similarly, \Sec~\ref{conc} will be used in \Sec~\ref{prune} to prove \Prop~\ref{Prop_prune}.
Subsequently in \Sec~\ref{Sec_Prop_ASS} we prove \Prop~\ref{Prop_ASS}.
The proof of \Prop~\ref{Prop_int} follows in \Sec~\ref{Sec_Prop_int}.
In \Sec~\ref{Sec_applications} we prove the results stated in \Sec~\ref{Sec_intro} and also point out a few further applications of the theorems from \Sec~\ref{Sec_main}.
Two further extensions of our results can be found in \Sec~\ref{klcondp}.  

\section{Groundwork}\label{Sec_groundwork}

\subsection{A generalised model}\label{Sec_generalised}
To facilitate the various parts of the proof we introduce one unified random factor graph model and supply a few tools for analysing it.
The generic model has variable nodes $V_n=\{x_1,\ldots,x_n\}$ and factor nodes $F_m=\{a_1,\ldots,a_m\}$.
Each variable node comes with a target degree $d_i\geq0$.
The sequence $(d_1,\ldots,d_n)$ is denoted by $\sd$.
Similarly, each factor node $a_i$ comes with a target degree $k_i\geq0$ and we let $\sk=(k_1,\ldots,k_m)$.
The degrees are required to satisfy the condition
\begin{align}\label{eqdegCond}
\sum_{i=1}^n d_i\geq\sum_{i=1}^m k_i.
\end{align}
Every $i\in[m]$ comes with a finite set $\Psi_i$ of weight functions $\Omega^{k_i}\to(0,\infty)$, each of which is equipped with a probability measure $P_i$.
Let $\sP=(P_1,\ldots,P_m)$.

The random factor graph $\G(\sd,\sk,\sP,\theta)$ is now defined as follows.
Let $\vGamma$ be a random maximal matching of the complete bipartite graph with vertex classes
\begin{align*}
\bigcup_{i=1}^n\cbc{x_i}\times[d_i]&&\mbox{and}&&\bigcup_{i=1}^m\cbc{a_i}\times[k_i].
\end{align*}
Then the bipartite graph underlying $\G(\sd,\sk,\sP,\theta)$ is obtained from $\vGamma$ by contracting the vertex sets $\cbc{x_i}\times[d_i]$ and $\cbc{a_j}\times[k_j]$ for all $i\in[n]$ and all $j\in[m]$.
Thus, the construction is similar to the well known pairing model for random graphs with given degree sequences.
Strictly speaking, the result of this process is a bipartite multigraph.
We turn this multigraph into a factor graph by drawing for each $a_i$ a weight function $\psi_{a_i}$ from the distribution $P_{i}$ independently.
Furthermore, we add few unary factor nodes $p_1,\ldots,p_\theta$.
For each $p_i$  we let $\partial p_i=\{x_{i}\}$.
Moreover, with $\vomega_i\in\Omega$ drawn independently and uniformly, the weight function of $p_i$ reads
\begin{align*}
\PSI_{p_i}(\sigma)&=\vecone\cbc{\sigma=\vomega_i}.
\end{align*}
The random factor graph induces a Boltzmann distribution  and partition function defined via \eqref{eqZ}. 
Furthermore, $\G(\sd,\sk,\sP,\theta)$ induces the reweighted factor graph distribution $\hat\G(\sd,\sk,\sP,\theta)$ 
defined by 
\begin{align} \label{eq_hatG}
\pr\brk{\hat\G(\sd,\sk,\sP,\theta)\in\cA}&=\frac{\Erw[Z(\G(\sd,\sk,\sP)\vecone\{\G(\sd,\sk,\sP,\theta)\in\cA\}]}{\Erw[Z(\G(\sd,\sk,\sP,\theta)]}&&\mbox{for any event }\cA.
\end{align}
Further, given $\sigma\in\Omega^{V_n}$ we define $\G^*(\sd,\sk,\sP,\theta,\sigma)$ by
\begin{align} \label{eq_starG}
\pr\brk{\G^*(\sd,\sk,\sP,\theta,\sigma)\in\cA}&=\frac{\Erw[\psi_{\G(\sd,\sk,\sP,\theta)}(\sigma)\vecone\{\G(\sd,\sk,\sP,\theta)\in\cA\}]}{\Erw[\psi_{\G(\sd,\sk,\sP,\theta)}(\sigma)]}&&\mbox{for any event }\cA.
\end{align}
Finally, we obtain an induced distribution $\hat\vsigma(\sd,\sk,\sP,\theta)$ on assignments via
\begin{align} \label{eq_hatsigma}
\pr\brk{\hat\vsigma(\sd,\sk,\sP,\theta)=\sigma}&=\frac{\Erw[\psi_{\G(\sd,\sk,\sP,\theta)}(\sigma)]}{\Erw[Z(\G(\sd,\sk,\sP,\theta))]}.
\end{align}

\subsection{Getting started}
The factor graph model and the corresponding Boltzmann distribution facilitate delicate correlations between the spins of different vertices. To cope with them technically, we are in the lucky position that any finite probability space can be partitioned into finitely many sets (so-called \textit{pure states}) such that a given probability measure behaves like a product measure on these states.
\begin{lemma}[Regularity Lemma, \cite{COPS}]
\label{lem_regularitylemma}
For any finite set $\Omega$ and for all $\eps > 0$ there are $L > 0$ and $N > 0$ such that for all $n \geq N$ and all $\mu \in \cP( \Omega^n )$ we find a partition $S_1, \ldots, S_\ell$ of $\Omega^n$ into finitely many parts ($1 \leq \ell \leq L$) such that
\begin{itemize}
    \item $\sum_{i=1}^\ell \mu(S_i) \geq 1 - \eps$,
    \item for all $i$ we find $\mu(S_i) > 0$ and $\Erw \brk{ \dTV(\mu_{j,k}[\cdot \mid S_i] - \mu_{j}[\cdot \mid S_i] \otimes \mu_{k}[\cdot \mid S_i])} \leq \eps$.
\end{itemize}
\end{lemma}
The regularity lemma itself deals with pairwise interactions between vertex spins. It turns out, that this pairwise approximate independence generalizes to an approximate independence between any bounded number of vertex spins.
\begin{lemma}[Symmetry, \cite{Will2}]
\label{lem_pairwise_symmetry}
For any finite set $\Omega$ and any measure $\mu \in \cP( \Omega^n )$ we find that for any $k \geq 2$
\begin{align*}
    \Erw \brk{ \dTV \bc{ \mu_{i,j}, \mu_i \otimes \mu_j }} = o(1) \Longrightarrow \Erw \brk{ \dTV \bc{ \mu_{i_1, \ldots i_k}, \bigotimes_{i=1}^k \mu_i } } = o(1).
\end{align*}
\end{lemma}

We are left to find a partition of $\Omega^n$ into pure-states. It turns out that the \textit{pinning operation} (that is, assigning specific values to a small number of variables) yields a regular partition.

\begin{lemma}[Pinning Lemma, Lemma 3.5 of \cite{CKPZ}] \label{Lemma_tpinning}
Let $\Omega$ be a finite set. For all $\eps > 0$ there is a number $T = T(\eps, \Omega)$ such that for any $n > T$ and any probability measure $\mu \in \cP(\Omega^n)$ we find the following. We create a random probability measure $\check \mu \in \cP(\Omega^n)$ as follows. 
\begin{itemize}
    \item Draw a sample $\check \SIGMA$ from $\mu$.
    \item Independently, choose $\vec \Theta in (0, T)$ uniformly at random.
    \item Create a random subset $U$ of $[n]$ by including each $i \in [n]$ independently with probability $\vec \Theta / n$.
    \item Finally, define $$ \check \mu (\sigma)  = \frac{\mu (\sigma) \vecone \cbc{ \forall i \in U: \check \SIGMA_i = \sigma_i }}{ \mu( \cbc{ \tau \in \Omega^n : \forall i \in U: \tau_i = \check \SIGMA_i } ) }.$$
\end{itemize}

Then, with probability at least $1 - \eps$ we find $$ \Erw_{i, j} \brk{ \dTV \bc{ \check \mu_{i,j}, \check \mu_i \otimes \check \mu_j }} < \eps. $$
\end{lemma}

The following lemma evinces that if the free energy in $\G^*$ is larger than the first moment bound, the free energy in $\G$ is strictly smaller than this bound.
\begin{lemma} \label{lem_free_energy_null_star}
We have
\begin{align*}
    &\Erw \brk{\log Z(\G^*(\sd,\sk,\sP,\theta,\vsigma^*))} = \log \Erw \brk{Z(\G(\sd,\sk,\sP,\theta))} + o(n) \Leftrightarrow \\
    & \qquad \qquad \Erw \brk{\log Z(\G(\sd,\sk,\sP,\theta))} = \log \Erw \brk{Z(\G(\sd,\sk,\sP,\theta))} + o(n)
\end{align*}
\end{lemma}

\Lem~\ref{lem_free_energy_null_star} is an immediate consequence of \Lem~\ref{condp_rs_phase}.

Throughout this paper, we will use the standard Landau notation and introduce $\tilde O(\cdot)$ to hide logarithmic factors. Moreover, if $(\cE_n)_n$ denotes a sequence of events we say that $(\cE_n)_n$ holds \textit{with high probability} (\whp) if $\lim_{n \to \infty} \Pr \brk{\cE_n} = 1$.
The proofs in the subsequent sections require the weight functions to be bounded and not too small, which we ensure by imposing the condition that they take values in $(\varepsilon,2)$ which can be safely assumed by {\bf SYM}.

\subsection{Adding factor nodes}\label{Sec_Noela}
Let $\sd=(d_1,\ldots,d_n)$, $\sk=(k_1,\ldots,k_m)$ and $(\Psi_1,P_1),\ldots,(\Psi_m,P_m)$ be as before.
The aim in this section is to compare the random factor graph model with these parameters with a model with one extra factor node.
Hence, let $\sk^+=(k_1,\ldots,k_m,k_{m+1})$ be a degree sequence obtained from $\sk$ by adding one more entry.
Additionally, let $(\Psi_{m+1},P_{m+1})$ be a set of possible weight functions for the new factor node together with a probability distribution on that set.
The aim of the following proposition is to show that $\G^*(\sd,\sk^+,\sP, \theta, \sigma)$ can essentially be obtained by first creating $\G^*(\sd,\sk,\sP, \theta, \sigma)$ and then adding one extra factor node.
While such a description is trivially valid in the realm of binomial factor graph models, in the present setting of given degree sequences matters turn out to be quite delicate.
In particular, we need to assume the following. 
\begin{description}
\item[SYM$'$] There exist reals $\eps,\xi>0$ such that for every $\psi\in\bigcup_{1\leq i\leq m+1}\Psi_i$, $j\in[k_\psi]$, $\omega\in\Omega$ we have
\begin{align*}
q^{1-k_\psi}\sum_{\sigma\in\Omega^{k_\psi}}\vecone\cbc{\sigma_j=\omega}\psi(\sigma)&=\xi,&\min_{\sigma\in\Omega^{k_\psi}}\psi(\sigma)>\eps.
\end{align*} 
\end{description}
In particular, {\bf SYM} holds for $P_i$, $i\in[m+1]$.
The following proposition constitutes one of the key tools that will be required in the following sections.
\begin{proposition} \label{prop_coupling_int}
For any fixed $C>0,\eps>0$ the following is true.
Suppose that all degrees satisfy $d_i\leq C$ for $i\in[n]$, $k_j \leq C$ for $j \in [m]$, that
\begin{align*}
\sum_{i=1}^n d_i - \sum_{i=1}^m k_i \geq \eps n,
\end{align*}
and that {\bf SYM$'$} is satisfied. 
Moreover, assume that $\sigma \in \Omega^{V_n}$ is such that for all $\omega \in \Omega$ we have
\begin{align*}	\abs{\sum_{i=1}^n d_i \bc{\vecone \cbc{\sigma_i = \omega}-1/q}} = O(\sqrt n\log n) \qquad \text{and} \qquad \abs{\sum_{i=1}^n \vecone \cbc{\sigma_i = \omega} - \frac{n}{q}} = O(\sqrt n\log n).
\end{align*}
Then there exists a coupling of $\G^*(\sd,\sk,P,\theta,\sigma)$ and $\G^*(\sd,\sk^+,P,\theta,\sigma)$ such that
\begin{align*}
\Pr \brk{\G^*(\sd,\sk,P,\theta,\sigma) = \G^*(\sd,\sk^+,P,\theta,\sigma)-a_{m+1}} &= 1 - \tilde O(n^{-1}), \\
\Pr \brk{\abs{\G^*(\sd,\sk,P,\theta,\sigma) \triangle \G^*(\sd,\sk^+,P,\theta,\sigma)} < \sqrt{n}\log n} &= 1 - O(n^{-2}).
\end{align*}
\end{proposition}

We also need an estimate of the total variation distance of the two random factor graph models when {\bf SYM} is not assumed for the last factor node.

\begin{description}
\item[SYM$''$] There exist reals $\eps,\xi>0$ such that for every $\psi\in\bigcup_{1\leq i\leq m}\Psi_i$, $j\in[k_\psi]$, $\omega\in\Omega$ we have
\begin{align*}
q^{1-k_\psi}\sum_{\sigma\in\Omega^{k_\psi}}\vecone\cbc{\sigma_j=\omega}\psi(\sigma)&=\xi,&\min_{\sigma\in\Omega^{k_\psi}}\psi(\sigma)>\eps.
\end{align*} 
\end{description}

\begin{proposition} \label{prop_coupling_int_m}
For any fixed $C>0,\eps>0$ the following is true.
Suppose that all degrees satisfy $d_i\leq C$ for $i\in[n]$, $k_j \leq C$ for $j \in [m]$, that
\begin{align*}
\sum_{i=1}^n d_i - \sum_{i=1}^m k_i \geq \eps n,
\end{align*}
and that {\bf SYM$''$} is satisfied. 
Moreover, assume that $\sigma \in \Omega^{V_n}$ is such that for all $\omega \in \Omega$ we have
\begin{align*}	\abs{\sum_{i=1}^n d_i \bc{\vecone \cbc{\sigma_i = \omega}-1/q}} = O(\sqrt n\log n) \qquad \text{and} \qquad \abs{\sum_{i=1}^n \vecone \cbc{\sigma_i = \omega} - \frac{n}{q}} = O(\sqrt n\log n).
\end{align*}
Then there exists a coupling of $\G^*(\sd,\sk,P,\theta,\sigma)$ and $\G^*(\sd,\sk^+,P,\theta,\sigma)$ such that
\begin{align*}
\Pr \brk{\G^*(\sd,\sk,P,\theta,\sigma) = \G^*(\sd,\sk^+,P,\theta,\sigma)-a_{m+1}} &= 1 - \tilde O(n^{-1/2}), \\
\Pr \brk{\abs{\G^*(\sd,\sk,P,\theta,\sigma) \triangle \G^*(\sd,\sk^+,P,\theta,\sigma)} < \sqrt{n}\log n} &= 1 - O(n^{-2}).
\end{align*}
\end{proposition}

\noindent
A key feature of \Prop~\ref{prop_coupling_int_m} is that we do not need to assume {\bf SYM$'$} for the new factor node $a_{m+1}$.
To prove \Prop~\ref{prop_coupling_int} we introduce a more accessible construction of the graph $\vG^*(\sd, \sk, \sP,\theta, \sigma)$.
Let
\begin{align}\label{eqNoela_1}
	\Delta &= \sum_{i=1}^n d_i - \sum_{i=1}^m k_i \geq \eps n,&\Delta^+&=\Delta-k_{m+1}.
\end{align}
Additionally, for each $i \in [\Delta]$ we introduce a unary factor node $b_{i}$ whose weight function is just the constant $1$.
Hence, the overall number of factor nodes becomes $m+\Delta$.
Like in the pairing model of random graphs with given degree sequences we further introduce sets
\begin{align*}
\cX&=\bigcup_{i=1}^n\cbc{x_i}\times[d_i],&\cA&=\bigcup_{i=1}^m\cbc{a_i}\times[k_i],&\cA^+&=\bigcup_{i=1}^{m+1}\cbc{a_i}\times[k_i],&
\cD&=\cbc{b_1,\ldots,b_\Delta},&\cD^+&=\cbc{b_1,\ldots,b_{\Delta^+}}
\end{align*}
of clones of variable and factor nodes.
Moreover, given the assignment $\sigma \in \Omega^n$ let $\chi \in \Omega^{\cX}$ be the induced assignment on the variable clones.

We now consider the following experiment whose outcome is a factor graph $\vG^\sharp(\sd, \sk, \sP,\sigma)$.
\begin{description}
\item[SHARP1] Generate a random assignment $\vy^\sharp\in\Omega^{\cA\cup\cD}$ as follows.
Draw $\vy^\flat$ from the distribution
\begin{align*}
    \pr\brk{\vy^\flat = y}& = \prod_{i=1}^m\frac{ \Erw\brk{\vec \psi_{a_i}(y_{a_i})}}{\sum_{y' \in \Omega^{k_i}}\Erw\brk{\vec \psi_{a_i}(y')}} q^{-\Delta}&&(y\in\Omega^{\cA\cup\cD}),&&\mbox{and then choose}\\
    \pr\brk{\vy^\sharp = y} &= \pr\brk{\vy^\flat = y \Big\vert \rho_{\vy^\flat} = \rho_\chi},
\end{align*}
where $\rho_\tau$ denotes the empirical distribution of spins under configuration $\tau\in\Omega^{\cX}$ and $y_{a_i}$ denotes the restriction of $y$ to $\{a_i\}\times[k_i]$.
\item[SHARP2] 
Given $\vy^\sharp = y$, for $i \in [m]$ independently, choose weight functions according to
\begin{align*}
    \pr\brk{\vpsi_{a_i}^\sharp \in \cE\mid\vy^\sharp=y} = \frac{\Erw\brk{\vec \psi_{a_i}(y_{a_i}) \vecone\cbc{\vec \psi_{a_i} \in \cE}}}{\Erw\brk{\vec \psi_{a_i}(y_{a_i}) }},
\end{align*}
where $y_{a_i}$ denotes the restriction of $y$ to $\{a_i\}\times[k_i]$.
\item[SHARP3] 
Finally, choose a bijection $\vg^\sharp:\cX\to \cA\cup\cD$ uniformly from the set of all bijections $g$ such that $y \circ g = \chi$; thus, for any such $g$ we have
\begin{align*}
    \pr\brk{\vg^\sharp = g\mid\vy^\sharp=y} = \prod_{z=1}^q\frac{1}{\abs{y^{-1}(z)}!}.
\end{align*}
\end{description}
We denote the result of this procedure by $\vG^\sharp(\sd, \sk, \sP,\sigma)$.
From this graph we obtain $\vG^\sharp(\sd, \sk, \sP,\theta, \sigma)$ by adding unary factor nodes $p_1, \ldots, p_\theta$ adjacent to $x_1, \ldots, x_\theta$ with weight functions $\tau \mapsto \vecone\cbc{\tau = \sigma_i}$.
Analogously we define $\chi^+,\vy^{\flat,+},\vy^{\sharp,+},\PSI_{a_i}^{\sharp,+},\vg^{\sharp,+}$ for the degree sequence $(k_1,\ldots,k_{m+1})$.
These give rise to the factor graph $\vG^{\sharp}(\sd,\sk^+,\sP,\theta,\sigma)$.
\begin{lemma} \label{lem_pebble_dist}
The random factor graphs $\G^*(\sd,\sk,\sP,\theta,\sigma)$ and $\G^{\sharp}(\sd,\sk,\sP,\theta,\sigma)$ are identically distributed.
So are   $\G^*(\sd,\sk^+,\sP,\theta,\sigma)$ and $\G^{\sharp}(\sd,\sk^+,\sP,\theta,\sigma)$.
\end{lemma}
\begin{proof}
It suffices to prove the second statement.
Hence, let {$g:\cX\to\cA^+\cup\cD^+$} be a bijection and write {$y = \chi \circ g^{-1}$} for the induced assignment on {$\cA^+\cup\cD^+$}.
\begin{align*}
    \pr&\brk{\vG^\sharp(\sd, \sk^+, \sP,\theta,\sigma) \in \cbc{g}\times \prod_{i=1}^{m+1} \cE_i }=
		\pr\brk{\vy^{\sharp,+} = y} \pr\brk{\vg^{\sharp,+ } = g \vert \vy^{\sharp,+} = y }\prod_{i=1}^{m+1}  \pr\brk{\vec \psi_{a_i}^{\sharp,+} \in \cE_i \vert \vy^{\sharp,+} = y}\\
   		&=\frac{1}{\pr\brk{\rho_{\vy^{\flat,+}}= \rho_\chi}\prod_{\tau=1}^q\abs{y^{-1}(\tau)}!}\bc{\prod_{i=1}^{m+1}\frac{ \Erw\brk{\vec \psi_{a_i}(y_{a_i})}}{\sum_{\tau\in\Omega^{k_i}}
				\Erw\brk{\vec \psi_{a_i}(\tau)}}}\bc{\prod_{i=1}^{m+1} \frac{\Erw\brk{\vec \psi_{a_i}(y_{a_i})
					\vecone\cbc{\vpsi_{a_i} \in \cE_i}}}{\Erw\brk{\vec \psi_{a_i}(y_{a_i}) }}}.
\end{align*}
Moreover, with $f$ ranging over all bijections {$\cX\to\cA^+\cup\cD^+$}, 
\begin{align*}
   \pr\brk{\vG^*(\sd, \sk^+, \sP,\theta,\sigma) \in \cbc{g}\times \prod_{i=1}^{m+1} \cE_i } &=  \frac{\Erw\brk{\vec\psi_{\vG(\sd, \sk^+, \sP,\theta)}(\sigma)\vecone\cbc{\vG(\sd, \sk^+, \sP,\theta) \in  \cbc{g}\times \prod_{i=1}^{m+1} \cE_i }}}{\Erw\brk{\vec\psi_{\vG(\sd, \sk^+, \sP,\theta)}(\sigma)}}\\
   &=  \frac{\prod_{i=1}^{m+1} \Erw\brk{\vec \psi_{a_i}(y_{a_i})\vecone\cbc{\vec \psi_{a_i} \in \cE_i}}}{\bc{\sum_{i=1}^n d_i}!\sum_{f}\pr\brk{\vg = f}\prod_{i=1}^{m+1} \Erw\brk{\vec \psi_{a_i}({(\chi \circ f^{-1})_{a_i})}}}\\
   &= \bc{ \prod_{i=1}^{m+1} \frac{\Erw\brk{\vec \psi_{a_i}(y_{a_i}) \vecone\cbc{\vec \psi_{k_i} \in \cE_i}}}{\Erw\brk{\vec \psi_{a_i}(y_{a_i}) }}} \frac{\prod_{i=1}^{m+1} \Erw\brk{\vec \psi_{a_i}(y_{a_i})}}{\sum_{f}\prod_{i=1}^{m+1}\Erw\brk{\vec \psi_{a_i}({(\chi \circ f^{-1})_{a_i}})}}.
\end{align*}
It thus remains to show that
\begin{align}\label{eqNoelasDistr1}
    \pr\brk{\rho_{\vy^{\flat,+}}= \rho_\chi}\prod_{i=1}^{m+1}\bc{\sum_{\tau\in \Omega^{k_i}}\Erw\brk{\vec \psi_{a_i}(\tau)}}\prod_{\tau\in\Omega}\abs{y^{-1}(\tau)}! = \sum_{f}\prod_{i=1}^{m+1}\Erw\brk{\vec \psi_{a_i}({(\chi \circ f^{-1})_{a_i}})}.
\end{align}
On the right hand side we may alternatively sum over all possible images {$\chi \circ f^{-1}$} that arise from bijections $f$.
Observe that each different {$\chi \circ f^{-1}$} can arise from exactly $\prod_{z=1}^q \abs{\chi^{-1}(z)}!$ many different $f$, as permuting the images of clones within a color class does not change the induced image on the {factor} side.
Moreover, we can only see {$\chi \circ f^{-1}$ with $\rho_{\chi \circ f^{-1}} = \rho_{\chi}$ which means $\abs{\chi^{-1}(z)} = \abs{(\chi \circ f^{-1})^{-1}(z)}$} for all $z \in \Omega$.
Therefore,
\begin{align*}
 {\sum_{f}\prod_{i=1}^{m+1}\Erw\brk{\vec \psi_{a_i}((\chi \circ f^{-1})_{a_i})}} &= \prod_{z=1}^q \abs{y^{-1}(z)}! {\sum_{y': \rho_{y'} = \rho_\chi}\prod_{i=1}^{m+1}\Erw\brk{\vec \psi_{a_i}(y'_{a_i})}}.
\end{align*}
Further, by the definition of $\vy^{\flat,+}$,
\begin{align*}
   \pr\brk{\rho_{\vy^{\flat,+}}= \rho_\chi}\prod_{i=1}^{m+1}\bc{\sum_{\tau \in \Omega^{k_i}}\Erw\brk{\vec \psi_{a_i}(\tau)}} &= \sum_{y': \rho_{y'} = \rho_\chi}\prod_{i=1}^{m+1}\Erw\brk{\vec \psi_{a_i}(y'_{a_i})}, 
\end{align*}
which establishes \eqref{eqNoelasDistr1} and thus the lemma.
\end{proof}

We prove \Prop~\ref{prop_coupling_int} by showing that the assignments observed on the factor nodes can be coupled so that they agree with probability $1-\tilde O(1/n)$.
Let $\vy^{\sharp}$ and $\vy^{\sharp,+}$ denote the assignments drawn as per {\bf SHARP1} for the two graphs. 
Furthermore, let $\cA$ denote the set of clones of $a_1,\ldots,a_m$ and let $\vy^{\sharp}_{\cA}, \vy^{\sharp,+}_{\cA}$ signify the restrictions of $\vy^{\sharp},\vy^{\sharp,+}$ to $\cA$.
Moreover, let us call $y\in\Omega^{\cA}$ {\em extendible} if
\begin{align*}
\sum_{\alpha\in\cA}\vecone\{y_\alpha=\tau\}\leq\rho_\chi(\tau)\sum_{i=1}^nd_i\qquad\mbox{for all }\tau\in\Omega.
\end{align*}
Thus, the extendible $y$ are the conceivable outcomes of $\vy^{\sharp}_{\cA},\vy^{\sharp,+}_{\cA}$.

As a first step we deal with ``atypical'' extendible $y$.
To this end we finally introduce for $i\in[m]$
\begin{align*}
\vY^{\sharp/\flat}_i(\tau)&=\sum_{j=1}^{k_i}\vecone\cbc{\vy^{\sharp/\flat}_{(a_i,j)} = \tau},& A&=\sum_{i=1}^{m}k_i.
\end{align*}
Thus, $\vY^{\sharp/\flat}_i(\tau)$ counts occurrences of $\tau$ among the clones of factor node $a_i$ under $\vy^{\sharp/\flat}$ from {\bf SHARP1}. 
\begin{lemma} \label{lem_config_concentration}
Assume the assumptions of \Prop~\ref{prop_coupling_int} to hold. We have 
\begin{align*}
    \pr\brk{ \abs{\sum_{i=1}^{m}\vY^{\sharp}_i(\tau) - \frac{A}{q}} >\sqrt{A}\ln A\ln\ln A}\leq n^{-3}.    
\end{align*}
\end{lemma}

\begin{proof}
Due to \SYM'\ 
we have $\Erw\brk{\vY^\flat_i(\tau)} = k_i/q$ for all $\tau\in\Omega$.
Therefore, Stirling's formula yields that there exists $c>0$ such that 
\begin{align*}
    \pr\brk{\rho_{\vy^\flat} = \rho_\chi} = \Omega\bc{ n^{-(q-1)/2}n^{-c\ln n}}.
\end{align*}
Hence, 
\begin{align}\nonumber
\pr&\brk{ \abs{\sum_{i=1}^{m}\vY^{\sharp}_i(\tau) - \frac{A}{q}} > \sqrt{A}\ln A\ln\ln A}\leq\frac{\pr\brk{ \abs{\sum_{i=1}^{m}\vY^\flat_i(\tau) - \frac{A}{q}} > \sqrt{A}\ln A\ln\ln A}}{\pr\brk{\rho_{\vy^\flat} = \rho_\chi}}\\
    &\leq O\bc{n^{(q-1)/2}n^{c\ln n}} \pr\brk{ \abs{\sum_{i=1}^{m}\vY^\flat_i(\tau) - \frac{A}{q}} > \sqrt{A}\ln A\ln\ln A}.
\label{eqNoelaChernoff1}\end{align}
Moreover, the fact that the factor degrees are bounded and the Azuma--Hoeffding inequality imply that
\begin{align}\label{eqNoelaChernoff2}
\pr\brk{ \abs{\sum_{i=1}^{m}\vY^\flat_a(z) - \frac{A}{q}} > \sqrt{A}\ln A\ln\ln A} \leq 2 \exp\bc{-\frac{2}{C^2}\log^2m\bc{\log\log m}^2}=O(n^{- \log n (\log \log n)}).
\end{align}
Combining \eqref{eqNoelaChernoff1} and \eqref{eqNoelaChernoff2} completes the proof.
\end{proof}

Let $\cY$ be the set of all extendible $y \in \Omega^{\cA}$ such that for all $\tau \in \Omega$,
\begin{align} \label{eq_config_concentration}
    \abs{\sum_{i=1}^{m}\sum_{j=1}^{k_i}\vecone\cbc{y_{(a_i,j)}= \tau} - \frac{A}{q}} \leq \sqrt{A}\ln\bc{A}\ln\ln\bc{A}.
\end{align}
\begin{lemma}\label{lem_coupling_pebble}
Suppose that {\bf SYM$'$} is satisfied.
There is a coupling of $\vy^{\sharp},\vy^{\sharp,+}$ and of $\vG^\sharp(\sd, \sk, \sP,\theta,\sigma)$, $\vG^\sharp( \sd, \sk^+, \sP,\theta,\sigma)$ such that
\begin{align*}
    \pr\brk{\vG^\sharp( \sd, \sk, \sP,\theta,\sigma) = \vG^\sharp( \sd, \sk^+, \sP,\sigma)\mid \vy^{\sharp}_{\cA},\vy^{\sharp,+}_{\cA}\in\cY}= 1-\tilde{O}\bc{\frac{1}{n}}.
\end{align*}
\end{lemma}

We prove \Lem~\ref{lem_coupling_pebble} in several steps.
The first step is to calculate the following ratio.

\begin{claim}\label{Claim_pebble_1}
Suppose that $\alpha,\beta\in\cP(\Omega)$ satisfy $\dTV(\alpha,q^{-1}\vecone),\dTV(\beta,q^{-1}\vecone)=O(n^{-1/2}\log n\log\log n)$ and $\dTV(\alpha,\beta)=O(1/n)$.
Then
\begin{align*}
\binom{\Delta^+}{\alpha\Delta^+}\binom{\Delta^+}{\beta\Delta^+}^{-1}=1+q\Delta^+\sum_{\tau\in\Omega}(\alpha(\tau)-1/q)(\alpha(\tau)-\beta(\tau))+\tilde O(1/n).
\end{align*}
\end{claim}
\begin{proof}
By Stirling's formula,
\begin{align}\label{eq_pebble_7}
\binom{\Delta^+}{\alpha\Delta^+}=\bc{2\pi \Delta^+}^{-(q-1)/2} \exp\bc{-\Delta^+\sum_{\tau \in \Omega}\alpha(\tau) \log\bc{\alpha(\tau)} - \frac{1}{2}\sum_{\tau \in \Omega}\log\bc{\alpha(\tau)} + O\bc{\frac{1}{n}}}.
\end{align}
Moreover, applying Taylor's formula to the entropy function, we obtain
\begin{align}\label{eq_pebble_8}
&-\sum_{\tau\in\Omega} \alpha(\tau)\ln\alpha(\tau) = \log q -  \frac{q}{2} \sum_{\tau\in\Omega} \bc{\alpha(\tau)-\frac{1}{q}}^2  + \frac{q^2}{6} \sum_{\tau \in \Omega} \bc{\alpha(\tau)-\frac{1}{q}}^3  +  \tilde{O}\bc{\frac{1}{n^2}}.
\end{align}
Of course, estimates similar to \eqref{eq_pebble_7} and \eqref{eq_pebble_8} apply to $\binom{\Delta^+}{\beta\Delta^+}$.
Combining them, we obtain
\begin{align}\nonumber 
\binom{\Delta^+}{\alpha\Delta^+}\binom{\Delta^+}{\beta\Delta^+}^{-1}&=\exp\bigg[-\Delta^+\bc{\frac{q}{2}\sum_{\tau\in\Omega}(\alpha(\tau)-1/q)^2-(\beta(\tau)-1/q)^2+\frac{q}{3}\bc{\alpha(\tau)-1/q}^3-(\beta(\tau)-1/q)^3}\\
		&\qquad-\frac{1}{2}\sum_{\tau\in\Omega}\ln\frac{\alpha(\tau)}{\beta(\tau)}+O(1/n)\bigg].\label{eq_pebble_8a}
\end{align}
Furthermore,
\begin{align}\nonumber
\sum_{\tau\in\Omega}(\alpha(\tau)-1/q)^2-(\beta(\tau)-1/q)^2&=\sum_{\tau\in\Omega}(\alpha(\tau)-\beta(\tau))^2+2(\alpha(\tau)-1/q)(\beta(\tau)-\alpha(\tau))\\
	&=O(1/n^2)+2\sum_{\tau\in\Omega}(\alpha(\tau)-1/q)(\beta(\tau)-\alpha(\tau)),\label{eq_pebble_8b}\\
\sum_{\tau\in\Omega}(\alpha(\tau)-1/q)^3-(\beta(\tau)-1/q)^3&=-\sum_{\tau\in\Omega}(\beta(\tau)-\alpha(\tau))^3+3(\alpha(\tau)-1/q)^2(\beta(\tau)-\alpha(\tau))+3(\alpha(\tau)-1/q)(\beta(\tau)-\alpha(\tau))^2\nonumber\\
&=\tilde O(1/n^2)
\label{eq_pebble_8c}\end{align}
Plugging \eqref{eq_pebble_8b} and \eqref{eq_pebble_8c} into \eqref{eq_pebble_8a}, we obtain
\begin{align*}
\binom{\Delta^+}{\alpha\Delta^+}\binom{\Delta^+}{\beta\Delta^+}^{-1}&=\exp\bc{q\Delta^+\sum_{\tau\in\Omega}(\alpha(\tau)-1/q)(\alpha(\tau)-\beta(\tau))+\tilde O(1/n)}.
\end{align*}
Expanding the exponential series completes the proof.
\end{proof}
\begin{claim}\label{Claim_pebble_2}
For $y\in\cY$ we have $\frac{\pr[\vy_{\cA}^\sharp=y]}{\pr[\vy_{\cA}^{\sharp,+}=y]}=1+\tilde O(n^{-1/2})\sum_{j=1}^{k_{m+1}}\sum_{\tau\in\Omega}\abs{\pr\brk{\vy^{\sharp,+}_{m+1,j}=\tau}-1/q}+\tilde O(n^{-1})$.
\end{claim}
\begin{proof} 
For any $y\in\cY$ we have, by the definition of $\vy^\sharp$ from {\bf SHARP1},
\begin{align}\nonumber
\pr\brk{\vy^{\sharp}_{\cA} = y } &=\frac{\pr\brk{\vy^{\flat}_{\cA}=y,\rho_{\vy^{\flat}}=\rho_\chi}}{\pr\brk{\rho_{\vy^{\flat}=\rho_\chi}}}
= \sum_{y' \in \Omega^{k_{m+1}}}
\frac{\pr\brk{\vy^{\flat}_{\cA}=y,\,\forall i\in[k_{m+1}]:\vy^{\flat}_{b_i}=y'_i,\rho_{\vy^{\flat}} = \rho_{\chi}}}
		{\pr\brk{\rho_{\vy^\flat}=\rho_\chi}}\\
    &= \frac{\pr\brk{\vy^{\flat}_{\cA}=y}}{\pr\brk{\rho_{\vy^{\flat}}=\rho_\chi}}
			\sum_{y' \in \Omega^{k_{m+1}}} \pr\brk{\forall i\in[k_{m+1}]:\vy^{\flat}_{b_i}=y_i'}
				\pr\brk{\rho_{(y,y',\vy^\flat_{\Delta^+})} = \rho_{\chi}} \label{eq_pebble_1}.
\end{align}
Analogously,
\begin{align}\label{eq_pebble_2}
\pr\brk{\vy^{\sharp,+}_{\cA} = y} &=\frac{\pr\brk{\vy^{\flat, +}_{\cA}=y}}{\pr\brk{\rho_{\vy^{\flat,+}}=\rho_\chi}}
	\sum_{y' \in \Omega^{k_{m+1}}} \pr\brk{\forall i\in[k_{m+1}]:\vy^{\flat,+}_{b_{a_{m+1}}}=y'_i }  \pr\brk{\rho_{(y,y',\vy^{\flat,+}_{\Delta^+})} = \rho_{\chi}}.
\end{align}
Set 
\begin{align*}
    \alpha_{y,y'}(\tau)&=\frac1{\Delta^+}\brk{\sum_{i=1}^n d_i \vecone\cbc{\sigma_i = \tau} - \sum_{i=1}^{m}\sum_{j=1}^{k_i} \vecone\cbc{y_{(a_i,j)} = \tau} - \sum_{j=1}^{k_{m+1}}\vecone\cbc{y'_{j}=\tau}},\\
   \alpha_y'(\tau)&=\frac1{\Delta^+}\brk{\sum_{i=1}^n d_i \vecone\cbc{\sigma_i = \tau} - \sum_{i=1}^{m}\sum_{j=1}^{k_i} \vecone\cbc{y_{(a_i,j)} = \tau} - \frac{k_{m+1}}{q}},\\
    \alpha_{y'}''(\tau)&=\frac1{\Delta^+} \sum_{j=1}^{k_{m+1}}\bc{\frac{1}{q} - \vecone\cbc{y'_{j}=\tau}}
\end{align*}
so that $ \alpha_{y,y'}(\tau) =  \alpha_{y}'(\tau) + \alpha_{y'}''(\tau)$.
Then
\begin{align*}
  \pr\brk{\rho_{(y,y',\vy^\flat_{\Delta^+})} = \rho_{\chi}}  = \pr\brk{\Mult\bc{\Delta^+, \frac{1}{q}, \ldots, \frac{1}{q}} = \Delta^+\alpha_{y,y'}} = q^{-\Delta^+}\binom{\Delta^+}{\bc{\alpha_{y,y'}(\tau)\Delta^+}_{\tau \in \Omega}}.
\end{align*}
Moreover, because $y \in \cY$ we have
\begin{align*}
  \alpha_{y}'(\tau) = 1/q + O\bc{\frac{\log n \log \log n}{\sqrt{n}}} \hspace{1 cm} \text{and} \hspace{1 cm}   \alpha_{y'}''(\tau) = O\bc{\frac{1}{n}}.
\end{align*}
Claim~\ref{Claim_pebble_1}, \eqref{eq_pebble_1} and \eqref{eq_pebble_2} therefore yield
\begin{align}\label{eq_peblle_666}
\frac{\pr\brk{\vy^{\sharp}_{\cA} = y }}{    \pr\brk{\vy^{\sharp,+}_{\cA} = y }}
	&=\frac{\pr\brk{\rho_{\vy^{\flat,+}}=\rho_{\chi}}}{\pr\brk{\rho_{\vy^{\flat}}=\rho_{\chi}}}\bc{1+\tilde O(1/n)+\tilde O(n^{-1/2})\sum_{\tau\in\Omega}\sum_{j=1}^{k_{m+1}}\abs{\pr\brk{\vy^{\flat,+}_{m+1,j}=\tau}-1/q}}.
\end{align}

We finally need to compare $\pr\brk{\rho_{\vy^{\flat}}=\rho_\chi}$ and $\pr\brk{\rho_{\vy^{\flat,+}}=\rho_\chi}$. This can be done in a similar way to the previous calculation.
For $y' \in \Omega^{k_{m+1}}$ and $\tau \in\Omega$, write
\begin{align*}
 \alpha_{y'}^-(\tau) = \frac{\sum_{i=1}^n d_i \vecone\cbc{\sigma_i = \tau} - \sum_{h=1}^{k_{m+1}} \vecone\cbc{y'_h=\tau}}{\sum_{i=1}^nd_i-k_{m+1}}.
\end{align*}
Moreover, let $\vy^{\flat,-}$ be the vector $\vy^{\flat}$ with the components corresponding to $b_{1},\ldots,b_{k_{m+1}}$ removed.
Then
\begin{align}\label{different_plus1}
\frac{\pr\brk{\rho_{\vy^{\flat}}=\rho_\chi}}{\pr\brk{\rho_{\vy^{\flat+}}=\rho_\chi}} &=
\frac{\sum_{y'' \in \Omega^{k_{m+1}}}\Erw\brk{\PSI_{m+1}(y'')}}{q^{k_{m+1}}} \cdot
 \frac{\sum_{y' \in \Omega^{k_{m+1}}}\pr\brk{\rho_{\vy^{\flat, -}}=\bc{\alpha_{y'}^-(\tau)}_{\tau \in \Omega}}}
	{\sum_{y' \in \Omega^{k_{m+1}}}\pr\brk{\rho_{\vy^{\flat,-}}=\bc{\alpha_{y'}^-(\tau)}_{\tau \in \Omega}}\Erw\brk{\PSI_{m+1}(y')} }.
\end{align}
We next compare the probabilities to hit certain colour statistics if factor node $m+1$ is excluded. As before,
\begin{align}\label{eq_pebble_3}
   & \pr\brk{\rho_{\vy^{\flat,-}}=\bc{\alpha_{y'}^-(\tau)}_{\tau \in \Omega}} = \sum_{y^+ \in \Omega^{\cA}}\pr\brk{\Mult\bc{\Delta^+, \frac{1}{q}, \ldots, \frac{1}{q}} = \Delta^+ \alpha_{y^+,y'}} \pr\brk{\vy_{\cA}^{\flat} = y^+}.
\end{align}
To estimate \eqref{eq_pebble_3} we notice that for any $y^+ \in\Omega^{\cA}$ with
\begin{align}\label{eq_pebble_3a}
\abs{\sum_{i=1}^{m}\sum_{h=1}^{k_i} \vecone\cbc{y^+_{(a_i,h)}=\tau} - |\cA|/q} > \sqrt{|\cA|}\ln|\cA|\ln\ln |\cA|
\end{align}
and any $y' \in\Omega^{k_{m+1}}$, we find
\begin{align}\label{eq_pebble_4}
   \abs{ \Delta^+\alpha_{y^+,y'}(\tau) - \frac{\Delta^+}{q}} = \Omega\bc{\sqrt{|\cA|} \log |\cA| \log\log |\cA|}
\end{align}
Further, if \eqref{eq_pebble_4} is satisfied, then the Chernoff bound implies that there is a constant $\delta > 0$ such that
\begin{align}
    &\pr\brk{\Mult\bc{\Delta^+, \frac{1}{q}, \ldots, \frac{1}{q}} = \Delta^+\alpha_{y^+,y'}}\pr\brk{\vy^{\flat}_{\cA} = y^+} = O\bc{n^{-\delta \log n (\log\log n)^2}}\pr\brk{\vy^{\flat}_{\cA} = y^+}.\nonumber
\end{align}
Hence,
\begin{align*}
\frac{\pr\brk{\rho_{\vy^{\flat}}=\rho_\chi}}{\pr\brk{\rho_{\vy^{\flat, +}}=\rho_\chi}} &=\frac{q^{-k_{m+1}}
		\sum_{y'\in\Omega^{k_{m+1}}}\sum_{y^+\in\cY}\pr\brk{\vy^{\flat}_{\cA}=y^+}\binom{\Delta^+}{\Delta^+\alpha_{y^+,y'}}}
			{\sum_{y'\in\Omega^{k_{m+1}}}\pr\brk{\vy^{\flat,+}_{m+1}=y'}\sum_{y^+\in\cY}\pr\brk{\vy^{\flat}_{\cA}=y^+}\binom{\Delta^+}{\Delta^+\alpha_{y^+,y'}}}+\tilde O(1/n)
\end{align*}
Thus, Claim~\ref{Claim_pebble_1} yields
\begin{align}\label{eq_pebble_667}
\frac{\pr\brk{\rho_{\vy^{\flat}}=\rho_\chi}}{\pr\brk{\rho_{\vy^{\flat, +}}=\rho_\chi}} &= 
1+\tilde O(n^{-1/2})\sum_{j=1}^{k_{m+1}}\sum_{\tau\in\Omega}\abs{\pr\brk{\vy^{\sharp,+}_{m+1,j}=\tau}-1/q}+\tilde O(n^{-1}).
\end{align}
Combining \eqref{eq_pebble_1}, \eqref{eq_pebble_2}, \eqref{eq_peblle_666} and \eqref{eq_pebble_667}, we obtain the assertion.
\end{proof}

\begin{proof}[Proof of \Lem~\ref{lem_coupling_pebble}]
Claim~\ref{Claim_pebble_2} and assumption {\bf SYM'} yield
\begin{align*}
\dTV(\vy^{\sharp}_{\cA},\vy^{\sharp,+}_{\cA})=\tilde{O}(1/n).
\end{align*}
The coupling lemma for the total variation distance (see i.e. \cite{Levin_2017}) therefore yields a coupling under which $\vy^{\sharp}_{\cA}$ and $\vy^{\sharp,+}_{\cA}$ differ with probability $\tilde{O}(1/n)$.
The construction {\bf SHARP2--3} finally extends this coupling to the desired coupling of
$\vG^\sharp(\sd, \sk, \sP,\theta,\sigma)$ and $\vG^\sharp(\sd, \sk^+, \sP,\theta,\sigma)$.
\end{proof}

\begin{lemma} \label{lem_coupling_pebble2}
There is a coupling of $\vy^{\sharp},\vy^{\sharp,+}$ and of $\vG^\sharp(\sd, \sk, \sP,\theta,\sigma)$, $\vG^\sharp( \sd, \sk^+, \sP,\theta,\sigma)$ such that
\begin{align*}
    \pr\brk{\abs{\vG^\sharp( \sd, \sk, \sP,\theta,\sigma) \triangle \vG^\sharp( \sd, \sk^+, \sP,\sigma)}>\sqrt n\log n\mid \vy^{\sharp}_{\cA},\vy^{\sharp,+}_{\cA}\in\cY}= O\bc{n^{-2}}.
\end{align*}
\end{lemma}

\begin{proof}
Let us denote by $\cbc{\psi \in \Psi, k \in \NN_{\geq 2}, \tau \in \Omega^k}$ the set of all possible triples of weight function, arity and neighbourhood spins for a factor node. Further, let $q_{\psi, k, \tau}$ denote the probability to observe such a triple. Since each factor node's arity is bounded, $\Omega$ is a finite set and there exist only finitely many different weight functions, the number of distinct weight function, arity and neighbourhood triples is also finite. Thus $q_{\psi, k, \sigma}>\eps$ for some arbitrarily small $\eps>0$. Since there are $m = \Theta(n)$ many factor nodes, the Chernoff bound for the binomial distribution ensures that each distinct $(\psi, k, \tau)$ occurs in both $\vG^\sharp(\sd, \sk, \sP,\theta,\sigma)$ and $\vG^\sharp(\sd, \sk^+, \sP,\theta,\sigma)$ for any choice of $\sd, \sk^+, \sP,\theta,\sigma$ at least 
\begin{align*}
    \d n q_{\psi, k, \tau}/ \k - \sqrt{n}\log n
\end{align*}
often with probability $1-O(n^{-2})$. Therefore, we can couple $\vG^\sharp(\sd, \sk, \sP,\theta,\sigma)$ and $\vG^\sharp(\sd, \sk^+, \sP,\theta,\sigma)$ in such a way that they differ in at most $\sqrt{n} \log n$ factor nodes with probability $1-O(n^{-2})$ whence the lemma follows.
\end{proof}

\begin{proof}[Proof of \Prop~\ref{prop_coupling_int}]
The proposition is an immediate consequence of \Lem s~\ref{lem_pebble_dist}-\ref{lem_coupling_pebble2}.
\end{proof}

\begin{proof}[Proof of \Prop~\ref{prop_coupling_int_m}]
By \Lem~\ref{lem_pebble_dist}, \ref{lem_config_concentration}, Claims \ref{Claim_pebble_1} and \ref{Claim_pebble_2} and assuming \SYM'', we have
\begin{align*}
    \frac{\pr[\vy_{\cA}^\sharp=y]}{\pr[\vy_{\cA}^{\sharp,+}=y]}=1+\tilde O(n^{-1/2})\sum_{j=1}^{k_{m+1}}\sum_{\tau\in\Omega}\abs{\pr\brk{\vy^{\sharp,+}_{m+1,j}=\tau}-1/q}+\tilde O(n^{-1}) = 1+\tilde O(n^{-1/2})
\end{align*}
Thus, $\dTV(\vy^{\sharp}_{\cA},\vy^{\sharp,+}_{\cA})=\tilde O(n^{-1/2})$ and the construction {\bf SHARP2--3} extends this coupling to the desired coupling of
$\vG^\sharp(\sd, \sk, \sP,\theta,\sigma)$ and $\vG^\sharp(\sd, \sk^+, \sP,\theta,\sigma)(\sigma)$.
\Lem~\ref{lem_coupling_pebble2} establishes the second statement of the lemma.
\end{proof}

\begin{proof}[Proof of \Prop~\ref{Prop_NoelaRough}]
The proposition is a special case of \Prop~\ref{prop_coupling_int}.
\end{proof}

\subsection{Adding a variable}
We add a variable node with its adjacent factor nodes to $\G^*(\sd, \sk, \sP, \theta, \sigma)$ as follows.
Let $\sd^+$ be the sequence $\sd$ extended by the degree of a new the variable node $x_{n+1}$.
Similarly, let $\sk^+$ be the sequence $\sk$ with the degrees of the factor nodes $a_1'=a_{m+1}, \dots,a_{d_{n+1}}'= a_{m+d_{n+1}}$ appended.
Also let $h_i\in[k_{m+i}]$ for each $i\in[d_{n+1}]$ and let $\psi_{a_i'}$ signify the weight function of $a_i'$.
Furthermore, let $\check \G^*(\sd^+,\sk^+,\sP,\sigma)$ be the random factor graph that results from the following experiment. 
\begin{description}
\item[PLUS1] choose $\sigma_{x_{n+1}}\in\Omega$ uniformly at random.
\item[PLUS2] draw a random factor graph $\G^*(\sd^+,\sk^+,\sP,(\sigma,\sigma_{x_{n+1}}))$ given that the clones $x_{n+1}\times[d_{n+1}]$ are connected to $(a'_{1},h_1),\ldots (a'_{d_{n+1}},h_{d_{n+1}})$ in this order.
\end{description}
As in the previous subsection, we ask how the factor graph $\check \G^*(\sd^+,\sk^+,\sP,\sigma)-x_{n+1}-a_1'-\ldots-a_{d_{n+1}}'$ obtained by removing $x_{n+1},a_1',\ldots,a_{d_{n+1}}'$ compares to $\G^*(\sd,\sk,\sP,\sigma)$.
We need the following assumption.
\begin{description}
\item[SYM$'''$] There exist reals $\eps,\xi>0$ such that for every $\psi\in\bigcup_{1\leq i\leq m+d_{n+1}}\Psi_i$, $j\in[k_\psi]$, $\omega\in\Omega$ we have
\begin{align*}
q^{1-k_\psi}\sum_{\sigma\in\Omega^{k_\psi}}\vecone\cbc{\sigma_j=\omega}\psi(\sigma)&=\xi,&\min_{\sigma\in\Omega^{k_\psi}}\psi(\sigma)>\eps.
\end{align*} 
\end{description}

\begin{proposition}\label{Prop_oplus}
For any fixed $C>0,\eps>0$ the following is true.
Suppose that all degrees satisfy $d_i\leq C$ for $i\in[n]$, $k_j \leq C$ for $j \in [m]$, that
\begin{align*}
\sum_{i=1}^n d_i - \sum_{i=1}^m k_i \geq \eps n,
\end{align*}
and that {\bf SYM$'''$} is satisfied. 
Moreover, assume that $\sigma \in \Omega^{V_n}$ is such that for all $\omega \in \Omega$ we have
\begin{align*}	\abs{\sum_{i=1}^n d_i \bc{\vecone \cbc{\sigma_i = \omega}-1/q}} = O(\sqrt n\log n) \qquad \text{and} \qquad \abs{\sum_{i=1}^n \vecone \cbc{\sigma_i = \omega} - \frac{n}{q}} = O(\sqrt n\log n).
\end{align*}
Then there is a coupling of $\G^*(\sd, \sk, \sP, \theta, \sigma)$ and $\check \G^*(\sd^+, \sk^+, \sP, \theta, \sigma,\sigma_{x_{n+1}})$ such that
\begin{align*}
    &\Pr \brk{\G^*(\sd, \sk, \sP, \theta, \sigma) = \check \G^*(\sd^+, \sk^+, \sP, \theta, \sigma,\sigma_{x_{n+1}}) - x_{n+1} - \sum_{i=1}^{d_{n+1}} a_{i}'} = 1 - \tilde O \bc{n^{-1}}, \\
    &\Pr \brk{\abs{\G^*(\sd, \sk, \sP, \theta, \sigma) \triangle \check \G^*(\sd^+, \sk^+, \sP, \theta, \sigma,\sigma_{x_{n+1}})}\leq\sqrt n\log n} = 1 - \tilde O \bc{n^{-2}}.
\end{align*}
\end{proposition}

The proof of \Prop~\ref{Prop_oplus} is based on the arguments from the previous section.
Specifically, we introduce an auxiliary factor graph model in which the new variable $x_{n+1}$ and the new factor nodes $a_i'$ are replaced by a single factor node $a_0'$ of degree $k_0'=\sum_{i=1}^{d_{n+1}}k_{m+i}-d_{n+1}$.
Moreover, the weight function of $a_0'$ is defined as
\begin{align*}
\psi_{a_0'}(\tau)&=\sum_{\chi\in\Omega}\prod_{i=1}^{d_{n+1}}\sum_{\tau\in\Omega^{k_{m+i}}}\vecone\{\tau_{h_i}=\chi\}\psi_{a_i'}(\tau).
\end{align*} 
Let $\tilde \G^*= G(\sd,(\sk,k_0'),\sP,\sigma)$ be the random factor graph with the additional factor node $a_0'$.

\begin{lemma}\label{Lemma_a0'}
Under the assumptions of \Prop~\ref{Prop_oplus} there exists a coupling of $\tilde \G^*$ and $\G^*(\sd,\sk,\sP,\sigma)$ such that
\begin{align*}
\pr\brk{\tilde \G^*-a_0'=\G^*(\sd,\sk,\sP,\sigma)}&=1-\tilde O(1/n),\\
\pr\brk{\abs{(\tilde \G^*-a_0')\triangle\G^*(\sd,\sk,\sP,\sigma)}>\sqrt n\log(n)/2}&=O(n^{-2}).
\end{align*}
\end{lemma}
\begin{proof}
We reiterate the argument from \Sec~\ref{Sec_Noela} for the $\tilde \G^*$ model.
The assumption {\bf SYM$'''$} ensures that the random factor graph model $\tilde \G^*$ satisfies the assumption {\bf SYM$'$} from \Sec~\ref{Sec_Noela}.
Indeed, for the factor nodes $a_1,\ldots,a_m$ this is an immediate consequence of {\bf SYM$'''$}.
Moreover, with respect to $a_0'$ we fix $i\in[d_{n+1}]$, $j\in[k_{m+i}]\setminus\{h_i\}$ and $\omega\in\Omega$.
Then
\begin{align}\label{eqLemma_a0'_1}
\sum_{\chi\in\Omega}\prod_{t=1}^{d_{n+1}}\sum_{\tau\in\Omega^{k_{m+t}}}\vecone\{\tau_{h_t}=\chi\wedge(t\neq i\vee\tau_j=\omega)\}\psi_{a_i'}(\tau)=q^{\sum_{\ell=1}^{d_{n+1}}k_{m+\ell}-d_{n+1}}\xi^{d_{n+1}}.
\end{align}
In particular, the expression on the r.h.s.\ is independent of $\omega$.
Applying {\bf SYM$'$} for $a_1,\ldots,a_m$ and \eqref{eqLemma_a0'_1} for $a_0'$ and reiterating the proof of \Prop~\ref{prop_coupling_int}, we obtain the assertion.

Indeed, {\bf \SYM$'''$}, \eqref{eqLemma_a0'_1} and Claims \ref{Claim_pebble_1}, \ref{Claim_pebble_2} are everything we need to prove the statement. By Claims \ref{Claim_pebble_1} and \ref{Claim_pebble_2} such a coupling exists if we manage to prove that for each $\tau \in \Omega$ the probability of observing color $\tau$ at any variable $x_i$ for $i \in [n+1]$ connected to $a_0'$ under $\sigma$ is $1/q$.
If $i \in [n]$ this is an immediate consequence of {\bf \SYM$'''$}. If $i = n+1$, this follows from \eqref{eqLemma_a0'_1}.
\end{proof}

\begin{remark} \label{remark_v}
Because the factor graphs are random, we may assume without loss that the distributions $\PSI_k$ are invariant under permutations of the arguments, that is for any permutation $\kappa$ of $[k]$ and for any $\psi \in \Psi_k$, the weight function $\psi^{\kappa}(\sigma)=\psi(\sigma_{\kappa_1}, \ldots, \sigma_{\kappa_k})$ satisfies $\Pr \brk{\PSI_k = \psi} = \Pr \brk{\PSI_k = \psi^{\kappa}}$.
\end{remark}

\begin{lemma}\label{Lemma_a0'_2}
Under the assumptions of \Prop~\ref{Prop_oplus} there exists a coupling of $\tilde \G^*$ and $\check \G^*$ such that
\begin{align*}
\pr\brk{\tilde \G^*=\check \G^*}&=1-\tilde O(1/n),\\
\pr\brk{\abs{\tilde \G^*\triangle \check \G^*}>\sqrt n\log(n)/2}&=O(n^{-2}).
\end{align*}
\end{lemma}

\begin{proof}
In the first step, we claim that the distributions of $\check \G^*$ and $\tilde \G^*$ are identical conditioned on vertex $x_{n+1}$ having an identical spin, i.e. let $\omega \in \Omega$, then
\begin{align}
    \label{eq_lema0_eq1}\Pr \brk{\check \G^* = g | \check \SIGMA^*_{x_{n+1}} = \omega} = \Pr \brk{\tilde \G^* = g | \tilde \SIGMA^*_{x_{n+1}} = \omega}.
\end{align}
Indeed, by the definition of $\check \G^*$ and $\tilde \G^*$ and Bayes theorem we find for any assignment $\sigma \in \Omega^{V_n}$
\begin{align}
    \label{eq_lema0_eq2} \Pr \brk{\check \G^* = g | \check \SIGMA_{x_{n+1}} = \omega} &= \frac{\Pr\brk{\check \G = g} \psi_g(\sigma, \omega)}{\Erw \brk{\psi_{\check \G}(\sigma, \omega)}}, \\
    \Pr \brk{\tilde \G^* = g | \tilde \SIGMA_{x_{n+1}} = \omega} &= \frac{\Pr \brk{\tilde \G = g} \psi_g(\sigma) \Pr \label{eq_lema0_eq3} \bc{\tilde \SIGMA_{x_{n+1}} = \omega |\tilde \G^* = g}}{\Erw \brk{\psi_{\tilde \G}(\sigma)} \Pr \brk{\tilde \SIGMA_{x_{n+1}} = \omega}} = \frac{\Pr \brk{\tilde \G = g} \psi_g(\sigma, \omega)}{\Erw \brk{\psi_{\tilde \G}(\sigma)} \Pr \brk{\tilde \SIGMA_{x_{n+1}} = \omega}}
\end{align}
Moreover,
\begin{align}
    \label{eq_lema0_eq4} \Pr \brk{\tilde \SIGMA_{x_{n+1}} = \omega} = \Erw \brk{\frac{\psi_{\tilde \G^*}(\sigma, \omega)}{\sum_{\chi \in \Omega}\psi_{\tilde \G^*}(\sigma, \chi)}} = \frac{\Erw \brk{\psi_{\tilde \G}(\sigma) \frac{\psi_{\tilde \G}(\sigma, \omega)}{\sum_{\chi \in \Omega}\psi_{\tilde \G}(\sigma, \chi)}}}{\Erw \brk{\psi_{\tilde \G}(\sigma)}} = \frac{\Erw \brk{\psi_{\tilde \G}(\sigma,\omega)}}{\Erw \brk{\psi_{\tilde \G}(\sigma)}} = \frac{\Erw \brk{\psi_{\check \G}(\sigma,\omega)}}{\Erw \brk{\psi_{\tilde \G}(\sigma)}}
\end{align}
Therefore, \eqref{eq_lema0_eq1} follows from \eqref{eq_lema0_eq2} -- \eqref{eq_lema0_eq4} and the fact that by definition $\Pr\brk{\check \G = g}=\Pr\brk{\tilde \G = g}$.
We now need to get a handle on the distribution of $\tilde \SIGMA^*_{x_{n+1}}$ and $\check \SIGMA^*_{x_{n+1}}$. Clearly, we find by construction $\Pr \brk{\check \SIGMA_{x_{n+1}} = \omega}=1/q$. 
We claim that
\begin{align}
    \label{eq_lema0_eq5} \Pr \brk{\tilde \SIGMA_{x_{n+1}} = \omega} = 1/q + \tilde O (n^{-1/2}).
\end{align}
By assumption there is $(\eps_\omega)_{\omega \in \Omega}$ such that $\eps_\omega = \tilde O(n^{-1/2})$ for all $\omega \in \Omega$ with the property that the marginal distribution on a cavity with color $\omega$ is $1/q + \eps_{\omega}.$ 
It turns out that this is enough to prove the claim. By Remark~\ref{remark_v} without loss of generality, suppose that $h_i = 1$ for all $i \in d_{n+1}$. Then,
\begin{align}
    \notag \Pr \brk{\tilde \SIGMA_{x_{n+1}} = \omega} &\propto \prod_{i=1}^{d_{n+1}} \sum_{\sigma \in \Omega^{k_{m+i}}} \vecone \cbc{\sigma_1 = \omega} \psi_{a'_i}(\sigma) \prod_{j=2}^{k_{m+i}} \frac 1q + \eps_{\sigma_j} \\
    &= \notag \prod_{i=1}^{d_{n+1}} \sum_{\sigma \in \Omega^{k_{m+i}}} \vecone \cbc{\sigma_1 = \omega} \psi_{a'_i}(\sigma) \bc{q^{-k_{m+i}+1} + q^{-k_{m+i}+2} \sum_{j=2}^{k_{m+i}} \eps_{\sigma_j} + O \bc{\norm{\eps}^2}} \\
    &= \notag \prod_{i=1}^{d_{n+1}} \xi + \sum_{\sigma \in \Omega^{k_{m+i}}} \vecone \cbc{\sigma_1 = \omega} \psi_{a'_i}(\sigma) q^{-k_{m+i}+2} \sum_{j=2}^{k_{m+i}} \eps_{\sigma_j} + O \bc{\norm{\eps}^2} \\
    &= \label{eq_lema0_eq6} \xi^{d_{n+1}} + \sum_{i=1}^{d_{n+1}} \sum_{\sigma \in \Omega^{k_{m+i}}} \vecone \cbc{\sigma_i=\omega} \psi_{a'_i} \sum_{j=2}^{k_{m+i}} \eps_{\sigma_j} + O \bc{\norm{\eps}^2}
\end{align}
Since $\xi^{d_{n+1}}$ does not depend on $\omega$ and since $d_{n+1},k_{m+i} , \dots, k_{ m + d_{n+1}}, \PSI_{a'_1}, \dots, \PSI_{a'_{d_{n+1}}}$ are all bounded, we find that \eqref{eq_lema0_eq6} implies \eqref{eq_lema0_eq5}.

Thus, we are left to consider two cases.
\begin{description}
\item[Case $\tilde \SIGMA_{x_{n+1}} = \check \SIGMA_{x_{n+1}}$] This occurs with probability $1 - \tilde O(n^{-1/2})$. In this case, we find $$ \Pr \brk{\check \G^* = g \mid \tilde \SIGMA_{x_{n+1}} = \check \SIGMA_{x_{n+1}}} =  \Pr \brk{\tilde \G^* = g \mid \tilde \SIGMA_{x_{n+1}} = \check \SIGMA_{x_{n+1}}} $$ trivially by the above.
\item[Case $\tilde \SIGMA_{x_{n+1}} \neq \check \SIGMA_{x_{n+1}}$] By \Prop~\ref{prop_coupling_int_m} there is a coupling of $\check \G^*$ and $\tilde \G^*$ such that $$ \Pr \brk{ \check \G^* \neq \tilde \G^* \mid \tilde \SIGMA_{x_{n+1}} \neq \check \SIGMA_{x_{n+1}}} = \tilde O(n^{-1/2}).$$ 
\end{description}
Hence,
\begin{align*}
    \pr\brk{\tilde \G^* \neq \check \G^*} = \pr\brk{\tilde \G^* \neq \check \G^*| \tilde \SIGMA = \check \SIGMA} \Pr \brk{\tilde \SIGMA = \check \SIGMA} + \pr\brk{\tilde \G^* \neq \check \G^*| \tilde \SIGMA \neq \check \SIGMA} \Pr \brk{\tilde \SIGMA \neq \check \SIGMA} = \tilde O \bc{1/n},
\end{align*}
implying the first part of the lemma. The second part of the lemma follows from \Lem~\ref{lem_coupling_pebble2}.
\end{proof}

\begin{proof}[Proof of \Prop~\ref{Prop_NoelaRoughVar}]
The proposition is a special case of \Prop~\ref{Prop_oplus}.
\end{proof}

\subsection{Nishimori redux}\label{Sec_Nishiredux}
The general models from \Sec~\ref{Sec_groundwork} satisfy the following 'Nishimori identity'.
\begin{proposition}\label{Lemma_genNishi}
For any event $\cA$ and for any $\sigma\in\Omega^{V_n}$ we have
\begin{align*}
\Erw\brk{\vecone\{\hat\G(\sd,\sk,\sP,\theta)\in\cA\}\mu_{\hat\G(\sd,\sk,\sP,\theta)}(\sigma)}=
\pr\brk{\hat\vsigma(\sd,\sk,\sP,\theta)=\sigma}\pr\brk{\G^*(\sd,\sk,\sP,\theta,\sigma)\in\cA}.
\end{align*}
Furthermore, $\hat\vsigma(\sd,\sk,\sP,\theta)$ and $\vsigma^*$ are mutually contiguous, as are $\hat\G(\sd,\sk,\sP,\theta)$ and $\G^*(\sd,\sk,\sP,\theta,\vsigma^*)$.
\end{proposition}

The proof of the proposition can be found in \Sec~\ref{mc}.

\begin{proof}[Proof of \Prop~\ref{Prop_Nishi}]
This is an immediate consequence of \Prop~\ref{Lemma_genNishi}.
\end{proof}

\begin{lemma} \label{lem_cav1}
Let $\cC_x=\cbc{i \in [n]: \exists h \in [d_i]: (i,h) \in \cC}$ be the set of cavity variables and let $\vx \in \cC_x$ denote a randomly chosen cavity where $\pr\brk{\vx=i} = \abs{\cbc{h \in [d_i]: (i,h) \in \cC}}/|\cC|$. Moreover, abbreviate $\G^\ast_{\eps, n} = \G^\ast(\underline{d}, \underline{k}, \underline{P}, \theta, \sigma)$. Under the assumptions of Proposition \ref{prop_coupling_int_m}, for any $\omega \in \Omega$, there exists a sequence $\eps_\omega = \tilde{O}(n^{-1/2})$ such that
\begin{align*}
\pr\brk{\mu_{\G^\ast_{\eps, n}, \vx}(\omega) = \frac{1}{q} + \eps_\omega} = 1-o(1).
\end{align*}
\end{lemma}

\begin{proof}
By Lemma \ref{lem_config_concentration}, we have $\pr\brk{\vy^\sharp_{\cA} \in \cY} = 1- O(n^{-3})$ and therefore for any $\omega \in \Omega$,
\begin{align*}
\pr\brk{\abs{\sum_{i=1}^\Delta \vecone\cbc{\vy^\sharp_{A+i} = \omega} - \frac{\Delta}{q}} > \sqrt{\Delta} \log \Delta \log\log \Delta } = O\bc{n^{-3}}.
\end{align*}
Therefore, with probability $1-O(n^{-3})$, $\pr\brk{\abs{(\SIGMA^\ast\vert_{\cC_x})^{-1}(\omega) - \Delta/q} >\sqrt{\Delta} \log \Delta \log\log \Delta} = O(n^{-3})$. By contiguity, this in turn implies that
\begin{align*}
\pr\brk{\abs{(\hat{\SIGMA}\vert_{\cC_x})^{-1}(\omega) - \Delta/q} >\sqrt{\Delta} \log \Delta \log\log \Delta} = o(1).
\end{align*}
and therefore by \Prop~\ref{Prop_Nishi} we find
\begin{align*}
\Erw\brk{\mu_{\G^\ast(\SIGMA^\ast), \cC}\bc{\cbc{\sigma \in [q]^{\cC}: \abs{(\sigma^{-1}(\omega) - \Delta/q} >\sqrt{\Delta} \log \Delta \log\log \Delta}} } = o(1).
\end{align*}
\end{proof}

Recall $\vY^\sharp_i(\tau) = \sum_{j=1}^{k_i}\vecone\cbc{\vy_{(a_i,j)}^\sharp = \tau}$ for $i \in [m]$ and $A = \sum_{i=1}^mk_i$ from above. We now correspondingly denote by
\begin{align*}
\vC^\sharp(\tau) = \sum_{j=1}^\Delta \vecone\cbc{\vy_{b_j}^\sharp= \tau}
\end{align*} 
the number of cavities of each colour $\tau$.
The following lemma provides that the spin distribution on the cavities is close to uniform.

\begin{lemma} \label{lem_cav2}
For all $\tau \in [q]$, we have
\begin{align*}
\pr\brk{\abs{\vC^\sharp(\tau) - \frac{\Delta}{q}} = O\bc{\sqrt{n}\log n \log\log n}} = O(n^{-3}).
\end{align*}
\end{lemma}

\begin{proof}
By Lemma \ref{lem_config_concentration}, we have $\pr\brk{\vy^\sharp_{\cA} \in \cY} = 1- O(n^{-3})$. Moreover, $\sum_{i=1}^n d_i\vecone\cbc{\sigma_i = \tau} = \sum_{i=1}^nd_i/q + O(\sqrt{n}\log n)$ and therefore
\begin{align*}
\sum_{i=1}^m \vY_i^\sharp(\tau) + C^\sharp(\tau) = \sum_{i=1}^n d_i\vecone\cbc{\sigma_i = \tau} = \sum_{i=1}^nd_i/q + O(\sqrt{n}\log n).
\end{align*} 
Rearranging, we see that 
\begin{align*}
C^\sharp(\tau) =  \sum_{i=1}^nd_i/q - A/q + O(\sqrt{n}\log n) + O\bc{\sqrt{A}\log A \log\log A} = \Delta/q + O(\sqrt{n}\log n \log \log n)
\end{align*}
with probability at least $1-O(n^{-3})$.
\end{proof}

\section{Variation of measures} \label{vom}

This section is entirely self-contained.
Fix a number $q\in\mathbb Z_{>1}$ of colours and a nonempty index set $\vmIndSp\subseteq\ZZ_{\ge 0}$.
Fix a degree $d_\vmInd\in\ZZ_{\ge 0}$ for each $\vmInd\in\vmIndSp$ such that $\cD\setminus\{0\}\neq\emptyset$ holds for the set $\cD=\{d_\vmInd:\vmInd\in\vmIndSp\}$ of degrees.
Further, for each $\vmInd\in\vmIndSp$ fix a measure $\mu_\vmInd\in\cP([q]^{d_\vmInd})$ satisfying the assumption {\bf SPAN}, i.e.~for all $\vmCol\in[q]$ we have $\vmCol 1_{[d_\vmInd]}\in\vmAssSp_{\vmInd}$ where $\vmAssSp_\vmInd\subseteq[q]^{d_\vmInd}$ denotes the support of $\mu_\vmInd$ and using the shorthand $1_{[d_\vmInd]}=(1)_{h\in[d_\vmInd]}$.
Analogously, the family $(d_\vmInd,\mu_\vmInd)_{\vmInd\in\vmIndSp}$ satisfies {\bf SPAN} iff $\mu_\vmInd\in\cP([q]^{d_\vmInd})$ satisfies {\bf SPAN} for all $\vmInd\in\vmIndSp$.
For $\vmIndP\in\cP(\vmIndSp)$ let $\vmRInd_{\vmIndP}$ have law $\vmIndP$, further $\vd_{\vmIndP}=d_{\vmRInd_{\vmIndP}}$ and
\begin{align*}
\vmIndPSp=\vmIndPSp(\vmIndSp)=\left\{\vmIndP\in\cP(\vmIndSp):\Erw\left[\vd_{\vmIndP}\right]\in\RR_{>0}\right\}
\end{align*}
denote the measures with finite positive degree expectation, i.e.~exactly the measures $\vmIndP$ for which $\hat \vmIndP\in\cP(\vmIndSp)$ given by the Radon-Nikodym derivative $\Erw[\bm d_{\vmIndP}]^{-1}d_\vmInd$ with respect to $\vmIndP$ is well-defined.
For $\vmInd\in\vmIndSp$ and $\vmPrior\in\cP([q])$ let $\mu_{\vmPrior,\vmInd}\in\cP(\vmAssSp_\vmInd)$ be given by
\begin{align*}
\mu_{\vmPrior,\vmInd}(\vmAss)=Z_{\vmPrior,\vmInd}^{-1}\mu_\vmInd(\vmAss)\prod_{h\in[d_\vmInd]}\vmPrior(\vmAss_h)\textrm{, }Z_{\vmPrior,\vmInd}=\sum_{\vmAss}\mu_\vmInd(\vmAss)\prod_{h\in[d_\vmInd]}\vmPrior(\vmAss_h)\textrm{, }
\end{align*}
for $\vmAss\in\vmAssSp_\vmInd$.
Let $\mu_{\vmPrior,\vmInd}\vert_h\in\cP([q])$, $h\in[d_\vmInd]$, denote the marginal on the $h$-th coordinate and further $\mu_{\vmPrior,\vmInd}\vert_*=\sum_hd_\vmInd^{-1}\mu_{\vmPrior,\vmInd}\vert_h$ if $d_\vmInd>0$.
The central quantity of this section is
\begin{align*}
\iota_{\vmIndP}:\cP([q])\rightarrow\cP([q])\textrm{, }
\vmPrior\mapsto\Erw\left[\mu_{\vmPrior,\vmRInd_{\hat\vmIndP}}\vert_*\right]
\end{align*}
for $\vmIndP\in\vmIndPSp$.
Notice that $\iota_{\vmIndP}$ is well-defined since we always have $\pr[\bm d_{\hat\vmIndP}=0]=0$.
\begin{proposition}\label{pp_mr_homeo}
Let $(d_\vmInd,\mu_\vmInd)_{\vmInd\in\vmIndSp}$ be a family satisfying {\bf SPAN}.
Then for any choice of $\vmIndP\in\vmIndPSp$ the map $\iota_{\vmIndP}$ is a homeomorphism.
\end{proposition}
For $\vmIndP\in\vmIndPSp$ let $\cM_{\vmIndP}=\prod_{\vmInd}\cP(\vmAssSp_\vmInd)$ denote the set of all families of measures that are absolutely continuous with respect to $(\mu_\vmInd)_\vmInd$, for all $\vmInd$ in the support of $\vmIndP$.
For given assignment distributions $\nu\in\cM_{\vmIndP}$ let
\begin{align*}
\rho_{\vmIndP}(\nu)=\rho_{\vmIndP,\nu}=\Erw\left[\nu_{{\vmRInd}_{\hat\vmIndP}}\vert_*\right]\in\cP([q])
\end{align*}
denote their (expected) relative colour frequencies.
Further, for $\rho\in\cP([q])$ let $\cM_{\vmIndP,\rho}=\rho_{\vmIndP}^{-1}(\rho)\subseteq\cM_{\vmIndP}$ denote the assignment distributions $\nu$ that are absolutely continuous with respect to $\mu=(\mu_\vmInd)_\vmInd$ with colour frequencies $\rho$.
Let $\vmColSp\subseteq[q]$ denote the support of $\rho$ and $\cP^\circ(\vmColSp)\subseteq\cP([q])$ the laws with support $\vmColSp$.
Further, let $\mu_{\vmColSp,\vmInd}\in\cP^\circ(\vmAssSp_\vmInd\cap\vmColSp^{d_\vmInd})$ denote the law of $\vmRAss_\vmInd|\vmRAss_\vmInd\in\vmColSp^{d_\vmInd}$ with $\vmRAss_\vmInd$ being a sample from $\mu_\vmInd$ for given $\vmInd\in\vmIndSp$.
Finally, the conditional relative entropy on the fibre given by $\rho$ is
\begin{align*}
f_{\vmIndP,\rho}:\cM_{\vmIndP,\rho}\rightarrow\RR_{\ge 0}\cup\{\infty\}\textrm{, }
\nu\mapsto\Erw\left[\KL{\nu_{\vmRInd_{\vmIndP}}}{\mu_{\vmColSp,\vmRInd_{\vmIndP}}}\right]\textrm{.}
\end{align*}
\begin{proposition}\label{pp_mr_min}
Assume that $(d_\vmInd,\mu_\vmInd)_{\vmInd\in\vmIndSp}$ satisfies {\bf SPAN}.
Then for any choice of $\vmIndP\in\vmIndPSp$ and $\rho\in\cP([q])$ with $\vmPrior=\iota^{-1}_{\vmIndP}(\rho)$ the assignment distribution family $(\mu_{\vmPrior,\vmInd})_\vmInd\in\cM_{\vmIndP,\rho}$ is the unique minimiser of $f_{\vmIndP,\rho}$.
\end{proposition}
For the last assertion we equip $\vmIndPSp$ with the metric given by
\begin{align*}
\vmIndPD(\vmIndP,\vmIndP')=\sum_{\vmInd\in\vmIndSp}(d_\vmInd+1)\left|\vmIndP(\vmInd)-\vmIndP'(\vmInd)\right|
\end{align*}
for $\vmIndP$, $\vmIndP'\in\vmIndPSp$, where adding one is required since $d_\vmInd=d_{\vmInd'}=0$ is possible for $\vmInd\neq\vmInd'$.
\begin{proposition}\label{pp_mr_cont}
If {\bf SPAN} holds, then $\iota$, $\iota^{-1}:\vmIndPSp\times\cP([q])\rightarrow\cP([q])$ are continuous.
\end{proposition}
\subsection{Proof strategy}
In the first part of the proof we derive Proposition \ref{pp_mr_homeo} and Proposition \ref{pp_mr_min} without the continuity results, which we postpone together with the proof of Proposition \ref{pp_mr_cont} to Section \ref{pp_cont}.
In the first part we start with the assumption that $d_\vmInd>0$ for all $\vmInd$ in the support of $\vmIndP\in\vmIndPSp$ and then discuss isolated vertices in Section \ref{pp_zero_proof}.
Further, we first restrict to $\iota_{\vmIndP}:\cP^\circ([q])\rightarrow\cP^\circ([q])$ and then extend the results to the boundary in Section \ref{pp_boundary_proof}.
In this restricted setup Section \ref{pp_onepoint} covers the case $|\vmIndSp|=1$, Section \ref{pp_twopoint} is dedicated to the case $|\vmIndSp|=2$, in Section \ref{pp_finite} we discuss finite index sets $\vmIndSp$ and in Section \ref{pp_infinite} we finally extend the results to countable sets $\vmIndSp$.
\subsection{One point masses}\label{pp_onepoint}
Fix $\vmInd\in\vmIndSp$ such that $d=d_\vmInd>0$. With $\rho\in\cP^\circ([q])$ our atoms are given by
\begin{align*}
\iota_\vmInd&:\cP^\circ([q])\rightarrow\cP^\circ([q])\textrm{, }\vmPrior\mapsto\mu_{\vmPrior,\vmInd}\vert_*\textrm{, }\\
\rho_\vmInd&:\cP(\vmAssSp_\vmInd)\rightarrow\cP([q])\textrm{, }\nu\mapsto\nu\vert_*\textrm{, }\\
f_{\vmInd,\rho}&:\cM_{\vmInd,\rho}\rightarrow\mathbb R_{\ge 0}\textrm{, }\nu\mapsto\KL{\nu}{\mu_\vmInd}\textrm{, }\\
\cM_{\vmInd,\rho}&=\left\{\nu\in\cP(\vmAssSp_\vmInd):\nu\vert_*=\rho\right\}\textrm{.}
\end{align*}
Notice that for $\vmPrior\in\cP^\circ([q])$ we indeed have $\mu_{\vmPrior,\vmInd}\vert_*\in\cP^\circ([q])$ since $\vmCol 1_{[d]}\in\vmAssSp_\vmInd$ for all $\vmCol\in[q]$.
We split $\iota_\vmInd$ into $\iota_{1,\vmInd}:\cP^\circ([q])\rightarrow\cM_\vmInd$, $\vmPrior\mapsto\mu_{\vmPrior,\vmInd}$, with $\cM_\vmInd=\im(\iota_{1,\vmInd})$ and $\iota_{2,\vmInd}:\cM_\vmInd\rightarrow\cP^\circ([q])$, $\nu\mapsto\nu\vert_*$, i.e.~$\iota_{2,\vmInd}$ is the restriction of $\rho_\vmInd$ to $\cM_\vmInd$.
Notice that $\cM_\vmInd\subseteq\cP^\circ(\vmAssSp_\vmInd)$, which we equip with $\|\cdot\|_1$ inherited from $\mathbb R^{\vmAssSp_\vmInd}$. Further, for any $\vmPrior\in\cP^\circ([q])$ and $\vmCol\in[q]$ we have
\begin{align*}
\frac{\vmPrior(\vmCol)}{\vmPrior(q)}=\left(\frac{\mu_{\vmPrior,\vmInd}(\vmCol 1_{[d]})\mu_\vmInd(q1_{[d]})}{\mu_{\vmPrior,\vmInd}(q1_{[d]})\mu_\vmInd(\vmCol 1_{[d]})}\right)^{1/d}\textrm{,}
\end{align*}
so uniqueness of $\vmPrior$ for given $\mu_{\vmPrior,\vmInd}$ follows with a normalization argument and hence $\iota_{1,\vmInd}$ is a bijection.
Next, we show that for any $\rho\in\cP^\circ([q])$ there exists a unique minimiser $\mu_{\min,\rho}\in\cM_{\vmInd,\rho}\cap\cP^\circ(\vmAssSp_\vmInd)$ of $f_{\vmInd,\rho}$.
First, notice that $\cM_{\vmInd,\rho}\neq\emptyset$ since $\mu_{\mathrm{b},\rho}\in\cM_{\vmInd,\rho}$ with $\mu_{\mathrm{b},\rho}=\sum_{\vmCol\in[q]}\rho(\vmCol)e_{\vmAssSp_\vmInd,\vmCol 1_{[d]}}$ and using the shorthand
$e_{\vmAssSp_\vmInd,\vmCol 1_{[d]}}=(\vecone\{\vmAss=\vmCol 1_{[d]}\})_{\vmAss\in\vmAssSp_\vmInd}\in\RR^{\vmAssSp_\vmInd}$ for standard basis vectors.
From this we obtain that $\cM_{\vmInd,\rho}=\mu_{\mathrm{b},\rho}+\cV_{\vmInd,\rho}$ with
\begin{align*}
\cV_{\vmInd,\rho}=\left\{v\in\cV_\vmInd:\mu_{\mathrm{b},\rho}+v\ge 0_{\vmAssSp_\vmInd}\right\}\textrm{, }
\cV_\vmInd=\left\{v\in\mathbb R^{\vmAssSp_\vmInd}:\sum_\vmAss v(\vmAss)=0\textrm{, }v\vert_*=0_{[q]}\right\}\textrm{.}
\end{align*}
In the following sections we use the shorthand $0$ for $0_{\vmAssSp_\vmInd}$ by an abuse of notation since the definition of $\cV_\vmInd$ determines the underlying space.
Since $\cV_\vmInd$ is a linear subspace of $\mathbb R^{\vmAssSp_\vmInd}$, the sets $\cV_{\vmInd,\rho}$ and $\cM_{\vmInd,\rho}$ are polytopes, so in particular they are convex and compact. Further, notice that for any $v\in\mathbb R^{\vmAssSp_\vmInd}$ and $h\in[d]$ we have $\sum_\vmAss v(\vmAss)=\sum_\vmCol v\vert_h(\vmCol)$ and hence $\sum_\vmAss v(\vmAss)=\sum_\vmCol v\vert_*(\vmCol)$, which suggests that $\cV_\vmInd=\{v\in\mathbb R^{\vmAssSp_\vmInd}:v\vert_*=0_{[q]}\}$, i.e.~$\cV_\vmInd$ is the kernel of the linear map $v\vert_*=Wv$ given by the matrix $W=(|\vmAss^{-1}(\vmCol)|/d)_{\vmAss\in\vmAssSp_\vmInd,\vmCol\in[q]}$.
The fact that the column vectors $(W_{\vmCol^*1_{[d]},\vmCol})_{\vmCol\in[q]}=e_{[q],\vmCol}$ are exactly the unit vectors for $\vmCol^*\in[q]$ shows that $W$ is surjective and thereby the kernel has dimension $|\vmAssSp_\vmInd|-q$.
Hence, if $\vmAssSp_\vmInd=\{\vmCol 1_{[d]}:\vmCol\in[q]\}$, then $W$ is bijective, further $\cM_{\vmInd,\rho}=\{\mu_{\mathrm{b},\rho}\}$ and $\mu_{\min,\rho}=\mu_{\mathrm{b},\rho}\in\cP^\circ(\vmAssSp_\vmInd)$ is the unique minimizer of $f_{\vmInd,\rho}$.
Otherwise, for $\varepsilon\in(0,1)$ let $\mu_{\varepsilon,\rho}\in\mathbb R^{\vmAssSp_\vmInd}$ be given by
\begin{align*}
\mu_{\varepsilon,\rho}(\vmAss)=\varepsilon\mu_\vmInd(\vmAss)+\sum_{\vmCol\in[q]}\alpha(\vmCol)e_{\vmAssSp_\vmInd,\vmCol 1_{[d]}}(\vmAss)\textrm{, }\alpha(\vmCol)=\rho(\vmCol)-\varepsilon\mu_\vmInd\vert_*(\vmCol)\textrm{,}
\end{align*}
for $\vmAss\in\vmAssSp_\vmInd$. Notice that $\mu_{\varepsilon,\rho}\vert_*=\rho$, $\varepsilon+\sum_\vmCol\alpha(\vmCol)=1$, and that $\alpha\ge 0$ for $\varepsilon$ sufficiently small since $\rho\in\cP^\circ([q])$, which gives $\mu_{\varepsilon,\rho}\in\cM_{\vmInd,\rho}\cap\cP^\circ(\vmAssSp_\vmInd)$, meaning that $\mu_{\varepsilon,\rho}$ is in the relative interior of $\cM_{\vmInd,\rho}$.
For any $\nu\in\cM_{\vmInd,\rho}\setminus\cP^\circ(\vmAssSp_\vmInd)$ the derivative of $f_{\vmInd,\rho}$ at $\nu$ in the direction $(\mu_{\varepsilon,\rho}-\nu)\in\cV_\vmInd$ is $-\infty$ by the properties of the relative entropy, hence any minimizer of $f_{\vmInd,\rho}$ has to be in $\cM_{\vmInd,\rho}\cap\cP^\circ(\vmAssSp_\vmInd)$.
On the other hand, the minimizer $\mu_{\min,\rho}$ exists since $\cM_{\vmInd,\rho}$ is compact and is unique since $f_{\vmInd,\rho}$ is strictly convex (with convex domain). Hence, let $\cM_{\min}=\{\mu_{\min,\rho}:\rho\in\cP^\circ([q])\}\subseteq\cP^\circ(\vmAssSp_\vmInd)$.
Notice that for any $\nu\in\cM_{\min}$ and $\vmAss^*\in\vmAssSp_\vmInd\setminus\{\vmCol 1_{[d]}:\vmCol\in[q]\}$ we know that $\nu$ is a stationary point of $f_{\vmInd,\nu\vert_*}$ (since $\vmAss^*$ exists), hence by evaluating the first derivative of $f_{\vmInd,\nu\vert_*}$ at $\nu$ in the direction $v=de_{\vmAssSp_\vmInd,\vmAss^*}-\sum_\vmCol|\vmAss^{*-1}(\vmCol)|e_{\vmAssSp_\vmInd,\vmCol 1_{[d]}}\in\cV_\vmInd$ we obtain
\begin{align*}
\sum_\vmAss\left(\ln\left(\frac{\nu(\vmAss)}{\mu_\vmInd(\vmAss)}\right)+1\right)v(\vmAss)=0\textrm{.}
\end{align*}
Rearranging yields $\nu(\vmAss^*)=\mu_\vmInd(\vmAss^*)\prod_\vmCol w(\vmCol)^{|\vmAss^{*-1}(\vmCol)|}=\mu_\vmInd(\vmAss^*)\prod_{h\in[d]}w(\vmAss^*_h)$ with
\begin{align*}
w(\vmCol)=\left(\frac{\nu(\vmCol 1_{[d]})}{\mu_\vmInd(\vmCol 1_{[d]})}\right)^{1/d}\in\mathbb R_{>0}\textrm{, }\vmCol\in[q]\textrm{.}
\end{align*}
Further, this equation trivially holds for any choice of $\nu$ and $\vmCol\in[q]$ with $\vmAss^*=\vmCol 1_{[d]}$.
Hence with $\vmPrior\in\cP^\circ([q])$ given by
$\vmPrior\propto w$, i.e.~$\vmPrior=Z^{-1}w$ with $Z=\sum_\vmCol w(\vmCol)$, a normalization argument applied to $\nu\in\cM_{\min}$ shows that $\nu=\mu_{\vmPrior,\vmInd}$, so $\cM_{\min}\subseteq\cM_\vmInd$.
Conversely, for any $\nu=\mu_{\vmPrior,\vmInd}\in\cM_\vmInd$, $\vmPrior\in\cP^\circ([q])$, we have $\nu\in\cM_{\vmInd,\nu\vert_*}$.
If $|\cM_{\vmInd,\nu\vert_*}|=1$, then $\nu$ is the unique minimizer and hence $\nu\in\cM_{\min}$, otherwise $\nu$ is in the relative interior of $\cM_{\vmInd,\nu\vert_*}$ and evaluating the first derivatives of $f_{\vmInd,\nu\vert_*}$ at $\nu$ in any direction $v\in\cV_\vmInd$ yields
\begin{align*}
\sum_\vmAss\left(\ln\left(\frac{\nu(\vmAss)}{\mu_\vmInd(\vmAss)}\right)+1\right)v(\vmAss)
=\sum_\vmAss\left(\sum_\vmCol|\vmAss^{-1}(\vmCol)|\ln(\vmPrior(\vmCol))-\ln(Z_{\vmPrior,\vmInd})\right)v(\vmAss)
=0\textrm{.}
\end{align*}
This shows that $\cM_{\min}=\cM_\vmInd$ and hence $\iota_{2,\vmInd}$ is bijective with inverse $\rho\mapsto\mu_{\min,\rho}$, which completes the proof.
\subsection{Two point masses}\label{pp_twopoint}
Assume that $\vmIndP$ is supported on two indices $\vmInd_1$, $\vmInd_2\in\vmIndSp$.
We let $d_1=d_{\vmInd_1}$, $d_2=d_{\vmInd_2}$ for transparency and use analogous shorthands for all related quantities throughout this section.
Further, we will continue to use the atoms introduced in Section \ref{pp_onepoint}.
For given $\rho\in\cP^\circ([q])$ let $\vmPrior_1=\iota_1^{-1}(\rho)$, $\vmPrior_2=\iota_2^{-1}(\rho)\in\cP^\circ([q])$ and notice that $\mu_\rho=(\mu_{\rho,1},\mu_{\rho,2})\in\cM_{\vmIndP,\rho}$ with $\mu_{\rho,1}=\mu_{\vmPrior_1,\vmInd_1}\in\cP^\circ(\vmAssSp_1)$ and $\mu_{\rho,2}=\mu_{\vmPrior_2,\vmInd_2}\in\cP^\circ(\vmAssSp_2)$.
Further, we have $\cM_{\vmIndP,\rho}=\mu_\rho+\cV_{\vmIndP,\rho}$ with
\begin{align*}
\cV_{\vmIndP,\rho}=\{v\in\cV_{\vmIndP}:\mu_\rho+v\ge 0\}\textrm{, }
\cV_{\vmIndP}=\left\{v\in\mathbb R^{\vmAssSp_1}\times\mathbb R^{\vmAssSp_2}:\sum_{\vmAss\in\vmAssSp_1}v_1(\vmAss)=\sum_{\vmAss\in\vmAssSp_2}v_2(\vmAss)=0\textrm{, }\rho_{\vmIndP,v}=0_{[q]}\right\}\textrm{.}
\end{align*}
Hence, the set $\cM_{\vmIndP,\rho}$ is a polytope.
As in Section \ref{pp_onepoint}, notice that $\sum_\vmAss v_1(\vmAss)=0$ implies $\sum_\vmCol v_1\vert_*(\vmCol)=0$, so using $\sum_\vmCol\rho_{\vmIndP,v}(\vmCol)=0$ we obtain $\sum_\vmCol v_2\vert_*(\vmCol)=0$ which suggests $\sum_\vmAss v_2(\vmAss)=0$ and further yields
\begin{align*}
\cV_{\vmIndP}=\left\{v\in\mathbb R^{\vmAssSp_1}\times\mathbb R^{\vmAssSp_2}:\sum_{\vmAss\in\vmAssSp_1}v_1(\vmAss)=0\textrm{, }\rho_{\vmIndP,v}=0_{[q]}\right\}\textrm{.}
\end{align*}
The linear map $W$ whose kernel is $\cV_{\vmIndP}$ is given by
\begin{align*}
W=\begin{pmatrix}
1_{\vmAssSp_1}^{\matTr} & 0_{\vmAssSp_2}^{\matTr}\\
\vmIndP(\vmInd_1)W_1 & \vmIndP(\vmInd_2)W_2
\end{pmatrix}\textrm{,}
\end{align*}
using $\matTr$ to denote the transpose and where $W_1$, $W_2$ are the matrices from Section \ref{pp_onepoint} corresponding to $v_1\vert_*$ and $v_2\vert_*$. To see that $W$ is surjective fix $w\in\mathbb R^{[q]_0}$, let $v_1\in\mathbb R^{\vmAssSp_1}$ be any choice with $\sum_\vmAss v_1(\vmAss)=w(0)$ and use surjectivity of $\vmIndP(\vmInd_2)W_2$ to determine a preimage $v_2$ of $w_{[q]}-\vmIndP(\vmInd_1)v_1\vert_*$.
Hence, the $(|\vmAssSp_1|+|\vmAssSp_2|-q-1)$ dimensional kernel $\cV_{\vmIndP}$ of $W$ is never trivial since $q>1$. As in Section \ref{pp_onepoint} for any boundary point $\nu\in\cM_{\vmIndP,\rho}\setminus(\cP^\circ(\vmAssSp_1)\times\cP^\circ(\vmAssSp_2))$ the derivative of $f_{\vmIndP,\rho}$ at $\nu$ in the direction $\mu_\rho-\nu$ is $-\infty$, hence we have $\mu_{\min,\rho}\in\cM_{\vmIndP,\rho}\cap(\cP^\circ(\vmAssSp_1)\times\cP^\circ(\vmAssSp_2))$ for the unique minimizer $\mu_{\min,\rho}$ of the strictly convex map $f_{\vmIndP,\rho}$ (with convex and compact domain $\cM_{\vmIndP,\rho}$).
Further, since we have at least one degree of freedom, the point $\mu_{\min,\rho}$ is a stationary point of $f_{\vmIndP,\rho}$ and in particular the first derivatives of $f_{\vmIndP,\rho}$ at $\mu_{\min,\rho}$ in the directions $v\in\cV_{\vmIndP}$ vanish.
Now, with $\mu_{\min,\rho}=(\mu_{\min,\rho,1},\mu_{\min,\rho,2})$ minimizing $f_{\vmIndP,\rho}$ the component $\mu_{\min,\rho,1}$ obviously needs to be the unique minimizer of $f_{\vmInd_1,\mu_{\min,\rho,1}\vert_*}$ and $\mu_{\min,\rho,2}$ the unique minimizer of $f_{\vmInd_2,\mu_{\min,\rho,2}\vert_*}$.
Since $\mu_{\min,\rho}$ is in the relative interior of $\cM_{\vmIndP,\rho}$ we know that $\mu_{\min,\rho,1}\vert_*$, $\mu_{\min,\rho,2}\vert_*\in\cP^\circ([q])$ and can hence use Section \ref{pp_onepoint} to obtain $\vmPrior_1$, $\vmPrior_2\in\cP^\circ([q])$ with $\mu_{\min,\rho,1}=\mu_{\vmPrior_1,\vmInd_1}$ and $\mu_{\min,\rho,2}=\mu_{\vmPrior_2,\vmInd_2}$.
Now, fix $\vmCol\in[q]^2$ with $\vmCol_1\neq \vmCol_2$ and let $v\in\cV_{\vmIndP}$ be given by
\begin{align*}
v_1(\vmCol_11_{[d_1]})=d_2\vmIndP(\vmInd_2)\textrm{, }
v_1(\vmCol_21_{[d_1]})=-d_2\vmIndP(\vmInd_2)\textrm{, }
v_2(\vmCol_11_{[d_2]})=-d_1\vmIndP(\vmInd_1)\textrm{, }
v_2(\vmCol_21_{[d_2]})=d_1\vmIndP(\vmInd_1)\textrm{, }
\end{align*}
and $v_1(\vmAss)=0$, $v_2(\vmAss)=0$ otherwise. The derivative of $f_{\vmIndP,\rho}$ at $\nu=\mu_{\min,\rho}$ in the direction $v$ then yields
\begin{align*}
0&=\vmIndP(\vmInd_1)\sum_{j}\ln\left(\frac{\nu_1(\vmCol_j1_{[d_1]})}{\mu_1(\vmCol_j1_{[d_1]})}\right)v_1(\vmCol_j1_{[d_1]})
+\vmIndP(\vmInd_2)\sum_{j}\ln\left(\frac{\nu_2(\vmCol_j1_{[d_2]})}{\mu_2(\vmCol_j1_{[d_2]})}\right)v_2(\vmCol_j1_{[d_2]})\textrm{.}
\end{align*}
Rearranging gives $\vmPrior_1(\vmCol_1)/\vmPrior_1(\vmCol_2)=\vmPrior_2(\vmCol_1)/\vmPrior_2(\vmCol_2)$. Since this result holds for all $\vmCol\in[q]^2$ a normalization argument suggests that $\vmPrior_1=\vmPrior_2$, which shows that $\iota_{\vmIndP}$ is surjective.
Further, for fixed $\vmPrior\in\cP^\circ([q])$ we can evaluate the directional derivatives of $f_{\vmIndP,\iota_{\vmIndP}(\vmPrior)}$ at $\mu_\vmPrior=(\mu_{\vmPrior,\vmInd_1},\mu_{\vmPrior,\vmInd_2})$ directly to see that $\mu_\vmPrior$ is indeed a stationary point.
This establishes a one-to-one correspondence between $\mu_\vmPrior$ and $\iota_{\vmIndP}(\vmPrior)$, but since we have seen that $\vmPrior$ can be uniquely reconstructed from any of the $\mu_{\vmPrior,1}$, $\mu_{\vmPrior,2}$ this completes the proof.
\subsection{Finite supports}\label{pp_finite}
The arguments in Section \ref{pp_twopoint} directly extend to the case where $\vmIndP$ has finite support.
In particular $f_{\vmIndP,\rho}$ is strictly convex with convex and compact domain, which establishes the existence of a unique minimizer $\mu_{\min,\rho}$ for any $\rho\in\cP^\circ([q])$.
Analogous arguments to the ones above show that $\mu_{\min,\rho}\in\prod_\vmInd\cP^\circ(\vmAssSp_\vmInd)$ with $\vmInd$ in the support of $\vmIndP$.
From this we obtain $\vmPrior_\vmInd\in\cP^\circ([q])$ with $\mu_{\min,\rho,\vmInd}=\mu_{\vmPrior_\vmInd,\vmInd}$ since the components $\mu_{\min,\rho,\vmInd}$ also have to be minimizers for $\mu_{\min,\rho,\vmInd}\vert_*$ as discussed in Section \ref{pp_twopoint}.
But now, for any two distinct $\vmInd_1$, $\vmInd_2$ with $\vmIndP_{12}$ denoting the law of $\vmRInd_{\vmIndP}|\vmRInd_{\vmIndP}\in\{\vmInd_1,\vmInd_2\}$ and $\vmIndP^{\mathrm{c}}_{12}$ denoting the law of $\vmRInd_{\vmIndP}|\vmRInd_{\vmIndP}\not\in\{\vmInd_1,\vmInd_2\}$, further
$\mu_{12}=(\mu_{\vmPrior_1,\vmInd_1},\mu_{\vmPrior_2,\vmInd_2})$, $\rho_{12}=\rho_{\vmIndP_{12},\mu_{12}}$, $\mu^{\mathrm{c}}_{12}=(\mu_{\vmPrior_\vmInd,\vmInd})_{\vmInd\not\in\{\vmInd_1,\vmInd_2\}}$ and $\rho^{\mathrm{c}}_{12}=\rho_{\vmIndP^{\mathrm{c}}_{12},\mu^{\mathrm{c}}_{12}}$ we obtain $\rho=\pr[\vmRInd_{\hat \vmIndP}\in\{\vmInd_1,\vmInd_2\}]\rho_{12}+\pr[\vmRInd_{\hat \vmIndP}\not\in\{\vmInd_1,\vmInd_2\}]\rho^{\mathrm{c}}_{12}$ and further
\begin{align*}
f_{\vmIndP,\rho}((\mu_{\vmPrior_\vmInd,\vmInd})_\vmInd)=\pr[\vmRInd_{\vmIndP}\in\{\vmInd_1,\vmInd_2\}]f_{\vmIndP_{12},\rho_{12}}(\mu_{12})+\pr[\vmRInd_{\vmIndP}\not\in\{\vmInd_1,\vmInd_2\}]f_{\vmIndP^{\mathrm{c}}_{12},\rho^{\mathrm{c}}_{12}}(\mu^{\mathrm{c}}_{12})\textrm{.}
\end{align*}
But then we necessarily have $\vmPrior_{\vmInd_1}=\vmPrior_{\vmInd_2}$ since otherwise we could use Section \ref{pp_twopoint} to obtain the unique minimizer of $f_{\vmIndP_{12},\rho_{12}}$ and use it to replace $\mu_{12}$, thereby effectively decreasing $f_{\vmIndP,\rho}$ without changing $\rho_{12}$ and hence also $\rho$.
This shows that $\mu_{\min,\rho}=(\mu_{\vmPrior,\vmInd})_\vmInd$ for some $\vmPrior\in\cP^\circ([q])$
and thereby $\iota_{\vmIndP}$ is surjective.
To see injectivity we follow Section \ref{pp_twopoint} and show that $(\mu_{\vmPrior,\vmInd})_\vmInd$ is a stationary point of $f_{\vmIndP,\iota_{\vmIndP}(\vmPrior)}$ by evaluating the directional derivatives of $f_{\vmIndP,\iota_{\vmIndP}(\vmPrior)}$ at $(\mu_{\vmPrior,\vmInd})_\vmInd$, and thereby is the unique minimizer.
\subsection{Infinite supports}\label{pp_infinite}
Fix $\vmIndP$ with countably infinite support and $\rho\in\cP^\circ([q])$.
Without loss of generality we may assume $\vmIndSp=\mathbb Z_{>0}$.
With $\varepsilon_n=\pr[\vmRInd_{\hat \vmIndP}\not\in[n]]$, $\vmIndP_{n}^{\mathrm{c}}$ denoting the law of $\vmRInd_{\vmIndP}|\vmRInd_{\vmIndP}\not\in[n]$, further $\mu^{\mathrm{c}}_n=(\mu_\vmInd)_{\vmInd\not\in[n]}$ and $\rho^{\mathrm{c}}_n=\rho_{\vmIndP_{n}^{\mathrm{c}},\mu^{\mathrm{c}}_n}$ we have
\begin{align*}
\rho_n=\frac{1}{1-\varepsilon_n}\rho-\frac{\varepsilon_n}{1-\varepsilon_n}\rho^{\mathrm{c}}_n\in\cP^\circ([q])
\end{align*}
for $n$ sufficiently large since $\rho_n\rightarrow\rho$ for $n\rightarrow\infty$.
For $\vmInd\in[n]$ let $\vmPrior_\vmInd=\iota^{-1}_\vmInd(\rho_n)$ and $\vmPrior_\vmInd=u_{[q]}$ otherwise, where $u_{[q]}=q^{-1}1_{[q]}\in\cP^\circ([q])$ denotes the uniform distribution over $[q]$.
Then with $\nu=(\mu_{\vmPrior_\vmInd,\vmInd})_\vmInd$ we have $\rho=\rho_{\vmIndP,\nu}$ and further
\begin{align*}
f_{\vmIndP,\rho}(\nu)
&=\Erw\left[\vecone\{\vmRInd_{\vmIndP}\in[n]\}\KL{\nu_{\vmRInd_{\vmIndP}}}{\mu_{\vmRInd_{\vmIndP}}}\right]\in\mathbb R_{>0}\textrm{.}
\end{align*}
This shows that $\cM_{\vmIndP,\rho}^\circ=\{\nu\in\cM_{\vmIndP,\rho}:f_{\vmIndP,\rho}(\nu)<\infty\}$ is non-empty. Since $f_{\vmIndP,\rho}$ is convex $\cM^\circ_{\vmIndP,\rho}$ is convex and $f_{\vmIndP,\rho}$ is strictly convex on $\cM^\circ_{\vmIndP,\rho}$ which shows uniqueness of the minimizer $\mu_{\min,\rho}$ given its existence.

With the discussion above and analogous to Section \ref{pp_finite} we consider $\cM_{\vmIndP,\rho}=\nu^*+\cV_{\vmIndP,\rho}$ with $\nu^*\in\cM_{\vmIndP,\rho}^\circ$, $\cV_{\vmIndP,\rho}=\{v\in\cV_{\vmIndP}:\nu^*+v\ge 0\}$ and
\begin{align*}
\cV_{\vmIndP}=\left\{v\in\prod_\vmInd\mathbb R^{\vmAssSp_\vmInd}:\forall \vmInd\,\sum_\vmAss v_\vmInd(\vmAss)=0\textrm{, }\rho_{\vmIndP,v}=0_{[q]}\right\}
\end{align*}
as (infinite dimensional) polytope. Notice that $v\mapsto\rho_{\vmIndP,v}$ is continuous with respect to the product topology since it is continuous for the restriction to finite domains and we have uniform tail bounds since $v_\vmInd\vert_*$ is uniformly bounded. This shows that $\cM_{\vmIndP,\rho}\subseteq\cM_{\vmIndP}$ is closed and hence compact (and metrizable) since $\cM_{\vmIndP}$ is.
Now, fix a minimizing sequence $\nu_n\in\cM_{\vmIndP,\rho}$, $n\in\mathbb Z_{>0}$, of $f_{\vmIndP,\rho}$.
Using sequential compactness of $\cM_{\vmIndP}$ we find a converging subsequence of $(\nu_n)_n$ with limit $\nu\in\cM_{\vmIndP,\rho}$ and restrict to this subsequence without loss of generality.
This shows that $f_{\vmIndP,\rho}(\nu)\ge\inf_{\nu'} f_{\vmIndP,\rho}(\nu')$ is well-defined.
Now, assume that $f_{\vmIndP,\rho}(\nu)>\inf_{\nu'}f_{\vmIndP,\rho}(\nu')$.
Then there exists $n^*$ such that the contribution to $f_{\vmIndP,\rho}(\nu)$ for $\vmRInd_{\vmIndP}\in[n^*]$ is greater than $\inf_{\nu'}f_{\vmIndP,\rho}(\nu')$.
But the contributions to $f_{\vmIndP,\rho}(\nu_n)$ for $\vmRInd_{\vmIndP}\in[n^*]$ converge to the contribution to $f_{\vmIndP,\rho}(\nu)$ for $\vmRInd_{\vmIndP}\in[n^*]$ due to continuity, hence for all sufficiently large $n$ these contributions are bounded away from $\inf_{\nu'}f_{\vmIndP,\rho}(\nu')$ and thereby $f_{\vmIndP,\rho}(\nu_n)$ is bounded away from $\inf_{\nu'}f_{\vmIndP,\rho}(\nu')$ since the tails are non-negative, which is a contradiction to $\nu_n$ being a minimizing sequence.
Hence $\nu$ is a minimizer of $f_{\vmIndP,\rho}$, which establishes that $\nu=\mu_{\min,\rho}\in\cM^\circ_{\vmIndP,\rho}$ is the unique minimizer.
Since $\rho\in\cP^\circ([q])$ is fully supported we know that the colour frequencies of $\mu_{\min,\rho}$ conditional to $\vmRInd_{\vmIndP}\in[n]$ are fully supported for $n$ sufficiently large.
But then the decomposition of $\rho$ and $f_{\vmIndP,\rho}$ with respect to $[n]$ and analogous to Section \ref{pp_finite} allows to use the finite support results for $[n]$ to obtain $\vmPrior\in\cP^\circ([q])$ such that $\mu_{\min,\rho,\vmInd}=\mu_{\vmPrior,\vmInd}$ for $\vmInd\in[n]$, due to local optimality of the minimizer $\mu_{\min,\rho}$ as discussed before. Since this argument holds for any $n$ sufficiently large, we obtain $\vmPrior_n\in\cP^\circ([q])$ for any such choice and further $\vmPrior_n=\vmPrior$ since $\mu_{\vmPrior_n,\vmInd}=\mu_{\min,\rho,\vmInd}=\mu_{\vmPrior_{n'},\vmInd}$ for any $\vmInd\in[n]\subseteq[n']$ and $n\le n'$. This shows that $\mu_{\min,\rho}=(\mu_{\vmPrior,\vmInd})_\vmInd$ and further that $\iota_{\vmIndP}$ is surjective.

To see injectivity fix $\vmPrior\in\cP^\circ([q])$ and let $\rho=\iota_{\vmIndP}(\vmPrior)\in\cP^\circ([q])$, $\nu^*=(\mu_{\vmPrior,\vmInd})_\vmInd$.
Notice that
\begin{align*}
f_{\vmIndP,\rho}(\nu^*)
&=\Erw\left[\KL{\nu^*_{\vmRInd_{\vmIndP}}}{\mu_{\vmRInd_{\vmIndP}}}\right]
=\Erw\left[-\ln\left(Z_{\vmPrior,\vmRInd_{\vmIndP}}\right)-\bm d_{\vmIndP}\CE{\nu^*_{\vmRInd_{\vmIndP}}\vert_*}{\vmPrior}\right]\\
&\le\Erw\left[-\ln\left(Z_{\vmPrior,\vmRInd_{\vmIndP}}\right)\right]
\le-\ln\left(\min_\vmCol\vmPrior(\vmCol)\right)\Erw\left[\bm d_{\vmIndP}\right]\in\mathbb R_{>0}\textrm{,}
\end{align*}
i.e.~$\nu^*\in\cM^\circ_{\vmIndP,\rho}$.
With $\cM_{\vmIndP,\rho}=\nu^*+\cV_{\vmIndP,\rho}$ as before and for any $v\in\cV_{\vmIndP}$ we have
\begin{align*}
\Erw\left[\sum_\vmAss\left(\ln\left(\frac{\nu^*_{\vmRInd_{\vmIndP}}(\vmAss)}{\mu_{\vmRInd_{\vmIndP}}(\vmAss)}\right)+1\right)v_{\vmRInd_{\vmIndP}}(\vmAss)\right]
&=\Erw\left[\sum_{\vmAss,h}v_{\vmRInd_{\vmIndP}}(\vmAss)\ln(\vmPrior(\vmAss_h))\right]
=\Erw\left[\sum_{\vmCol\in[q]}\ln(\vmPrior(\vmCol))\bm d_{\vmIndP}v_{\vmRInd_{\vmIndP}}\vert_*(\vmCol)\right]\\
&=\sum_{\vmCol\in[q]}\ln(\vmPrior(\vmCol))\Erw[\bm d_{\vmIndP}]\rho_{\vmIndP,v}(\vmCol)=0\textrm{.}
\end{align*}
For any $\nu\in\cM_{\vmIndP,\rho}$ and $\vmInd$ in the support of $\vmIndP$ with $\nu_\vmInd\neq\nu^*_\vmInd$ the strict convexity of the relative entropy yields that
\begin{align*}
\KL{\nu_\vmInd}{\mu_\vmInd}>\KL{\nu^*_\vmInd}{\mu_\vmInd}+\sum_\vmAss\left(\ln\left(\frac{\nu^*_\vmInd(\vmAss)}{\mu_\vmInd(\vmAss)}\right)+1\right)(\nu_\vmInd(\vmAss)-\nu^*_\vmInd(\vmAss))\textrm{,}
\end{align*}
i.e.~the relative entropy is strictly above its tangent at $\nu^*$. Combining these arguments gives $f_{\vmIndP,\rho}(\nu)>f_{\vmIndP,\rho}(\nu^*)$ for any $\nu\in\cM_{\vmIndP,\rho}\setminus\{\nu^*\}$, i.e.~$\nu^*=\mu_{\min,\rho}$ which completes the proof (since $\vmPrior$ can be reconstructed from $\nu^*$).
\subsection{Extension to the boundary}\label{pp_boundary_proof}
Let $\vmColSp\subseteq[q]$ be non-empty.
Notice that $\iota_{\vmIndP}(\cP^\circ(\vmColSp))\subseteq\cP^\circ(\vmColSp)$ since we have $\mu_\vmInd(\vmCol 1_{[d_\vmInd]})>0$ for all $\vmCol\in[q]$.
This shows that the restriction $\iota_{\vmIndP}:\cP^\circ(\vmColSp)\rightarrow\cP^\circ(\vmColSp)$ is a bijection for $|\vmColSp|=1$ since then both the domain and the image have size $1$.
Otherwise, we use the results from the preceeding sections with $q_{\vmColSp}=|\vmColSp|$, $\mu_\vmInd$ replaced by $\mu_{\vmColSp,\vmInd}$ (which still satisfies $\mu_{\vmColSp,\vmInd}(\vmCol 1_{[d_\vmInd]})>0$ for $\vmCol\in\vmColSp$) and $P$ to see that the corresponding map $\iota_{\vmColSp,\vmIndP}:\cP^\circ(\vmColSp)\rightarrow\cP^\circ(\vmColSp)$ is a bijection and for any $\rho\in\cP^\circ(\vmColSp)$ with $\vmPrior=\iota^{-1}_{\vmColSp,\vmIndP}(\rho)$ the assignment distribution $(\mu_{\vmPrior,\vmInd})_\vmInd$ is the unique minimizer of the corresponding map $f_{\vmColSp,\vmIndP,\rho}$.
However, for any $\vmPrior\in\cP^\circ(\vmColSp)$ we have $\mu_{\vmPrior,\vmInd}=\mu_{\vmColSp,\vmPrior,\vmInd}$ and thereby $\iota_{\vmColSp,\vmIndP}=\iota_{\vmIndP}$, $f_{\vmColSp,\vmIndP,\rho}=f_{\vmIndP,\rho}$ on $\cP^\circ(\vmColSp)$ (up to relabeling colours).
This shows that $\iota_{\vmIndP}:\cP([q])\rightarrow\cP([q])$ is a bijection and $(\mu_{\vmPrior,\vmInd})_\vmInd$ is the unique minimizer of $f_{\vmIndP,\iota_{\vmIndP}(\vmPrior)}$.
Finally, notice that the choice of $\mu_{\vmColSp,\vmInd}$ over $\mu_\vmInd$ in the definition of $f_{\vmIndP,\rho}$ for $\rho\in\cP^\circ(\vmColSp)$ is only relevant for the case where the support of $\vmIndP$ is infinite since these two versions of $f_{\vmIndP,\rho}$ only differ by an additive constant whenever the alternative definition of $f_{\vmIndP,\rho}$ is finite.
\subsection{Including zero}\label{pp_zero_proof}
Assume that $\vmIndP\in\vmIndPSp$ is such that $\pr[\bm d_{\vmIndP}=0]>0$.
By the definition of $\vmIndPSp$ we have $\pr[\bm d_{\vmIndP}=0]<1$. Further, we have $\cP([q]^0)=\{\mu_0\}$, i.e.~the one-point mass $\mu_0$ on the empty assignment is the only possible choice for $d_\vmInd=0$.
Now, let $\vmIndP^\circ$ be the law of $\vmRInd_{\vmIndP}|\bm d_{\vmIndP}>0$ and notice that $\hat \vmIndP^\circ=\hat \vmIndP$ which immediately gives $\iota_{\vmIndP}=\iota_{\vmIndP^\circ}$.
Further, since $\cP([q]^0)$ carries only one element the contributions to $f_{\vmIndP,\rho}$, $\rho\in\cP([q])$, for $d_\vmInd=0$ are $0$ and further $f_{\vmIndP,\rho}=\pr[\bm d_{\vmIndP}>0]f_{\vmIndP^\circ,\rho}$ since $f_{\vmIndP,\rho}$ only formally depends on the coordinates $\vmInd$ with $d_\vmInd=0$.
Thereby the results of the preceeding sections for $\vmIndP^\circ$ directly translate to $\vmIndP$.
\subsection{Continuity}\label{pp_cont}
First, we discuss continuity for fixed $\vmIndP\in\vmIndPSp$.
For this purpose we consider the decomposition $\iota_{\vmIndP}=\iota_3\circ\iota_2\circ\iota_1$ with
\begin{align*}
\iota_1&:\cP([q])\rightarrow\prod_{\vmInd\in\vmIndSp}\cP(\vmAssSp_\vmInd)\textrm{, }\vmPrior\mapsto(\mu_{\vmPrior,\vmInd})_{\vmInd\in\vmIndSp}\textrm{,}\\
\iota_2&:\prod_{\vmInd\in\vmIndSp}\cP(\vmAssSp_\vmInd)\rightarrow\cP([q])^{\vmIndSp}\textrm{, }\nu\mapsto(\nu_\vmInd\vert_*)_{\vmInd\in\vmIndSp}\textrm{,}\\
\iota_3&:\cP([q])^{\vmIndSp}\rightarrow\cP([q])\textrm{, }\rho\mapsto\Erw\left[\rho_{{\vmRInd}_{\hat \vmIndP}}\right]\textrm{.}
\end{align*}
We consider both $\im(\iota_1)$ and $\im(\iota_2)$ equipped with the inherited product topology, which is metrizable since $\vmIndSp$ is countable and therefore all topological spaces in question are compact and metrizable.
Thanks to the properties of the product topology both $\iota_1$ and $\iota_2$ are continuous, i.e.~since $\vmPrior\mapsto\mu_{\vmPrior,\vmInd}$, $\nu\mapsto\nu\vert_*$ and the projections are continuous.
This suggests that the restrictions $\iota_1:\cP([q])\rightarrow\im(\iota_1)$ and $\iota_2:\im(\iota_1)\rightarrow\iota_2(\im(\iota_1))$ are homeomorphisms, since they are continuous bijections of compact metrizable spaces and where Section \ref{pp_onepoint} is already sufficient to obtain bijectivity (with Section \ref{pp_boundary_proof} and Section \ref{pp_zero_proof}).
Continuity of $\iota_3$ was discussed in Section \ref{pp_infinite}, which concludes the proof that $\iota_{\vmIndP}$ is continuous.
But since $\iota_{\vmIndP}$ is then a continuous bijection of compact metric spaces it is a homeomorphism.

Next, we show that $\iota:\vmIndPSp\times\cP([q])\rightarrow\cP([q])$ is continuous.
For this purpose let $\vmIndP$, $\vmIndP_{n}\in\vmIndPSp$ and $\vmPrior$, $\vmPrior_n\in\cP([q])$, $n\in\mathbb Z_{>0}$, with $(\vmIndP_{n},\vmPrior_n)\rightarrow(\vmIndP,\vmPrior)$ for $n\rightarrow\infty$.
Further let $\rho_n=\iota(\vmIndP_{n},\vmPrior_n)$ and $\rho=\iota(\vmIndP,\vmPrior)$.
Standard arguments show that $\vmIndP_{n}\rightarrow \vmIndP$ with respect to the metric $\vmIndPD$ yields $\Erw[\bm d_{\vmIndP_{n}}]\rightarrow\Erw[\bm d_{\vmIndP}]$ and further $\hat\vmIndP_{n}\rightarrow\hat \vmIndP$ in $\|\cdot\|_1$. For any $\varepsilon\in(0,1)$ and $n$ sufficiently large such that $\pr[\vmRInd_{\hat \vmIndP}>n]<\varepsilon$ this gives
\begin{align*}
\|\rho_n-\rho\|_1
&\le\Erw\left[\left\|\mu_{\vmPrior_n,\vmRInd_{\hat \vmIndP}}\vert_*-\mu_{\vmPrior,\vmRInd_{\hat \vmIndP}}\vert_*\right\|_1\right]+
q\|\hat\vmIndP_{n}-\hat \vmIndP\|_1\\
&\le\Erw\left[\left\|\mu_{\vmPrior_n,\vmRInd_{\hat \vmIndP}}\vert_*-\mu_{\vmPrior,\vmRInd_{\hat \vmIndP}}\vert_*\right\|_1
\vecone\{\vmRInd_{\hat \vmIndP}\in[n]\}\right]+2\varepsilon+q\|\hat\vmIndP_{n}-\hat \vmIndP\|_1\rightarrow 2\varepsilon
\end{align*}
which shows that $\iota$ is continuous.
For the reverse direction fix a sequence $(\vmIndP_{n},\rho_n)\rightarrow(\vmIndP,\rho)$ using $\vmPrior_n=\iota^{-1}(\vmIndP_{n},\rho_n)$ and $\vmPrior=\iota^{-1}(\vmIndP,\rho)$.
Let $\varepsilon\in(0,1)$, $\cR=\iota_{\vmIndP}(\cB_{\varepsilon}(\vmPrior))$, and using that $\iota_{\vmIndP}$ is a homeomorphism let $\delta\in(0,1)$ be sufficiently small such that $\cB_\delta(\rho)\subseteq\cR$. For $n$ sufficiently large we have $\|\hat\vmIndP_{n}-\hat \vmIndP\|_1<\delta/(2q)$ and $\rho_n\in\cB_{\delta/2}(\rho)$, so
\begin{align*}
\|\iota_{\vmIndP}(\vmPrior_n)-\rho\|_1
<\frac{1}{2}\delta+q\frac{1}{2q}\delta=\delta\textrm{,}
\end{align*}
thereby $\iota_{\vmIndP}(\vmPrior_n)\in\cB_\delta(\rho)\subseteq\cR$ and hence $\vmPrior_n\in\cB_{\varepsilon}(\vmPrior)$. This completes the proof.
\section{Local Limit Theorem} \label{llt}

This section is mostly self-contained and only depends on the results obtained in Section \ref{vom} as well as Chapter 5 in \cite{bhattacharya2010}.
We start with the setup in Section \ref{vom}, i.e.~we fix a number $q\in\ZZ_{>1}$ of colours and a family $(d_\vmInd,\mu_\vmInd)_{\vmInd\in\vmIndSp}$ satisfying {\bf SPAN} with $\cD\setminus\{0\}\neq\emptyset$.
Following Chapter 5 in \cite{bhattacharya2010} let $\lltLat\subseteq\ZZ^{q-1}$ denote the lattice spanned by $(|\vmAss^{-1}(\vmCol)|)_{\vmCol\in[q-1]}\in\ZZ^{q-1}$ for $\vmAss\in\vmAssSp_\vmInd$ and $\vmInd\in\vmIndSp$, i.e.~the set of all points obtained from (finite) linear combinations with integer coefficients, and notice that $\lltLat$ indeed has full rank due to {\bf SPAN}.
Hence, Theorem 21.1 in \cite{bhattacharya2010} ensures the existence of a lattice basis $b_\vmCol$, $\vmCol\in[q-1]$, i.e.~$\lltLat=\sum_{\vmCol\in[q-1]}\ZZ b_\vmCol$.
Using the proof of this theorem and {\bf SPAN}, we notice that there exists a unique choice of the $b_\vmCol$ such that $(b_\vmCol)_\vmCol$ is lower triangular with positive diagonal $h=(b_{\vmCol,\vmCol})_{\vmCol}\in\ZZ_{>0}^{q-1}$.
The set of boxes $\lltBox_{\lltCFT}=\lltCFT+\prod_{\vmCol\in[q-1]}[-h_\vmCol/2,h_\vmCol/2)$ centered at $\lltCFT\in\lltLat$ is a partition of $\RR^{q-1}$ with $\lltBox_{\lltCFT}\cap\lltLat=\lltCFT$.
Further, since each $b_\vmCol$ is obtained from a finite linear combination of colour frequencies there exists a (not necessarily unique) \emph{finite} subset $\vmIndSp^\circ\subseteq\vmIndSp$ that spans $\lltLat$ in the sense above.

For given $\vmPrior\in\cP([q])$, a number $n\in\ZZ_{>0}$ of vertices, indices $\svmInd=(\vmInd_{\lltVL})_{\lltVL\in[n]}\subseteq\vmIndSp$ and $\lltVL\in[n]$
we use the shorthands $d_\lltVL=d_{\vmInd_\lltVL}$, $\vmAssSp_{\lltVL}=\vmAssSp_{\vmInd_\lltVL}$, further $Z_{\vmPrior,\lltVL}=Z_{\vmPrior,\vmInd_\lltVL}$ and $\mu_{\vmPrior,\lltVL}=\mu_{\vmPrior,\vmInd_\lltVL}$ for brevity, further let $\mu_{\lltVL}=\mu_{u_{[q]},\lltVL}$ and also omit the subscript if $\vmPrior=u_{[q]}$ for derived quantities since this corresponds to no variation.
Let $\lltRVL_n\sim u_{[n]}$ denote the uniformly random vertex, $\vmRInd_{\svmInd}=\vmInd_{\lltRVL_n}$ and $\vd_{\svmInd}=d_{\lltRVL_n}$.
The main results of this section apply to sequences $(\vmIndSp_n)_{n\in\ZZ_{>0}}$ of families
$\vmIndSp_n\subseteq\vmIndSp^n$ satisfying the following assumptions.
\begin{description}
\item[GEN] There exists $\epsGen\in(0,1)$ and a subset $\vmIndSp^\circ\subseteq\vmIndSp$ that spans $\lltLat$ such that
$\pr[\vmRInd_{\svmInd}=\vmInd]\ge\epsGen$ for all $\vmInd\in\vmIndSp^\circ$, $\svmInd\in\vmIndSp_n$ and $n\in\ZZ_{>0}$.
\item[VAR] There exists $\secB\in\RR_{>0}$ such that
$\Erw[\vd_{\svmInd}^2]\le\secB$ for all $\svmInd\in\vmIndSp_n$ and $n\in\ZZ_{>0}$.
\item[SKEW] There exists a sequence $\thiB_n\in\RR_{>1}$, $n\in\ZZ_{>0}$ such that
$\thiB_n=o(\sqrt{n/\log(n)^3})$ and $\Erw[\vd_{\svmInd}^3]\le\thiB_n$ for all $\svmInd\in\vmIndSp_n$ and $n\in\ZZ_{>0}$.
\end{description}
We may assume without loss of generality that $\vmIndSp^\circ$ is minimal and in particular $d_{\vmInd}>0$ for all $\vmInd\in\vmIndSp^\circ$.
Alternatively, we could define $\vmIndSp_n$ to be maximal for given $\epsGen$, $\vmIndSp^\circ$, $\secB$ and $\thiB_n$.
Further, notice that the empty set satisfies all assumptions.

Now, for given $\vmPrior\in\cP([q])$, $n\in\ZZ_{>0}$ and $\svmInd\in\vmIndSp^n$ let $\lltRHEA_{\vmPrior,\svmInd}$ denote a sample from
$\bigotimes_{\lltVL\in[n]}\mu_{\vmPrior,\lltVL}$, i.e.~$\lltRHEA_{\vmPrior,\svmInd}=(\vmRAss_{\vmPrior,\svmInd,i})_{i\in[n]}$ with the components being independent, and let $\lltRCFA_{\vmPrior,\svmInd}$ be the corresponding absolute colour frequencies, i.e.~for $\vmCol\in[q]$ given by
\begin{flalign*}
\lltRCFA_{\vmPrior,\svmInd}(\vmCol)=\sum_{\lltVL\in[n],h\in[d_\lltVL]}\vecone\{\vmRAss_{\vmPrior,\svmInd,\lltVL,h}=\vmCol\}=\sum_{\lltVL\in[n]}|\vmRAss_{\vmPrior,\svmInd,\lltVL}^{-1}(\vmCol)|
\end{flalign*}
and further let $\lltECFA_{\vmPrior,\svmInd}=\Erw[\lltRCFA_{\vmPrior,\svmInd}]$ denote the expectation.
Notice that $\Erw[\vd_{\svmInd}]n$ is the total degree and in particular
$\sum_{\vmCol\in[q]}\lltRCFA_{\vmPrior,\svmInd}(\vmCol)=\Erw[\vd_{\svmInd}]n$.
The first result allows to control the tails of the colour frequencies.
\begin{proposition}\label{llt_tails}
Let $(\vmIndSp_n)_{n\in\ZZ_{>0}}\subseteq\vmIndSp^n$ satisfy {\bf GEN} and {\bf VAR}.
Then there exist constants $c$, $c'\in\RR_{>0}$ such that for all $n\in\ZZ_{>0}$, all $\svmInd\in\vmIndSp$ and all $r\in\RR_{>0}$ we have
\begin{flalign*}
\pr\left[\|\lltRCFA_\svmInd-\lltECFA\|_1\ge \sqrt{\Erw[\vd_{\svmInd}]n}r\right]\le c'\exp(-cr^2)\textrm{.}
\end{flalign*}
\end{proposition}
For $\svmInd\in\vmIndSp^n$ let $\vmIndP_{\svmInd}\in\cP(\vmIndSp)$ denote the law of $\vmRInd_{\svmInd}$ and notice that $\vd_{\svmInd}=\vd_{\vmIndP_{\svmInd}}$ is consistent.
If $\Erw[\vd_{\svmInd}]$ is positive let $\lltRCFR_{\vmPrior,\svmInd}=\frac{1}{\Erw[\vd_{\svmInd}]n}\lltRCFA_{\vmPrior,\svmInd}\in\cP([q])$ denote the relative colour frequencies for given $\vmPrior\in\cP([q])$.
With Chapter 5 in \cite{bhattacharya2010} (Section 3.5 in \cite{durrett2010}) it is immediate that $\pr[\lltRCFA_{\vmPrior,\svmInd,[q-1]}\in\lltLat]=1$ for any $\vmPrior\in\cP([q])$ and $\svmInd\in\vmIndSp^n$,
where we use the shorthand $\lltRCFA_{\vmPrior,\svmInd,[q-1]}=(\lltRCFA_{\vmPrior,\svmInd,\vmCol})_{\vmCol\in[q-1]}$ here and in the remainder.
Using $\sum_{\vmCol\in[q]}\lltRCFA_{\vmPrior,\svmInd}(\vmCol)=\Erw[\vd_{\svmInd}]n$ we extend the lattice $\lltLat$ to $\lltLatA_{\svmInd}\in\ZZ^q$, hence $\pr[\lltRCFA_{\vmPrior,\svmInd}\in\lltLatA_\svmInd]=1$, scale and truncate it to obtain $\pr[\lltRCFR_{\vmPrior,\svmInd}\in\lltLatR_\svmInd]=1$ with $\lltLatR_\svmInd=\frac{1}{\Erw[\vd_{\svmInd}]n}\lltLatA_\svmInd\cap\cP([q])$ and conclude with $\lltLatP_\svmInd=\iota^{-1}_\svmInd(\lltLatR_\svmInd)$ using the shorthand $\iota_\svmInd=\iota_{\vmIndP_\svmInd}$ for the homeomorphism introduced in Section \ref{vom} (notice that indeed $\vmIndP_\svmInd\in\vmIndPSp$).
Finally, let $\Sigma_{\svmInd,\vmPrior}=\frac{1}{\Erw[\vd_{\svmInd}]n}\Cov(\lltRCFA_{\vmPrior,\svmInd,[q-1]})$.
The following theorem determines the local limits in the large deviation regime.
\begin{theorem}\label{llt_global}
Fix a compact set $\cP^*\subseteq\cP^\circ([q])$ and a sequence $(\vmIndSp_n)_n\subseteq\vmIndSp^n$ satisfying {\bf SPAN}, {\bf GEN}, {\bf VAR} and {\bf SKEW}.
Then uniformly for all $\svmInd\in\vmIndSp_n$ and $\vmPrior\in\lltLatP_\svmInd\cap\cP^*$ the covariance matrix $\Sigma_{\svmInd,\vmPrior}$ is positive definite with $\|\Sigma_{\svmInd,\vmPrior}^{-1}\|_2^{-1},\|\Sigma_{\svmInd,\vmPrior}\|_2=\Theta(1)$ and further using $\lltCFR=\iota_{\svmInd}(\vmPrior)$ we have
\begin{flalign*}
\pr[\lltRCFR_\svmInd=\lltCFR]&=\left(1+O\left(\thiB_n\sqrt{\frac{\ln(n)^3}{n}}\right)\right)
\frac{\prod_\vmCol h(\vmCol)}{\sqrt{\Erw[\vd_\svmInd]n}^{q-1}}\frac{\exp\left(-n\Erw\left[\KL{\mu_{\vmPrior,\vmRInd_\svmInd}}{\mu_{\vmRInd_\svmInd}}\right]\right)}{\sqrt{(2\pi)^{q-1}\det(\Sigma_{\svmInd,\vmPrior})}}\textrm{.}
\end{flalign*}
\end{theorem}
For frequencies close to the expectation $\lltECFR_\svmInd$, with $\lltECFR_{\vmPrior,\svmInd}=\Erw[\lltRCFR_{\vmPrior,\svmInd}]$, Theorem \ref{llt_global} can be simplified to remove the dependency on $\vmPrior$.
\begin{theorem}\label{llt_local}
Fix a sequence $r_n\in\RR_{>0}$ with $r_n^2n=\Omega(1)$, $n\thiB_nr_n^3=o(1)$ and a family $(\vmIndSp_n)_n\subseteq\vmIndSp^n$ satisfying {\bf SPAN}, {\bf GEN}, {\bf VAR} and {\bf SKEW}.
Then uniformly for all $\svmInd\in\vmIndSp_n$ and $\lltCFR\in\lltLatR_\svmInd$ with $\|\lltCFR-\lltECFR_\svmInd\|_2<r_n$ we have
\begin{flalign*}
\pr[\lltRCFR_\svmInd=\lltCFR]&=\left(1+O\left(\left(\sqrt{\frac{\ln(n)}{n}}^3+r_n^3\right)\thiB_nn\right)\right)
\frac{\prod_\vmCol h(\vmCol)}{\sqrt{\Erw[\vd_\svmInd]n}^{q-1}}\phi_\svmInd\left(\sqrt{\Erw[\vd_\svmInd]n}(\lltCFR-\lltECFR_\svmInd)_{[q-1]}\right)\textrm{,}
\end{flalign*}
where $\phi_\svmInd$ denotes the density of the normal distribution
$\cN(0_{[q-1]},\Sigma_\svmInd)$.
\end{theorem}
\subsection{Proof of Proposition \ref{llt_tails}}
Notice that $\lltRCFA_\svmInd(\vmCol)$ is a sum of independent bounded random variables $|\vmRAss_{\svmInd,\lltVL}^{-1}(\vmCol)|\in[d_{\lltVL}]_0$ with $\lltVL\in[n]$, $\vmCol\in[q]$, and hence Hoeffding's inequality with the usual transition to $\|\cdot\|_\infty$ and $\|\cdot\|_1$ yields
\begin{flalign*}
\pr\left[\|\lltRCFA_\svmInd-\lltECFA_\svmInd\|_1\ge \sqrt{\Erw[\vd_\svmInd]n}r\right]\le 2q\exp\left(-\frac{2\Erw[\vd_\svmInd]nr^2}{q^2\sum_\lltVL d_\lltVL^2}\right)\le 2q\exp\left(-\frac{2\epsGen}{q^2\secB}r^2\right)\textrm{.}
\end{flalign*}
\subsection{Proof of Theorem \ref{llt_global}}
The core idea of the proof is to determine the asymptotics of the point probabilities for colour frequencies $\lltCFR\in\lltLatR_\svmInd$ by replacing the original law $\bigotimes_\lltVL\mu_\lltVL$ with the frequency specific law $\bigotimes_\lltVL\mu_{\vmPrior,\lltVL}$, $\vmPrior=\iota_\svmInd^{-1}(\lltCFR)$, that is centered around $\lltCFR$.
The first result introduces the variation of measure.
For this purpose let $\lltCFA_\lltHEA\in\ZZ_{\ge 0}^q$ and $\lltCFR_\lltHEA\in\cP([q])$ denote the absolute and relative colour frequencies of an assignment $\lltHEA\in\prod_{\lltVL\in[n]}[q]^{d_\lltVL}$ for $\svmInd\in\vmIndSp^n$ with $\Erw[\vd_\svmInd]>0$ or equivalently $\vmIndP_\svmInd\in\vmIndPSp$.
\begin{lemma}\label{llt_global_measure_switch}
Assume that $\vmIndSp$ satisfies {\bf SPAN}.
Then for all $n\in\ZZ_{>0}$, all $\svmInd\in\vmIndSp^n$ with $\vmIndP_\svmInd\in\vmIndPSp$ and $\lltHEA\in\prod_{\lltVL\in[n]}[q]^{d_\lltVL}$, using $\vmPrior=\iota^{-1}_{\svmInd}(\lltCFR_\lltHEA)$ we have
\begin{flalign*}
\pr[\lltRHEA_\svmInd=\lltHEA]&=\exp\left(-n\Erw\left[\KL{\mu_{\vmPrior,\vmRInd_\svmInd}}{\mu_{\vmRInd_\svmInd}}\right]\right)\pr[\lltRHEA_{\vmPrior,\svmInd}=\lltHEA]\textrm{.}
\end{flalign*}
\end{lemma}
The gist in the proof of Lemma \ref{llt_global_measure_switch}, which is postponed to Section \ref{llt_global_measure_switch_proof}, is that $\lltECFA_{\vmPrior,\svmInd}=\lltCFA_{\lltHEA}$ by the design of $\iota_\svmInd$.
Lemma \ref{llt_global_measure_switch} directly implies that 
\begin{align*}
\pr[\lltRCFR_\svmInd=\lltCFR_\lltHEA]
=\exp\left(-n\Erw\left[\KL{\mu_{\vmPrior,\vmRInd_\svmInd}}{\mu_{\vmRInd_\svmInd}}\right]\right)
\pr[\lltRCFR_{\vmPrior,\svmInd}=\lltCFR_\lltHEA]
\end{align*}
and hence Theorem \ref{llt_global} is an immediate consequence from the following proposition which reflects the asymptotic point probability of \emph{exactly} the expectation.
\begin{proposition}\label{llt_global_prop}
With $\cP^*\subseteq\cP^\circ([q])$ and $(\vmIndSp_n)_n\subseteq\vmIndSp^n$ from Theorem \ref{llt_global}
and uniformly for all $\svmInd\in\vmIndSp_n$ and $\vmPrior\in\lltLatP_\svmInd\cap\cP^*$ the covariance matrix $\Sigma_{\svmInd,\vmPrior}$ is positive definite with $\|\Sigma_{\svmInd,\vmPrior}^{-1}\|_2^{-1},\|\Sigma_{\svmInd,\vmPrior}\|_2=\Theta(1)$ and further
\begin{flalign*}
\pr[\lltRCFR_{\vmPrior,\svmInd}=\lltECFR_{\vmPrior,\svmInd}]&=\left(1+O\left(\thiB_n\sqrt{\frac{\ln(n)^3}{n}}\right)\right)
\frac{\prod_\vmCol h(\vmCol)}{\sqrt{\Erw[\vd_\svmInd]n}^{q-1}}\frac{1}{\sqrt{(2\pi)^{q-1}\det(\Sigma_{\svmInd,\vmPrior})}}\textrm{.}
\end{flalign*}
\end{proposition}
We show this result in Section \ref{llt_global_prop_proof} using the characteristic function inversion formula in Chapter 5 of \cite{bhattacharya2010}, or equivalently a vanilla version of the saddle point method.
\subsection{Proof of Lemma \ref{llt_global_measure_switch}}\label{llt_global_measure_switch_proof}
First, notice that $\pr[\lltRHEA_\svmInd=\lltHEA]=0$ if and only if $\pr[\lltRHEA_{\vmPrior,\svmInd}=\lltHEA]$ since $\vmPrior\in\cP([q])$ and $\iota_{\svmInd}(\vmPrior)\in\cP([q])$ always have the same support.
For the non-trivial case with $\vmColSp\subseteq[q]$ denoting the support of $\vmPrior$ we have
\begin{flalign*}
\pr[\lltRHEA_\svmInd=\lltHEA]
&=\prod_{\lltVL}\mu_{\lltVL}(\vmAss_\lltVL)
=c\prod_{\vmCol\in\vmColSp}\vmPrior(\vmCol)^{\lltCFA(\vmCol)}\prod_{\lltVL}\frac{\mu_{\lltVL}(\vmAss_\lltVL)}{Z_{\vmPrior,\lltVL}}
=c\prod_{\lltVL,h}\vmPrior(\vmAss_{\lltVL,h})\prod_{\lltVL}\frac{\mu_{\lltVL}(\vmAss_\lltVL)}{Z_{\vmPrior,\lltVL}}
=c\pr[\lltRHEA_{\vmPrior,\svmInd}=\lltHEA]\textrm{,}\\
c&=\left(\prod_\lltVL Z_{\vmPrior,\lltVL}\right)\left(\prod_{\vmCol\in\vmColSp}\vmPrior(\vmCol)^{-\lltCFA_\lltHEA(\vmCol)}\right)
=\left(\prod_\vmInd Z_{\vmPrior,\vmInd}^{n\pr[\vmRInd_\svmInd=\vmInd]}\right)\left(\prod_{\vmCol\in\vmColSp}\vmPrior(\vmCol)^{-\lltCFA_\lltHEA(\vmCol)}\right)\textrm{.}
\end{flalign*}
For $\vmInd\in\vmIndSp$ let $\vmRAss^*_{\vmPrior,\vmInd}$ be a sample from $\mu_{\vmPrior,\vmInd}$ and $\lltRCFA^*_{\vmPrior,\vmInd}=(|\vmRAss^{*-1}_{\vmPrior,\vmInd}(\vmCol)|)_{\vmCol\in[q]}$ the colour frequencies.
Now, we use $\lltCFR_\lltHEA=\iota_\svmInd(\vmPrior)=\lltECFR_{\vmPrior,\svmInd}$ to obtain 
$\ln(c)=-n\Erw[\alpha(\vmRInd_\svmInd)]$ with
\begin{align*}
\alpha(\vmInd)
&=\Erw\left[\sum_{\vmCol\in\vmColSp}\lltRCFA^*_{\vmPrior,\vmInd}(\vmCol)\log(\vmPrior(\vmCol))-\log(Z_{\vmPrior,\vmInd})\right]
=\Erw\left[\log\left(\frac{\prod_{h\in[d_\vmInd]}\vmPrior(\vmRAss^*_{\vmPrior,\vmInd,h})}{Z_{\vmPrior,\vmInd}}\right)\right]
=\KL{\mu_{\vmPrior,\vmInd}}{\mu_\vmInd}\textrm{.}
\end{align*}
\subsection{Proof of Proposition \ref{llt_global_prop}}\label{llt_global_prop_proof}
We start with some standard results based on Section 21 in \cite{bhattacharya2010}.
Notice that due to {\bf GEN} the lattice $\lltLat$ is the \emph{minimal lattice} for $\lltRCFA_{\vmPrior,\svmInd,[q-1]}$.
Slightly deviating from \cite{bhattacharya2010} we let the \emph{dual basis} be given by $(b^*_\vmCol)_\vmCol=B^*=2\pi (B^{-1})^{\matTr}$ with $B=(b_\vmCol)_\vmCol$, so $B^*$ is upper triangular, further $B^{*\matTr}B=2\pi I_{[q-1]}$ with $I_{[q-1]}$ denoting the identity, and $b^*_{\vmCol}(\vmCol)=2\pi/h(\vmCol)$ for $\vmCol\in[q-1]$ which yields the \emph{fundamental domain} $\lltBox^*=\prod_\vmCol[-\pi/h(\vmCol),\pi/h(\vmCol)]$. Hence, translating the point probability for $\lltRCFR_{\vmPrior,\svmInd}$ to $\lltRCFA_{\vmPrior,\svmInd}$, reducing it to $\lltRCFA_{\vmPrior,\svmInd,[q-1]}$ and using the \emph{inversion formula in the lattice case} gives
\begin{align*}
\pr[\lltRCFR_{\vmPrior,\svmInd}=\lltECFR_{\vmPrior,\svmInd}]
&=\frac{\prod_\vmCol h(\vmCol)}{(2\pi)^{q-1}}\int_{\lltBox^*}\lltCFInt_{\svmInd,\vmPrior}(\lltCFAng)\intD\lltCFAng\textrm{, }\\
\lltCFInt_{\svmInd,\vmPrior}(\lltCFAng)
&=\exp\left(-\cpxI\lltCFAng^\matTr\lltECFA_{\svmInd,\vmPrior,[q-1]}\right)\Erw\left[\exp\left(\cpxI\lltCFAng^\matTr\lltRCFA_{\svmInd,\vmPrior,[q-1]}\right)\right]\textrm{.}
\end{align*}
With the shorthand $\lltECFA^*_{\vmPrior,\svmInd}=\Erw[\lltRCFA^*_{\vmPrior,\svmInd}]$ for the expectation of $\lltRCFA^*_{\vmPrior,\svmInd}$ from Section \ref{llt_global_measure_switch_proof} and since $\lltRCFA_{\vmPrior,\svmInd}$ is a sum of independent random vectors we have
\begin{align*}
\lltCFInt_{\svmInd,\vmPrior}(\lltCFAng)&=\prod_{\vmInd\in\vmIndSp}\lltCFInt^*_{\vmPrior,\vmInd}(\lltCFAng)^{n\pr[\vmRInd_\svmInd=\vmInd]}\textrm{, }
\lltCFInt^*_{\vmPrior,\vmInd}(\lltCFAng)=\Erw\left[\exp\left(\cpxI\lltCFAng^\matTr\left(\lltRCFA^*_{\vmPrior,\vmInd}-\lltECFA^*_{\vmPrior,\vmInd}\right)_{[q-1]}\right)\right]\textrm{.}
\end{align*}
Now, we follow the standard scheme in that we first bound the tails at constant distance, then establish subgaussian tails and finally use a normal approximation to obtain the material contribution, with some careful bookkeeping along the way to obtain suitable error bounds.
\begin{lemma}\label{llt_global_tails}
For all $r\in\RR_{>0}$ there exists a constant $c\in\RR_{>0}$ such that for all $n\in\ZZ_{>0}$, $\svmInd\in\vmIndSp_n$, $\vmPrior\in\cP^*$ and $\lltCFAng\in\lltBox^*\setminus\cB_r(0_{[q-1]})$ we have
$|\lltCFInt_{\svmInd,\vmPrior}(\lltCFAng)|\le\exp(-cn)$.
\end{lemma}
\begin{proof}
First, since $\lltCFInt^*_{\vmPrior,\vmInd}$ is a characteristic function we have
$|\lltCFInt^*_{\vmPrior,\vmInd}(\lltCFAng)|\le 1$.
This shows that $|\lltCFInt_{\svmInd,\vmPrior}(\lltCFAng)|\le|\lltCFInt^\circ_{\vmPrior}(\lltCFAng)|^{\epsGen n}$ with
$\lltCFInt^\circ_{\vmPrior}(\lltCFAng)=\prod_{\vmInd\in\vmIndSp^\circ}\lltCFInt^*_{\vmPrior,\vmInd}(\lltCFAng)$.
The unique maximizer of $|\lltCFInt^\circ_{\vmPrior}(\lltCFAng)|$ is $0_{[q-1]}$ on the closure $\lltBox^{\mathrm{c}}$ of $\lltBox^*$ for all $\vmPrior\in\cP^\circ([q])$, due to Lemma 21.6 in \cite{bhattacharya2010} and the fact that $\vmIndSp^\circ$ spans $\lltLat$. Considering $\lltCFInt^\circ(\vmPrior,\lltCFAng)=\lltCFInt^\circ_{\vmPrior}(\lltCFAng)$ as a function of both $\vmPrior$ and $\lltCFAng$, we notice that $\lltCFInt^\circ$ and further $|\lltCFInt^\circ|$ are both continuous on the compact set $\cP^*\times\lltBox^{\mathrm{c}}$, so the latter attains its maximum $M_r\in[0,1)$ on $\cP^*\times(\lltBox^{\mathrm{c}}\setminus\cB_r(0_{[q-1]}))$ for sufficiently small $r$.
If $r$ is too large the assertion is trivially true, otherwise take $c'\in(M_r,1)$ to obtain $|\lltCFInt_{\svmInd,\vmPrior}(\lltCFAng)|\le\exp(-cn)$ with $c=-\epsGen\log(c')\in\RR_{>0}$, valid for all required $n$, $\svmInd$, $\vmPrior$ and $\lltCFAng$.
\end{proof}
With the coarse tail bound in place we establish subgaussian tails.
For this purpose we take a closer look at the atoms $\lltCFInt^*_{\vmPrior,\vmInd}$.
Let $\Sigma^*_{\vmPrior,\vmInd}=\Cov(\lltRCFA^*_{\vmPrior,\vmInd,[q-1]})$ denote the corresponding covariance.
\begin{lemma}\label{llt_global_atoms}
For all $\vmInd\in\vmIndSp$ with $d_\vmInd>0$,
$\vmPrior\in\cP^\circ([q])$ and $\lltCFAng\in\RR^{q-1}$ with $\|\lltCFAng\|_\infty<\frac{\pi}{2d_\vmInd}$ we have
\begin{align*}
\left|\ln(\lltCFInt^*_{\vmPrior,\vmInd}(\lltCFAng))+\frac{1}{2}\lltCFAng^{\matTr}\Sigma^*_{\vmPrior,\vmInd}\lltCFAng\right|\le \frac{d_\vmInd^3}{3\cos(d_\vmInd\|\lltCFAng\|_\infty)^3}\|\lltCFAng\|_1^3\textrm{.}
\end{align*}
Further, there exists $c_\vmInd\in\RR_{>0}$ such that $c_\vmInd\le\|\Sigma^{*-1}_{\vmPrior,\vmInd}\|_2^{-1}\le\|\Sigma^*_{\vmPrior,\vmInd}\|_2\le d_\vmInd^2$ for all $\vmPrior\in\cP^*$.
\end{lemma}
\begin{proof}
Using $d=d_\vmInd$ recall that $\|\lltRCFA^*_{\vmPrior,\vmInd}\|_1=d$ almost surely and thereby
$|\lltCFAng^{\matTr}\lltRCFA^*_{\vmPrior,\vmInd,[q-1]}|\le d\|\lltCFAng\|_\infty<\pi/2$.
With $\lltCFInt^*_{\vmPrior,\vmInd}(\lltCFAng)=a+\cpxI b=re^{\cpxI\alpha}$ and
since the cosine is even, non-negative and decreasing on $[0,\pi/2]$ we have
\begin{align*}
r=|\lltCFInt^*_{\vmPrior,\vmInd}(\lltCFAng)|=\left|\Erw\left[\exp\left(\cpxI\lltCFAng^{\matTr}\lltRCFA^*_{\vmPrior,\vmInd,[q-1]}\right)\right]\right|
\ge |a|=a=\sum_\lltCFA\pr[\lltRCFA^*_{\vmPrior,\vmInd}=\lltCFA]\cos(|\lltCFAng^{\matTr}\lltCFA_{[q-1]}|)\ge\cos(d\|\lltCFAng\|_\infty)\textrm{.}
\end{align*}
With $a>0$ the choice of $\alpha$ with $|\alpha|<\pi/2$ is unique, and thereby $\lltCFLog^*_{\vmPrior,\vmInd}(\lltCFAng)=\ln(\lltCFInt^*_{\vmPrior,\vmInd}(\lltCFAng))=\ln(r)+\cpxI\alpha$ is well-defined.
By direct computation we obtain that the first derivatives of $\lltCFLog^*_{\vmPrior,\vmInd}$ at $0_{[q-1]}$ vanish, the second partial derivatives yield $\Sigma^*_{\vmPrior,\vmInd}$ and for the third partial derivatives we get
\begin{align*}
\frac{\partial\lltCFLog^*_{\vmPrior,\vmInd}(\lltCFAng)}{\prod_{i\in[3]}\partial\lltCFAng(\vmCol_i)}&=-\cpxI\Erw\left[c(\lltRCFA,\vmCol)\frac{\exp\left(\cpxI\lltCFAng^{\matTr}\left(\sum_i\lltRCFA_i\right)_{[q-1]}\right)}{\lltCFInt^*_{\vmPrior,\vmInd}(\lltCFAng)^3}\right]\textrm{,}\\
c(\lltCFA,\vmCol)&=\lltCFA_1(\vmCol_1)\left(\lltCFA_1(\vmCol_2)-\lltCFA_2(\vmCol_2)\right)\left(\lltCFA_1(\vmCol_3)+\lltCFA_2(\vmCol_3)-2\lltCFA_3(\vmCol_3)\right)\textrm{,}
\end{align*}
with i.i.d. copies $\lltRCFA_i\sim\lltRCFA^*_{\vmPrior,\vmInd}$ for $i\in[3]$ and $\vmCol\in[q-1]^3$.
This gives $|c(\lltCFA,\vmCol)|\le 2d^3$ uniformly, hence the third partial derivatives can be upper bounded by $\frac{2d^3}{|\lltCFInt^*_{\vmPrior,\vmInd}(\lltCFAng)|^3}\le\frac{2d^3}{\cos(d\|\lltCFAng\|_\infty)^3}$, which proves the first assertion using Taylor's theorem.

Recall that the covariance $\Sigma^*_{\vmPrior,\vmInd}$ is positive semi-definite.
With $\lltRCFA^\circ_{\vmPrior,\vmInd}(\vmCol)=(\lltRCFA^*_{\vmPrior,\vmInd}(\vmCol)-\lltECFA_{\vmPrior,\vmInd}^*(\vmCol))\in[-d,d]$ almost surely for all $\vmCol\in[q-1]$ and $\vmPrior$, this gives $|v^\matTr\lltRCFA^\circ_{\vmPrior,\vmInd}|\le d$ and hence $v^{\matTr}\Sigma^*_{\vmPrior,\vmInd}v\le d^2$ for all $v\in\RR^{q-1}$ with $\|v\|_2=1$.

For the lower bound notice that there exists $\vmCol^*\in[q-1]$ with $|v(\vmCol^*)|=\|v\|_\infty\ge 1/\sqrt{q-1}$ by equivalence of norms and $\|v\|_2=1$. With $\cE=\{e_{[q],q},e_{[q],\vmCol^*}\}$ this yields
\begin{align*}
v^{\matTr}\Sigma^*_{\vmPrior,\vmInd}v
&=
\Erw\left[\left(v^{\matTr}\lltRCFA^\circ_{\vmPrior,\vmInd,[q-1]}\right)^2\vecone\left\{\lltRCFA^*_{\vmPrior,\vmInd}\not\in\cE\right\}\right]
+\pr\left[\lltRCFA^*_{\vmPrior,\vmInd}\in\cE\right]\Erw\left[\left(v^{\matTr}\lltRCFA^\circ_{\vmPrior,\vmInd,[q-1]}\right)^2\middle|\lltRCFA^*_{\vmPrior,\vmInd}\in\cE\right]\textrm{.}
\end{align*}
For the conditional expectation we use $\alpha_1x_1^2+\alpha_2x_2^2\ge\alpha_1\alpha_2(x_1-x_2)^2$ valid for all $\alpha\in\cP([2])$ and $x\in\RR^2$, which gives
\begin{align*}
v^{\matTr}\Sigma^*_{\vmPrior,\vmInd}v
&\ge\pr\left[\lltRCFA^*_{\vmPrior,\vmInd}\in\cE\right]\pr\left[\lltRCFA_{\vmPrior,\vmInd}^*=e_{[q],\vmCol^*}\middle|\lltRCFA^*_{\vmPrior,\vmInd}\in\cE\right]\pr\left[\lltRCFA_{\vmPrior,\vmInd}^*=e_{[q],q}\middle|\lltRCFA^*_{\vmPrior,\vmInd}\in\cE\right]\|v\|_\infty^2\\
&\ge\pr\left[\lltRCFA_{\vmPrior,\vmInd}^*=e_{[q],\vmCol^*}\right]\pr\left[\lltRCFA_{\vmPrior,\vmInd}^*=e_{[q],q}\right]\|v\|_\infty^2
\ge\frac{\mu_\vmInd(\vmCol^*1_{[d]})\vmPrior(\vmCol^*)^{d}\mu_\vmInd(q1_{[d]})\vmPrior(q)^{d}}{q-1}\textrm{.}
\end{align*}
With $\mu_{\min,\vmInd}=\min_{\vmCol\in[q]}\mu_\vmInd(\vmCol 1_{[d]})$ and $\lltEpsP=\min_{\vmPrior\in\cP^*}\min_{\vmCol\in[q]}\vmPrior(\vmCol)\in(0,1)$ this gives $v^{\matTr}\Sigma^*_{\vmPrior,\vmInd}v\ge\frac{\mu_{\min,\vmInd}^2\lltEpsP^{2d}}{q-1}$ for all required $v$, $\vmPrior$ and $\vmInd$.
\end{proof}
With Lemma \ref{llt_global_atoms} we are ready to establish the subgaussian tails.
\begin{corollary}\label{llt_global_subgaussian}
There exists a constant $c\in\RR_{>0}$ such that $|\lltCFInt_{\svmInd,\vmPrior}(\lltCFAng)|\le\exp(-c\|\lltCFAng\|_2^2n)$ for all $n\in\ZZ_{>0}$, $\svmInd\in\vmIndSp_n$, $\vmPrior\in\cP^*$ and $\lltCFAng\in\lltBox^*$.
\end{corollary}
\begin{proof}
Fix $\vmInd\in\vmIndSp^\circ$ and $r\in(0,1)$ sufficiently small to obtain a good approximation of $\ln(\lltCFInt^*_{\vmPrior,\vmInd}(\lltCFAng))$ for $\lltCFAng\in\cB_r(0_{[q-1]})$ using Lemma \ref{llt_global_atoms}, say $|\ln(\lltCFInt^*_{\vmPrior,\vmInd}(\lltCFAng))+\frac{1}{2}\lltCFAng^{\matTr}\Sigma^*_{\vmPrior,\vmInd}\lltCFAng|\le\frac{1}{4}c_\vmInd\|\lltCFAng\|_2^2$ uniformly for all $\vmPrior\in\cP^*$, i.e.~the relative error is at most $1/2$. With $\ln(\lltCFInt^*_{\vmPrior,\vmInd}(\lltCFAng))=a+\cpxI b$ this gives $|a+\frac{1}{2}\lltCFAng^{\matTr}\Sigma^*_{\vmPrior,\vmInd}\lltCFAng|\le\frac{1}{4}\lltCFAng^{\matTr}\Sigma^*_{\vmPrior,\vmInd}\lltCFAng$, so $|\lltCFInt^*_{\vmPrior,\vmInd}(\lltCFAng)|=e^a\le\exp(-\frac{1}{4}\lltCFAng^{\matTr}\Sigma^*_{\vmPrior,\vmInd}\lltCFAng)\le\exp(-\frac{1}{4}c_\vmInd\|\lltCFAng\|_2^2)$, and further $|\lltCFInt_{\svmInd,\vmPrior}(\lltCFAng)|\le\exp(-\frac{1}{4}c_\vmInd\epsGen\|\lltCFAng\|_2^2n)$.
For $\lltCFAng\in\lltBox^*\setminus\cB_r(0_{[q-1]})$ we use the constant $c'\in\RR_{>0}$ from Lemma \ref{llt_global_tails} to obtain $|\lltCFInt_{\svmInd,\vmPrior}(\lltCFAng)|\le\exp(-\frac{c'}{r^2}r^2n)\le\exp(-\frac{c'}{r^2}\|\lltCFAng\|_2^2n)$, and taking the minimum of the two choices completes the proof.
\end{proof}
Since the assertion indicates that the integral is of order $\sqrt{n}^{-(q-1)}$,
we fix $\varepsilon_{\mathrm{a},n}=c^*\sqrt{\ln(n)/n}$ for some large $c^*\in\RR_{>0}$ and set $\cB_{\mathrm{a},n}=\cB_{\varepsilon_{\mathrm{a},n}}(0_{[q-1]})$, since then
\begin{align*}
\left|\int_{\lltBox^*}\vecone\{\lltCFAng\not\in\cB_{\mathrm{a},n}\}\lltCFInt_{\svmInd,\vmPrior}(\lltCFAng)\intD\lltCFAng\right|
\le n^{-cc^{*2}}\prod_\vmCol h(\vmCol)=o\left(\sqrt{n}^{-q}\right)
\end{align*}
uniformly in $\svmInd\in\vmIndSp_n$ and $\vmPrior\in\cP^*$.
The remainder of the proof is dedicated to the material contributions $\cB_{\mathrm{a},n}$.
First, we extend Lemma \ref{llt_global_atoms} to $\lltCFInt_{\svmInd,\vmPrior}$. For this purpose let
\begin{align*}
\lltCFLog_{\svmInd,\vmPrior}(\lltCFAng)=\ln(\lltCFInt_{\svmInd,\vmPrior}(\lltCFAng))
=n\Erw\left[\ln\left(\lltCFInt^*_{\vmPrior,\vmRInd_\svmInd}(\lltCFAng)\right)\right]\textrm{,}
\end{align*}
which is defined for sufficiently small $\lltCFAng$ (depending on $\svmInd$, $\vmPrior$) as shown in Lemma \ref{llt_global_atoms}.
\begin{lemma}\label{llt_global_second_order}
Uniformly for all $\svmInd\in\vmIndSp_n$, $\vmPrior\in\cP^*$ and $\lltCFAng\in\cB_{\mathrm{a},n}$ we have
\begin{align*}
\left|\lltCFLog_{\svmInd,\vmPrior}(\lltCFAng)+\frac{1}{2}\Erw[\vd_\svmInd]n\lltCFAng^{\matTr}\Sigma_{\svmInd,\vmPrior}\lltCFAng\right|=O\left(\thiB_n\sqrt{\frac{\ln(n)^3}{n}}\right)\textrm{.}
\end{align*}
Further, there exists $c\in(0,1)$ with $c\le\|\Sigma_{\svmInd,\vmPrior}^{-1}\|_2^{-1}\le\|\Sigma_{\svmInd,\vmPrior}\|_2\le c^{-1}$ for all $\vmPrior\in\cP^*$, $\svmInd\in\vmIndSp_n$ and $n\in\ZZ_{>0}$.
\end{lemma}
\begin{proof}
For given $\svmInd\in\vmIndSp_n$ the maximum degree $d_{\max,\svmInd}$ satisfies
$d_{\max,\svmInd}^3\le n\Erw[\vd_\svmInd^3]\le n\thiB_n$, so we have $d_{\max,\svmInd}\le(n\thiB_n)^{1/3}=o(\sqrt{n/\ln(n)})$ and further $d_{\max,\svmInd}\varepsilon_{\mathrm{a},n}=o(1)$ uniformly in $\svmInd$ (and $\vmPrior$). So, with Lemma \ref{llt_global_atoms} and equivalence of norms we obtain $n_0\in\ZZ_{>0}$, $c\in\RR_{>0}$ such that
\begin{align*}
\left|\ln(\lltCFInt^*_{\vmPrior,\vmInd}(\lltCFAng))+\frac{1}{2}\lltCFAng^{\matTr}\Sigma^*_{\vmPrior,\vmInd}\lltCFAng\right|\le cd_\vmInd^3\varepsilon_{\mathrm{a},n}^3
\end{align*}
for all $\lltCFAng\in\cB_{\mathrm{a},n}$, $\vmPrior\in\cP^*$, $\vmInd$ in $\svmInd\in\vmIndSp_n$ and $n\in\ZZ_{\ge n_0}$. With the definition of $\lltCFLog_{\svmInd,\vmPrior}$, $\Sigma_{\svmInd,\vmPrior}$ and the triangle inequality this gives
\begin{align*}
\left|\lltCFLog_{\svmInd,\vmPrior}(\lltCFAng)+\frac{1}{2}\Erw[\vd_\svmInd]n\lltCFAng^{\matTr}\Sigma_{\svmInd,\vmPrior}\lltCFAng\right|\le nc\Erw[\vd_\svmInd^3]\varepsilon_{\mathrm{a},n}^3\le nc\thiB_n\varepsilon_{\mathrm{a},n}^3=O\left(\thiB_n\sqrt{\frac{\ln(n)^3}{n}}\right)
\end{align*}
uniformly in $\lltCFAng\in\cB_{\mathrm{a},n}$, $\vmPrior\in\cP^*$ and $\svmInd\in\vmIndSp_n$.
Finally, using Lemma \ref{llt_global_atoms}, $\vmInd\in\vmIndSp^\circ$, and with $v\in\RR^{q-1}$, $\|v\|_2=1$, we have
\begin{align*}
\frac{\epsGen c_\vmInd}{\secB}\le\frac{1}{\Erw[\vd_\svmInd]}\pr[\vmRInd_\svmInd=\vmInd]c_\vmInd\le v^{\matTr}\Sigma_{\svmInd,\vmPrior}v\le\frac{\Erw[\vd_\svmInd^2]}{\Erw[\vd_\svmInd]}\le\frac{\secB}{\epsGen}
\end{align*}
uniformly for all $\svmInd$, $\vmPrior$.
\end{proof}
For the sake of transparency let $\lltCFInt_{\svmInd,\vmPrior}(\lltCFAng)=\lltCFInt_{\mathrm{r},\svmInd,\vmPrior}(\lltCFAng)+\cpxI\lltCFInt_{\mathrm{i},\svmInd,\vmPrior}(\lltCFAng)$,
$\lltCFLog_{\svmInd,\vmPrior}(\lltCFAng)=\lltCFLog_{\mathrm{r},\svmInd,\vmPrior}(\lltCFAng)+\cpxI \lltCFLog_{\mathrm{i},\svmInd,\vmPrior}(\lltCFAng)$ be the decompositions into real and imaginary part, so in particular
$\lltCFInt_{\mathrm{r},\svmInd,\vmPrior}(\lltCFAng)=\exp\left(\lltCFLog_{\mathrm{r},\svmInd,\vmPrior}(\lltCFAng)\right)\cos(\lltCFLog_{\mathrm{i},\svmInd,\vmPrior}(\lltCFAng))$.
As directly implied by the left hand side of the inversion formula we only need to evaluate the integral over $\lltCFInt_{\mathrm{r},\svmInd,\vmPrior}(\lltCFAng)$ since the integral over $\lltCFInt_{\mathrm{i},\svmInd,\vmPrior}(\lltCFAng)$ vanishes. Using $|a|\le|z|$ for $z=a+\cpxI b$ we obtain the bound on the tails for the real part from the bound on the tails of the complex integral.
Further, Lemma \ref{llt_global_second_order} suggests that $\lltCFLog_{\mathrm{r},\svmInd,\vmPrior}(\lltCFAng)=-\frac{1}{2}\Erw[\vd_{\svmInd}]n\lltCFAng^{\matTr}\Sigma_{\svmInd,\vmPrior}\lltCFAng+O(\thiB_n\sqrt{\ln(n)^3/n})$ and $\lltCFLog_{\mathrm{i},\svmInd,\vmPrior}(\lltCFAng)=O(\thiB_n\sqrt{\ln(n)^3/n})=o(1)$. Using $\cos(x)=1-O(x^2)$ and $\exp(x)=1+O(x)$ this gives
\begin{align*}
\int_{\lltBox^*}\lltCFInt_{\mathrm{r},\svmInd,\vmPrior}(\lltCFAng)\intD\lltCFAng
&=\left(1+O\left(\thiB_n\sqrt{\frac{\ln(n)^3}{n}}\right)\right)\int_{\cB_{\mathrm{a},n}}\exp\left(-\frac{1}{2}\Erw[\vd_{\svmInd}]n\lltCFAng^{\matTr}\Sigma_{\svmInd,\vmPrior}\lltCFAng\right)\intD\lltCFAng+o\left(\sqrt{n}^{-q}\right)\textrm{.}
\end{align*}
Rescaling with $v=\sqrt{\Erw[\vd_{\svmInd}]n}\lltCFAng$ gives a Gaussian integral.
With $\vv_{\svmInd,\vmPrior}\sim\cN(0,\Sigma_{\svmInd,\vmPrior}^{-1})$ reflecting the corresponding normal and $\cV_{\svmInd,\vmPrior}=\sqrt{\Erw[\vd_{\svmInd}]n}\cB_{\mathrm{a},n}$ the corresponding event this gives
\begin{align*}
\int_{\lltBox^*}\lltCFInt_{\mathrm{r},\svmInd,\vmPrior}(\lltCFAng)\intD\lltCFAng
&=\left(1+O\left(\thiB_n\sqrt{\frac{\ln(n)^3}{n}}\right)\right)\frac{\sqrt{(2\pi)^{q-1}}}{\sqrt{\Erw[\vd_{\svmInd}]n}^{q-1}\sqrt{\det(\Sigma_{\svmInd,\vmPrior})}}\pr[\vv_{\svmInd,\vmPrior}\in\cV_{\svmInd,\vmPrior}]+o\left(\sqrt{n}^{-q}\right)\textrm{.}
\end{align*}
With $\Erw[\vd_{\svmInd}]\ge\epsGen$ we obtain $c\in\RR_{>0}$ with $\|v\|_2\ge c\sqrt{n}\varepsilon_{\mathrm{a},n}$ for all $v\not\in\cV_{\svmInd,\vmPrior}$ and $\svmInd$, $\vmPrior$.
Further, with Lemma \ref{llt_global_second_order} we have $c'\in(0,1)$ to bound the eigenvalues of $\Sigma_{\svmInd,\vmPrior}$ uniformly, suggesting the existence of constants $c$, $c'\in\RR_{>0}$ such that $\pr[\vv_{\svmInd,\vmPrior}\not\in\cV_{\svmInd,\vmPrior}]\le c'\exp(-c\varepsilon_{\mathrm{a},n}^2n)$ uniformly for all $\vmPrior\in\cP^*$, $\svmInd\in\vmIndSp_n$ and $n\in\ZZ_{>0}$.
With the definition of $\varepsilon_{\mathrm{a},n}$ this gives $\pr[\vv_{\svmInd,\vmPrior}\not\in\cV_{\svmInd,\vmPrior}]\le c'n^{-cc^{*2}}$, hence for some fixed large $c^*\in\RR_{>0}$ we have $\pr[\vv_{\svmInd,\vmPrior}\not\in\cV_{\svmInd,\vmPrior}]=o(\sqrt{n}^{-1})$ uniformly. Now, since $\Erw[\vd_{\svmInd}]=\Theta(1)$ uniformly and $\det(\Sigma_{\svmInd,\vmPrior})=\Theta(1)$ uniformly the dominant contribution is of order $\sqrt{n}^{-(q-1)}$. Hence, extracting the material part gives
\begin{align*}
\int_{\lltBox^*}\lltCFInt_{\mathrm{r},\svmInd,\vmPrior}(\lltCFAng)\intD\lltCFAng
&=\left(1+O\left(\thiB_n\sqrt{\frac{\ln(n)^3}{n}}\right)+o(\sqrt{n}^{-1})\right)\frac{\sqrt{(2\pi)^{q-1}}}{\sqrt{\Erw[\vd_{\svmInd}]n}^{q-1}\sqrt{\det(\Sigma_{\svmInd,\vmPrior})}}\textrm{.}
\end{align*}
Here, the fact that $\thiB_n\ge 1$ completes the proof.
\subsection{Proof of Theorem \ref{llt_local}}
We split the proof into two parts. The first part is dedicated to a local limit theorem for $\vmRPrior_\svmInd=\iota_{\svmInd}^{-1}(\lltRCFR_\svmInd)\in\lltLatP_\svmInd$ around $u_{[q]}$, and in the second part we translate the result to $\lltRCFR_\svmInd$.
\begin{proposition}\label{llt_local_p}
Uniformly for all $\svmInd\in\vmIndSp_n$ and $\vmPrior\in\lltLatP_\svmInd$ with $\|\vmPrior-u_{[q]}\|_2<r_n$ we have
\begin{align*}
\pr[\vmRPrior_\svmInd=\vmPrior]
&=\left(1+O\left(\left(\sqrt{\frac{\ln(n)}{n}}^3+r_n^3\right)\thiB_nn\right)\right)
\frac{\det(\Sigma_\svmInd^{-1}))\prod_\vmCol h(\vmCol)}{\sqrt{\Erw[\vd_\svmInd ]n}^{q-1}}\phi_\svmInd\left(\sqrt{\Erw[\vd_\svmInd ]n}(q\vmPrior-1_{[q]})_{[q-1]}\right)\textrm{,}
\end{align*}
where $\phi_\svmInd$ denotes the density of $\cN(0,\Sigma_\svmInd^{-1})$.
\end{proposition}
\begin{proof}
With Theorem \ref{llt_global} we expand the exponent to second order, control the resulting errors and proceed analogously for the determinant by expanding to zeroth order.
For given $\svmInd$ and $\vmPrior$ let $\alpha_{\svmInd}(\vmPrior)=\Erw[\vd_\svmInd ]^{-1}\Erw[\KL{\mu_{\vmPrior,\svmInd }}{\mu_{\svmInd }}]$, notice that $\KL{\mu_{\vmPrior,\vmInd}}{\mu_\vmInd}=0$ if $d_\vmInd=0$ and further
\begin{align*}
\alpha_\svmInd(\vmPrior)=\Erw[\vd_\svmInd ]^{-1}\Erw\left[\sum_{\vmCol\in[q]}\lltECFA^*_{\vmPrior,\vmRInd_\svmInd}(\vmCol)\ln(\vmPrior(\vmCol))-\ln(Z_{\vmPrior,\vmRInd_\svmInd})\right]=\sum_\vmCol\lltECFR_{\svmInd,\vmPrior}(\vmCol)\ln(\vmPrior(\vmCol))-\frac{\Erw[\ln(Z_{\vmPrior,\vmRInd_\svmInd})]}{\Erw[\vd_\svmInd]}\textrm{,}
\end{align*}
thereby removing dependencies on the product spaces $[q]^d$.
Recall that $\vmPrior=u_{[q]}$ is the unique global minimizer of $\alpha_\svmInd$ as discussed in Section \ref{vom}, hence the first derivatives vanish. For the sake of completeness and later use we provide the derivatives.
For transparency we use the shorthand $f^{(\vmCol_1,\dots,\vmCol_k)}=\frac{\partial f}{\partial \vmPrior(\vmCol_k)\cdots\partial \vmPrior(\vmCol_1)}$ to denote the $k$-th partial derivatives of the extension of a map $f:\cP^{\circ}([q])\rightarrow\RR$ to $\RR_{>0}^q$.
With the shorthand $\Sigma^\circ_{\vmPrior,\vmInd}=\Cov(\lltRCFA^*_{\vmPrior,\vmInd})$ (as opposed to $\Sigma^*_{\vmPrior,\vmInd}=\Cov(\lltRCFA^*_{\vmPrior,\vmInd,[q-1]})$ in Section \ref{llt_global_prop_proof})
and for $\vmInd\in\vmIndSp$ with $d_\vmInd>0$ the derivatives at $\vmPrior\in\cP^\circ([q])$ are
\begin{align*}
Z_{\vmPrior,\vmInd}^{(\vmCol)}=\frac{Z_{\vmPrior,\vmInd}\lltECFA^*_{\vmPrior,\vmInd}(\vmCol)}{\vmPrior(\vmCol)}\textrm{, }
\lltECFA^{*(\vmCol_2)}_{\vmPrior,\vmInd}(\vmCol_1)=\frac{\Sigma^\circ_{\vmPrior,\vmInd,\vmCol_1,\vmCol_2}}{\vmPrior(\vmCol_2)}\textrm{.}
\end{align*}
With $L_\vmPrior=(\ln(\vmPrior(\vmCol)))_\vmCol $ and $\Sigma^\circ_{\svmInd,\vmPrior}=\frac{1}{\Erw[\vd_\svmInd ]}\Erw[\Sigma^\circ_{\vmPrior,\vmRInd_\svmInd}]$, i.e.~$\Sigma_{\svmInd,\vmPrior}=(\Sigma^\circ_{\svmInd,\vmPrior})_{[q-1]\times[q-1]}$, this gives
\begin{align*}
\lltECFR^{(\vmCol_2)}_{\svmInd,\vmPrior}(\vmCol_1)=\frac{\Sigma^\circ_{\svmInd,\vmPrior,\vmCol_1,\vmCol_2}}{\vmPrior(\vmCol_2)}\textrm{, }
\alpha^{(\vmCol)}_{\svmInd}(\vmPrior)=\frac{\sum_{\vmCol'}\Sigma^\circ_{\svmInd,\vmPrior,\vmCol,\vmCol'}L_{\vmPrior,\vmCol'}}{\vmPrior(\vmCol)}\textrm{.}
\end{align*}
For given $\vmCol\in[q]^3$ we have $\Sigma_{\vmPrior,\vmInd,\vmCol_1,\vmCol_2}^{\circ(\vmCol_3)}=\frac{S^\circ_{\vmPrior,\vmInd,\vmCol}}{\vmPrior(\vmCol_3)}$ with $S^\circ_{\vmPrior,\vmInd,\vmCol}=\Erw[\prod_i(\lltRCFA^*_{\vmPrior,\vmInd}(\vmCol_i)-\lltECFA^*_{\vmPrior,\vmInd}(\vmCol_i))]$.
Hence, with $S_{\svmInd,\vmPrior}^\circ=\frac{1}{\Erw[\vd_\svmInd ]}\Erw[S^\circ_{\vmPrior,\vd_\svmInd }]$ the derivatives on the next level are given by
\begin{align*}
\lltECFR^{(\vmCol_2,\vmCol_3)}_{\svmInd,\vmPrior}(\vmCol_1)&=\frac{S^\circ_{\svmInd,\vmPrior,\vmCol}-\delta_{\vmCol_2,\vmCol_3}\Sigma^{\circ}_{\svmInd,\vmPrior,\vmCol_1,\vmCol_2}}{\vmPrior(\vmCol_2)\vmPrior(\vmCol_3)}\textrm{, }
\alpha^{(\vmCol_1,\vmCol_2)}_{\svmInd}(\vmPrior)
=\frac{\Sigma^\circ_{\svmInd,\vmPrior,\vmCol}+\sum_{\vmCol'}(S^\circ_{\svmInd,\vmPrior,\vmCol,\vmCol'}-\delta_{\vmCol_1,\vmCol_2}\Sigma^\circ_{\svmInd,\vmPrior,\vmCol_1,\vmCol'})L_{\vmPrior,\vmCol'}}{\vmPrior(\vmCol_1)\vmPrior(\vmCol_2)}\end{align*}
using the Kronecker symbol.
For $\vmCol\in[q]^4$ we have $S_{\vmPrior,\vmInd,\vmCol_1,\vmCol_2,\vmCol_3}^{\circ(\vmCol_4)}=\frac{F^\circ_{\vmPrior,\vmInd,\vmCol}}{\vmPrior(\vmCol_4)}$ with
\begin{align*}
F^{\circ}_{\vmPrior,\vmInd,\vmCol}=\Erw\left[\prod_{i\in[4]}(\lltRCFA^*_{\vmPrior,\vmInd}(\vmCol_i)-\lltECFA^*_{\vmPrior,\vmInd}(\vmCol_i))\right]-\sum_{i\in[3]}\Sigma^\circ_{\vmPrior,\vmInd,\vmCol_i,\vmCol_4}\Sigma^\circ_{\vmPrior,\vmInd,\vmCol_{[3]\setminus\{i\}}}\textrm{,}
\end{align*}
recall that $\Sigma^\circ_{\vmPrior,\vmInd}$ is symmetric, so e.g.~$\Sigma^\circ_{\vmPrior,\vmInd,\vmCol_1,\vmCol_2}=\Sigma^\circ_{\vmPrior,\vmInd,\vmCol_2,\vmCol_1}$, and hence $F^{\circ}_{\vmPrior,\vmInd}$ is symmetric in that $F^\circ_{\vmPrior,\vmInd,\vmCol\circ\sigma}=F^\circ_{\vmPrior,\vmInd,\vmCol}$ for all permutations $\sigma:[4]\rightarrow[4]$.
With $F^\circ_{\svmInd,\vmPrior,\vmCol}=\frac{1}{\Erw[\vd_\svmInd ]}\Erw[F^\circ_{\vmPrior,\vmRInd_\svmInd ,\vmCol}]$ and $\vmCol\in[q]^3$ this yields
\begin{align*}
\alpha_{\svmInd}^{(\vmCol)}(\vmPrior)&=\frac{2S^\circ_{\svmInd,\vmPrior,\vmCol}+
\sum_{\vmCol'}F^\circ_{\svmInd,\vmPrior,\vmCol,\vmCol'}L_{\vmPrior,\vmCol'}
-\alpha^{(\vmCol)}_{2,\svmInd}(\vmPrior)
+2\delta_{\vmCol_1,\vmCol_2}\delta_{\vmCol_1,\vmCol_3}\sum_{\vmCol'}\Sigma^\circ_{\svmInd,\vmPrior,\vmCol_1,\vmCol'}L_{\vmPrior,\vmCol'}}{\vmPrior(\vmCol_1)\vmPrior(\vmCol_2)\vmPrior(\vmCol_3)}\textrm{,}\\
\alpha^{(\vmCol)}_{2,\svmInd}(\vmPrior)&=\sum_{i<j}\delta_{\vmCol_i,\vmCol_j}\left(\Sigma^\circ_{\svmInd,\vmPrior,\vmCol_i,\vmCol_{[3]\setminus\{i,j\}}}+\sum_{\vmCol'}S^\circ_{\svmInd,\vmPrior,\vmCol_i,\vmCol_{[3]\setminus\{i,j\}},\vmCol'}L_{\vmPrior,\vmCol'}\right)
\end{align*}
Recall that all $\Sigma^\circ_{\svmInd,\vmPrior}$, $S^\circ_{\svmInd,\vmPrior}$ and $F^\circ_{\svmInd,\vmPrior}$ are invariant to permutations of the indicies and further $\sum_{\vmCol'}\Sigma^\circ_{\svmInd,\vmPrior,\vmCol,\vmCol'}=0$, $\sum_{\vmCol'}S^\circ_{\svmInd,\vmPrior,\vmCol,\vmCol'}=0$, $\sum_{\vmCol'}F^\circ_{\svmInd,\vmPrior,\vmCol,\vmCol'}=0$ for all suitable $\vmCol$ respectively, i.e.~the ``column'' sum for any given dimension and choice of remaining indicies vanishes.
On the one hand, since $L_{\vmPrior}\equiv\ln(q^{-1})$ for $\vmPrior=u_{[q]}$ all inner products involving $L_\vmPrior$ vanish, i.e. all first derivatives vanish and further $(\alpha_{\svmInd}^{(\vmCol)}(u_{[q]}))_\vmCol =q^2\Sigma^\circ_{\svmInd,u_{[q]}}$.
On the other hand, this means that the inner product with $L_\vmPrior$ equals the inner product with $L_\vmPrior+c1_{[q]}$ for any $c\in\RR$. Since we discuss $\alpha_{\svmInd}$ locally around $\vmPrior=u_{[q]}$ we choose $c=\ln(q)$ and let $L^\circ_\vmPrior=L_\vmPrior+c1_{[q]}=(\ln(\vmPrior(\vmCol)/q^{-1}))_\vmCol$.

Now, since $r_n=o(n^{-1/3})$ we can fix any small compact neighbourhood $\cP^*\subseteq\cP^\circ([q])$ of $u_{[q]}$ to obtain $\cB_n\subseteq\cP^*$ with $\cB_n=\cB_{r_n}(u_{[q]})$ for $n\in\ZZ_{\ge n_0}$ and some $n_0\in\ZZ_{>0}$, so in particular we get some $\lltEpsP\in(0,q^{-1})$ close to $q^{-1}$ with $\vmPrior(\vmCol)\ge\lltEpsP$ for all $\vmPrior\in\cB_n$ and $n\in\ZZ_{\ge n_0}$.
This e.g.~takes care of the denominator of the third partial derivatives.

Recall from the proof of Lemma \ref{llt_global_atoms} that $\|\lltRCFA^*_{\vmPrior,\vmInd}-\lltECFA^*_{\vmPrior,\vmInd}\|_\infty\le d_\vmInd$ almost surely for all $\vmInd\in\vmIndSp$ and $\vmPrior\in\cP([q])$, and further that $\Erw[\vd_\svmInd ]\ge\epsGen$. This gives $|\Sigma^\circ_{\vmPrior,\vmInd,\vmCol}|\le d_\vmInd^2$, $|S^\circ_{\vmPrior,\vmInd,\vmCol}|\le d_\vmInd^3$ and $|F^{\circ}_{\vmPrior,\vmInd,\vmCol}|\le 4d_\vmInd^4$ for all suitable $\vmCol$ respectively, and uniformly in $\vmPrior$ and $\vmInd$, so
$|\Sigma^\circ_{\svmInd,\vmPrior,\vmCol}|\le\epsGen^{-1}\Erw[\vd_\svmInd ^2]\le\epsGen^{-1}\secB$,
$|S^\circ_{\svmInd,\vmPrior,\vmCol}|\le\epsGen^{-1}\thiB_n$ and
$|F^\circ_{\svmInd,\vmPrior,\vmCol}|\le 4\epsGen^{-1}\thiB_nd_{\max,\svmInd}$ uniformly in $\vmPrior$ and $\svmInd$, where we recall $d_{\max,\svmInd}\in\ZZ_{>0}$ from the proof of Lemma \ref{llt_global_second_order}, in particular that $d_{\max,\svmInd}^3\le n\thiB_n$ and hence $d_{\max,\svmInd}r_n=o(1)$.
Finally, due to the restriction to $\cP^*$ and with equivalence of norms we get a global constant $c_{\mathrm{l}}\in\RR_{>1}$ with $\|L^\circ_\vmPrior\|_1\le c_{\mathrm{l}}\|\vmPrior-u_{[q]}\|_2\le c_{\mathrm{l}}r_n$ for all $\vmPrior\in\cB_n$ and $n\in\ZZ_{\ge n_0}$. Using these bounds we get $\alpha^{(\vmCol)}_{\svmInd}(\vmPrior)=O(\thiB_n)$ with the order given by the first contribution and uniformly in $\vmCol\in[q]^3$, $\svmInd\in\vmIndSp_n$ and $\vmPrior\in\cB_n$.
Now, Taylor's theorem with equivalence of norms yields
\begin{align*}
\left|\Erw[\vd_\svmInd ]n\alpha_{\svmInd}(\vmPrior)-\frac{q^2}{2}\Erw[\vd_\svmInd ]n(\vmPrior-u_{[q]})^{\matTr}\Sigma^\circ_{\svmInd,u_{[q]}}(\vmPrior-u_{[q]})\right|=O\left(n\thiB r_n^3\right)
\end{align*}
uniformly in $\vmPrior\in\cB_n$ and $\svmInd\in\vmIndSp_n$.
Now, for $\vmCol\in[q-1]$ let $b_\vmCol =e_{[q],\vmCol}-e_{[q],q}$ denote a basis of $1_{[q]}^{\perp}$ and further $B=(b_{\vmCol}(\vmCol^*))_{\vmCol^*\in[q],\vmCol\in[q-1]}$ the corresponding transformation, then $(\vmPrior-u_{[q]})=B(\vmPrior-u_{[q]})_{[q-1]}$ and the precision matrix of our normal distribution is given by $B^{\matTr}\Sigma^\circ_{\svmInd,u_{[q]}}B$.
On the other hand, since $1_{[q]}$ is both a row and column eigenvector of $\Sigma^\circ_{\svmInd,u_{[q]}}$ with eigenvalue $0$ we have $\Sigma^\circ_{\svmInd,u_{[q]}}=B\Sigma_{\svmInd,u_{[q]}}B^{\matTr}=B\Sigma_{\svmInd}B^{\matTr}$.
Hence, with $B^{\matTr}B=I_{[q-1]}+1_{[q-1]}1_{[q-1]}^{\matTr}$ we obtain $B^{\matTr}\Sigma^\circ_{\svmInd,u_{[q]}}B=\Sigma_{\svmInd}$.

With the exponent in place we turn to the asymptotics of the determinant $f_{\svmInd}(\vmPrior)=\det(\Sigma_{\svmInd,\vmPrior})$.
Interpreting the matrix entries $\Sigma_{\svmInd,\vmPrior,\vmCol}=\Sigma^\circ_{\svmInd,\vmPrior,\vmCol}$, $\vmCol\in[q]^2$, as functions in $\vmPrior$ the discussion above shows that $|\Sigma_{\svmInd,\vmPrior,\vmCol}-\Sigma_{\svmInd,u_{[q]},\vmCol}|=O(\thiB r_n)$ uniformly in $\vmCol\in[q]^2$, $\vmPrior\in\cB_n$ and $\svmInd\in\vmIndSp_n$.
Due to the assumption $nr_n^2=\Omega(1)$ we have $\thiB r_n=O(n\thiB r_n^3)$.
Using the Leibniz formula to view $f_{\svmInd}(\vmPrior)$ as a polynomial and taking derivatives in $\Sigma_{\svmInd,\vmPrior,\vmCol}$, $\vmCol\in[q-1]^2$ (as opposed to $\vmPrior(\vmCol)$) we obtain
\begin{align*}
|f_{\svmInd}(\vmPrior)-f_{\svmInd}(u_{[q]})|=O\left(\left(\sum_{\vmCol}\left|\Sigma_{\svmInd,\vmPrior,\vmCol}\right|\right)^{q-2}\sum_\vmCol |\Sigma_{\svmInd,\vmPrior,\vmCol}-\Sigma_{\svmInd,u_{[q]},\vmCol}|\right)
=O\left(n\thiB r_n^3\right)
\end{align*}
since we already showed that $\left|\Sigma_{\svmInd,\vmPrior,\vmCol}\right|\le\epsGen^{-1}\secB$ uniformly in $\vmCol\in[q-1]^2$, $\vmPrior\in\cB_n$ and $\svmInd\in\vmIndSp_n$. With $\det(\Sigma_{\svmInd})=\Theta(1)$ uniformly in $\svmInd\in\vmIndSp_n$ as derived in the proof of Theorem \ref{llt_global} the assertion follows.
\end{proof}
In the remainder of the proof we approximate $\iota_{\svmInd}$ to first order and control the errors in the exponent, while the remainder already agrees with the assertion in Theorem \ref{llt_local}.
In the proof of Proposition \ref{llt_local_p} we have already established the first and second partial derivatives of $\iota_{\svmInd}(\vmPrior)=\lltECFR_{\svmInd,\vmPrior}$, and further the bounds on $\vmPrior$, $S^\circ_{\svmInd,\vmPrior}$ and $\Sigma^\circ_{\svmInd,\vmPrior}$ required to derive $\lltECFR_{\svmInd,\vmPrior}^{(\vmCol_2,\vmCol_3)}(\vmCol_1)=O(\thiB_n)$ uniformly for all $\vmCol\in[q]^3$, $\vmPrior\in\cB_n$ and $\svmInd\in\vmIndSp_n$. Hence, Taylor's theorem yields
\begin{align*}
\|\lltECFR_{\svmInd,\vmPrior}-\tilde\lltCFR_{\svmInd,\vmPrior}\|=O\left(\thiB_nr_n^2\right)\textrm{, }
\tilde\lltCFR_{\svmInd,\vmPrior}=\lltECFR_{\svmInd}+q\Sigma^\circ_{\svmInd,u_{[q]}}(\vmPrior-u_{[q]})\textrm{,}
\end{align*}
uniformly for all $\vmPrior\in\cB_n$ and $\svmInd\in\vmIndSp_n$, with $\thiB_nr_n^2=o(\frac{1}{r_nn})=o(\sqrt{n}^{-1})$ since $r_n=\Omega(\sqrt{n}^{-1})$.
Recall that the eigenvalues of $\Sigma^\circ_{\svmInd,u_{[q]}}$ can be upper bounded by $\Erw[\vd_\svmInd ^2]/\Erw[\vd_\svmInd ]$, and are hence uniformly bounded, and further that $1_{[q]}$ is an eigenvector of $\Sigma^\circ_{\svmInd,u_{[q]}}$ (with eigenvalue $0$), so $\Sigma^\circ_{\svmInd,u_{[q]}}$ maps $1_{[q]}^\perp$ into $1_{[q]}^{\perp}$. This shows that (for large enough $n$) the linear approximation $\tilde\lltCFR_{\svmInd,\vmPrior}$ is in $\cP^\circ([q])$ with $\|\tilde\lltCFR_{\svmInd,\vmPrior}-\lltECFR_{\svmInd}\|=O(r_n)$.
With $\tilde\lltCFR_{\svmInd,\vmPrior}\in\cP^\circ([q])$ we can safely project onto the first $(q-1)$ coordinates to obtain $(\tilde\lltCFR_{\svmInd,\vmPrior}-\lltECFR_{\svmInd})_{[q-1]}=q\tilde B^{\matTr}\Sigma^\circ_{\svmInd,u_{[q]}}B(\vmPrior-u_{[q]})_{[q-1]}$ with $B$ introduced in the proof of Proposition \ref{llt_local_p} and $\tilde B=(e_{[q],\vmCol}(\vmCol^*))_{\vmCol^*\in[q],\vmCol\in[q-1]}$.
With $\Sigma^\circ_{\svmInd,u_{[q]}}=B\Sigma_{\svmInd}B^{\matTr}$, $\tilde B^{\matTr}B=I_{[q-1]}$ and $B^{\matTr}B=I_{[q-1]}+1_{[q-1]}1_{[q-1]}^{\matTr}$ we have $(\tilde\lltCFR_{\svmInd,\vmPrior}-\lltECFR_{\svmInd})_{[q-1]}=q\Sigma_{\svmInd}(\vmPrior-u_{[q]})_{[q-1]}$.
For one, we already obtained uniform bounds on the eigenvalues of $\Sigma_{\svmInd}$ in the proof of Theorem \ref{llt_global} and hence a constant $c\in(0,1)$ such that $c\|(\vmPrior-u_{[q]})_{[q-1]}\|\le\|(\tilde\lltCFR_{\svmInd,\vmPrior}-\lltECFR_{\svmInd})_{[q-1]}\|\le c^{-1}\|(\vmPrior-u_{[q]})_{[q-1]}\|$ for all $\vmPrior\in\cP^*$, $\svmInd\in\vmIndSp_n$, $n\in\ZZ_{>0}$ whenever $\tilde\lltCFR_{\svmInd,\vmPrior}\in\cP([q])$. More than that, this map is invertible and allows to substitute $q(\vmPrior-u_{[q]})_{[q-1]}$ in the exponent $\alpha_{\mathrm{p},\svmInd}(\vmPrior)$, i.e.
\begin{align*}
\alpha_{\mathrm{p},\svmInd}(\vmPrior)=-\frac{1}{2}\Erw[\vd_\svmInd ]n q^2(\vmPrior-u_{[q]})_{[q-1]}^{\matTr}\Sigma_{\svmInd}(\vmPrior-u_{[q]})_{[q-1]}
=-\frac{1}{2}\Erw[\vd_\svmInd ]n(\tilde\lltCFR_{\svmInd,\vmPrior}-\lltECFR_{\svmInd})_{[q-1]}^{\matTr}\Sigma_{\svmInd}^{-1}(\tilde\lltCFR_{\svmInd,\vmPrior}-\lltECFR_{\svmInd})_{[q-1]}\textrm{.}
\end{align*}
Using the bounds on the fluctuations $\|\lltECFR_{\svmInd,\vmPrior}-\tilde\lltCFR_{\svmInd,\vmPrior}\|$ this gives
\begin{align*}
|\alpha_{\mathrm{r},\svmInd}(\vmPrior)-\alpha_{\mathrm{p},\svmInd}(\vmPrior)|
&=\frac{1}{2}\Erw[\vd_\svmInd ]n\left|(\lltECFR_{\svmInd,\vmPrior}-\tilde\lltCFR_{\svmInd,\vmPrior})_{[q-1]}^{\matTr}\Sigma_{\svmInd}^{-1}(\lltECFR_{\svmInd,\vmPrior}-\tilde\lltCFR_{\svmInd,\vmPrior})_{[q-1]}\right|
=O\left((\thiB_n)^2r_n^4n\right)=o\left(\thiB_nr_n^3n\right)\textrm{,}\\
\alpha_{\mathrm{r},\svmInd}(\vmPrior)&=-\frac{1}{2}\Erw[\vd_\svmInd ]n(\lltECFR_{\svmInd,\vmPrior}-\lltECFR_{\svmInd})_{[q-1]}^{\matTr}\Sigma_{\svmInd}^{-1}(\lltECFR_{\svmInd,\vmPrior}-\lltECFR_{\svmInd})_{[q-1]}\textrm{,}
\end{align*}
since $r_n^2n=\Omega(1)$ and hence $\thiB_nr_n=O(\thiB_nr_n^3n)=o(1)$.
Hence, the relative error made by approximating the exponent $\alpha_{\mathrm{\vmPrior},\svmInd}(\vmPrior)$ with $\alpha_{\mathrm{r},\svmInd}(\vmPrior)$ is strictly smaller than the existing bound.

To be thorough, fix a sequence of radii $r_n$ with $r_n^2n=\Omega(1)$ and $\thiB r_n^3n=o(1)$, bounding the fluctuations $\|\lltCFR-\lltECFR_{\svmInd}\|_2$ (as opposed to $\|\vmPrior-u_{[q]}\|$ which was the case so far) and let $\cB_{\svmInd}=\cB_{r_n}(\lltECFR_{\svmInd})$ be the corresponding ball.
Let $c^*\in\RR_{>0}$ be large and $r'_n=c^*r_n$, then our existing results hold for $r'_n$ respectively $\cB'_n=\cB_{r'_n}(u_{[q]})$.
The two-sided bounds for $\tilde\lltCFR_{\svmInd,\vmPrior}$ imply that all $\tilde\lltCFR\in\tilde{\cB}_{\svmInd}$, $\tilde{\cB}_{\svmInd}=\cB_{2r_n}(\lltECFR_{\svmInd})$, are covered by $\cB'_n$. But since the fluctuations $\|\lltECFR_{\svmInd,\vmPrior}-\tilde\lltCFR_{\svmInd,\vmPrior}\|=o(\sqrt{n}^{-1})=o(r_n)$ are very small, all $\lltCFR\in\cB_{\svmInd}$ are covered by $\cB'_n$, which completes the proof since $r_n=\Theta(r'_n)$.

\section{Assignment distributions} \label{assd}

This section extends the results from Section \ref{vom} and Section \ref{llt}.
For this purpose fix a non-trivial family $(d_\vmInd,\mu_\vmInd)_{\vmInd\in\vmIndSp}$ satisfying {\bf SPAN} and a sequence $(\vmIndSp_n)_{n\in\ZZ_{>0}}\subseteq\vmIndSp^n$ satisfying {\bf GEN}, {\bf VAR} and {\bf SKEW}.
For $\vmIndP\in\vmIndPSp$ and $\vmPrior\in\cP([q])$ let the expected assignment distribution $\acEAD_{\vmIndP,\vmPrior}\in\cP(\acSupp)$ be given by $\acEAD_{\vmIndP,\vmPrior}(\vmInd,\vmAss)=\vmIndP(\vmInd)\mu_{\vmPrior,\vmInd}(\vmAss)$ for $(\vmInd,\vmAss)\in\acSupp$ with $\acSupp=\{(\vmInd,\vmAss):\vmInd\in\vmIndSp,\vmAss\in[q]^{d_\vmInd}\}$. As before, we usually omit the subscript $\vmPrior$ if $\vmPrior=u_{[q]}$.
We consider the distributions $\acEAD_{\vmIndP,\vmPrior}$ elements of $\acADSp=\{\alpha\in\cP(\acSupp):\alpha\vert_1\in\vmIndPSp\}$ equipped with the metric
\begin{align*}
\acADD(\acAD,\acAD')
=\sum_{\vmInd\in\vmIndSp}(d_\vmInd+1)\sum_{\vmAss\in[q]^{d_\vmInd}}\left|\acAD(\vmInd,\vmAss)-\acAD'(\vmInd,\vmAss)\right|
\end{align*}
for $\acAD$, $\acAD'\in\acADSp$.
Further, based on the insights from Section \ref{llt} we let $\lltECFR_{\vmPrior,\vmIndP}=\iota_\vmIndP(\vmPrior)$ to stress the interpretation as expected colour frequencies and recall that $\lltECFR_{\vmPrior,\vmIndP_\svmInd}=\lltECFR_{\vmPrior,\svmInd}$ for non-trivial sequences $\svmInd\in\vmIndSp^n$.

For $n\in\ZZ_{>0}$, non-trivial $\svmInd\in\vmIndSp^n$ and $\lltHEA\in\prod_\lltVL[q]^{d_\lltVL}$ let $\acAD_\lltHEA\in\acADSp$ denote the assignment frequencies, i.e.
\begin{align*}
\acAD_\lltHEA(\vmInd,\vmAss)=\frac{1}{n}\left|\left\{\lltVL\in[n]:\vmInd_\lltVL=\vmInd,\vmAss_\lltVL=\vmAss\right\}\right|
\end{align*}
for $(\vmInd,\vmAss)\in\acSupp$, where we keep the dependence on $\svmInd$ implicit.
Finally, for $\vmPrior\in\cP([q])$ and $\lltCFR$ in the support of $\lltRCFR_{\vmPrior,\svmInd}$ we let $\lltRHEA_{\vmPrior,\svmInd,\lltCFR}=(\lltRHEA_{\vmPrior,\svmInd}|\lltRCFR_{\vmPrior,\svmInd}=\lltCFR)$, further $\acRAD_{\vmPrior,\svmInd}=\acAD_{\lltRHEA_{\vmPrior,\svmInd}}$ and $\acRAD_{\vmPrior,\svmInd,\lltCFR}=\acAD_{\lltRHEA_{\vmPrior,\svmInd,\lltCFR}}$.
The main result of this section ensures that, given colour frequencies $\lltCFR$ close to their expectation $\lltECFR_\vmIndP$ and a sequence $\svmInd$ with frequencies $\vmIndP_\svmInd$ close to the reference $\vmIndP$, the assignment distribution $\acRAD_{\svmInd,\lltCFR}$ is close to the expected unconditional assignment distribution $\acEAD_\vmIndP$ of the reference $\vmIndP$ with very high probability.
\begin{proposition}\label{ac_main_prop}
Fix $(\vmIndSp_n)_n\subseteq\vmIndSp^n$ satisfying {\bf SPAN}, {\bf GEN}, {\bf VAR} and {\bf SKEW}, a reference distribution $\vmIndP\in\vmIndPSp$ and $\varepsilon\in\RR_{>0}$.
Then there exists $\delta$, $c$, $c'\in\RR_{>0}$ such that
for all $n\in\ZZ_{>0}$, all $\svmInd\in\vmIndSp_n$ with $\vmIndP_\svmInd\in\cB_\delta(\vmIndP)$, and
all $\lltCFR\in\cB_\delta(\lltECFR_{\vmIndP})$ in the support of $\lltRCFR_{\svmInd}$ we have
\begin{align*}
\pr\left[\acADD(\acRAD_{\svmInd,\lltCFR},\acEAD_\vmIndP)\ge\varepsilon\right]\le c'\exp(-cn)\textrm{.}
\end{align*}
\end{proposition}
The proof of Proposition \ref{ac_main_prop} builds intuition for the construction in Section \ref{vom}, in particular for the distributions $\mu_{\vmPrior,\vmInd}$.
\subsection{Proof strategy}
Consider the specified $(\vmIndSp_n)_n$, $\vmIndP$ and $\varepsilon$ fixed in the remainder.
Further, for given $\lltEpsP\in\RR_{>0}$ let $\cP^*=\{\vmPrior\in\cP([q]):\min_{\vmCol\in[q]}\vmPrior(\vmCol)\ge\lltEpsP\}$ and for $\svmInd\in\vmIndSp_n$ let $\cP^*_\svmInd$ be the set of distributions $\vmPrior\in\cP^*$ with $\lltECFR_{\vmPrior,\svmInd}$ in the support of $\lltRCFR_\svmInd$.
Our first result is a corollary to Proposition \ref{llt_global_prop}.
\begin{fact}\label{ac_cf_max}
For fixed $\lltEpsP\in\RR_{>0}$ and uniformly over all $\svmInd\in\vmIndSp_n$ and $\vmPrior\in\cP^*_\svmInd$ we have
\begin{align*}
\pr[\lltRCFR_{\vmPrior,\svmInd}=\lltECFR_{\vmPrior,\svmInd}]=\Theta\left(\sqrt{n}^{-(q-1)}\right)\textrm{.}
\end{align*}
\end{fact}
The proof is postponed to Section \ref{ac_cf_max}. The next result deals with the unconditional case for the adjusted measures.
\begin{lemma}\label{ac_loc_uncond}
There exist constants $\delta$, $c$, $c'\in\RR_{>0}$ such that 
for all $n\in\ZZ_{>0}$, all $\svmInd\in\vmIndSp_n$ with $\vmIndP_{\svmInd}\in\cB_{\delta}(\vmIndP)$ and all $\vmPrior\in\cP^*$ we have
\begin{align*}
\pr\left[\acADD\left(\acRAD_{\vmPrior,\svmInd},\acEAD_{\vmPrior,\svmInd}\right)\ge\varepsilon\right]\le c'\exp(-cn)\textrm{.}
\end{align*}
\end{lemma}
The proof is postponed to Section \ref{ac_loc_uncond_proof}.
Combining Fact \ref{ac_cf_max} and Lemma \ref{ac_loc_uncond} allows to derive bounds for the conditional probability, still for the adjusted measures.
For this purpose let $\svmInd\in\vmIndSp_n$ and $\vmPrior\in\cP^*_\svmInd$ we let $\lltRHEA^*_{\vmPrior,\svmInd}=\lltRHEA_{\vmPrior,\svmInd,\lltECFR_{\vmPrior,\svmInd}}$
and $\acRAD^*_{\vmPrior,\svmInd}=\acAD_{\lltRHEA^*_{\vmPrior,\svmInd}}$, and
further use $\acEAD_{\vmPrior,\svmInd}=\acEAD_{P_{\svmInd},\vmPrior}$ for consistency.
\begin{lemma}\label{ac_loc}
For all $\lltEpsP\in\RR_{>0}$ there exist constants $\delta$, $c$, $c'\in\RR_{>0}$ such that
for all $n\in\ZZ_{>0}$, $\svmInd\in\vmIndSp_n$ with $\vmIndP_\svmInd\in\cB_\delta(\vmIndP)$ and all $\vmPrior\in\cP^*_\svmInd$  we have
\begin{align*}
\pr\left[\acADD(\acRAD^*_{\vmPrior,\svmInd},\acEAD_{\vmPrior,\svmInd})\ge\varepsilon\right]\le c'\exp(-cn)\textrm{.}
\end{align*}
\end{lemma}
The proof is postponed to Section \ref{ac_loc_proof}.
Finally, the following fact justifies the discussion of the adjusted measures.
\begin{fact}\label{ac_dist_change}
For all $n\in\ZZ_{>0}$, $\svmInd\in\vmIndSp_n$ and $\vmPrior\in\cP^*_\svmInd$ the assignments $\acRAD_{\svmInd,\lltECFR_{\vmPrior,\svmInd}}$ and $\acRAD^*_{\vmPrior,\svmInd}$ have the same law.
\end{fact}
The proof is postponed to Section \ref{ac_dist_change_proof}.
Lemma \ref{ac_loc} combined with Fact \ref{ac_dist_change} yielsd concentration results for the assignment distributions given their colour frequencies $\lltECFR_{\vmPrior,\svmInd}$.
Hence, the only part left to show is that the local concentration points $\acEAD_{\vmPrior,\svmInd}$ are close to the reference $\acEAD_{\vmIndP}$ if $\vmIndP_{\svmInd}$ is close to $\vmIndP$ and $\vmPrior$ is close to $u_{[q]}$.
The details are presented in Section \ref{ac_main_prop_proof}.
\subsection{Proof of Fact \ref{ac_cf_max}}\label{ac_cf_max_proof}
By construction we satisfy the assumptions of Proposition \ref{llt_global_prop}.
But as thoroughly discussed e.g.~in Section \ref{llt} we also have $\Erw[\vd_\svmInd]$, $\det(\Sigma_{\svmInd,\vmPrior})=\Theta(1)$ uniformly in $\svmInd\in\vmIndSp_n$ which completes the proof.
Notice that $\pi\in\cP^*\cap\lltLatP_\svmInd$ suffices to show this result, since Proposition \ref{llt_global_prop} then implies that $\lltECFR_{\vmPrior,\svmInd}$ is in the support of $\lltRCFR_\svmInd$ for sufficiently large $n$.
\subsection{Proof of Lemma \ref{ac_loc_uncond}}\label{ac_loc_uncond_proof}
We consider $\lltEpsP$ fixed throughout this section.
Further, fix $\varepsilon_{\mathrm{f}}$, $\delta\in(0,1)$, $n\in\ZZ_{>0}$, $\svmInd\in\vmIndSp_n$ with $\vmIndP_{\svmInd}\in\cB_\delta(\vmIndP)$ and $\vmPrior\in\cP^*$.
Further let
\begin{align*}
\vmIndSp_-=\{\vmInd\in\vmIndSp:\pr[\vmRInd_{\vmIndP}=\vmInd]<\varepsilon_{\mathrm{f}}\}\textrm{, }
\vmIndSp_+=\vmIndSp\setminus\vmIndSp_-
\end{align*}
denote the partition into measures of low frequency and high frequency respectively.
Notice that $|\vmIndSp_+|\in\ZZ_{>0}$ for $\varepsilon_{\mathrm{f}}$ sufficiently small and let
$d_{\max}=\max\{d_\vmInd:\vmInd\in\vmIndSp_+\}$. With $\Delta=\acADD$ we consider the corresponding split
\begin{align*}
\Delta\left(\acRAD_{\vmPrior,\svmInd},\acEAD_{\vmPrior,\svmInd}\right)
&=\Delta_-\left(\acRAD_{\vmPrior,\svmInd},\acEAD_{\vmPrior,\svmInd}\right)+\Delta_+\left(\acRAD_{\vmPrior,\svmInd},\acEAD_{\vmPrior,\svmInd}\right)\textrm{,}\\
\Delta_{\pm}\left(\acRAD_{\vmPrior,\svmInd},\acEAD_{\vmPrior,\svmInd}\right)
&=\sum_{\vmInd\in\vmIndSp_\pm,\vmAss\in\vmAssSp_\vmInd}(d_\vmInd+1)\left|\acRAD_{\vmPrior,\svmInd}(\vmInd,\vmAss)-\acEAD_{\vmPrior,\svmInd}(\vmInd,\vmAss)\right|\textrm{.}
\end{align*}
Recall that $\acRAD_{\vmPrior,\svmInd}\vert_1=\acEAD_{\vmPrior,\svmInd}\vert_1=\vmIndP_{\svmInd}$,
so with $\acRAD_{\vmPrior,\svmInd,\vmInd}\in\cP(\vmAssSp_\vmInd)$ given by $\acRAD_{\vmPrior,\svmInd,\vmInd}(\vmAss)=\acRAD_{\vmPrior,\svmInd}(\vmInd,\vmAss)/\vmIndP_\svmInd(\vmInd)$ for $\vmAss\in\vmAssSp_\vmInd$ and $\vmInd$ in the support of $\vmIndP_\svmInd$ denoting the law conditional to $\vmRInd_\svmInd=\vmInd$ we have
\begin{align*}
\Delta_{\pm}\left(\acRAD_{\vmPrior,\svmInd},\acEAD_{\vmPrior,\svmInd}\right)
&=\sum_{\vmInd\in\vmIndSp_\pm}\vmIndP_\svmInd(\vmInd)(d_\vmInd+1)\left\|\acRAD_{\vmPrior,\svmInd,\vmInd}-\mu_\vmInd\right\|_1\textrm{.}
\end{align*}
Since we can uniformly bound the norm and $\vmIndP_{\svmInd}\in\cB_\delta(\vmIndP)$ we obtain
\begin{align*}
\Delta_-\left(\acRAD_{\vmPrior,\svmInd},\acEAD_{\vmPrior,\svmInd}\right)
\le 2\sum_{\vmInd\in\vmIndSp_-}\vmIndP_\svmInd(\vmInd)(d_\vmInd+1)
<2\sum_{\vmInd\in\vmIndSp_-}\vmIndP(\vmInd)(d_\vmInd+1)+2\delta\textrm{.}
\end{align*}
Since $\vmIndP\in\vmIndPSp$ has a finite first moment the latter expectation tends to $0$ for $\varepsilon_{\mathrm{f}}\rightarrow 0$, so for $\varepsilon_{\mathrm{f}}$ sufficiently small and $\delta=\varepsilon_{\mathrm{f}}/2$ we have $\Delta_-\left(\acRAD_{\vmPrior,\svmInd},\acEAD_{\vmPrior,\svmInd}\right)\le\varepsilon/2$ almost surely and thereby
\begin{align*}
p=\pr\left[\Delta\left(\acRAD_{\vmPrior,\svmInd},\acEAD_{\vmPrior,\svmInd}\right)\ge\varepsilon\right]
\le\pr\left[\Delta_+\left(\acRAD_{\vmPrior,\svmInd},\acEAD_{\vmPrior,\svmInd}\right)\ge\varepsilon/2\right]
=\pr\left[\Delta_+\left(\acRAD_{\vmPrior,\svmInd},\acEAD_{\vmPrior,\svmInd}\right)\ge p_+c\varepsilon\right]
\end{align*}
with $c=(2p_+)^{-1}$ and $p_+=\pr[\vmRInd_\svmInd \in\vmIndSp_+]$.
Writing both sides of $\Delta_+\left(\acRAD_{\vmPrior,\svmInd},\acEAD_{\vmPrior,\svmInd}\right)\ge p_+c\varepsilon$ as expectations with respect to $\vmRInd_\svmInd$ yields
\begin{align*}
p\le\sum_{\vmInd\in\vmIndSp_+}\pr\left[(d_\vmInd+1)\left\|\acRAD_{\vmPrior,\svmInd,\vmInd}-\mu_{\vmPrior,\vmInd}\right\|_1\ge c\varepsilon\right]
\le\sum_{\vmInd\in\vmIndSp_+}\pr\left[\left\|\acRAD_{\vmPrior,\svmInd,\vmInd}-\mu_{\vmPrior,\vmInd}\right\|_1\ge c'\varepsilon\right]
\end{align*}
with $c'=c/(d_{\max}+1)$ and where we notice that $\vmIndP_{\svmInd}(\vmInd)>\vmIndP(\vmInd)-\delta\ge\varepsilon_{\mathrm{f}}/2$ for all $\vmInd\in\vmIndSp_+$ in the support of $\vmRInd_\svmInd$.
Recall that for all $\vmAss\in\vmAssSp_\vmInd$ the frequency $\acRAD_{\vmPrior,\svmInd,\vmInd}(\vmAss)$ is a sum of $\vmIndP_{\svmInd}(\vmInd)n$ i.i.d.~random variables with expectation $\mu_{\vmPrior,\vmInd}(\vmAss)$, so Hoeffding's inequality for $\varepsilon'\in\RR_{\ge 0}$ yields
\begin{align*}
\pr\left[\left|\acRAD_{\vmPrior,\svmInd,\vmInd}(\vmAss)-\mu_{\vmPrior,\vmInd}(\vmAss)\right|\ge\varepsilon'\right]
\le 2\exp\left(-2\varepsilon'^2 \vmIndP_{\svmInd}(\vmInd)n\right)
\le 2\exp\left(-\varepsilon_{\mathrm{f}}\varepsilon'^2n\right)\textrm{.}
\end{align*}
Standard arguments yield a bound for the $\|\cdot\|_\infty$ norm and further
\begin{align*}
\pr\left[\left\|\acRAD_{\vmPrior,\svmInd,\vmInd}-\mu_{\vmPrior,\vmInd}\right\|_1\ge\varepsilon'\right]
\le 2q^{d_{\max}}\exp\left(-\frac{\varepsilon_{\mathrm{f}}}{q^{2d_{\max}}}\varepsilon'^2n\right)\textrm{.}
\end{align*}
This uniform bounds directly implies
\begin{align*}
p\le 2|\vmIndSp_+|q^{d_{\max}}\exp\left(-\frac{\varepsilon_{\mathrm{f}}c'^2}{q^{2d_{\max}}}\varepsilon^2n\right)
\end{align*}
and thereby completes the proof.
Finally, notice that $\vmPrior\in\cP^*$ was not required.
\subsection{Proof of Lemma \ref{ac_loc}}\label{ac_loc_proof}
For fixed $\lltEpsP\in\RR_{>0}$ and with Fact \ref{ac_cf_max} we obtain $c\in\RR_{>0}$ and such that
\begin{align*}
\pr\left[\acADD(\acRAD^*_{\vmPrior,\svmInd},\acEAD_{\vmPrior,\svmInd})\ge\varepsilon\right]
\le c\sqrt{n}^{q-1}\pr\left[\acADD(\acRAD_{\vmPrior,\svmInd},\acEAD_{\vmPrior,\svmInd})\ge\varepsilon\right]
\end{align*}
for all sufficiently large $n$, $\svmInd\in\vmIndSp_n$ and $\vmPrior\in\cP^*_\svmInd$.
Now, we summon Lemma \ref{ac_loc_uncond} to obtain $\delta$, $c_1$, $c_2$ such that for all $\svmInd\in\vmIndSp_n$ with $\vmIndP_{\svmInd}\in\cB_\delta(\vmIndP)$ and $\vmPrior\in\cP^*_\svmInd$ we have
\begin{align*}
\pr\left[\acADD(\acRAD^*_{\vmPrior,\svmInd},\acEAD_{\vmPrior,\svmInd})\ge\varepsilon\right]
\le cc_1\sqrt{n}^{q-1}\exp\left(-c_2n\right)
=\exp\left(-\left(c_2-n^{-1}\ln\left(cc_1\sqrt{n}^{q-1}\right)\right)n\right)\textrm{.}
\end{align*}
Hence, we fix a constant $c\in(0,c_2)$ and $n^*$ sufficiently large such that for all $n\in\ZZ_{\ge n^*}$ the leading coefficient in the exponent exceeds $c$, so for all $\svmInd\in\vmIndSp_n$ with $\vmIndP_{\svmInd}\in\cB_\delta(\vmIndP)$ and $\vmPrior\in\cP^*_\svmInd$ we have
\begin{align*}
\pr\left[\acADD(\acRAD^*_{\vmPrior,\svmInd},\acEAD_{\vmPrior,\svmInd})\ge\varepsilon\right]
\le\exp\left(-cn\right)\textrm{.}
\end{align*}
Finally, we set $c'=\exp(cn^*)$ which ensures that $c'\exp(-cn)\ge 1$ for all $n<n^*$ and hence the assertion holds.
\subsection{Proof of Fact \ref{ac_dist_change}}\label{ac_dist_change_proof}
For assignments $\lltHEA$ with $\lltCFR_\lltHEA=\lltECFR_{\vmPrior,\svmInd}$, using Lemma \ref{llt_global_measure_switch} and $\lltCFR=\lltECFR_{\vmPrior,\svmInd}$ we have
\begin{align*}
\pr\left[\lltRHEA_{\svmInd,\lltCFR}=\lltHEA\right]
&=\frac{\pr[\lltRHEA_{\svmInd}=\lltHEA]}{\pr[\lltRCFR_\svmInd=\lltCFR]}
=\frac{\exp\left(-n\Erw\left[\KL{\mu_{\vmPrior,\vmRInd_\svmInd }}{\mu_{\vmRInd_\svmInd }}\right]\right)\pr[\lltRHEA_{\vmPrior,\svmInd}=\lltHEA]}{\exp\left(-n\Erw\left[\KL{\mu_{\vmPrior,\vmRInd_\svmInd }}{\mu_{\vmRInd_\svmInd }}\right]\right)\pr[\lltRCFR_{\vmPrior,\svmInd}=\lltCFR]}
=\pr\left[\lltRHEA^*_{\vmPrior,\svmInd}=\lltHEA\right]\textrm{,}
\end{align*}
which directly translates to the distributions and thereby completes the proof.
\subsection{Proof of Proposition \ref{ac_main_prop}}\label{ac_main_prop_proof}
Fix suitable $(\vmIndSp_n)_n$, $\vmIndP$ and $\varepsilon$.
Further, fix some small $\lltEpsP\in(0,1)$ and let $\eta=\iota^{-1}$.
Since $\eta$ is continuous due to Proposition \ref{pp_mr_cont}, the preimage $\eta^{-1}(\cB_{\lltEpsP}(u_{[q]}))$ is open and $(\vmIndP,\lltECFR_{\vmIndP})\in\eta^{-1}(\cB_{\lltEpsP}(u_{[q]}))$ since $\eta(\vmIndP,\lltECFR_{\vmIndP})=u_{[q]}$.
From this we obtain $\delta_1\in\RR_{>0}$ such that $\cB_{\delta_1}(\vmIndP)\times\cB_{\delta_1}(\lltECFR_{\vmIndP})\subseteq\eta^{-1}(\cB_{\lltEpsP}(u_{[q]}))$.

With Lemma \ref{ac_loc} we obtain $\delta_2$, $c$, $c'$ such that for all $n\in\ZZ_{>0}$, all $\svmInd\in\vmIndSp_n$ with $\vmIndP_{\svmInd}\in\cB_{\delta_2}(\vmIndP)$ and all $\vmPrior\in\cB_{\lltEpsP}(u_{[q]})$ with $\lltECFR_{\vmPrior,\svmInd}$ in the support of $\lltRCFR_\svmInd$ we have
\begin{align*}
\pr\left[\acADD(\acRAD^*_{\vmPrior,\svmInd},\acEAD_{\vmPrior,\svmInd})\ge\varepsilon\right]
\le c'\exp\left(-cn\right)\textrm{.}
\end{align*}
Using Fact \ref{ac_dist_change} and $\lltCFR=\lltECFR_{\vmPrior,\svmInd}$ immediately yields
\begin{align*}
\pr\left[\acADD(\acRAD_{\svmInd,\lltCFR},\acEAD_{\vmPrior,\svmInd})\ge\varepsilon\right]
\le c'\exp\left(-cn\right)\textrm{.}
\end{align*}
Now, let $\delta=\min(\delta_1,\delta_2)$, $n\in\ZZ_{>0}$, $\svmInd\in\vmIndSp_n$ with $\vmIndP_{\svmInd}\in\cB_\delta(\vmIndP)$ and $\rho\in\cB_\delta(\lltECFR_{\vmIndP})$ in the support $\lltRCFR_\svmInd$.
By the above we have $\vmPrior=\eta(\vmIndP_{\svmInd},\rho)\in\cB_{\lltEpsP}(u_{[q]})$ and further with $\vmIndP_{\svmInd}\in\cB_\delta(\vmIndP)$ we obtain
\begin{align*}
\pr\left[\acADD(\acRAD_{\svmInd,\lltCFR},\acEAD_{\vmPrior,\svmInd})\ge\varepsilon\right]
\le c'\exp\left(-cn\right)\textrm{.}
\end{align*}
Now, we're left to show that the conditional assignment distribution expectations $\acEAD_{\vmPrior,\svmInd}$ are close to the unconditional expectation $\acEAD_{\vmIndP}$.
For this purpose notice that by using the triangle inequality and normalization of $\mu_{\vmPrior,\vmInd}$ we have
\begin{align*}
\acADD(\acEAD_{\vmPrior,\svmInd},\acEAD_{\vmIndP})
&=\sum_{\vmInd,\vmAss}(d_\vmInd+1)\left|\vmIndP_{\svmInd}(\vmInd)\mu_{\vmPrior,\vmInd}(\vmAss)-\vmIndP(\vmInd)\mu_{\vmInd}(\vmAss)\right|
\le\vmIndPD(\vmIndP_{\svmInd},\vmIndP)+\sum_{\vmInd}\vmIndP(\vmInd)(d_{\vmInd}+1)\|\mu_{\vmPrior,\vmInd}-\mu_{\vmInd}\|_1\textrm{.}
\end{align*}
For $\vmIndSp^*\subseteq\vmIndSp$ sufficiently large (but still finite) we use the uniform bounds for the norm on the $\vmIndSp\setminus\vmIndSp^*$ contribution to the expectation and an upper bound $d_{\mathrm{cap}}\in\RR_{>0}$ for the degrees of $\vmIndSp^*$.
Further, since $\vmPrior\mapsto(\mu_{\vmPrior,\vmInd})_{\vmInd\in\vmIndSp^*}$ is continuous with $u_{[q]}\mapsto(\mu_\vmInd)_{\vmInd\in\vmIndSp^*}$ we can also control the norm on $\vmIndSp^*$ and thereby find $\delta\in\RR_{>0}$ such that for all $\svmInd\in\vmIndSp_n$ with $\vmIndP_{\svmInd}\in\cB_\delta(\vmIndP)$ and $\vmPrior\in\cB_\delta(u_{[q]})$ we have
\begin{align*}
\acADD(\acEAD_{\vmPrior,\svmInd},\acEAD_{\vmIndP})
&\le\delta+2\Erw\left[(\vd_{\vmIndP}+1)\vecone\{\vmRInd_{\vmIndP}\not\in\vmIndSp^*\}\right]
+(d_{\mathrm{cap}}+1)\Erw\left[\vecone\{\vmRInd_{\vmIndP}\in\vmIndSp^*\}\|\mu_{\vmPrior,\vmRInd_{\vmIndP}}-\mu_{\vmRInd_{\vmIndP}}\|_1\right]
<\frac{1}{3}\varepsilon+\frac{1}{3}\varepsilon+\frac{1}{3}\varepsilon=\varepsilon\textrm{.}
\end{align*}
Finally, we combine the two arguments to obtain the result as follows.
First, choose $\delta_1\in(0,1)$ sufficiently small such that $\acADD(\acEAD_{\vmPrior,\svmInd},\acEAD_{\vmIndP})<\varepsilon/2$ for all $\vmPrior\in\cB_{\delta_1}(u_{[q]})$ and $\svmInd\in\vmIndSp_n$ with $\vmIndP_{\svmInd}\in\cB_{\delta_1}(\vmIndP)$.
Further, for $\varepsilon/2$ and $\lltEpsP=\delta_1$ the first argument provides $\delta_2$, $c$, $c'$ such that
for all $n\in\ZZ_{>0}$, $\svmInd\in\vmIndSp_n$ with $\vmIndP_{\svmInd}\in\cB_{\delta_2}(\vmIndP)$ and all $\rho\in\cB_{\delta_2}(\lltECFR_{\vmIndP})$ in the support of $\lltRCFR_\svmInd$ we have $\vmPrior=\eta(\vmIndP_{\svmInd},\rho)\in\cB_{\delta_1}(u_{[q]})$ and
\begin{align*}
\pr\left[\acADD(\acRAD_{\svmInd,\lltCFR},\acEAD_{\vmPrior,\svmInd})\ge\varepsilon/2\right]
\le c'\exp\left(-cn\right)\textrm{.}
\end{align*}
Now, let $\delta=\min(\delta_1,\delta_2)$. Then for all $n\in\ZZ_{>0}$, all $\svmInd\in\vmIndSp_n$ with $\vmIndP_{\svmInd}\in\cB_{\delta}(\vmIndP)$ and all $\rho\in\cB_\delta(\lltECFR_{\vmIndP})$ in the support of $\lltRCFR_\svmInd$ we have $\vmPrior=\eta(\vmIndP_{\svmInd},\rho)\in\cB_{\delta_1}(u_{[q]})$,
which gives $\acADD(\acEAD_{\vmPrior,\svmInd},\acEAD_{\vmIndP})<\varepsilon/2$, so using the triangle inequality $\acADD(\acRAD_{\svmInd,\lltCFR},\acEAD_{\vmIndP})\ge\varepsilon$ implies $\acADD(\acRAD_{\svmInd,\lltCFR},\acEAD_{\vmPrior,\svmInd})\ge\varepsilon/2$ and thereby
\begin{align*}
\pr\left[\acADD(\acRAD_{\svmInd,\lltCFR},\acEAD_{\vmIndP})\ge\varepsilon\right]
\le\pr\left[\acADD(\acRAD_{\svmInd,\lltCFR},\acEAD_{\vmPrior,\svmInd})\ge\varepsilon/2\right]
\le c'\exp(-cn)\textrm{.}
\end{align*}

\section{Degree distributions} \label{dd}

Recall the degree distributions introduced in Section \ref{Sec_rfg}, let $\d=\Erw[\vd]$, $\k=\Erw[\vk]$, $\m_n=\d n/\k$,  $\vt^*_n=(\vm,(\vd_i)_{i\in[n]},(\vk_i)_{i\in[\vm]})$ and $\cT^*_n$ denote the support of $\vt^*_n$.
For $t\in\cT^*_n$ we use $t=(m_t,d_t,k_t)$ to specify the components.
Further, let $\cE_n$ denote the event
\begin{align*}
\sum_{i=1}^n\vd_i=\sum_{i=1}^{\vm}\vk_i
\end{align*}
and $\cN$ the values of $n$ with $\pr[\cE_n]>0$. Finally, for $n\in\cN$ let $\vt_n=(\vt^*_n|\cE_n)$ denote the degree sequences for which $\G$ is well-defined and $\cT_n$ the support of $\vt_n$.

Let $\varepsilon_{\mathrm{deg}}$ be such that {\bf DEG} holds, further $\alpha=2+\varepsilon_{\mathrm{deg}}$ and $\cP_{\mathrm{deg}}=\{p\in\cP(\ZZ_{\ge 0}):\Erw[\vx_p^{\alpha}]\in\RR_{\ge 0}\}$ with $\bm x_p\sim p$.
Notice that the map
\begin{align*}
\Delta(p,p')=\sum_{x}x|p'(x)-p(x)|+\left|\Erw\left[\vx_p^2\right]-\Erw\left[\vx_{p'}^2\right]\right|+\left|\Erw\left[\vx_p^{\alpha}\right]-\Erw\left[\vx_{p'}^\alpha\right]\right|
\end{align*}
with $p$, $p'\in\cP_{\mathrm{deg}}$ defines a metric on $\cP_{\mathrm{deg}}$.
This metric induces a metric on the product space $\cT_{\mathrm{rel}}=\mathbb R_{\ge 0}\times\mathcal P_{\mathrm{deg}}^2$ given by
\begin{align*}
\Delta(\tau,\tau')=|\tau_{\mathrm{r}}-\tau'_{\mathrm{r}}|+\Delta(\tau_{\mathrm{v}},\tau'_{\mathrm{v}})+\Delta(\tau_{\mathrm{f}},\tau'_{\mathrm{f}})
\end{align*}
for $\tau=(\tau_{\mathrm{r}},\tau_{\mathrm{v}},\tau_{\mathrm{f}})$, $\tau'=(\tau'_{\mathrm{r}},\tau'_{\mathrm{v}},\tau'_{\mathrm{f}})\in\mathcal T_{\mathrm{rel}}$.
With $p_{\mathrm{d}}$, $p_{\mathrm{k}}$ denoting the laws of $\bm d$ and $\bm k$ respectively we notice that $\tau^*=(\bar d/\bar k,p_{\mathrm{d}},p_{\mathrm{k}})\in\mathcal T_{\mathrm{rel}}$.
For $n\in\ZZ_{>0}$ and $t\in\cT^*_n$ we let $\tau(t)=(m_t/n,p_{\mathrm{d},t},p_{\mathrm{k},t})\in\mathcal T_{\mathrm{rel}}$ with $p_{\mathrm{d},t}$ denoting the relative frequencies of the degrees on the variable side, or equivalently the law of $\bm d_{t}=d_{t,\bm i}$ with $\bm i$ uniform on $[n]$, and $p_{\mathrm{k},t}$ denoting the relative frequencies of the degrees on the factor side, or equivalently the law of $\bm k_t=k_{t,\bm a_t}$ with $\bm a_t$ uniform on $[m_t]$.
For the case $m_t=0$ we let $p_{\mathrm{k},t}$ be the one-point mass on $0$.
Notice that for given $n\in\cN$ and $t\in\cT_n$ the number $m_t\in\ZZ_{\ge 0}$ of factors may still be arbitrarily large. We say that a sequence $f_n:\cT_n\rightarrow\RR$, $n\in\cN$, is sublinear in the number of factors if there exists a constant $c\in\RR_{>0}$ such that $|f_n(t)|\le c+cm_t/n$ for all $t\in\cT_n$ and $n\in\cN$.
\begin{proposition}\label{dd_typ_prop}
Assume that {\bf DEG} holds. Then there exists $r_n\in\mathbb R_{>0}$ with $r_n=o(1)$
such that for all sequences $f_n:\cT_n\rightarrow\RR$, $n\in\cN$, that are sublinear in the number of factors we have
\begin{align*}
\Erw[f_n(\vt_n)]=\Erw[f_n(\vt_n)\vecone\{\tau(\vt_n)\in\cB_{r_n}(\tau^*)\}]+o(1)
=\Erw[f_n(\vt_n)|\tau(\vt_n)\in\cB_{r_n}(\tau^*)]+o(1)\textrm{.}
\end{align*}
\end{proposition}
As a byproduct of the proof we will see that $|\cN|=\infty$, so taking limits is reasonable.
Using Proposition \ref{dd_typ_prop} we consider $r_n$ fixed and use $\cT^\circ_n$ to denote the typical valid degree sequences, i.e.~valid degree sequences $t\in\cT_n$ with $\tau(t)\in\mathcal B_{r_n}(\tau^*)$. In particular, we are free to choose $r_n$ such that uniform bounds on the various quantities are enforced, e.g.~$\frac{\d}{2\k}\le m_t/n\le\frac{3\d}{2\k}$ by choosing $r_n\le\frac{\d}{2\k}$ for all $n\in\cN$, uniform lower bounds for the point probabilities in finite subsets of the supports $\cD$, $\cK$ of $\vd$, $\vk$, bounds on the moments and so on.
Details on further implications can be found in Section \ref{ddp_typ_prop}.
%
%
\subsection{Proof strategy}\label{ddp_ps}
The main ingredient to the proof of Proposition \ref{dd_typ_prop} is the following result.
\begin{proposition}\label{ddp_typ_whp}
Assume that {\bf DEG} holds. Then there exists $r_n=o(1)$ such that $\tau(\bm t_{n})\in\mathcal B_{r_n}(\tau^*)$ with high probability.
\end{proposition}
We split the proof of Proposition \ref{ddp_typ_whp} into three parts.
In the first part we determine the order of the probability that $\bm t^*_n\in\mathcal T_n$ and show that $|\mathcal N|=\infty$.
\begin{lemma}\label{ddp_prop_valid}
Assume that {\bf DEG} holds. Then we have $\mathbb P[\bm t^*_n\in\mathcal T_n]=\Theta(\sqrt{n}^{-1})$.
\end{lemma}
Notice that the proof only requires existence of the second moments.
Next, we show that $\tau(\bm t^*_n)$ is typically close to $\tau^*$.
\begin{lemma}\label{ddp_prop_typ}
Assume that {\bf DEG} holds. Then there exists $r_n=o(1)$ such that $\tau(\bm t^*_{n})\in\mathcal B_{r_n}(\tau^*)$ with high probability.
\end{lemma}
We then use a fairly general argument to show how Proposition \ref{ddp_typ_whp} is immediately implied by Lemma \ref{ddp_prop_valid} and Lemma \ref{ddp_prop_typ}.
Finally, we derive Proposition \ref{dd_typ_prop} from Proposition \ref{ddp_typ_whp}.
%
%
\subsection{Proof of Lemma \ref{ddp_prop_valid}}\label{ddp_valid}
The relevant quantities for the proof are the total variable degree $\bm d_{\mathrm{tot},n}$, $n\in\mathbb Z_{\ge 0}$, the total factor degree $\bm k_{\mathrm{tot},m}$, $m\in\mathbb Z_{\ge 0}$, and the number of factors $\bm m_n=\Po(\bar m_n)$, i.e.
\begin{flalign*}
\bm d_{\mathrm{tot},n}=\sum_{i\in[n]}\bm d_i\textrm{, }
\bm k_{\mathrm{tot},m}=\sum_{a\in[m]}\bm k_a\textrm{, }
\bm m'_n=\sum_{i\in[n]}\bm m'_{i}
\end{flalign*}
with $\bm m'_{i}\sim\Po(\bar d/\bar k)$, $i\in\mathbb Z_{>0}$, independent of anything else, hence $\bm m_n\sim\bm m'_n$ by the properties of the Poisson distribution. Hence, all relevant quantities are sums of i.i.d.~non-negative integer random variables with slightly more than the second moment, which allows to treat them simultaneously using Theorem 3.5.2 in \cite{durrett2010} and the discussion prior to the Theorem.

In particular, we need to distinguish four cases depending on whether or not $\bm d$ and $\bm k$ are degenerate.
To be thorough, notice that $\mathcal D\setminus\{0\}\neq\emptyset$ and $\mathcal K\setminus\{0\}\neq\emptyset$ since $\bar d\in\mathbb R_{>0}$ and $\bar k\in\mathbb R_{>0}$. Hence, we have $\mathbb P[\bm d=\bar d]=1$ with $\bar d\in\mathbb Z_{>0}$ if $\bm d$ is degenerate, and otherwise $\mathbb P[\bm d\in d^*+h_{\mathrm{d}}\mathbb Z]=1$ for some $d^*\in\mathcal D$ and $h_{\mathrm{d}}\in\mathbb Z_{>0}$ denoting the span of $\bm d$ as introduced in Section 3.5 of \cite{durrett2010}. In the latter case we say that $\bm d$ is lattice.
Obviously, the same holds for $\bm k$, while $\Po(\bar d/\bar k)$ is always lattice with span $1$.

In order to treat the random variables above simultaneously we let $\bm x\in\mathbb Z_{\ge 0}$ with $\bar x=\mathbb E[\bm x]\in\mathbb R_{>0}$ and $\sigma^2=\Var(\bm x)\in\mathbb R_{\ge 0}$. Further, for $n\in\mathbb Z_{>0}$ we let $\bm s_n=\sum_{i\in[n]}\bm x_i$ with $\bm x_i\sim \bm x$, $i\in\mathbb Z_{>0}$, being i.i.d.~random variables.
If $\bm x$ is degenerate then we have $\mathbb P[\bm x=\bar x]=1$ with $\bar x\in\mathbb Z_{>0}$ and further $\mathbb P[\bm s_n=\bar xn]=1$ for $n\in\mathbb Z_{>0}$.
Otherwise, we have $\mathbb P[\bm x\in x^*+h\mathbb Z]=1$ with $x^*\in\mathbb Z_{\ge 0}$ such that $\mathbb P[\bm x=x^*]\in(0,1)$ and $h$ being the span of $\bm x$.
In this case, as discussed in \cite{durrett2010}, we have $\mathbb P[\bm s_n\in\mathcal L_n]=1$ with $\mathcal L_n=x^*n+h\mathbb Z$, and the following local limit theorem.
\begin{theorem}\label{ddp_llt}
For $\bm x$ lattice with finite variance and the notions introduced above we have
\begin{align*}
\lim_{n\rightarrow\infty}&\sup_{s\in\mathcal L_n}\left|\frac{\sqrt{n}}{h}\mathbb P[\bm s_n=s]-\phi_n(s)\right|=0\textrm{,}\\
\phi_n(s)&=\frac{1}{\sqrt{2\pi\sigma^2}}\exp\left(-\frac{v_{n,s}^2}{2\sigma^2}\right)\textrm{, }
v_{n,s}=\frac{1}{\sqrt{n}}(s-\bar xn)\textrm{, }s\in\mathbb R_{\ge 0}\textrm{.}
\end{align*}
\end{theorem}
As we will see in the following, Theorem \ref{ddp_llt} has immediate consequences for the distribution of $\bm s_n$ that facilitate the proof of Lemma \ref{ddp_prop_valid}. Now, we are ready for the discussion of the four cases.

First, assume that we are in the biregular case, i.e.~both $\bm d$ and $\bm k$ are degenerate.
Then $\bm d_{\mathrm{tot},n}=\bar dn$ and $\bm k_{\mathrm{tot},m}=\bar km$ are degenerate as well, which implies that $t\in\mathcal T_n$ iff $m_t=\bar m_n$ and hence $|\mathcal T_n|=1$.
For $n\in\bar k\mathbb Z_{>0}$ we have $\bar m_n\in\mathbb Z_{>0}$ and further $\mathbb P[\bm t^*_n\in\mathcal T_n]=\mathbb P[\bm m_n=\bar m_n]>0$ so $|\mathcal N|=\infty$.
Further, for any $n\in\mathcal N$ we must have $\bar m_n\in\mathbb Z_{>0}$, and further saw that
\begin{align*}
\mathbb P[\bm t^*_n\in\mathcal T_n]=\mathbb P[\bm m_n=\bar m_n]=\mathbb P[\bm m'_n=\bar m_n]\textrm{.}
\end{align*}
Hence, we can use the local limit theorem \ref{ddp_llt} for $\bm m'_n$ at $\bar m_n$, i.e.~$h=1$, $\sigma^2=\bar d/\bar k$, $v_{\bar m_n}=0$ and hence $\phi_n(\bar m_n)=c$ with $c=\sqrt{2\pi\sigma^2}^{-1}\in\mathbb R_{>0}$ which gives
\begin{align*}
\mathbb P[\bm t^*_n\in\mathcal T_n]=\frac{1}{\sqrt{n}}\left(c+o(1)\right)=\Theta(\sqrt{n}^{-1})\textrm{.}
\end{align*}
Next, we consider the case that $\bm d$ is degenerate and $\bm k$ is lattice.
Hence, we have $\bar d\in\mathbb Z_{>0}$, $\mathcal D=\{\bar d\}$ and $\mathbb P[\bm d_{\mathrm{tot},n}=\bar dn]=1$ for $n\in\mathbb Z_{>0}$, on the variable side.
Further, we have $\mathbb P[\bm k_{\mathrm{tot},m}\in\mathcal L_m]=1$ with $\mathcal L_m=k^*m+h_{\mathrm{k}}\mathbb Z$ for $m\in\mathbb Z_{>0}$ on the factor side, where $k^*\in\mathcal K\setminus\{0\}$ and $h_{\mathrm{k}}\in\mathbb Z_{>0}$ denotes the span of $\bm k$.
Now, for $n\in k^*\mathbb Z_{>0}$ we have $m=\frac{\bar dn}{k^*}\in\mathbb Z_{>0}$ and hence $t\in\mathcal T_n$, where $t$ is given by $m_t=m$, $d_{t,i}=\bar d$ for $i\in[n]$ and $k_{t,a}=k^*$ for $a\in[m]$, so $n\in\mathcal N$ and hence $|\mathcal N|=\infty$.
Further, for any $n\in\mathcal N$ there exists $t^*\in\mathcal T_n$, so $\bar dn\in\mathcal L_{m_{t^*}}$.
But by definition we have $\mathcal L_m=\mathcal L_{m_{t^*}}$ for any $m\in m_{t^*}+h_{\mathrm{k}}\mathbb Z$, so $\bar dn\in\mathcal L_m$. Now, fix some large radius $r\in\mathbb R_{>0}$ and let $\mathcal M_n$ be given by all $m\in m_{t^*}+h_{\mathrm{k}}\mathbb Z$ with $|m-\bar m_n|<r\sqrt{n}$. Due to the lattice structure this gives $|\mathcal M_n|=\Theta(\sqrt{n})$. Further, notice that for $m\in\mathcal M_n$ we have $\phi_n(m)=\Theta(1)$ uniformly since $|v_m|<r$ in the local limit theorem for $\bm m'_n$, so $\mathbb P[\bm m'_n=m]=\Theta(\sqrt{n}^{-1})$ uniformly and thereby $\mathbb P[\bm m'_n\in\mathcal M_n]=\Theta(1)$. For any $m\in\mathcal M_n$ we have $|\bar dn-\bar km|=\bar k|\bar m-m|<\bar kr\sqrt{n}$, so the required total degree $\bar dn$ is sufficiently close to the expected total degree $\bar km$ on the factor side.
Now, since we have $m=\Theta(n)$ uniformly for all $m\in\mathcal M_n$ and $\bar dn\in\mathcal L_m$, the local limit theorem for $\bm k_{\mathrm{tot},m}$ gives $\mathbb P[\bm k_{\mathrm{tot},m}=\bar dn]=\Theta(\sqrt{n}^{-1})$ uniformly for all $m\in\mathcal M_n$. This shows that $\mathbb P[\bm t^*_n\in\mathcal T_n]=\Omega(\sqrt{n}^{-1})$. To see that $\mathbb P[\bm t^*_n\in\mathcal T_n]=O(\sqrt{n}^{-1})$ we only have to notice that $\mathbb P[\bm k_{\mathrm{tot},m}=\bar dn]=O(\sqrt{n}^{-1})$ uniformly for all $m\in\mathbb Z_{\ge\varepsilon n}$ for any fixed $\varepsilon\in\mathbb R_{>0}$ and that $\mathbb P[\bm m_n<\varepsilon n]=o(\sqrt{n}^{-1})$ using the well-known Poisson tails.

Now, assume that $\bm d$ is lattice and $\bm k$ is degenerate.
Let $d^*\in\mathcal D\setminus\{0\}$ and $h_{\mathrm{d}}\in\mathbb Z_{>0}$ denote the span of $\bm d$.
Notice that for any $n\in\bar k\mathbb Z_{>0}$ we have $m=d^*n/\bar k\in\mathbb Z_{>0}$, so the corresponding sequences $t$ are in $\mathcal T_n$ and hence $|\mathcal N|=\infty$.
Further, for any $n\in\mathcal N$ we fix $t^*\in\mathcal T_n$ and notice that $\bar k m_{t^*}\in\mathcal L_n$ with $\mathcal L_n=d^*n+h_{\mathrm{d}}\mathbb Z$, so $\bar km\in\mathcal L_n$ for any $m\in m_{t^*}+h_{\mathrm{d}}\mathbb Z$.
We repeat the previous construction with $r\in\mathbb R_{>0}$ to obtain $\mathbb P[\bm m'_n\in\mathcal M_n]=\Theta(1)$, and again for any $m\in\mathcal M_n$ we have $|\bar dn-\bar k m|<\bar kr\sqrt{n}$. But this time, the consequence is that $\mathbb P[\bm d_{\mathrm{tot},n}=\bar km]=\Theta(\sqrt{n}^{-1})$ uniformly for all $m\in\mathcal M_n$, so $\mathbb P[\bm t^*_n\in\mathcal T_n]=\Omega(\sqrt{n}^{-1})$. For the upper bound notice that we have the uniform bound $\mathbb P[\bm d_{\mathrm{tot},n}=\bar k m]=O(\sqrt{n}^{-1})$ for any choice of $m$ and hence $\mathbb P[\bm t^*_n\in\mathcal T_n]=O(\sqrt{n}^{-1})$.

We turn to the final case that both $\bm d$ and $\bm k$ are lattice.
Let $d^*\in\mathcal D\setminus\{0\}$ and $k^*\in\mathcal K\setminus\{0\}$, further $h_{\mathrm{d}}$ and $h_{\mathrm{k}}$ denote the spans as before. With $n\in k^*\mathbb Z_{>0}$ and $m=d^*n/k^*\in\mathbb Z_{>0}$ we get $|\mathcal N|=\infty$.
Further, for any $n\in\mathcal N$ there exists $t^*\in\mathcal T_n$, so $\mathcal L_{\mathrm{d},n}\cap\mathcal L_{\mathrm{k},m_{t^*}}\neq\emptyset$, where $\mathcal L_{\mathrm{d},n}=d^*n+h_{\mathrm{d}}\mathbb Z$ and $\mathcal L_{\mathrm{k},m}=k^*m+h_{\mathrm{k}}\mathbb Z$. Fix $s^*\in\mathcal L_{\mathrm{d},n}\cap\mathcal L_{\mathrm{k},m_{t^*}}$, then we have $s^*\in\mathcal L_{\mathrm{d},n}\cap\mathcal L_{\mathrm{k},m}$ for any $m\in m_{t^*}+h_{\mathrm{k}}\mathbb Z$ since then $\mathcal L_{\mathrm{k},m}=\mathcal L_{\mathrm{k},m_{t^*}}$.
Hence, we are free to repeat the previous constrruction for given $r$ to obtain $\mathcal M_n$ with $\mathbb P[\bm m_n\in\mathcal M_n]=\Theta(1)$.
In the next step we need to improve on $s^*$, so for fixed $m\in\mathcal M_n$ we notice that we have $s\in\mathcal L_{\mathrm{d},n}\cap\mathcal L_{\mathrm{k},m}$ for any $s\in s^*+h_{\mathrm{e}}\mathbb Z$ with $h_{\mathrm{e}}=\gcd(h_{\mathrm{d}},h_{\mathrm{k}})$. So, for fixed and large $r'\in\mathbb R_{>0}$ let $\mathcal E_n$ be given by $s\in s^*+h_{\mathrm{e}}\mathbb Z$ with $|s-\bar dn|<r'\sqrt{n}$. But then for any $m\in\mathcal M_n$ and $s\in\mathcal E_n$ we have $|s-\bar km|<r'\sqrt{n}+r\bar k\sqrt{n}$, so $s$ is sufficiently close to the expected total degree on the factor side.
Now, we can summon the local limit theorem for $\bm d_{\mathrm{tot},n}$ to get $\mathbb P[\bm d_{\mathrm{tot},n}\in\mathcal E_n]=\Theta(1)$ and the local limit theorem for $\bm k_{\mathrm{tot},m}$ to get $\mathbb P[\bm k_{\mathrm{tot},m}=s]=\Theta(\sqrt{n}^{-1})$ uniformly for all $s\in\mathcal E_n$ and $m\in\mathcal M_n$.
This gives $\mathbb P[\bm t^*_n\in\mathcal T_n]=\Omega(\sqrt{n}^{-1})$, while $\mathbb P[\bm d_{\mathrm{tot},n}=s]=O(\sqrt{n}^{-1})$ uniformly gives $\mathbb P[\bm t^*_n\in\mathcal T_n]=O(\sqrt{n}^{-1})$.
%
%
\subsection{Proof of Lemma \ref{ddp_prop_typ}}\label{ddp_typ}
We split the metric into the seven individual contributions and consider them separately.
For this purpose let $n\in\mathbb Z_{>0}$ and
\begin{align*}
\mathcal T_{\mathrm{r},n}&=\{t\in\mathcal T^*_n:|m_t/n-\bar d/\bar k|<r_{\mathrm{r},n}\}\textrm{,}\\
\mathcal T_{\mathrm{v1},n}&=\left\{t\in\mathcal T^*_n:\sum_dd\left|p_{\mathrm{d},t}(d)-p_{\mathrm{d}}(d)\right|<r_{\mathrm{v1},n}\right\}\textrm{, }
\mathcal T_{\mathrm{f1},n}=\left\{t\in\mathcal T^*_n:\sum_kk\left|p_{\mathrm{k},t}(k)-p_{\mathrm{k}}(k)\right|<r_{\mathrm{f1},n}\right\}\textrm{,}\\
\mathcal T_{\mathrm{v2},n}&=\left\{t\in\mathcal T^*_n:\left|\mathbb E[\bm d_t^2]-\mathbb E[\bm d^2]\right|<r_{\mathrm{v2},n}\right\}\textrm{, }
\mathcal T_{\mathrm{f2},n}=\left\{t\in\mathcal T^*_n:\left|\mathbb E[\bm k_t^2]-\mathbb E[\bm k^2]\right|<r_{\mathrm{f2},n}\right\}\textrm{,}\\
\mathcal T_{\mathrm{v3},n}&=\left\{t\in\mathcal T^*_n:\left|\mathbb E[\bm d_t^\alpha]-\mathbb E[\bm d^\alpha]\right|<r_{\mathrm{v3},n}\right\}\textrm{, }
\mathcal T_{\mathrm{f3},n}=\left\{t\in\mathcal T^*_n:\left|\mathbb E[\bm k_t^\alpha]-\mathbb E[\bm k^\alpha]\right|<r_{\mathrm{f3},n}\right\}\textrm{,}
\end{align*}
for some sequences of radii and with $\alpha=2+\varepsilon_{\mathrm{deg}}$.
Since $m_{\bm t^*_n}$ is $\Po(\bar m_n)$ we can use the standard Poisson bounds, e.g.~Theorem 2.1 with Remark 2.6 in \cite{janson2000}, to see that $\bm t^*_n\in\mathcal T_{\mathrm{r},n}$ with high probability for any $r_{\mathrm{r},n}=\omega(\sqrt{n}^{-1})$ with $r_{\mathrm{r},n}=o(1)$.
Further, we notice that
\begin{align*}
s_{t}=n\mathbb E[\bm d_t^2]=\sum_{i\in[n]}d_{t,i}^2\textrm{, so }
s_{\bm t^*_n}&=\sum_{i\in[n]}\bm d_{i}^2
\end{align*}
and thereby $s_{\bm t^*_n}$ is the sum over the i.i.d.~random variables $\bm d_{i}^2$, $i\in[n]$.
Hence, we use the weak law of large numbers, e.g.~Chapter 10.2 in \cite{feller1968},  applied to $\frac{1}{n}s_{\bm t^*_n}$ considered as the average over the i.i.d.~$\bm d^2_i$, $i\in\mathbb Z_{>0}$, with finite \emph{first} moment $\mathbb E[\bm d^2]\in\mathbb R_{>0}$ to obtain $r_{\mathrm{v2},n}=o(1)$ such that $\bm t^*_n\in\mathcal T_{\mathrm{v2},n}$ with high probability.
The discussion of $\mathcal T_{\mathrm{v3},n}$ is completely analogous.
Next, for $r_{\mathrm{e},n}\in\mathbb R_{>0}$ consider the event
\begin{align*}
\mathcal E_n=\left\{t\in\mathcal T^*_n:\sum_{d>0}\left|p_{\mathrm{d},t}(d)-p_{\mathrm{d}}(d)\right|<r_{\mathrm{e},n}\right\}\textrm{.}
\end{align*}
Let $\alpha'(d)=d^{-(1+\frac{1}{2}\varepsilon_{\mathrm{deg}})}$ for $d\in\mathcal D\setminus\{0\}$.
Let $a=\sum_d\alpha'(d)\in\mathbb R_{>0}$ and $\alpha=a^{-1}\alpha'\in\mathcal P(\mathcal D\setminus\{0\})$.
Then we have $\mathbb E[\vecone\{\bm d>0\}\alpha(\bm d)^{-2}]=a^2\mathbb E[\bm d^{2+\varepsilon_{\mathrm{deg}}}]$ and further
\begin{align*}
\mathbb P\left[\bm t^*_n\not\in\mathcal E_n\right]
&=\mathbb P\left[\sum_{d>0}|p_{\mathrm{d},\bm t^*_n}(d)-p_{\mathrm{d}}(d)|\ge\sum_{d>0}\alpha(d)r_{\mathrm{e},n}\right]
\le\sum_{d>0}\mathbb P\left[|p_{\mathrm{d},\bm t^*_n}(d)-p_{\mathrm{d}}(d)|\ge\alpha(d)r_{\mathrm{e},n}\right]\\
&\le\sum_{d>0}\frac{\Var(np_{\mathrm{d},\bm t^*_n}(d))}{\left(\alpha(d)nr_{\mathrm{e},n}\right)^2}
=\sum_{d>0}\frac{p_{\mathrm{d}}(d)(1-p_{\mathrm{d}}(d))}{\alpha(d)^2nr_{\mathrm{e},n}^2}
\le\frac{a^2\mathbb E[\bm d^{2+\varepsilon_{\mathrm{deg}}}]}{nr_{\mathrm{e},n}^2}\textrm{,}
\end{align*}
where we used that $np_{\mathrm{d},\bm t^*_n}(d)$ is binomial with size $n$ and success probability $p_{\mathrm{d}}(d)$.
Hence, we can choose any $r_{\mathrm{e},n}=\omega(\sqrt{n}^{-1})$ with $r_{\mathrm{e},n}=o(1)$ to obtain $\bm t^*_n\in\mathcal E_n$ with high probability.
With $\alpha=2+\varepsilon_{\mathrm{deg}}$, $r_{\mathrm{e},n}=o(n^{-1/\alpha})$, $d_{\max,n}=c_nn^{1/\alpha}=\omega(1)$, $c_n=(\mathbb E[\bm d^\alpha]+r_{\mathrm{v3},n})^{1/\alpha}=\Theta(1)$, and $t\in\mathcal E_n\cap\mathcal T_{\mathrm{v3},n}$ Markov's inequality implies that
\begin{flalign*}
\mathbb P[\bm d_t\ge d_{\max,n}]\le\frac{\mathbb E[\bm d_t^{\alpha}]}{d_{\max,n}^\alpha}<\frac{\mathbb E[\bm d^{\alpha}]+r_{\mathrm{v3},n}}{c_n^\alpha n}=\frac{1}{n}
\end{flalign*}
and hence $\mathbb P[\bm d_t\ge d_{\max,n}]=0$, which further yields
\begin{flalign*}
\sum_dd|p_{\mathrm{d},t}(d)-p_{\mathrm{d}}(d)|=\sum_{d<d_{\max,n}}d|p_{\mathrm{d},t}(d)-p_{\mathrm{d}}(d)|+\mathbb E[\vecone\{\bm d\ge d_{\max,n}\}\bm d]<d_{\max,n}r_{\mathrm{e},n}+o(1)=o(1)\textrm{,}
\end{flalign*}
meaning that there exists $r_{\mathrm{v1},n}=o(1)$ such that $\mathcal T_{\mathrm{v1},n}\subseteq\mathcal E_n\cap\mathcal T_{\mathrm{v3},n}$.
This shows the existence of radii $r_{\mathrm{r},n}$, $r_{\mathrm{v},n}=o(1)$ such that jointly $\bm t^*_n\in\mathcal T_{\mathrm{r},n}$ and $p_{\mathrm{d},\bm t^*_n}\in\mathcal B_{r_{\mathrm{v},n}}(p_{\mathrm{d}})$ with high probability.
Due to symmetry we obtain $r_{\mathrm{f},m}=o(1)$ (in the number of factors) such that $p_{\mathrm{k},\bm t^*_n}\in\mathcal B_{r_{\mathrm{f},m}}(p_{\mathrm{k}})|m_{\bm t^*_n}=m$ with high probability in $m$ and further uniformly in $n$.
But since we have $m_t=\Theta(n)$ for $t\in\mathcal T_{\mathrm{r},n}$ uniformly, we obtain radii $r_{\mathrm{f},n}=o(1)$ depending only on $n$ by taking the supremum of $r_{\mathrm{f},m_t}$ over $t\in\mathcal T_{\mathrm{r},n}$.
With $r_n=r_{\mathrm{r},n}+r_{\mathrm{v},n}+r_{\mathrm{f},n}=o(1)$ this immediately gives $\tau(\bm t^*_n)\in\mathcal B_{r_n}(\tau^*)$ with high probability.
%
%
\subsection{Proof of Proposition \ref{ddp_typ_whp}}\label{ddp_typ_cond}
Recall from the proof of Lemma \ref{ddp_prop_valid} that $\tau(\bm t_n)=\tau^*$ almost surely for all $n\in\cN$ if $\bm d$, $\bm k$ are degenerate, i.e.~the assertion holds for any choice of $r_n=o(1)$.
Otherwise, let $r_n=o(1)$ be a sequence obtained from Lemma \ref{ddp_prop_typ} such that $r_n=\omega(\ln(n)^{-1})$.
Notice that $3r_n\in o(1)$ and
\begin{flalign*}
\mathbb P\left[\tau(\bm t^*_n)\in\mathcal B_{r_n}(\tau^*)\right]
&\le\mathbb P\left[\bm t^*_n\in\mathcal T_{\mathrm{s},n}\right]
\le\mathbb P\left[\tau(\bm t^*_n)\in\mathcal B_{3r_n}(\tau^*)\right]\textrm{,}\\
\bm t^*_n\in\mathcal T_{\mathrm{s},n}&\textrm{ iff }
\bm m_n/n\in\mathcal B_{r_n}(\bar d/\bar k),\bm p_{\mathrm{d},n}\in\mathcal B_{r_n}(p_{\mathrm{d}}),\bm p_{\mathrm{k},\bm m_n}\in\mathcal B_{r_n}(p_{\mathrm{k}})\textrm{,}
\end{flalign*}
with $\bm p_{\mathrm{d},n}=p_{\mathrm{d},\bm t^*_n}$ denoting the relative frequencies of $(\bm d_i)_{i\in[n]}$, $\bm p_{\mathrm{k},m}$ denoting the relative frequencies of $(\bm k_a)_{a\in[m]}$ for given $m\in\mathbb Z_{\ge 0}$ and where we recall that $\bm t^*_n=(\bm m_n,(\bm d_i)_{i\in[n]},(\bm k_a)_{a\in[\bm m_n]})$.
In particular, the above shows that all three events occur with high probability and further $3r_n$ is also a suitable choice in the context of Lemma \ref{ddp_prop_typ}. For given $s$ and $m$ we use the shorthands
\begin{flalign*}
p_{\mathrm{m}}(m)&=\mathbb P[\bm m_n=m]\textrm{, }
P_{\mathrm{d}}(s)=\mathbb P\left[\sum_{i\in[n]}\bm d_i=s\right]\textrm{, }
P_{\mathrm{k}}(s,m)=\mathbb P\left[\sum_{a\in[m]}\bm k_a=s\right]\textrm{, }\\
P^+_{\mathrm{d}}(s)&=\mathbb P\left[\sum_{i\in[n]}\bm d_i=s,\bm p_{\mathrm{d},n}\in\mathcal B_{r_n}(p_{\mathrm{d}})\right]\textrm{, }
P^+_{\mathrm{k}}(s,m)=\mathbb P\left[\sum_{a\in[m]}\bm k_a=s,\bm p_{\mathrm{k},m}\in\mathcal B_{r_n}(p_{\mathrm{k}})\right]\textrm{,}
\end{flalign*}
further $P^-_{\mathrm{d}}(s)=P_{\mathrm{d}}(s)-P^+_{\mathrm{d}}(s)$, $P^-_{\mathrm{k}}(s,m)=P_{\mathrm{k}}(s,m)-P^+_{\mathrm{k}}(s,m)$ and $\mathcal M_n=n\mathcal B_{r_n}(\bar d/\bar k)$.
Using this notation we have
\begin{flalign*}
\mathbb P\left[\bm t^*_n\in\mathcal T_n\right]
&=\sum_{s,m}p_{\mathrm{m}}(m)P_{\mathrm{d}}(s)P_{\mathrm{k}}(s)\textrm{,}\\
\mathbb P\left[\bm t^*_n\in\mathcal T_n\cap\mathcal T_{\mathrm{s},n}\right]
&=\sum_{m\in\mathcal M_n}\sum_{s}p_{\mathrm{m}}(m)P^+_{\mathrm{d}}(s)P^+_{\mathrm{k}}(s,m)
\ge\mathbb P\left[\bm t^*_n\in\mathcal T_n,\bm m_n\in\mathcal M_n\right]-E_{\mathrm{d}}-E_{\mathrm{k}}\textrm{,}\\
E_{\mathrm{d}}&=\sum_{m\in\mathcal M_n}\sum_sp_{\mathrm{m}}(m)P^-_{\mathrm{d}}(s)P_{\mathrm{k}}(s,m)\textrm{, }
E_{\mathrm{k}}=\sum_{m\in\mathcal M_n}\sum_sp_{\mathrm{m}}(m)P_{\mathrm{d}}(s)P^-_{\mathrm{k}}(s,m)\textrm{, }
\end{flalign*}
where we exploited the dependency structure of $\bm t^*_n$.
With the Poisson bounds used in the proof of Lemma \ref{ddp_prop_typ}, $r_n=\omega(\ln(n)^{-1})$ and Lemma \ref{ddp_prop_valid} we have
\begin{flalign*}
\mathbb P\left[\bm t^*_n\in\mathcal T_n,\bm m_n\in\mathcal M_n\right]
&=\mathbb P\left[\bm t^*_n\in\mathcal T_n\right]-\mathbb P\left[\bm t^*_n\in\mathcal T_n,\bm m_n\not\in\mathcal M_n\right]
=\mathbb P\left[\bm t^*_n\in\mathcal T_n\right]-o(\sqrt{n}^{-1})\\
&=(1+o(1))\mathbb P\left[\bm t^*_n\in\mathcal T_n\right]\textrm{.}
\end{flalign*}
Now, assume that both $\bm d$ and $\bm k$ are lattice.
With $m\in\mathcal M_n$ and the proof of Lemma \ref{ddp_prop_valid}, respectively Theorem \ref{ddp_llt}, notice that $P_{\mathrm{k}}(s,m)=O(\sqrt{n}^{-1})$ uniformly in $s$, $m$ since $\bm k$ is lattice, so
\begin{flalign*}
E_{\mathrm{d}}=O(\sqrt{n}^{-1})\sum_{m\in\mathcal M_n}\sum_sp_{\mathrm{m}}(m)P^-_{\mathrm{d}}(s)
=O(\sqrt{n}^{-1})\mathbb P[\bm m_n\in\mathcal M_n]\mathbb P\left[\bm p_{\mathrm{d},n}\not\in\mathcal B_{r_n}(p_{\mathrm{d}})\right]=o(\sqrt{n}^{-1})
\end{flalign*}
since $\bm p_{\mathrm{d},n}\in\mathcal B_{r_n}(p_{\mathrm{d}})$ with high probability.
Further, we have $P_{\mathrm{d}}(s)=O(\sqrt{n}^{-1})$ uniformly since $\bm d$ is lattice and hence we obtain $E_{\mathrm{k}}=o(\sqrt{n}^{-1})$ analogously. With $\mathbb P(\bm t^*_n\in\mathcal T_n\cap\mathcal T_{\mathrm{s},n}]\le\mathbb P[\bm t^*_n\in\mathcal T_n]$ this gives
\begin{flalign*}
\mathbb P[\bm t^*_n\in\mathcal T_n\cap\mathcal T_{\mathrm{s},n}]
&=(1+o(1))\mathbb P[\bm t^*_n\in\mathcal T_n]
\end{flalign*}
with another application of Lemma \ref{ddp_prop_valid}, which shows that $\mathbb P[\bm t_n\in\mathcal T_{\mathrm{s},n}]=1+o(1)$ and thereby $\mathbb P[\tau(\bm t_n)\in\mathcal B_{3r_n}(\tau^*)]=1+o(1)$, establishing the assertion for the current case with $3r_n$.

Next, we consider the case that $\bm d$ is lattice and $\bm k$ is degenerate.
Then we have $\bm p_{\mathrm{k},m}=p_{\mathrm{k}}$ almost surely for all $m\in\mathbb Z_{>0}$ and hence $P^-_{\mathrm{k}}(s,m)=0$ for all $s$ and further $E_{\mathrm{k}}=0$ for $n$ sufficiently large.
Further, we have $P_{\mathrm{k}}(s,m)=\vecone\{s=\bar km\}$ and hence
\begin{flalign*}
E_{\mathrm{d}}=\sum_{m\in\mathcal M_n}p_{\mathrm{m}}(m)P^-_{\mathrm{d}}(\bar km)
=O(\sqrt{n}^{-1})\sum_{m\in\mathcal M_n}P^-_{\mathrm{d}}(\bar km)
=O(\sqrt{n}^{-1})\sum_{s}P^-_{\mathrm{d}}(s)
=o(\sqrt{n}^{-1})
\end{flalign*}
using the local limit theorem for $\bm m_n$ and $\bm p_{\mathrm{d},n}\in\mathcal B_{r_n}(p_{\mathrm{d}})$ with high probability. Following the discussion above this yields $\mathbb P[\tau(\bm t_n)\in\mathcal B_{3r_n}(\tau^*)]=1+o(1)$.

Finally, assume that $\bm d$ is degenerate and $\bm k$ is lattice, so in particular $\bm p_{\mathrm{d},n}=p_{\mathrm{d}}$ almost surely, i.e.~$P^-_{\mathrm{d}}(s)=0$ and further $E_{\mathrm{d}}=0$, and $P_{\mathrm{d}}(s)=\vecone\{s=\bar dn\}$, leaving us with
\begin{flalign*}
E_{\mathrm{k}}=\sum_{m\in\mathcal M_n}p_{\mathrm{m}}(m)P^-_{\mathrm{k}}(\bar dn,m)\textrm{.}
\end{flalign*}
Now, let $\bm m_1$, $\bm m_2$ be i.i.d.~with law $\Po(\bar m_n/2)$, i.e.~we consider $\bm m_n=\bm m_1+\bm m_2\sim\Po(\bar m_n)$ as derived random variable.
Analogously, we consider i.i.d.~copies $\bm k_{1,a}$, $\bm k_{2,a}$ with law $p_{\mathrm{k}}$ and $a\in\mathbb Z_{>0}$, which allows to consider $(\bm k_a)_{a\in[\bm m_n]}=((\bm k_{1,a})_{a\in[\bm m_1]},(\bm k_{2,a})_{a\in[\bm m_2]})$ as derived random variables. This immediately gives $\bm p_{\mathrm{k},\bm m_n}=\frac{\bm m_1}{\bm m_n}\bm p_{\mathrm{k1},\bm m_1}+\frac{\bm m_2}{\bm m_n}p_{\mathrm{k2},\bm m_2}$ and further
\begin{flalign*}
\Delta(\bm p_{\mathrm{k},\bm m_n},p_{\mathrm{k}})\le\frac{\bm m_1}{\bm m_n}\Delta(\bm p_{\mathrm{k1},\bm m_1},p_{\mathrm{k}})+\frac{\bm m_2}{\bm m_n}\Delta(\bm p_{\mathrm{k2},\bm m_2},p_{\mathrm{k}})\textrm{.}
\end{flalign*}
Hence, in the event that $\bm p_{\mathrm{k},\bm m_n}\not\in\mathcal B_{r_n}(p_{\mathrm{k}})$ we have $\bm p_{\mathrm{k1},\bm m_1}\not\in\mathcal B_{r_n}(p_{\mathrm{k}})$ or $\bm p_{\mathrm{k2},\bm m_2}\not\in\mathcal B_{r_n}(p_{\mathrm{k}})$.
Using corresponding shorthands for this decomposition we first obtain
\begin{flalign*}
E_{\mathrm{k}}&=\sum_{m\in\mathcal M_n}\sum_{m_1}p_{\mathrm{m1}}(m_1)p_{\mathrm{m2}}(m-m_1)P^-_{\mathrm{k}}(\bar dn,m)\\
&\le\sum_{m_1,m_2\in\mathcal M_{2,n}}p_{\mathrm{m1}}(m_1)p_{\mathrm{m2}}(m_2)P^-_{\mathrm{k}}(\bar dn,m_1+m_2)+o(\sqrt{n}^{-1})
\end{flalign*}
with $\mathcal M_{2,n}=\mathcal B_{nr_n}(\bar m_n/2)$ by using the Poisson bounds for both $\bm m_1$, $\bm m_2$ and an extension of the domain. As discussed above this further yields
\begin{flalign*}
E_{\mathrm{k}}&\le E_{\mathrm{k1}}+E_{\mathrm{k2}}+o(\sqrt{n}^{-1})\textrm{,}\\
E_{\mathrm{k1}}&=\sum_{m\in\mathcal M_{2,n}^2}p_{\mathrm{m1}}(m_1)p_{\mathrm{m2}}(m_2)\sum_sP^-_{\mathrm{k1}}(s,m_1)P_{\mathrm{k2}}(\bar dn-s,m_2)\textrm{,}\\
E_{\mathrm{k2}}&=\sum_{m\in\mathcal M_{2,n}^2}p_{\mathrm{m1}}(m_1)p_{\mathrm{m2}}(m_2)\sum_sP_{\mathrm{k1}}(s,m_1)P^-_{\mathrm{k2}}(\bar dn-s,m_2)\textrm{.}
\end{flalign*}
Since both $m_1$ and $m_2$ are uniformly linear in $n$ we can apply the local limit theorem to obtain
\begin{flalign*}
E_{\mathrm{k1}}=O(\sqrt{n}^{-1})\mathbb P\left[\bm m_1\in\mathcal M_{2,n},\bm p_{\mathrm{k1},\bm m_1}\not\in\mathcal B_{r_n}(p_{\mathrm{k}})\right]
\end{flalign*}
and the corresponding result for $E_{\mathrm{k2}}$.
At this point we notice that both the assertion of Proposition \ref{ddp_typ_whp} and  Lemma \ref{ddp_prop_typ} allow the choice of any arbitrarily flat sequence $r_n=o(1)$ and in particular such that the assertion of Lemma \ref{ddp_prop_typ} still holds with $r'_n=r_{2n}$ (where we may assume $r_{2n}\le r_n$ without loss of generality).
Hence, the observation that the models corresponding to $\bm m_1$ and $\bm m_2$ exactly reflect the model corresponding to $\bm m_{n/2}$ with radii $r'_{n/2}=r_n$ shows that we can choose $r_n$ such that $E_{\mathrm{k1}}$, $E_{\mathrm{k2}}=o(\sqrt{n}^{-1})$ and thus $E_{\mathrm{k}}=o(\sqrt{n}^{-1})$.
With these error bounds we also conclude for the last case that $\mathbb P[\tau(\bm t_n)\in\mathcal B_{3r_n}(\tau^*)]=1+o(1)$.
%
%
\subsection{Proof of Proposition \ref{dd_typ_prop}}
Let a sequence $f_n:\mathcal T_n\rightarrow\mathbb R$, $n\in\mathcal N$, be given that is sublinear in the number of factors and let $c\in\mathbb R_{>0}$ such that $|f_n(t)|\le c+c\frac{m_t}{n}$ for all $t\in\mathcal T_n$ and $n\in\mathcal N$.
Using Proposition \ref{ddp_typ_whp} we obtain $r_n=o(1)$ and let $\mathcal T^\circ_n$ denote the set of $t\in\mathcal T_n$ with $\tau(t)\in\mathcal B_{r_n}(\tau^*)$.
With this notation we have
\begin{flalign*}
\left|\mathbb E\left[f_n(\bm t_n)\vecone\{\bm t_n\not\in\mathcal T^\circ_n\}\right]\right|
&\le c\mathbb P[\bm t_n\not\in\mathcal T^\circ_n]+c\mathbb E\left[\frac{m_{\bm t_n}}{n}\vecone\{\bm t_n\not\in\mathcal T^\circ_n\}\right]\\
&\le o(1)+\frac{2c\bar m_n}{n}\mathbb P[\bm t_n\not\in\mathcal T^\circ_n]+c\mathbb E\left[\frac{m_{\bm t_n}}{n}\vecone\{\bm t_n\not\in\mathcal T^\circ_n,m_{\bm t_n}\ge 2\bar m_n\}\right]\textrm{.}
\end{flalign*}
Using the definition of $\bar m_n$ and Proposition \ref{ddp_typ_whp} we notice that the second contribution is also $o(1)$.
For the last contribution we recall the definition of $\bm t_n$, resolve the conditional expectation and use Lemma \ref{ddp_prop_valid} for the bound $\Theta(\sqrt{n}^{-1})$ in the denominator $\mathbb P[\bm t^*_n\in\mathcal T_n]$, while the nominator can be upper bounded by $\mathbb E[\bm m_nn^{-1}\vecone\{\bm m_n\ge 2\bar m_n\}]$.
From the definition of the Poisson distribution we have $\mathbb E[\bm m_n\vecone\{\bm m_n=m\}]=\bar m_n\mathbb P[\bm m_n=m-1]$, so we obtain the upper bound $\bar d/\bar k\mathbb P[\bm m_n+1\ge 2\bar m_n]$ which is exponentially small using standard Poisson bounds and hence $\sqrt{n}\mathbb P[\bm m_n+1\ge 2\bar m_n]=o(1)$.
%
%
\subsection{Properties of typical sequences}\label{ddp_typ_prop}
In this section we summarize a few properties of the typical sequences $t\in\cT_n$ for later usage.
First, notice that
\begin{flalign*}
\|p_{\mathrm{d},t}-p_{\mathrm{d}}\|_1\le 2\sum_{d>0}|p_{\mathrm{d},t}(d)-p_{\mathrm{d}}(d)|\le 2\Delta(\tau(t),\tau^*)<2r_n\textrm{,}
\end{flalign*}
so we can choose $r_n$ such that for any finite subset $\cD'\subseteq\cD$ and sufficiently small $\varepsilon\in(0,1)$ we have $p_{\mathrm{d},t}(d)\ge\varepsilon$ uniformly in $n$, $t\in\cT^\circ_n$ and $d\in\cD'$, further impose any absolute bound on the distance to $p_{\mathrm{d}}$ in $\|\cdot\|_1$ as well as the degree reweighted distance
\begin{flalign*}
\sum_dd|p_{\mathrm{d},t}(d)-p_{\mathrm{d}}(d)|\textrm{.}
\end{flalign*}
In particular, we also obtain convergence of the first moment since
\begin{flalign*}
|\Erw[\vd_t]-\Erw[\vd]|\le\sum_dd|p_{\mathrm{d},t}(d)-p_{\mathrm{d}}(d)|<r_n\textrm{.}
\end{flalign*}
Since we obviously have $\Erw[\vd_t^2]\rightarrow\Erw[\vd^2]$ and $\Erw[\vd_t^{\alpha}]\rightarrow\Erw[\vd^\alpha]$ with $\alpha=2+\varepsilon_{\mathrm{deg}}$ uniformly, we can choose $r_n$ to enforce uniform upper bounds $E^{(2)}$, $E^{(\alpha)}\in\mathbb R_{>0}$ uniformly in $n$ and $t\in\cT^\circ_n$.
As discussed in the proof of Proposition \ref{ddp_prop_typ} Markov's inequality then implies that $\max\{d_{t,i}:i\in[n]\}\le d_{\max,n}$ with $d_{\max,n}=(E^{(\alpha)}n)^{1/\alpha}$ uniformly in $t\in\cT^\circ_n$, so $\|\vd_t\|_\infty\le cn^{\beta}$ almost surely for some $c\in\RR_{\ge 0}$, $\beta\in(0,1/2)$ and uniformly in $t\in\cT^\circ_n$.
Combining these gives the uniform bound
\begin{flalign*}
\Erw\left[\vd_t^3\right]\le d_{\max,n}^{3-\alpha}E^{(\alpha)}=cn^{\beta}
\end{flalign*}
with $c\in\RR_{>0}$ given by the above and $\beta=(3-\alpha)/\alpha\in(0,1/2)$ (if $\varepsilon_{\mathrm{deg}}<1$ and obviously $\beta=0$ otherwise).
Since $m_t\sim\m_n$ uniformly for $t\in\cT^\circ_n$ the discussion above directly yields corresponding results for the factor side.

\section{Mutual contiguity} \label{mc}

This section is dedicated to the mutual contiguity part of Proposition \ref{Prop_Nishi}.
We start with the definition of contiguity.
Let two sequences $p_n$, $p^*_n\in\mathcal P(\Omega_n)$ for $n\in\mathbb Z_{>0}$ on the same spaces $\Omega_n$ be given.
Then $(p_n)_n$ is contiguous with respect to $(p^*_n)_n$ if for every $\varepsilon\in\mathbb R_{>0}$ there exists $n_0\in\mathbb Z_{>0}$ and $\delta\in\mathbb R_{>0}$ such that for all $n\in\mathbb Z_{\ge n_0}$ and all events $\mathcal E\subseteq\Omega_n$ with $p^*(\mathcal E)<\delta$ we have $p(\mathcal E)<\varepsilon$. If further $(p^*_n)_n$ is contiguous with respect to $(p_n)_n$ then the two sequences are mutually contiguous.

The factor graph model introduced in the following has the same law as the model discussed in Section \ref{Sec_Noela} for $\Theta=0$ and a prescribed way of obtaining the measures $\sP$.
The reason for explicitly introducing the model is to build the connection to Section \ref{vom} and Section \ref{llt} for one, and further simplifying the notation for brevity.
Another feature is that factor graphs are defined or all possible distribution sequences (corresponding to degree sequences in the standard case), which is useful and required to obtain concentration results in the upcoming sections.
\subsection{Product measure families}\label{mc_measures}
For $k\in\ZZ_{\ge 0}$ let $\mcFPsi_k=\RR_{>0}^{\mcColSp^k}$ denote the set of functions $\mcFpsi:\mcColSp^k\rightarrow\RR_{>0}$.
Further, fix a family $(\mcFdeg_\mcFInd,\mcFPsiP_\vmInd)_{\mcFInd\in\mcFIndSp}$ with $\mcFIndSp\subseteq\ZZ_{\ge 0}$ and $\mcFdeg_\mcFInd\in\ZZ_{\ge 0}$, $\mcFPsiP_\mcFInd\in\cP(\mcFPsi_{\mcFdeg_\mcFInd})$ for $\mcFInd\in\mcFIndSp$.
Let $\mcFRpsi_\mcFInd\sim\mcFPsiP_\mcFInd$,
$\mcFEpsi_\mcFInd=\Erw[\mcFRpsi_\mcFInd]$,
$\mcFEZ_{\mcFInd}=\sum_{\mcFHEA}\mcFEpsi_\mcFInd(\mcFHEA)$,
$\mcFExi_{\mcFInd}=\mcFEZ_{\mcFInd}q^{-\mcFdeg_\mcFInd}$ and
$\mcFED_\mcFInd=\mcFEZ_\mcFInd^{-1}\mcFEpsi_\mcFInd\in\cP(\mcColSp^{\mcFdeg_\mcFInd})$ if $\mcFEZ_\mcFInd>0$ and the one-point mass on the empty assignment otherwise.
The family $(\mcFdeg_\mcFInd,\mcFPsiP_\mcFInd)_{\mcFInd\in\mcFIndSp}$ satisfies {\bf BAL'} if
\begin{align*}
\sum_{\mcFHEA}\mcFED_\mcFInd(\mcFHEA)\prod_{h\in[\mcFdeg_\mcFInd]}p(\mcFHEA_h)\le q^{-\mcFdeg_\mcFInd}
\end{align*}
for all $p\in\cP(\mcColSp)$ and $\mcFInd\in\mcFIndSp$ with $\mcFdeg_\mcFInd>0$.
Further, notice that $(\mcFdeg_\mcFInd,\mcFED_\mcFInd)_{\mcFInd\in\mcFIndSp}$ satisfies {\bf SPAN} and the induced lattice $\lltLat$ discussed in Section \ref{llt} is $\mathbb Z^{q-1}$ and in particular $h\equiv 1$.
Analogous to the coupling in Section \ref{Sec_Noela} we introduce a new index $\mcFInd^\circ\in\ZZ_{\ge 0}$ with $\mcFdeg_{\mcFInd^\circ}=1$ and $\mcFPsiP_{\mcFInd^\circ}$ being the one-point mass on $\mcFpsi\equiv 1$, and let $\mcFIndSp^\circ=\mcFIndSp\cup\{\mcFInd^\circ\}$. Notice that this modification does not change the associated lattice and further {\bf SPAN} still holds.

For the sake of symmetry we also fix a family $(\mcVdeg_\mcVInd)_{\mcVInd\in\mcVIndSp}$ with $\mcVIndSp\subseteq\ZZ_{\ge 0}$ and $\mcVdeg_\mcVInd\in\ZZ_{\ge 0}$ for $\mcVInd\in\mcVIndSp$.
Further, let $\mcVED_\mcVInd\in\cP(\mcColSp^{\mcVdeg_\mcVInd})$ be given by $\mcVED_\mcVInd(\mcCol 1_{[\mcVdeg_\mcVInd]})=q^{-1}$ for $\mcCol\in\mcColSp$ and notice that $(\mcVdeg_\mcVInd,\mcVED_\mcVInd)_{\mcVInd\in\mcVIndSp}$ satisfies {\bf SPAN} by definition.
In the remainder we tacitly assume that both $\mcVIndSp$ and $\mcFIndSp$ are non-trivial, i.e.~not all degrees are zero.
\subsection{Distribution sequences}\label{mc_profiles}
For $n\in\ZZ_{>0}$ we let $\mcSeqSp_n$ denote the set of distribution sequences, i.e.
\begin{align*}
\mcSeqSp_n=\left\{(m,\mcVSInd,\mcFSInd):m\in\ZZ_{\ge 0},\mcVSInd\in\mcVIndSp^n,\mcFSInd\in\mcFIndSp^m\right\}\textrm{.}
\end{align*}
For $\mcT=(m,\mcVSInd,\mcFSInd)\in\mcSeqSp_n$ we use the same shorthands as in Section \ref{llt} and Section \ref{vom}, e.g.~$\mcVdeg_i=\mcVdeg_{\mcVInd_i}$ for $i\in[n]$.
Further, let $\mcVDTot_\mcT=\sum_{i\in[n]}\mcVdeg_i$, $\mcFDTot_\mcT=\sum_{i\in[m]}\mcFdeg_i$, $\mcDTot_\mcT=\max(\mcVDTot_\mcT,\mcFDTot_\mcT)$, $\mcDVDelta(\mcT)=\mcDTot_\mcT-\mcVDTot_\mcT$ and $\mcDFDelta(\mcT)=\mcDTot_\mcT-\mcFDTot_\mcT$ denote the total degrees and missing half-edges on both sides. We use the notions $\mcVHESp_\mcT$ and $\mcFHESp_\mcT$ from Section \ref{Sec_Noela}, but introduce two sets $\mcDVDeltaSp(\mcT)$, $\mcDFDeltaSp(\mcT)$ with $|\mcDVDeltaSp(\mcT)|=\mcDVDelta(\mcT)$, $\mcDVDeltaSp(\mcT)\cap\mcVHESp_\mcT=\emptyset$ and $|\mcDFDeltaSp(\mcT)|=\mcDFDelta(\mcT)$, $\mcDFDeltaSp(\mcT)\cap\mcFHESp_\mcT=\emptyset$.

As indicated above and in Section \ref{Sec_Noela} and using the shorthand $m^\circ=m+\mcDFDelta(\mcT)$, we let $\mcT^\circ=(m^\circ,\mcVSInd,\mcFSInd^\circ)$ with $\mcFSInd^\circ\in(\mcFIndSp^\circ)^{m+\mcDFDelta(\mcT)}$ given by $\mcFSInd^\circ_{[m]}=\mcFSInd$ and $\mcFInd_{i}=\mcFInd^\circ$ otherwise.
\subsection{Factor graphs}\label{mc_graphs}
For given $n\in\ZZ_{>0}$ and $\mcT=(m,\mcVSInd,\mcFSInd)\in\mcSeqSp_n$ a factor graph $\mcG$ is given by a bijection
$\mcBij:\mcVHESp_\mcT\cup\mcDVDeltaSp(\mcT)\rightarrow\mcFHESp_\mcT\cup\mcDFDeltaSp(\mcT)$ and weights $\mcFpsi_{\mcFL_i}\in\mcFPsi_{\mcFdeg_i}$ for each factor $\mcFL_i\in\mcFLSp_m$, $i\in[m]$.
Let
\begin{align*}
\mcFLSp(\mcG)=\{\mcFL:(\mcFL,h)\in\mcFHESp_\mcT,\mcBij^{-1}(\mcFL,h)\in\mcVHESp_\mcT\}
\end{align*}
denote the subset of factors that are not connected to the dummy variables $\mcDVDeltaSp(\mcT)$.
For $\mcVA\in\mcColSp^{\mcVLSp_n}$ let $\mcVSHEA_{\mcT,\mcVA}=(\mcVA_{\mcVL})_{(\mcVL,h)\in\mcVHESp_\mcT}$
denote the assignment to the half-edges excluding dummies and further $\mcCFA_{\mcT,\mcVA}$ the corresponding absolute colour frequencies, i.e.~$\mcCFA_{\mcT,\mcVA}(\mcCol)=\sum_{\mcVL,h}\vecone\{\mcVHEA_{\mcVL,h}=\mcCol\}$ for $\mcCol\in\mcColSp$, and notice that these notions do not depend on $\mcG$.
Further, let $\mcFSHEA_{\mcG,\mcVA}=\mcFSHEA_{\mcBij,\mcVA}\in\mcColSp^{\mcFHESp_\mcT\cup\mcDFDeltaSp(\mcT)}$ be given by $\mcFHEA_{\mcG,\mcVA,h}=\mcVHEA_{\mcT,\mcVA,\mcBij^{-1}(h)}$ for $h\in\mcFHESp_\mcT\cup\mcDFDeltaSp(\mcT)$ with $\mcBij^{-1}(h)\in\mcVHESp_\mcT$ and undefined otherwise.
Finally, let $\mcFpsi_\mcG(\mcVA)=\prod_{\mcFL\in\mcFLSp(\mcG)}\mcFpsi_{\mcFL}(\mcFHEA_{\mcG,\mcVA,\mcFL})$, with $\mcZ_\mcG=\sum_\mcVA\mcFpsi_\mcG(\mcVA)\in\RR_{>0}$, $\mcBMD_\mcG=\mcZ_\mcG^{-1}\mcFpsi_\mcG\in\cP(\mcColSp^{\mcVLSp_n})$ unchanged and $\mcRVAPost_\mcG\sim\mcBMD_\mcG$.
\subsection{Random factor graphs}\label{mc_graphs_2}
For $n\in\ZZ_{>0}$ and $t=(m,\mcVSInd,\mcFSInd)\in\mcSeqSp_n$ we obtain the null model $\mcRG_t$ by drawing a uniformly random bijection $\mcRBij:\mcVHESp_\mcT\cup\mcDVDeltaSp(\mcT)\rightarrow\mcFHESp_\mcT\cup\mcDFDeltaSp(\mcT)$ and independently drawing the weight functions $\mcFpsi_{\mcFL_i}$ from $\mcFPsiP_i$.
Using $\mcFEpsi_{\mcT}=\Erw[\mcFpsi_{\mcRG_\mcT}]$ the teacher-student scheme $\mcRG^*_\mcT(\mcVA)$ with ground truth $\mcVA\in\mcColSp^{\mcVLSp_n}$ is given by the Radon-Nikodym derivative $\mcFpsi_\mcG(\mcVA)/\mcFEpsi_\mcT(\mcVA)$ with respect to $\mcRG_\mcT$.
Further, using $\mcZE_\mcT=\Erw[\mcZ_{\mcRG_\mcT}]$ the Nishimori ground truth $\mcRVAN_\mcT\in\mcColSp^{\mcVLSp_n}$ is given by
$\pr[\mcRVAN_\mcT=\mcVA]=\mcFEpsi_\mcT(\mcVA)/\mcZE_\mcT$.
Finally, we use the shorthand $\mcRFSHEA^*_\mcT(\mcVA)=\mcFSHEA_{\mcRG^*_\mcT(\mcVA),\mcVA}$ to denote the assignment to the factor side half-edges for a given ground truth.
Notice that the models $\mcRG_\mcT$, $\mcRG^*_\mcT(\mcVA)$ and $\mcRG_{\mcT^\circ}$, $\mcRG^*_{\mcT^\circ}(\mcVA)$ are equal in that they show exactly the same behaviour and only differ in the explicit modelling of the dummy factors in the latter case.
\subsection{Typical distribution sequences}\label{mc_profiles_typ}
A sequence $(\mcSeq_n)_n\subseteq\mcSeqSp_n$ satisfies {\bf MC} if the following holds.
The family $(\mcFdeg_\mcFInd,\mcFPsiP_\mcFInd)_{\mcFInd\in\mcFIndSp}$ satisfies {\bf BAL'}.
There exists $(\vmIndSp_n)_n\subseteq\mcVIndSp^n$ satisfying {\bf GEN}, {\bf VAR} and {\bf SKEW} such that for all $n\in\ZZ_{>0}$ and $(m,\mcVSInd,\mcFSInd)\in\mcSeq_n$ we have $\mcVSInd\in\vmIndSp_n$.
There exists $(\vmIndSp_m)_m\subseteq\mcFIndSp^m$ satisfying {\bf GEN}, {\bf VAR} and {\bf SKEW} such that for all $n\in\ZZ_{>0}$ and $(m,\mcVSInd,\mcFSInd)\in\mcSeq_n$ we have $\mcFSInd\in\vmIndSp_m$.
Finally, for all $n\in\ZZ_{>0}$ and $(m,\mcVSInd,\mcFSInd)\in\mcSeq_n$ we have
\begin{align*}
\sum_{i\in[n]}\mcVdeg_i\ge\sum_{i\in[m]}\mcFdeg_i\textrm{,}
\end{align*}
using the conventions from Section \ref{llt}, e.g.~$\mcVdeg_i=\mcVdeg_{\mcVInd_i}$.

Again, notice that for any sequence $(\mcSeq_n)_n\subseteq\mcSeqSp_n$ satisfying {\bf MC} the sequence $(\mcSeq^\circ_n)_n$ given by $\mcSeq^\circ_n=\{\mcT^\circ:\mcT\in\mcSeq\}$, $n\in\ZZ_{>0}$, satisfies {\bf MC} as well, and $\mcFInd^\circ$ spans the lattice as does any other index with non-trivial degree.
Further, notice that these assumptions ensure the existence of $c\in(0,1)$ with $cm\le\mcFDTot_\mcT\le\mcVDTot_\mcT\le c^{-1}n$ using {\bf GEN} on the factor side for the first inequality and {\bf VAR} on the variable side for the last inequality, so $m\in O(n)$ uniformly for all $\mcT\in\mcSeq_n$.
Using {\bf VAR} on the factor side and {\bf GEN} on the variable side we obtain $c\in(0,1)$ with $m^\circ\ge\max(m,cn-c^{-1}m)\ge\frac{c}{1+c^{-1}}n$, so $m^\circ$ is uniformly linear in $n$ for all $\mcT\in\mcSeq_n$ since the arguments for $m\in O(n)$ also apply to $m^\circ$. Let $\mcDegDens_-$, $\mcDegDens_+\in\RR_{>0}$ be corresponding bounds, i.e.~$\mcDegDens_-\le m^\circ\le\mcDegDens_+n$ for all $\mcT\in\mcSeq_n$ and $n$.

The arguments in the remainder of this section will clarify that using $(\mcSeq^\circ_n)_n$ is not only an alternative modelling approach, but superior to using $(\mcSeq_n)_n$. Intuitively, factors are \emph{not pruned} (or missing in any sense), but \emph{replaced} by trivial factors such that the total degree imposed by the variable side is met.
Hence, for consistency we consider $\mcFLSp^\circ_\mcT=\{\mcFL_1,\dots,\mcFL_{m+\mcDFDelta(t)}\}$ an extension of $\mcFLSp_m$ for $t=(m,\mcVSInd,\mcFSInd)\in\mcSeq_n$, let $\mcDFDeltaSp(t)=\{(\mcFL_i,1):m<i\le m+\mcDFDelta(t)\}$ and $\mcFHESp^\circ_\mcT=\mcFHESp\cup\mcDFDeltaSp(t)$.
\subsection{Random distributions}\label{mc_degseq}
In order to complete the picture recall $\vm$, $\vm_\varepsilon$, $\vd$, $\vk$, $\vpsi_k$ from the introduction and $\vt_n$ from Section \ref{dd}. For the standard case we let $\mcVIndSp$ be the support of $\vd$ with $\mcVdeg_\mcVInd=\mcVInd$ for $\mcVInd\in\mcVIndSp$, i.e.~there is no distinction between labels and degrees.
Analogously, we let $\mcFIndSp$ be the support of $\vk$ with $\mcFdeg_\mcFInd=\mcFInd$ for $\mcFInd\in\mcFIndSp$ and $\mcFPsiP_\mcFInd$ be the law of $\vpsi_\mcFInd$.
Notice that $\mcSeqSp_n=\cT^*_n$, further $\cT^\circ_n$ satisfies {\bf MC} and $\mcT^\circ=\mcT$ for all $\mcT\in\cT_n$.
Hence, compared to the discussion in Section \ref{Sec_cavs} and Section \ref{Sec_Noela} we slightly change the model in that we do not condition on suitable degree sequences, but define the factor graphs for all possible degree sequences.
However, as discussed in Section \ref{Sec_cavs} (and as is evident from Section \ref{dd}) the consistency condition
\begin{align}\label{eqepsdegs}
\sum_{i=1}^n\vd_i\geq\sum_{i=1}^{\vm_\eps}\vk_i.
\end{align}
is satisfied with very high probability and hence the change of the model may be considered of purely technical nature.

For $\varepsilon\in(0,1)$ we let $\vt^*_{\varepsilon,n}=(\vm_\varepsilon,(\vd_i)_{i\in[n]},(\vk_i)_{i\in[\vm_\varepsilon]})$ be the analogue of $\vt^*_n$, but contrary to $\vt_n$ we let $\vt_{\varepsilon,n}=\vt^*_{\varepsilon,n}$ as discussed above. With $\tau^*=((1-\varepsilon)\d/\k,p_{\mathrm{d}},p_{\mathrm{k}})$ and the metric $\Delta$ from Section \ref{dd} it is immediate from the results of Section \ref{dd} that there exists $r_n=o(1)$ such that $\tau(\vt_{\varepsilon,n})\in\cB_{r_n}(\tau^*_{\varepsilon})$ with high probability and Proposition \ref{dd_typ_prop} also holds for $\varepsilon>0$, so we can define $\cT^\circ_{\varepsilon,n}$ analogously. As discussed in Section \ref{dd} we can choose $r_n$ such that uniform bounds hold for $t=(m,\sd,\sk)\in\cT^\circ_n$, and in particular $\sum_{i\in[n]}d_i>\sum_{i\in[m]}k_i$.
We define $\mcVIndSp$ and $\mcFIndSp$ as before and notice that as opposed to the boundary case $\varepsilon=0$ above we now have $\mcT^\circ\neq\mcT$ for all $\mcT\in\cT^\circ_n$.

The conditions imposed by $\cT^\circ_n$ ensure that the number of factors is asymptotically equivalent to $(1-\varepsilon)\d/\k$ and the total degrees are asymptotically equivalent to $\d n$ and $(1-\varepsilon)\d n$ on the variable and factor side respectively, hence the absolute frequency of $\mcFInd^\circ$ in $\mcT^\circ$ is asymptotically equivalent to $\varepsilon\d n$.
Hence, for $\mcT=(m,\mcVSInd,\mcFSInd)\in\cT^\circ_{\varepsilon,n}$ with $\mcT^\circ=(m^\circ,\mcVSInd,\mcFSInd^\circ)$ the number $m^\circ$ of factors including the dummy factors is asymptotically equivalent to $(\frac{1-\varepsilon}{\k}+\varepsilon)\d n$.
Further, the relative frequencies $\vmIndP_\mcVSInd$ (introduced in Section \ref{llt}) converge to $p_{\mathrm{d}}$ with respect to the metric $\vmIndPD$ introduced in Section \ref{vom},
and the frequencies $\vmIndP_{\mcFSInd}$ converge to $p_{\mathrm{k}}$ with respect to $\vmIndPD$.
Hence, the frequencies $\vmIndP_{\mcFSInd^\circ}$ converge to $p^\circ_{\mathrm{k}}\in\cP(\mcFIndSp^\circ)$ given by $p^\circ_{\mathrm{k}}(\mcFInd^\circ)=\varepsilon/(\varepsilon+(1-\varepsilon)/\k)$ and $p^\circ_{\mathrm{k}}(\mcFInd)=(1-p^\circ_{\mathrm{k}}(\mcFInd^\circ))p_{\mathrm{k}}(\mcFInd)$ otherwise.

In a nutshell, the arguments above stress the fact that we \emph{always only} consider factor graphs where the total degrees of the variable side and of the factor side \emph{are equal}, a change of perspective that is essential for the upcoming sections.
\subsection{Mutual contiguity}\label{mc_mr}
Mutual contiguity of $\mcRVA^*$ and $\mcRVAN_\mcT$ uniformly over $\mcT\in\mcSeq_n$ follows with standard arguments from the following proposition.
Further implications are discussed in Section \ref{mc_hence}
Finally, in Section \ref{mc_pinning} we will briefly discuss why these results are entirely invariant to pinning.
\begin{proposition}\label{mcp_main_prop}
For all sequences $(\mcSeq_n)_n\subseteq\mcSeqSp_n$ satisfying {\bf MC} and $\varepsilon\in(0,1)$ there exist $c\in(0,1)$, $r\in\RR_{>0}$ and $n_0\in\ZZ_{>0}$ such that for all $n\in\ZZ_{\ge n_0}$, all $\mcT\in\mcSeq_n$, all $\mcVA\in\cE_\mcT$, with
$\cE_\mcT=\{\mcVA\in\mcColSp^{\mcVLSp_n}:\|\mcCFA_{\mcT,\mcVA}-\mcVDTot_\mcT u_{\mcColSp}\|_2<r\sqrt{n}\}$, we have
$\pr[\mcRVA^*\in\cE_\mcT]$, $\pr[\mcRVAPost_{\mcRG^*_\mcT(\mcRVA^*)}\in\cE_\mcT]$,
$\pr[\mcRVAN_\mcT\in\cE_\mcT]>1-\varepsilon$ and $c<\pr[\mcRVA^*=\mcVA]/\pr[\mcRVAN_\mcT=\mcVA]<c^{-1}$ .
\end{proposition}
From now on we consider $(\mcSeq_n)_n\subseteq\mcSeqSp_n$ satisfying {\bf MC} fixed.
In order to show Proposition \ref{mcp_main_prop} we first determine the asymptotics of the normalization constant of
$\mcRVAN_\mcT$, i.e.~the first moment $\mcZE_\mcT=\Erw[\mcZ_{\mcRG_\mcT}]$.
\begin{proposition}\label{mcp_first_moment}
Uniformly for $\mcT=(m,\mcVSInd,\mcFSInd)\in\mcSeq_n$ we have $\mcZE_\mcT=\Theta(\mcZ^*_{\mcT})$ with $\mcZ^*_{\mcT}=q^n\prod_{i\in[m]}\xi_i$.
\end{proposition}
From the proof of Proposition \ref{mcp_first_moment} we directly obtain tail bounds and a local limit theorem for the colour frequencies of $\mcRVAN_\mcT$.
For brevity let $\mcRCFRN_\mcT=\frac{1}{\mcVDTot_\mcT}\mcCFA_{\mcT,\mcRVAN_\mcT}\in\cP(\mcColSp)$ denote the random relative color frequencies on the half-edges under $\mcRVAN_\mcT$.
Recall from Section \ref{llt} that we have $\pr[\mcRCFRN_\mcT\in\lltLatR_\mcT]=1$ for $\mcT=(m,\mcVSInd,\mcFSInd)\in\mcSeq_n$ and $n\in\ZZ_{>0}$, where $\lltLatR_\mcT=\lltLatR_{\mcVSInd}$ is the set induced by the lattice $\lltLat$ obtained from $(\mcVdeg_\mcVInd,\mcVED_\mcVInd)_{\mcVInd\in\mcVIndSp}$.
\begin{proposition}\label{mcp_tails}
There exist constants $c$, $c'\in\RR_{>0}$ such that for all $n\in\ZZ_{>0}$, $\mcT\in\mcSeq_n$ and $r\in\RR_{\ge 0}$ we have
\begin{align*}
\pr[\|\mcRCFRN_\mcT -u_{\mcColSp}\|_2\ge r]\le c'\exp(-cr^2n)\textrm{.}
\end{align*}
\end{proposition}
In the following we may use the notions for $\mcT=(m,\mcVSInd,\mcFSInd)\in\mcSeq_n$ implied by the notions introduced in Section \ref{vom} and Section \ref{llt} without explicitly introducing them, e.g.~$\mcRVInd_\mcT=\mcRVInd_\mcVSInd$, $\mcRFInd_\mcT=\mcRFInd_\mcFSInd$ for the random indices and $\mcRVdeg_\mcT=\mcRVdeg_\mcVSInd$, $\mcRFdeg_\mcT=\mcRFdeg_\mcFSInd$ for the random degrees.
In addition, let $\mcVCov{\mcT}=\Sigma_{\mcVSInd}$, $\mcFCov{\mcT}=\Sigma_{\mcFSInd^\circ}$ as introduced in Section \ref{llt} and notice that $\Sigma_{\mcFSInd^\circ}\neq\Sigma_{\mcFSInd}$ in general.
Further, let $\mcECov{\mcT}=\frac{\Erw[\mcRVdeg_\mcT]}{\Erw[\mcRVdeg_\mcT^2]}\mcVCov{\mcT}$ and let $\mcCov{\mcT}$ be given by $\mcCov{\mcT}^{-1}=\mcVCov{\mcT}^{-1}+\mcFCov{\mcT}^{-1}-\mcECov{\mcT}^{-1}$.
Let $\mcVDGCD=\gcd\{\mcVdeg_\mcVInd:\mcVInd\in\mcVIndSp\}$ denote the greatest common divisor of the attainable variable side degrees.
\begin{proposition}\label{mcp_llt}
For $r_n=\Theta(\sqrt{\ln(n)/n})$, uniformly in $\mcT\in\mcSeq_n$ and $\mcCFR\in\lltLatR_\mcT\cap\cB_{r_n}(u_{\mcColSp})$ we have
\begin{align*}
\pr[\mcRCFRN_\mcT=\mcCFR]=(1+o(1))\frac{\mcVDGCD^{q-1}}{\sqrt{\mcVDTot_\mcT}^{q-1}}\phi_\mcT\left(\sqrt{\mcVDTot_\mcT}(\mcCFR-u_{\mcColSp})_{[q-1]}\right)\textrm{,}
\end{align*}
where $\phi_\mcT$ denotes the density of $\cN(0_{[q-1]},\mcCov{\mcT})$ and $\mcCov{\mcT}$.
Further, $\mcCov{\mcT}^{-1}$ is positive definite and $\|\mcCov{\mcT}\|_2$, $\|\mcCov{\mcT}^{-1}\|_2=\Theta(1)$ uniformly in $\mcT\in\mcSeq_n$.
\end{proposition}
This local limit theorem for $\mcRCFRN_\mcT$ with the local limit theorem for $\mcRCFR^*_\mcT=\frac{1}{\mcVDTot_\mcT}\mcCFA_{\mcT,\mcRVA^*}\in\lltLatR_\mcT$ from Section \ref{llt} and the tail bounds above is sufficient to derive Proposition \ref{mcp_main_prop}.
%
%
\subsection{Proof of Proposition \ref{mcp_first_moment}}\label{mcp_first_moment_proof}
Fix parameters $\epsGen$, $\secB$ and $\thiB_n$ to satisfy the assumptions {\bf GEN}, {\bf VAR} and {\bf SKEW} jointly for the variable and factor side, this means in particular that $\thiB_n\in o(\sqrt{\ln(n)^3/n})$ is a uniform third moment bound for the variable side distribution sequences in $n$ and also for \emph{all} factor side sequences $t^\circ$, with $m^\circ$ ranging from $\mcDegDens_-n$ to $\mcDegDens_+n$.

For $\mcT\in\mcSeq_n$ and $\mcCFR$ in the support of $\mcRCFRN_\mcT$ we have $\binom{\mcVDTot_\mcT}{\mcVDTot_\mcT\mcCFR}\prod_\mcCol\mcCFR(\mcCol)^{\mcVDTot_\mcT\mcCFR(\mcCol)}\ge\binom{\mcVDTot_\mcT+q-1}{q-1}^{-1}$, i.e.~the maximal probability of the multinomial is at least the uniform. Hence, the uniform bounds on $\Erw[\mcRVdeg_{\mcT}]$ for $\mcT\in\mcSeq_n$ yield a uniform lower bound $\binom{\mcVDTot_\mcT}{\mcVDTot_\mcT\mcCFR}\prod_\mcCol\mcCFR(\mcCol)^{\mcVDTot_\mcT\mcCFR(\mcCol)}=\Omega(n^{-(q-1)})$.
With Proposition \ref{llt_tails} we have
\begin{align*}
\pr[\|\mcRCFR^*_\mcT-u_{\mcColSp}\|_2\ge r]\le c'\exp(-cr^2n)
\end{align*}
for all $n\in\ZZ_{>0}$, $\mcT\in\mcSeq_n$ and $r\in\RR_{\ge 0}$.
For $t=(m,\mcVSInd,\mcFSInd)\in\mcSeq_n$, with $m^\circ=m+\mcDFDelta(\mcT)$, $\mcCFA_{\mcFSHEA}$ denoting the colour frequencies of $\mcFSHEA\in\mcColSp^{\mcFHESp^\circ_\mcT}$ and using arguments analogous to Section \ref{Sec_Noela} this yields
\begin{align*}
\frac{\mcZE_\mcT}{\mcZ^*_\mcT}&=\sum_\mcVA\frac{q^{\mcVDTot_\mcT}}{q^n\binom{\mcVDTot_\mcT}{\mcCFA_{\mcT,\mcVA}}}\sum_{\mcFSHEA}\vecone\{\mcCFA_{\mcFSHEA}=\mcCFA_{\mcT,\mcVA}\}\prod_{i\in[m^\circ]}\mcFED_{i}(\mcFHEA_{\mcFL_i})
=r_{t,+}+r_{t,-}\textrm{,}\\
r_{t,+}&=\sum_{\mcCFA\in\cB_n}\frac{\pr[\mcCFA_{\mcT,\mcRVA^*}=\mcCFA]\pr[\mcCFA_{\mcRFSHEA^*_\mcT}=\mcCFA]}{\binom{\mcVDTot_\mcT}{\mcCFA}q^{-\mcVDTot_\mcT}}\textrm{,}\\
r_{t,-}&=q^{\mcVDTot_\mcT}\sum_{\mcCFA\not\in\cB_n}\frac{\pr[\mcCFA_{\mcT,\mcRVA^*}=\mcCFA]}{\binom{\mcVDTot_\mcT}{\mcCFA}}\sum_{\mcFSHEA}\vecone\{\mcCFA_{\mcFSHEA}=\mcCFA\}\prod_{i\in[m^\circ]}\mcFED_{i}(\mcFHEA_{\mcFL_i})\textrm{,}
\end{align*}
with $\cB_n=\mcVDTot_\mcT\cB_{r_n}(u_{\mcColSp})$ and $\mcRFSHEA^*_\mcT\sim\bigotimes_{i\in[m^\circ]}\mcFED_{i}$.
For $r_{t,-}$, using $\mcCFR_\mcCFA=\mcVDTot_\mcT^{-1}\mcCFA$ and {\bf BAL'} we get
\begin{align*}
r_{t,-}&=q^{\mcVDTot_\mcT}\sum_{\mcCFA\not\in\cB_n}\frac{\pr[\mcCFA_{\mcT,\mcRVA^*}=\mcCFA]}{\binom{\mcVDTot_\mcT}{\mcCFA}\prod_\mcCol\mcCFR_\mcCFA(\mcCol)^{\mcCFA(\mcCol)}}\sum_{\mcFSHEA}\vecone\{\mcCFA_{\mcFSHEA}=\mcCFA\}\prod_{i\in[m^\circ]}\left(\mcFED_{i}(\mcFHEA_{\mcFL_i})\prod_{h\in[\mcFdeg_{i}]}\mcCFR_\mcCFA(\mcFHEA_{\mcFL_i,h})\right)\\
&\le q^{\mcVDTot_\mcT}\sum_{\mcCFA\not\in\cB_n}\frac{\pr[\mcCFA_{\mcT,\mcRVA^*}=\mcCFA]}{\binom{\mcVDTot_\mcT}{\mcCFA}\prod_\mcCol\mcCFR_\mcCFA(\mcCol)^{\mcCFA(\mcCol)}}\sum_{\mcFSHEA}\prod_{i\in[m^\circ]}\left(\mcFED_{i}(\mcFHEA_{\mcFL_i})\prod_{h\in[\mcFdeg_{i}]}\mcCFR_\mcCFA(\mcFHEA_{\mcFL_i,h})\right)\\
&\le\sum_{\mcCFA\not\in\cB_n}\frac{\pr[\mcCFA_{\mcT,\mcRVA^*}=\mcCFA]}{\binom{\mcVDTot_\mcT}{\mcCFA}\prod_\mcCol\mcCFR_\mcCFA(\mcCol)^{\mcCFA(\mcCol)}}
=O(n^q\pr[\mcCFA_{\mcT,\mcRVA^*}\not\in\cB_n])
=O(n^q\exp(-cr_n^2n))
\end{align*}
uniformly in $\mcT\in\mcSeq_n$. Hence, for any $a\in\RR_{>0}$, all $c^*\in\RR_{>0}$ large enough and with $r_n=c^*\sqrt{\ln(n)/n}$ we have $r_{t,-}=o(n^{-a})$. This completes the discussion of the tails.

Next, we turn to the asymptotics of $r_{t,+}$.
Preparing the application of the local limit theorem \ref{llt_local} \emph{and} the large deviation result \ref{llt_global} jointly for the variable side and the factor side, we proceed with care.
First, recall the existence of sequences satisfying {\bf MC} that cover $\mcVSInd$ and $\mcFSInd^\circ$ respectively for all $t=(m,\mcVSInd,\mcFSInd)\in\mcSeq_n$ and $n\in\ZZ_{>0}$.
Further, we fix a sequence $R_{m^\circ}=\Theta(\sqrt{\ln(m^\circ)/m^\circ})$, with asymptotics in $m^\circ$, for the factor side and sufficiently large such that $\cB_n\subseteq\cB_{R_{m^\circ}}(u_{\mcColSp})$ for all sufficiently large $n$ and uniformly in $m^\circ$ for $t\in\mcSeq_n$.
Further, we fix a compact set $\cP^*\subseteq\cP^\circ(\mcColSp)$, covering $\cP^\circ(\mcColSp)$ but for a small residue at the boundary.  As discussed in the proof of theorem \ref{llt_local} using the first order approximation of the homeomorphism $\iota$ from Section \ref{vom}, we eventually have $\iota_{\mcVSInd}^{-1}(\cB_n)\subseteq\cP^*$ and $\iota_{\mcFSInd^\circ}^{-1}(\cB_n)\subseteq\cP^*$ for all $\mcT\in\mcSeq_n$ and $n$ sufficiently large.
Now, we first use the large deviation result \ref{llt_global} with the uniform error bounds.
Recalling that $\Erw[\mcRFdeg_{\mcT^\circ}]m^\circ=\mcFDTot_{\mcT^\circ}=\mcVDTot_\mcT$, using the notions from Section \ref{llt}, further for $\mcT=(m,\mcVSInd,\mcFSInd)\in\mcSeq_n$, $\mcCFA\in\cB_n$ in the support of $\mcCFA_{t,\mcRVA^*}$ and with $\mcCFR=\mcVDTot_\mcT^{-1}\mcCFA$, $\mcPrior=\iota^{-1}_{\mcVSInd}(\mcCFR)$, $\mcPrior'=\iota^{-1}_{\mcFSInd^\circ}(\mcCFR)$ we have
\begin{align*}
r_{\mcT,+}&=\left(1+O\left(\thiB_n\sqrt{\frac{\ln(n)^3}{n}}\right)\right)
\sum_{\mcCFA\in\cB_n}W_\mcT(\mcCFA)\textrm{,}\\
W_\mcT(\mcCFA)&=\frac{W_{\mathrm{V},\mcT}(\mcCFA)W_{\mathrm{F},\mcT}(\mcCFA)}{W_{\mathrm{E},\mcT}(\mcCFA)}
=\frac{\mcVDGCD^{q-1}}{\sqrt{\mcVDTot_\mcT^{q-1}}}\frac{\exp(-\mcVDTot_\mcT\alpha_\mcT(\mcCFA))}{\sqrt{2\pi}^{q-1}\sqrt{q^{q}\det(\mcCov{\mcVSInd,\mcPrior}\mcCov{\mcFSInd^\circ,\mcPrior'})}}\\
\alpha_\mcT(\mcCFA)&=\alpha_{\mathrm{v},\mcT}(\mcCFA)+\alpha_{\mathrm{f},\mcT}(\mcCFA)-\alpha_{\mathrm{e}}(\mcCFA)\\
W_{\mathrm{V},\mcT}(\mcCFA)&=\frac{\mcVDGCD^{q-1}}{\sqrt{\mcVDTot_\mcT^{q-1}}}\frac{\exp\left(-\mcVDTot_\mcT\alpha_{\mathrm{v},\mcT}(\mcCFA)\right)}{\sqrt{2\pi}^{q-1}\sqrt{\det(\mcCov{\mcVSInd,\mcPrior})}}\textrm{, }
\alpha_{\mathrm{v},\mcT}(\mcCFA)=\frac{1}{\Erw[\mcRVdeg_{\mcT^\circ}]}\Erw[\KL{\mcVED_{\mcPrior,\mcRVInd_{\mcT^\circ}}}{\mcVED_{\mcRVInd_{\mcT^\circ}}}]\\
W_{\mathrm{F},\mcT}(\mcCFA)&=\frac{1}{\sqrt{\mcVDTot_\mcT^{q-1}}}\frac{\exp\left(-\mcVDTot_\mcT\alpha_{\mathrm{f},t}(\mcCFA)\right)}{\sqrt{2\pi}^{q-1}\sqrt{\det(\mcCov{\mcFSInd^\circ,\mcPrior'})}}\textrm{, }
\alpha_{\mathrm{f},\mcT}(\mcCFA)=\frac{1}{\Erw[\mcRFdeg_{\mcT^\circ}]}\Erw[\KL{\mcFED_{\mcPrior',\mcRFInd_{\mcT^\circ}}}{\mcFED_{\mcRFInd_{\mcT^\circ}}}]\\
W_{\mathrm{E},\mcT}(\mcCFA)&=\frac{1}{\sqrt{\mcVDTot_\mcT^{q-1}}}\frac{\exp\left(-\mcVDTot_\mcT\alpha_{\mathrm{e}}(\mcCFA)\right)}{\sqrt{2\pi}^{q-1}\sqrt{q^{-q}}}\textrm{, }
\alpha_{\mathrm{e}}(\mcCFA)=\KL{\mcCFR}{u_{\mcColSp}}\textrm{,}
\end{align*}
Using {\bf BAL'} we notice that for all $\mcFInd\in\mcFIndSp^\circ$ and $p\in\cP(\mcColSp^{\mcFdeg_\mcFInd})$ we have $\mcFED_\mcFInd\vert_*=u_{\mcColSp}$, obtained from the fact that $u_{\mcColSp}$ is a maximizer of $\vmPrior\mapsto\sum_{\mcFHEA}\mcFED_{\mcFInd}(\mcFHEA)\prod_h\vmPrior(\mcFHEA_h)$, hence a stationary point, and taking the first derivatives.
Further, using $p_*=p\vert_*$, we have
\begin{align*}
\KL{p}{\mcFED_\mcFInd}
&=\KL{p}{\mcFED_{p_*,\mcFdeg_\mcFInd}}+\ln\left(\frac{q^{-\mcFdeg_\mcFInd}}{\mcFEZ_{p_*,\mcFdeg_\mcFInd}}\right)+\mcFdeg_\mcFInd\ln(q)-\CE{p}{p_*^{\otimes k}}\\
&=\KL{p}{\mcFED_{p_*,\mcFdeg_\mcFInd}}++\ln\left(\frac{q^{-\mcFdeg_\mcFInd}}{\mcFEZ_{p_*,\mcFdeg_\mcFInd}}\right)+\mcFdeg_\mcFInd\KL{p_*}{u_{\mcColSp}}
\ge \mcFdeg_\mcFInd\KL{p_*}{u_{\mcColSp}}\textrm{.}
\end{align*}
With this result, the convexity of the relative entropy under $\hat{\mcRFInd}_{\mcT^\circ}=\vmRInd_{\hat{\vmIndP}_{\mcFSInd^\circ}}$ from Section \ref{vom}, i.e.~for $\mcFInd\in\mcFIndSp^\circ$ given by $\pr[\hat{\mcRFInd}_{\mcT^\circ}=\mcFInd]=\frac{\mcFdeg_\mcFInd}{\Erw[\mcRFdeg_{\mcT^\circ}]}\pr[\mcRFInd_{\mcT^\circ}=\mcFInd]$, and the fact that 
$\mcCFR=\Erw[\mcFED_{\mcPrior',\hat{\mcRFInd}_{\mcT^\circ}}\vert_*]$ we obtain
\begin{align*}
\alpha_{\mathrm{f},\mcT}(\mcCFA)
\ge\frac{1}{\Erw[\mcRFdeg_{\mcT^\circ}]}\Erw[\mcRFdeg_{\mcT^\circ}\KL{\mcFED_{\mcPrior',\mcRFInd_{\mcT^\circ}}\vert_*}{u_{\mcColSp}}]
&\ge\alpha_{\mathrm{e}}(\mcCFA)
\end{align*}
Since we further have $\alpha_{\mathrm{f},\mcT}(\mcVDTot_\mcT u_{\mcColSp})-\alpha_{\mathrm{e}}(\mcVDTot_\mcT u_{\mcColSp})=0$ as discussed in Section \ref{llt} (notice that $\iota_{\mcFSInd^\circ}(u_{\mcColSp})=u_{\mcColSp}$ by {\bf BAL'}), this implies that the Hessian $H_\mcT$ of $f_\mcT(\mcCFR_{[q-1]})=\alpha_{\mathrm{f},\mcT}(\mcCFA)-\alpha_{\mathrm{e}}(\mcCFA)$ at $\mcCFR=u_{\mcColSp}$ is positive semi-definite.
With $B$ and composing the Hessians from the proof of Proposition \ref{llt_local} we have the Hessian $\mcECov{\mcT}^{-1}=qB^{\matTr}B$ (not depending on $\mcT$) for the latter contribution and $\mcFCov{\mcT}^{-1}$ for the former, so $H_\mcT=\mcFCov{\mcT}^{-1}-\mcECov{\mcT}^{-1}$.
Now, we follow the proof of Proposition \ref{llt_local} to obtain
\begin{align*}
\mcVDTot_\mcT\alpha_\mcT(\mcCFA)
&=\frac{1}{2}v^{\matTr}\mcCov{\mcT}^{-1}v+O\left(\thiB_n\sqrt{\frac{\ln(n)^3}{n}}\right)\textrm{, }
v=\sqrt{\mcVDTot_\mcT}^{-1}(\mcCFA-\mcVDTot_\mcT u_{\mcColSp})_{[q-1]}\textrm{.}
\end{align*}
The fact that $\mcCov{\mcT}^{-1}=\mcVCov{\mcT}^{-1}+H_\mcT$ shows that $\mcCov{\mcT}^{-1}$ is positive definite with $\|\mcCov{\mcT}\|_2\le\|\mcVCov{\mcT}\|_2=O(1)$ uniformly in $\mcT$. Further, since $\mcECov{\mcT}^{-1}=qI_{[q-1]}+q1_{[q-1]}1_{[q-1]}^{\matTr}$ is positive definite with eigenvalues $q$, $q^2$ (and determinant $q^q$) we get $\|\mcCov{\mcT}^{-1}\|_2\le\|\mcVCov{\mcT}^{-1}\|_2+\|\mcFCov{\mcT}^{-1}\|_2=O(1)$ uniformly in $\mcT$, which yields $\|\mcCov{\mcT}\|_2$, $\|\mcCov{\mcT}^{-1}\|_2=\Theta(1)$ uniformly in $\mcT$.
Using $\mcVCov{\mcT}=\frac{\Erw[\mcRVdeg_\mcT^2]}{\Erw[\mcRVdeg_{\mcT}]}\mcECov{\mcT}$,
we obtain $\det(\mcVCov{\mcT})=(\Erw[\mcRVdeg_\mcT^2]/\Erw[\mcRVdeg_{\mcT}])^{q-1}q^q$.
Following the proof we can take the asymptotics of the determinants to get
\begin{align*}
r_{t,+}&=\left(1+O\left(\thiB_n\sqrt{\frac{\ln(n)^3}{n}}\right)\right)\sum_{\mcCFA\in\cB_n}W_\mcT(\mcCFA)\textrm{,}\\
W_\mcT(\mcCFA)&=r^*_\mcT\frac{\mcVDGCD^{q-1}}{\sqrt{\mcVDTot_\mcT}^{q-1}}\frac{\exp\left(-\frac{1}{2}v^{\matTr}\mcCov{\mcT}^{-1}v\right)}{\sqrt{(2\pi)^{q-1}\det(\mcCov{\mcT})}}\textrm{, }
r^*_\mcT=\sqrt{\det\left(\frac{\Erw[\mcRVdeg_{\mcT}]}{\Erw[\mcRVdeg_\mcT^2]}\mcCov{\mcT}\mcFCov{\mcT}^{-1}\right)}\textrm{,}
\end{align*}
and notice that $r^*_\mcT=\Theta(1)$ uniformly since all eigenvalues of both matrices and the second moment are uniformly $\Theta(1)$.
Recall that $\mcCFA_{[q-1]}$ sits on a lattice of lengths $\mcVDGCD$ in all dimensions, hence $v$ is on a lattice with lengths $\mcVDGCD\sqrt{\mcVDTot_\mcT}^{-1}$ in all dimensions. Using the uniform bounds on the eigenvalues of $\mcCov{\mcT}$ we can approximate the Riemann sum by an integral over a growing domain of radius $c^*\sqrt{\ln(n)}$ (in $1_{[q]}^\perp$ with the $2$-norm), hence the error is of order $O(\sqrt{\ln(n)/n})$, i.e.~negligible. Choosing $c^*$ sufficiently large ensures that the extension of the domain comes at a negligible cost, say $\sqrt{n}^{-1}$, hence we have
\begin{align*}
\frac{\mcZE_\mcT}{\mcZ^*_\mcT}=\left(1+O\left(\thiB_n\sqrt{\frac{\ln(n)^3}{n}}\right)\right)r^*_\mcT=\Theta(1)\textrm{,}
\end{align*}
uniformly in $\mcT\in\mcSeq_n$. The constant $r^*_\mcT$ is of interest in its own right and provides further insights, but in this context we only need the uniform bounds.
%
%
\subsection{Proof of Proposition \ref{mcp_tails}}\label{mcp_tails_proof}
First, notice that the discussion in Section \ref{mcp_first_moment} directly translates to $\mcRCFRN_\mcT$ since
\begin{align*}
\pr[\mcRCFRN_\mcT=\mcCFR]&=\frac{\mcZ^*_\mcT}{\mcZE_\mcT}\frac{\pr[\mcCFA_{\mcT,\mcRVA^*}=\mcVDTot_\mcT\mcCFR]\pr[\mcCFA_{\mcRFSHEA^*_\mcT}=\mcCFA]}{\binom{\mcVDTot_\mcT}{\mcVDTot_\mcT\mcCFR}q^{-\mcVDTot_\mcT}}\\
&=\left(1+O\left(\thiB_n\sqrt{\frac{\ln(n)^3}{n}}\right)\right)r^{*-1}_\mcT\frac{\pr[\mcCFA_{\mcT,\mcRVA^*}=\mcVDTot_\mcT\mcCFR]\pr[\mcCFA_{\mcRFSHEA^*_\mcT}=\mcCFA]}{\binom{\mcVDTot_\mcT}{\mcVDTot_\mcT\mcCFR}q^{-\mcVDTot_\mcT}}
\end{align*}
uniformly in $\mcCFR\in\lltLatR_\mcT$ and $\mcT\in\mcSeq_n^\circ$.
Analogously to the bounds derived for $r_{-,\mcT}$ and with the relative error bounds above we find $c$, $c'\in\RR_{>0}$ such that $\pr[\|\mcRCFRN_\mcT -u_{\mcColSp}\|_2\ge r]\le c'n^q\exp(-cr^2n)$ for all sufficiently large $n$, $\mcT\in\mcSeq_n$ and $r\in\RR_{\ge 0}$. In particular, if $r\ge r_n=\ln(n)/\sqrt{n}$, then we can weaken $c$ to $c''\in(0,c)$ to maintain the bound since $\frac{q\ln(q)}{c\ln(n)^2}=o(1)$.
For $r<r_n$ we can use the discussion of $r_{+,\mcT}$ with uniform bounds on the smallest eigenvalue which is ensured to be uniformly bounded away from zero, uniform bounds on the leading coefficient and the integral approximation to obtain uniform bounds up to $r_n$ and use the bound above on the remainder, then taking the smaller constant for the exponent and the sum of coefficients.
This completes the proof for $n$ large.
For small $n\le n_0$ we notice that $\|\mcRCFRN_\mcT -u_{\mcColSp}\|_2\le 2$, so if the leading coefficient $c'$ is sufficiently large and the constant $c$ in the exponent sufficiently small, then the right-hand side is larger than $1$ for all choices of $n\le n_0$ and $r$ with $\pr[\|\mcRCFRN_\mcT -u_{\mcColSp}\|_2\ge r]>0$. This ensures existence of $c$, $c'$ such that the assertion holds.
%
%
\subsection{Proof of Proposition \ref{mcp_llt}}\label{mcp_llt_proof}
Proposition \ref{mcp_llt} is immediate from Section \ref{mcp_first_moment_proof} with the discussion in Section \ref{mcp_tails_proof}.
%
%
\subsection{Proof of Proposition \ref{mcp_main_prop}}\label{mcp_main_prop_proof}
Recall that the results of Section \ref{llt} are also valid for $\mcRCFR^*_\mcT$, hence for given $\varepsilon$ we can choose $r$ such that $\|\mcRCFRN_\mcT -u_{\mcColSp}\|_2<r(\Erw[\mcRVdeg_{\mcT}]\sqrt{n})^{-1}$ and $\|\mcRCFR^*_\mcT-u_{\mcColSp}\|_2<r(\Erw[\mcRVdeg_{\mcT}]\sqrt{n})^{-1}$ with probability at least $1-\varepsilon$, valid for all $n\in\ZZ_{>0}$ and $\mcT\in\mcSeq_n$ using the uniform bounds for $\Erw[\mcRVdeg_{\mcT}]$. Since the relative error bounds are uniform for both models, the Radon-Nikodym derivative $\pr[\mcRCFR^*=\mcCFR]/\pr[\mcRCFRN_\mcT=\mcCFR]$ is the ratio of the densities of the normal approximations up to a leading constant. This ratio can be uniformly bounded from above and below, uniformly for all sufficintly large $n$, $\mcT\in\mcSeq_n$ and all $\mcCFR$ in the $r(\Erw[\mcRVdeg_{\mcT}]\sqrt{n})^{-1}$ radius around $u_{\mcColSp}$. Finally, notice that
\begin{align*}
\pr[\mcRVAN_\mcT=\mcVA]=\frac{\mcZ^*_\mcT q^{\mcVDTot_\mcT}}{q^n\binom{\mcVDTot_\mcT}{\mcCFA_{\mcT,\mcVA}}}\pr[\mcCFA_{\mcRFSHEA^*_{\mcT}}=\mcCFA_{\mcT,\mcVA}]
\end{align*}
by the discussion at the beginning of Section \ref{mcp_first_moment_proof}, which means that $\mcRVAN_\mcT$ given $\mcCFA_{\mcT,\mcRVAN_\mcT}$ is uniform and hence equal to $\mcRVA^*$ given $\mcCFA_{\mcT,\mcRVA^*}$.
In particular the derivative of $\mcRVA^*$ with respect to $\mcRVAN_\mcT$ is the derivative of $\mcRCFR^*$ with respect to $\mcRCFRN_\mcT$ (constant on assignments with same color frequencies, to be precise).

We're left to show the assertion that $\mcRVAPost_{\mcRG^*_\mcT(\mcRVA^*)}\in\cE_\mcT$ with probability at least $1-\varepsilon$ uniformly. With $c$ denoting the upper bound on the derivative of $\mcRVA^*$ to $\mcRVAN_\mcT$ notice that
\begin{align*}
\pr[\mcRVAPost_{\mcRG^*_\mcT(\mcRVA^*)}\not\in\mathcal E_t]
&\le\Erw[\vecone\{\mcRVAPost_{\mcRG^*_\mcT(\mcRVA^*)}\not\in\mathcal E_t,\mcRVA^*\in\mathcal E_\mcT\}]+\varepsilon
\le c\Erw[\vecone\{\mcRVAPost_{\mcRG^*_\mcT(\mcRVAN)}\not\in\mathcal E_t,\mcRVAN_\mcT\in\mathcal E_\mcT\}]+\varepsilon\\
&\le c\Erw[\vecone\{\mcRVAPost_{\mcRG^*_\mcT(\mcRVAN)}\not\in\mathcal E_\mcT\}]+\varepsilon
=c\Erw[\vecone\{\mcRVAN_\mcT\not\in\mathcal E_\mcT\}]+\varepsilon
\le(c+1)\varepsilon
\end{align*}
uniformly in $n\in\ZZ_{>0}$ and $\mcT\in\mcSeq_n$.
So, if we now choose $r^*\in\RR_{>0}$ sufficiently large such that both $\mcRCFRN_\mcT$ and $\mcRCFR^*_\mcT$ attain frequencies in the corresponding ball with probability at least $1-(c+1)^{-1}\varepsilon$, then we obtain the $1-\varepsilon$ bound for $\mcRVAPost_{\mcRG^*_\mcT(\mcRVA^*)}$. Hence, the assertion holds with $r^*$ and the bound $c^*$ on the derivative corresponding to $r^*$. This completes the proof.
\subsection{Implications}\label{mc_hence}
The results in Section \ref{dd} directly impy mutual contiguity of $\mcRVA^*$ and $\mcRVAN_{\mcRDSC_n}$ with $\mcRDSC_n$ from Section \ref{dd}, since the assumptions of Proposition \ref{mcp_main_prop} are clearly met and the result is uniform in $\mcT\in\mcSeq_n$.
The fact $\mcRG^*_\mcT(\mcRVA^*)$ and $\mcRG^*_\mcT(\mcRVAN_\mcT)$ (or $\mcRGN_\mcT$ for that matter) conditional to a fixed ground truth obviously have the same law then yields mutual contiuity of the degree/assignment/factor graph triplets.
For the same reason we obtain joint mutual contiguity for the factor side half-edge assignments $\mcRFSHEA_\mcT(\mcRVA^*)$ and $\mcRFSHEA_\mcT(\mcRVAN)$, combined with $\mcRDSC_n$, further $\mcRVA^*$ and $\mcRVAN$ or the corresponding factor graph models.
\subsection{Pinning}\label{mc_pinning}
The pinned model is obtained from the regular model by fixing a subset $U\subseteq\mcVLSp_n$ and attaching constraints to the variables that fix the assignment to a uniformly random colour $\check\mcRVA_\mcVL\in\mcColSp$ for $\mcVL\in U$.
Since this process is independent of anything else we have $\Erw[\psi_{\mcRG_{\mcT,U}}(\mcVA)]=\Erw[\psi_{\mcRG_{\mcT}}(\mcVA)]\Erw[\prod_{\mcVL\in U}\vecone\{\check\mcRVA_\mcVL=\mcVA_\mcVL\}]=q^{-|U|}\Erw[\psi_{\mcRG_{\mcT}}(\mcVA)]$. The result immediately translates to the partition function, implying that $\mcRVAN_{t,U}$ and $\mcRVAN_{t}$ have the same law and thereby the mutual contiguity results also hold for pinned models.

\section{Typical assignments} \label{typ_ass}

In this section we derive results for the variable side and factor side half-edge assignments introduced in Section \ref{mc}.
We use the model, notions and notation introduced in Section \ref{mc}, further Section \ref{dd}, Section \ref{assd} and Section \ref{llt}.
For the model introduced in Section \ref{mc} degrees and labels coincide, i.e.~$\mcFInd=\mcFdeg$ and $\mcVInd=\mcVdeg$ using the corresponding notation.
The distributions $\mcFED_\mcFdeg$ are derived from $\mcFRpsi_\mcFdeg$ introduced in Section \ref{Sec_rfg}, while the distributions on the variable side are given by $\mcVED_\mcVdeg$ as introduced in Section \ref{mc}.
Combining Section \ref{dd} and Section \ref{assd}, we have $\acVdegP$, $\acFdegP\in\vmIndPSp$, let  $\acSupp=\{(d,\vmAss):d\in\ZZ_{\ge 0},\vmAss\in\mcColSp^d\}$ denote the joint support, further let
$\acVAD^*\in\cP(\acSupp)$ be given by $\acVAD^*(\mcVdeg,\mcVHEA)=\acVdegP(\mcVdeg)\mcVED_\mcVdeg(\mcVHEA)$ for $\mcVdeg$ in the support of $\vd$ and $\mcVHEA\in\mcColSp^\mcVdeg$, and let
$\acFAD^*\in\cP(\acSupp)$ be given by $\acFAD^*(\mcFdeg,\mcFHEA)=\acFdegP(\mcFdeg)\mcFED_\mcFdeg(\mcFHEA)$ for $\mcFdeg$ in the support of $\vk$ and $\mcFHEA\in\mcColSp^\mcFdeg$, i.e.~the expected assignment distributions on the variable and factor side. Based on Section \ref{assd} we let $\acADPD$ denote the metric on $\acADSp^2$ induced by $\acADD$ on $\acADSp$, i.e.~$\acADPD(\acAD,\acAD')=\acADD(\acAD_1,\acAD'_1)+\acADD(\acAD_2,\acAD'_2)$ for $\acAD$, $\acAD'\in\acADSp^2$.
Finally, we let $\acJAD^*=(\acVAD^*,\acFAD^*)\in\acADSp^2$.

For $n\in\ZZ_{>0}$, $\mcT=(m,\sd,\sk)\in\cT^*_n$, $\mcVA\in\mcColSp^{\mcVLSp_n}$ and $\mcFSHEA\in\mcColSp^{\mcFHESp_\mcT}$, we let $\acVAF{\mcT,\mcVA}\in\acADSp$ denote the variable side half-edge assignment distribution, i.e.~
\begin{flalign*}
\acVAF{\mcT,\mcVA}(\mcVdeg,\mcVHEA)=n^{-1}\left|\left\{i\in[n]:\mcVdeg_{i}=\mcVdeg,(\mcVA_{\mcVL_i})_{h\in[\mcVdeg_{i}]}=\mcVHEA\right\}\right|
\end{flalign*}
for $(\mcVdeg,\mcVHEA)\in\acSupp$, and $\acFAF{\mcT,\mcFSHEA}\in\acADSp$ denote the factor side half-edge assignment distribution, i.e.
\begin{flalign*}
\acFAF{\mcT,\mcFSHEA}(\mcFdeg,\mcFHEA)=m^{-1}\left|\left\{i\in[m]:\mcFdeg_{i}=\mcFdeg,\mcFHEA_i=\mcFHEA\right\}\right|
\end{flalign*}
for $(\mcFdeg,\mcFHEA)\in\acSupp$ if $m>0$ and the one-point mass on the empty assignment $(0,())$ if $m=0$. Finally let $\acJAF{\acS}=(\acVAF{\mcT,\mcVA},\acFAF{\mcT,\mcFHEA})\in\acADSp^2$ with $\acS=(\mcT,\mcVA,\mcFSHEA)$.
Now, for $n\in\cN$, $\mcT\in\cT_n$ and $\mcVA\in\mcColSp^{\mcVLSp_n}$
recall $\mcRFSHEA^*_\mcT(\mcVA)$ from Section \ref{mc}
and let $\acRS^*_\mcT=(\mcT,\mcRVA^*,\mcRFSHEA^*_{\mcT}(\mcRVA^*))$ and $\acRSN_\mcT=(\mcT,\mcRVAN,\mcRFSHEA^*_{\mcT}(\mcRVAN_\mcT))$ denote the coloured sequences for the two versions of the teacher-student scheme for given $\mcT\in\cT_n$ and further $\acRS^*=\acRS^*_{\vt_n}$, $\acRSN=\acRSN_{\vt_n}$.
Further, let $\acSSpV_n$ denote the set of valid coloured sequences $\acS=(\mcT,\mcVA,\mcFSHEA)$, i.e.~we have $\mcT\in\cT_n$, $\mcVA\in\mcColSp^{\mcVLSp_n}$ and $\mcFSHEA=\mcFSHEA_{\mcG,\mcVA}$ for some $\mcG$ in the support of $\mcRG_\mcT$.
Finally, for given $r\in\RR_{>0}$ let 
$\acSSpT_{n,r}=\{\acS\in\acSSpV_n:\acJAD_\acS\in\cB_r(\acJAD^*)\}$.

As before, a sequence $f_n(\acS)$ with $\acS\in\acSSpV_n$ and $n\in\cN$ is sublinear in the number of factors if there exists $c\in\RR_{>0}$ with $|f_n(\acS)|\le c+c\frac{m}{n}$ for all $\acS=(\mcT,\mcVA,\mcFSHEA)\in\acSSpV_n$ and $n\in\cN$.
\begin{proposition}\label{typ_ass_prop}
Assume that {\bf DEG} and {\bf BAL} hold. Then there exists $r_n\in\RR_{>0}$ with $r_n=o(1)$ such that
for all sequences $f_n(\acS)$ that are sublinear in the number of factors we have
\begin{flalign*}
\Erw[f_n(\acRS^*)]
&=\Erw[f_n(\acRS^*)\vecone\{\acRS^*\in\acSSpT_{n,r_n}\}]+o(1)
=\Erw[f_n(\acRS^*)|\acRS^*\in\acSSpT_{n,r_n}]+o(1)
\end{flalign*}
and the same holds for $\acRS^*$ replaced by $\acRSN$.
\end{proposition}
Using Proposition \ref{typ_ass_prop} we fix a suitable choice of $r_n$ and let $\acSSpT_n=\acSSpT_{n,r_n}$ denote the set of valid typical coloured sequences.
Notice that while we discuss the standard model for brevity, the entire section canonically translates to the case including dummy factors as discussed in Section \ref{mc}, where the reference assignment distribution $\acFAD^*$ is the distribution corresponding to $p_{\mathrm{k}}^\circ$.
%
%
\subsection{Half-edge assignments}\label{typ_assp_he}
As opposed to the definition of the teacher-student scheme and the discussion in Section \ref{typ_ass} we will work with assignments to the variable side half-edges directly, or equivalently with assignments to non-isolated variables.
While there is almost a one-to-one correspondence between assignments $\mcVA\in\mcColSp^{\mcVLSp_n}$ to the variables and the assignments $\mcVSHEA\in\mcColSp^{\mcVHESp_\mcT}$ to the half-edges given $\mcT\in\cT_n$, we discard assignments $\mcVA_{\mcVL_i}$ to isolated variables $\mcVL_i\in\mcVLSp_n$ with $\mcVdeg_i=0$, $i\in[n]$.
Hence, this transition needs to be justified.

Let $n_\mcT=\pr[\bm d_t>0]n$ denote the number of variables with non-trivial degree and
$\mcVSHEA_{\mcT,\mcVA}\in\mcColSp^{\mcVHESp_\mcT}$ for $\mcVA\in\mcColSp^{\mcVLSp_n}$ and $\mcT\in\cT_n$ be given by $\mcVSHEA_{\mcVL_i,h}=\mcVA_{\mcVL_i}$ for $i\in[n]$ and $h\in[\mcVdeg_i]$.
For $\mcVSHEA$ in the support of $\mcRVSHEA_\mcT=\mcVSHEA_{\mcT,\mcRVA^*}$ and $\mcG$ in the support of $\mcRG^*_\mcT$ the definitions of $\mcFSHEA_{\mcG,\mcVA}$, $\mcFpsi_{\mcG}(\mcVSHEA)$, $\acVAF{\mcVSHEA}$ and hence $\acJAD^*$ are completely analogous to the previous case and coincide.
However, notice that with $\mcZ'_\mcG=\sum_{\mcVSHEA}\mcFpsi_{\mcG}(\mcVSHEA)$ we have $\mcZ_\mcG=q^{n-n_\mcT}\mcZ'_{\mcG}$.
Let $\mcRG^*_\mcT(\mcVSHEA)$ be the teacher-student model with ground truth $\mcVSHEA$ (in the support of $\mcRVSHEA^*_\mcT$) be given by the Radon-Nikodym derivative $\frac{\mcFpsi_{\mcG}(\mcVSHEA)}{\Erw\left[\mcFpsi_{\mcRG_\mcT}(\mcVSHEA)\right]}$ with respect to $\mcRG_\mcT$, so $\mcRG^*_\mcT(\mcVSHEA)$ and $\mcRG^*_\mcT(\mcVA)$ have the same law for all $\mcVA\in\mcColSp^{\mcVLSp_n}$ with $\mcVSHEA=\mcVSHEA_{\mcT,\mcVA}$, implying that $\mcRFSHEA^*_\mcT(\mcVSHEA)$ and $\mcRFSHEA^*_\mcT(\mcVA)$ have the same law.
Further, with $\mcRVSHEAN_\mcT=\mcVSHEA_{\mcT,\mcRVAN_\mcT}$ we have
\begin{flalign*}
\pr[\mcRVSHEA^*_{\mcT}=\mcVSHEA]=q^{n-n_\mcT}q^{-n}=q^{-n_\mcT}\textrm{, }
\pr[\mcRVSHEAN_\mcT=\mcVSHEA]
=q^{n-n_\mcT}\frac{\Erw[\mcFpsi_{\mcRG_\mcT}(\mcVSHEA)]}{q^{n-n_\mcT}\Erw[\mcZ'_{\mcRG_\mcT}]}
=\frac{\Erw[\mcFpsi_{\mcRG_\mcT}(\mcVSHEA)]}{\Erw[\mcZ'_{\mcRG_\mcT}]}\textrm{,}
\end{flalign*}
i.e.~consistent definitions of $\mcRVSHEA^*_{\mcT}$ and $\mcRVSHEAN$.
The remaining notions directly translate,
hence with the discussion above it is obvious that Proposition \ref{typ_ass_prop} holds if and only if it holds on the half-edge level.
%
%
\subsection{Proof strategy}\label{typ_assp_ps}
We start with the main result that yields Proposition \ref{typ_ass_prop} as a corollary.
\begin{proposition}\label{typ_ass_conc}
Assume that {\bf DEG} and {\bf BAL} hold. For all $\varepsilon\in\RR_{>0}$ there exist constants $c$, $c'\in\RR_{>0}$ such that the following holds.
For all $n\in\cN$ and all $\mcT\in\cT^\circ_n$ we have
\begin{flalign*}
\pr\left[\acADPD(\acJAD_{\acRS^*_\mcT},\acJAD^*)\ge\varepsilon\right]\le c'\exp(-cn)\textrm{, }
\pr\left[\acADPD(\acJAD_{\acRSN_\mcT},\acJAD^*)\ge\varepsilon\right]\le c'\exp(-cn)\textrm{.}
\end{flalign*}
\end{proposition}
The proof of Proposition \ref{typ_ass_conc} is split into two parts.
In the first part we show that the colour frequencies in both models are close to uniform with very high probability.
In the second part we show that for colour frequencies sufficiently close to uniform the assignment distributions are indeed very close to the reference with very high probability.
\begin{lemma}\label{typ_ass_tails}
Assume that {\bf DEG} and {\bf BAL} hold. Then there exist constants $c$, $c'\in\RR_{>0}$ such that the following holds. For all $n\in\cN$, all $\mcT\in\cT^\circ_n$ and all $\varepsilon\in\RR_{>0}$ we have
\begin{flalign*}
\pr\left[\left\|\mcRCFR^*_\mcT-u_{\mcColSp}\right\|_1\ge\varepsilon\right]\le c'\exp(-c\varepsilon^2n)\textrm{, }
\pr\left[\left\|\hat{\bm\rho}_\mcT-u_{\mcColSp}\right\|_1\ge\varepsilon\right]\le c'\exp(-c\varepsilon^2n)\textrm{.}
\end{flalign*}
\end{lemma}
With the tail bounds in place we can focus on the center, i.e.~colour frequencies close to uniform.
\begin{lemma}\label{typ_ass_center}
Assume that {\bf DEG} and {\bf BAL} hold. Then for all $\varepsilon\in\RR_{>0}$ there exist $\delta$, $c$, $c'\in\RR_{>0}$ such that the following holds. For all $n\in\cN$ and all $\mcT\in\cT^\circ_n$ we have
\begin{flalign*}
\pr\left[\acADPD(\acJAD_{\acRS^*_\mcT},\acJAD^*)\ge\varepsilon\middle\vert\mcRCFR^*_\mcT\in\cB_\delta(u_{\mcColSp})\right]\le c'\exp(-cn)
\end{flalign*}
and the same holds for $\acRS^*_\mcT$, $\mcRCFR^*_\mcT$ replaced by $\acRSN_\mcT$, $\mcRCFRN_\mcT$.
\end{lemma}
Proposition \ref{typ_ass_conc} is then an almost immediate consequence of Lemma \ref{typ_ass_tails} and Lemma \ref{typ_ass_center}.
%
%
\subsection{Proof of Lemma \ref{typ_ass_tails}}\label{typ_assp_tails}
Notice that $\mcRVSHEA^*_{\mcT}\sim\bigotimes_{i\in[n]}\nu_{d_{i}}$ and recall that $\cT^\circ_n$ satisfies the assumptions in Section \ref{llt} on both the variable and the factor side.
Hence, Proposition \ref{llt_tails} yields constants $c$, $c'\in\RR_{>0}$ such that
\begin{flalign*}
\pr[\|\mcRCFR^*_\mcT-u_{\mcColSp}\|_1\ge\varepsilon]\le c'\exp(-c\varepsilon^2n)
\end{flalign*}
for all $n\in\cN$, $\mcT\in\cT^\circ_n$ and $\varepsilon\in\RR_{\ge 0}$.
Proposition \ref{mcp_tails} yields constants $c$, $c'\in\RR_{>0}$ such that the assertion holds for $\mcRCFRN_\mcT$.
Taking the maximum $c'$ and minimum $c$ completes the proof.
%
%
\subsection{Proof of Lemma \ref{typ_ass_center}}\label{typ_assp_center}
Fix $\mcT=(m,\sd,\sk)\in\cT^\circ_n$, recall that $\mcRVSHEA^*_{\mcT}\sim\bigotimes_{i\in[n]}\nu_{\mcVdeg_{i}}$ and let $\mcRFSHEA^*_\mcT\sim\bigotimes_{i\in[m]}\mu_{\mcFdeg_{i}}$ be independent of anything else.
For $\mcCFA$ in the support of $\mcRCFA^*_\mcT$ let $\mcRVSHEA_{\mcT,\mcCFA}=(\mcRVSHEA^*_{\mcT}|\mcRCFA^*_\mcT=\gamma)$ and $\mcRFSHEA_{\mcT,\mcCFA}=(\mcRFSHEA^*_\mcT|\mcCFA_{\mcRFSHEA^*_{\mcT}}=\mcCFA)$ denote the half-edge assignments on the variable side and factor side for given $\mcCFA$ (with $\gamma_{\mcFSHEA}$ denoting the colour frequencies of $\mcFSHEA$, as introduced in Section \ref{mc}).
Notice that both $(\mcRVSHEA^*_{\mcT},\mcRFSHEA^*_\mcT(\mcRVSHEA^*_{\mcT}))|\mcRCFA^*_\mcT=\mcCFA$ and $(\mcRVAN_\mcT,\mcRFSHEA^*_\mcT(\mcRVAN_\mcT))|\mcRCFAN_\mcT=\mcCFA$ have the same law as $(\mcRVSHEA_{\mcT,\mcCFA},\mcRFSHEA_{\mcT,\mcCFA})$.

With the results in Section \ref{assd} we obtain $\delta$, $c$, $c'\in\RR_{>0}$ such that for all $n\in\cN$, $\mcT\in\cT^\circ_n$ with $p_{\mathrm{d},t}\in\cB_\delta(p_{\mathrm{d}})$ and $p_{\mathrm{k},t}\in\cB_\delta(p_{\mathrm{k}})$ (with respect to the corrsesponding metric) and attainable $\mcCFR\in\cB_\delta(u_{\mcColSp}$ with $\mcCFA=\Erw[\mcRVdeg_\mcT]n\mcCFR$ we have
\begin{flalign*}
\pr\left[\acADPD(\acVAF{\mcT,\mcRVSHEA_{\mcT,\mcCFA}},\acVAD^*)\ge\varepsilon\right]\le c'\exp(-cn)\textrm{.}
\end{flalign*}
Further, a corresponding result holds on the factor side. By weakening the constants and using $m\sim\vm_n$ uniformly we obtain $\delta$, $c$, $c'\in\RR_{>0}$ to obtain uniform exponential tail bounds on both sides.
Further, since we have $(p_{\mathrm{d},t},p_{\mathrm{k},t})\rightarrow(p_{\mathrm{d}},p_{\mathrm{k}})$ uniformly in $\mcT\in\cT^\circ_n$ the assumptions $p_{\mathrm{d},t}\in\cB_\delta(p_{\mathrm{d}})$ and $p_{\mathrm{k},t}\in\cB_\delta(p_{\mathrm{k}})$ are redundant for sufficiently large $n$.
By readjusting the leading coefficient $c'$ the tail bounds are trivial for small $n$. Finally, the assertion follows from an $\varepsilon/2$ argument.
\subsection{Proof of Proposition \ref{typ_ass_conc}}
Using Lemma \ref{typ_ass_center} we obtain uniform exponential tail bounds for the center, i.e.~restricted to $\mcRCFR^*_\mcT\in\cB_\delta(u_{\mcColSp})$ and $\mcRCFRN_\mcT\in\cB_\delta(u_{\mcColSp})$ respectively for some $\delta\in\RR_{>0}$. With Lemma \ref{typ_ass_tails} we then obtain exponential tail bounds for $\mcRCFR^*_t\not\in\cB_\delta(u_{\mcColSp})$ and $\mcRCFRN_t\not\in\cB_\delta(u_{\mcColSp})$ respectively, which immediately yield the assertion by splitting the probability into the two regimes and weakening the constants.
\subsection{Proof of Proposition \ref{typ_ass_prop}}
With Proposition \ref{typ_ass_conc} we can construct a sequence $r_n\in\RR_{>0}$, $n\in\cN$, with $r_n=o(1)$ such that uniformly in $\mcT\in\cT^\circ_n$ we have $\acRS^*_\mcT\in\acSSpT_{n,r_n}$ and $\acRSN_\mcT\in\acSSpT_{n,r_n}$ with high probability.
Following the proof of Proposition \ref{dd_typ_prop} we can restrict to $\mcT\in\cT^\circ_n$, and since $f$ is bounded on this subset the assertion follows from the above.
\section{Proof of \Prop~\ref{Prop_IZ}}\label{mi}

This section is dedicated to the proof of Proposition \ref{Prop_IZ}.
The arguments rely on the results in Section \ref{dd} and Section \ref{mc}.
Hence, we tacitly assume that the assumptions {\bf DEG}, {\bf BAL} and further {\bf SYM} are satisfied.
Notice that {\bf DEG} implies the corresponding assumptions in Section \ref{mc} and {\bf BAL} implies {\bf BAL'}.

Recall the valid numbers $\cN$ of variables and for $n\in\cN$ the valid degree sequences $\cT_n$.
For $n\in\cN$, $t\in\cT_n$ and $\mcVA\in\mcColSp^{\mcVLSp_n}$ let $r^*_{\mcVA}$ denote the Radon-Nikodym derivative $\mcG\mapsto\mcFpsi_\mcG(\mcVA)/\mcFEpsi_t(\mcVA)$, $\mcG$ in the support of $\mcRG_t$, of the teacher-student scheme $\mcRG^*_t(\mcVA)$ with respect to the null model $\mcRG_t$, where $\mcFEpsi_\mcT=\Erw[\mcFpsi_{\mcRG_t}]$ denotes the expected total weight. Further, let $r^*$ denote the derivative $\mcG\mapsto\Erw[r^*_{t,\mcRVA^*}]$ of $\mcRG^*_t(\mcRVA^*)$ with respect to $\mcRG_t$. Notice that we can keep the dependence of $r^*_{\mcVA}$ and $r^*$ on $t$ implicit since $t$ is determined by $\mcG$, i.e.~the sets of factor graphs for distinct degree sequences are disjoint.
The mutual information given $t$ and the unconditional mutual information are given by
\begin{flalign*}
I(t)=I(\mcRVA^*,\mcRG^*_\mcT(\mcRVA^*))=\Erw\left[\ln\left(\frac{r^*_{\mcRVA^*}(\mcRG^*_\mcT(\mcRVA^*))}{r^*(\mcRG^*_\mcT(\mcRVA^*))}\right)\right]\textrm{, }
I=I(\mcRVA^*,\mcRG^*_{\vt_n}(\mcRVA^*))=\Erw\left[I(\vt_n)\right]\textrm{.}
\end{flalign*}
We obtain the following proposition as a corollary. For this purpose recall the notions from Section \ref{dd}, Section \ref{mc} and let $\Lambda(x)=x\ln(x)$.
\begin{proposition}\label{mi_prop}
Under {\bf DEG}, {\bf BAL} and {\bf SYM} we have
\begin{flalign*}
\frac{1}{n}I(\mcRVA^*,\mcRG^*_{\vt_n}(\mcRVA^*))
&=\ln(q)+\Erw\left[\frac{\d}{\k\mcFExi_{\vk}q^{\vk}}\sum_{\mcFHEA\in\mcColSp^{\vk}}\Lambda(\mcFRpsi_{\vk}(\mcFHEA))\right]-\frac{1}{n}\Erw\left[\ln\left(\mcZ_{\mcRG^*_{\vt_n}(\mcRVA^*)}\right)\right]+o(1)\textrm{.}
\end{flalign*}
\end{proposition}
%
%
\subsection{Preliminaries}\label{mip_intro}
Using Section \ref{mi} with respect to Proposition \ref{Prop_IZ} allows to restrict to $t\in\cT^\circ_n$.
With respect to Proposition \ref{mi_prop} and using the definitions of Section \ref{mi} we first notice that under {\bf SYM} the mutual information per variable is sublinear in the number of factors, i.e.~there exists $\varepsilon\in(0,1)$ with $\varepsilon^{m}\le\mcFpsi_\mcG(\mcVA)\le \varepsilon^{-m}$ uniformly for $\mcG$ in the support of $\mcRG_t$, $\mcVA\in\mcColSp^{\mcVLSp_n}$ and $t=(m,\sd,\sk)\in\cT_n$, hence the same holds for $\mcFEpsi_t(\mcVA)$, further $|\ln(r^*_{\mcVA}(\mcG))|$, $|\ln(r^*(\mcG))|\le m\ln(\varepsilon^{-2})$ and thereby
$|i^*_t|\le 2\ln\left(\varepsilon^{-2}\right)\frac{m}{n}$ with $i^*_\mcT=\frac{1}{n}I(t)$.
With $i^*=\Erw[i^*_{\vt_n}]$ we can hence use Proposition \ref{dd_typ_prop} to obtain $i^*=\Erw[i^*_{\vt_n}\vecone\{\vt_n\in\cT^\circ_n\}]+o(1)$ which again justifies the restriction to typical degree sequences.

For $t\in\cT_n$ we can rewrite $i^*_t$ as follows to extract the material contributions.
While the following steps can be traced algebraically using the definition of $i^*_t$, but we prefer to give the conceptual and more intuitive derivation using the conditional entropy $H(\vx|\vy)$ and cross entropy $\CE{\vx}{\vy}$, i.e.
\begin{flalign*}
i^*_t
&=\frac{1}{n}\KL{(\mcRVA^*,\mcRG^*_t(\mcRVA^*))}{\mcRVA^*\otimes\mcRG^*_t(\mcRVA^*)}
=\frac{1}{n}H(\mcRVA^*)-\frac{1}{n}H(\mcRVA^*|\mcRG^*_t(\mcRVA^*))\\
&=\frac{1}{n}H(\mcRVA^*)-\frac{1}{n}\Erw[\CE{\mcRVA^*|\mcRG^*_t(\mcRVA^*)}{\mcRVAPost_{\mcRG^*_t(\mcRVA^*)}}]+\frac{1}{n}\Erw[\KL{\mcRVA^*|\mcRG^*_t(\mcRVA^*)}{\mcRVAPost_{\mcRG^*_t(\mcRVA^*)}}]\\
&=\ln(q)+\left(\eta^*_t-\phi^*_t\right)+i^*_{\mathrm{err}}(t)\textrm{,}\\
\eta^*_t&=\frac{1}{n}\Erw\left[\ln\left(\psi_{\mcRG^*_t(\mcRVA^*)}(\mcRVA^*)\right)\right]\textrm{, }
\phi^*_\mcT=\frac{1}{n}\Erw\left[\ln\left(\mcZ_{\mcRG^*_t(\mcRVA^*)}\right)\right]\textrm{, }\\
i^*_{\mathrm{err}}(t)&=-\frac{1}{n}\Erw\left[\ln\left(r_{\mcRG^*_t(\mcRVA^*)}(\mcRVA^*)\right)\right]\textrm{, }
r_\mcG(\mcVA)=\Erw\left[\frac{\mcFEpsi_t(\mcVA)}{\mcFEpsi_t(\mcRVAPost_\mcG)}\right]\textrm{.}
\end{flalign*}
The quantities $\eta^*_t$ and $\phi^*_t$ reflect the split of $\mcBMD_\mcG(\mcVA)=\frac{\mcFpsi_\mcG(\mcVA)}{\mcZ_\mcG}$ into the weight $\mcFpsi_\mcG(\mcVA)$ and normalization constant $\mcZ_\mcG$, and $\phi^*_t$ already appears in the right hand side of the assertion.
Hence, we are left to derive the material contributions from $\eta^*_t$ and to show that the relative entropy per variable $i_{\mathrm{err}}*(t)$ is negligible, where $r_\mcG$ is the derivative of the posterior $\mcRVAPost_{\mcRG^*_t(\mcRVA^*)}$ with respect to the prior $\mcRVA^*|\mcRG^*_t(\mcRVA^*)$  given $\mcG$ from $\mcRG^*_t(\mcRVA^*)$ (notice the leading minus sign in the definition of $i^*_{\mathrm{err}}(t)$).
%
%
\subsection{The material contribution}\label{mip_material}
For given $t=(m,\sd,\sk)\in\cT^\circ_n$ we add the conditioning level for the factor side assignments, i.e.~$\eta^*_\mcT=\Erw[\eta^*_t(\mcRVA^*,\mcRFSHEA^*_t(\mcRVA^*))]$ with $\eta^*_t(\mcVA,\mcFSHEA)=\frac{1}{n}\Erw\left[\ln\left(\mcFpsi_{\mcRG^*_t(\mcVA,\mcFSHEA)}(\mcVA)\right)\right]$.
With the results from Section \ref{mip_intro} we notice that $\eta^*_t(\mcVA,\mcFSHEA)$ is sublinear in the number of factors, so we can use Proposition \ref{typ_ass_prop} to obatin 
$\eta^*_\mcT=\Erw[\eta^*_t(\mcRVA^*,\mcRFSHEA^*_t(\mcRVA^*))\vecone\{(\mcRVA^*,\mcRFSHEA^*_t(\mcRVA^*)\in\acSuppT_t\}]+o(1)$ uniformly for all $t\in\cT^\circ_n$.
However, by the very definition of $\mcRG^*_t(\mcVA,\mcFSHEA)$ we have $\mcFpsi_{\mcRG^*_t(\mcVA,\mcFSHEA)}(\mcVA)=\prod_{\mcFL\in\mcFLSp_m}\mcFpsi_{\mcRG^*_t(\mcVA,\mcFSHEA),\mcFL}(\mcFHEA_\mcFL)$, and the weights $\mcFpsi_{\mcRG^*_t(\mcVA,\mcFSHEA),\mcFL}$, $\mcFL\in\mcFLSp_m$, are drawn independently and independent of the bijection $\mcBij_{\mcRG^*_t(\mcVA,\mcFSHEA)}$.
As discussed in Section \ref{Sec_Noela}, for $\mcFdeg$ in the support of $\vk$ and $\mcFHEA\in\mcColSp^\mcFdeg$ let $p^*_{\mcFdeg,\mcFHEA}\in\cP((0,2)^{\mcColSp^\mcFdeg})$ be the law given by the derivative $\mcFpsi\mapsto\frac{\mcFpsi(\mcFHEA)}{\mcFEpsi_\mcFdeg(\mcFHEA)}\in\RR_{>0}$ with respect to $\vpsi_\mcFdeg$,
then we have $(\psi_{\mcRG^*_t(\mcVA,\mcFSHEA),a})_a\sim\bigotimes_{i\in[m]}p^*_{\mcFdeg_i,\mcFHEA_{\mcFL_i}}$ and further
\begin{flalign*}
\eta^*_t(\mcVA,\mcFSHEA)&=\frac{m}{n}\sum_{i\in[m]}\frac{1}{m}\Erw\left[\ln\left(\psi_{\mcRG^*_t(\mcVA,\mcFSHEA),i}(\mcFHEA_{\mcFL_i})\right)\right]
=\frac{m}{n}\Erw\left[\ln\left(\mcFRpsi^*_{\mcFdeg_{\vi},\mcFHEA_{\mcFL_{\vi}}}(\mcFHEA_{\mcFL_{\vi}})\right)\right]
\end{flalign*}
using $\mcFRpsi^*_{\mcFdeg,\mcFHEA}\sim p^*_{\mcFdeg,\mcFHEA}$ and $\vi\in[m]$ uniform. Since $\eta^*_t(\mcVA,\mcFSHEA)=\Theta(1)$ uniformly in $t\in\cT^\circ_n$ and for all $(\mcVA,\mcFSHEA)$ assuming {\bf SYM}, and further $(\mcFdeg_{\vi},\mcFHEA_{\mcFL_{\vi}})$ converges to $(\vk,\mcRFSHEA^*_{\vk})$, $\mcRFSHEA^*_\mcFdeg\sim\mcFED_\mcFdeg$, in total variation distance for $(t,\mcVA,\mcFSHEA)\in\acSSpT_n$, to be precise we have uniform bounds in $1$-norm for the laws given $t$ and bounds on the $1$-norm of the degree laws, this gives $\eta^*_t(\mcVA,\mcFSHEA)=\frac{\d}{\k}\Erw[\ln(\mcFRpsi^*_{\vk,\mcRFSHEA^*_{\vk}}(\mcRFSHEA^*_{\vk}))]+o(1)$ uniformly for all valid typical colored sequences $(t,\mcVA,\mcFSHEA)\in\acSSpT_n$.
With the discussion at the beginning of this section we obtain
$\eta^*_\mcT=\frac{\d}{\k}\Erw[\ln(\mcFRpsi^*_{\vk,\mcRFSHEA^*_{\vk}}(\mcRFSHEA^*_{\vk}))]+o(1)$ uniformly for all $t\in\cT^\circ_n$ and further $\Erw[\eta^*_{\vt_n}]=\frac{\d}{\k}\Erw[\ln(\mcFRpsi^*_{\vk,\mcRFSHEA^*_{\vk}}(\mcRFSHEA^*_{\vk}))]+o(1)$.
While this may be considered the natural form in terms of our proof strategy, the form of the assertion can be established by expanding the expectation over $\mcRFSHEA^*_{\vk}$ and using the derivative of $p^*_{\mcFdeg,\mcFHEA}$.
%
%
\subsection{The negligible contribution}\label{mip_negligible}
The discussion in Section \ref{mip_intro} allows to restrict to $t\in\cT^\circ_n$, but as before subllinearity in the number of factors can also be easily obtained for $i^*_{\mathrm{err}}(t)$. Further, the fact that $i^*_{\mathrm{err}}(t)=\frac{1}{n}\Erw[\KL{\mcRVA^*|\mcRG^*_t(\mcRVA^*)}{\mcRVAPost_{\mcRG^*_t(\mcRVA^*)}}]$ directly yields $I^*_{\mathrm{err}}(t)\ge 0$ by basic properties of the relative entropy respectively an application of Jensen's inequality to $x\ln(x)$.
Upper bounding $i^*_{\mathrm{err}}(t)$ is involved since we consider the relative entropy given $\mcRG^*_t(\mcRVA^*)$, a model that is not as accessible as say $\acRS^*_n$.

However, the derivative $r_\mcG$ is an expectation by design and hence we can apply Jensen's inequality to $-\ln(x)$ with respect to the inner expectation, yielding $-\ln(r_\mcG(\mcVA))\le\Erw[-\ln(\mcFEpsi_t(\mcVA)/\mcFEpsi_t(\mcRVAPost_f))]$ and hence
\begin{flalign*}
i_{\mathrm{err}}^*(t)
&\le\frac{1}{n}\Erw\left[\ln\left(\mcFEpsi_t(\mcRVAPost_{\mcRG^*_t(\mcRVA^*)})\right)\right]-\frac{1}{n}\Erw\left[\ln\left(\mcFEpsi_t(\mcRVA^*)\right)\right]
=\delta^*_0(t)-\delta^*_1(t)\textrm{,}\\
\delta^*_0(t)&=\Erw\left[\frac{1}{n}\ln\left(r_t(\mcRVA^*)\right)\right]\textrm{, }
\delta^*_1(t)=\Erw\left[\frac{1}{n}\ln\left(r_t(\mcRVAPost_{\mcRG^*_t(\mcRVA^*)})\right)\right]\textrm{, }
r_t(\mcVA)=\frac{\pr[\mcRVA^*=\mcVA]}{\pr[\mcRVAN_\mcT=\mcVA]}\textrm{.}
\end{flalign*}
Notice that $\delta_0^*(t)=n^{-1}\KL{\mcRVA^*}{\mcRVAN_t}$ and $|n^{-1}\ln(r_t(\mcVA))|=n^{-1}|\ln(\Erw[\mcFEpsi_t(\mcRVA^*)])-\ln(\mcFEpsi_t(\mcVA))|$, which yields $|n^{-1}\ln(r_t(\mcVA))|<c^*$ uniformly in $t\in\cT^\circ_n$ and $\mcVA\in\mcColSp^{\mcVLSp_n}$ for some $c^*\in\RR_{>0}$.
For given $\varepsilon\in(0,1)$ we summon Proposition \ref{mcp_main_prop} to obtain $c$, $r\in\mathbb R_{>0}$ such that uniformly in $t\in\cT^\circ_n$ we have $\pr[\mcRVA^*\not\in\cE_n]$, $\pr[\mcRVAPost_{\mcRG^*_t(\mcRVA^*)}\not\in\cE_n]<\varepsilon$ and $|\ln(r_t(\mcVA))|<c$ for all $\mcVA\in\cE_n$ with $\cE_n=\{\mcVA\in\mcColSp^{\mcVLSp_n}:\|\mcCFA_{t,\mcVA}-\mcVDTot_tu_{\mcColSp}\|_2<r\sqrt{n}\}$.
Then we have
$|i_{\mathrm{err}}^*(t)|\le 2\left(\frac{c}{n}+\varepsilon c^*\right)\sim 2c^*\varepsilon$. Taking $\varepsilon$ to $0$ shows that $i^*_{\mathrm{err}}(t)=o(1)$ uniformly in $t\in\cT^\circ_n$, so $\Erw[i^*_{\mathrm{err}}(\vt_n)]=o(1)$ and the assertion holds.

\begin{proof}[Proof of \Prop~\ref{Prop_IZ}]
Since standard arguments, i.e., Section 5 in \cite{Coja_2018} show that there exists a simple $\G$ with the desired degree sequences with positive probability, the proposition is an immediate consequence of \Sec~\ref{mip_intro}-\ref{mip_negligible}. 
\end{proof}

\section{Concentration} \label{conc}

In this section we focus on the central quantity discussed in this work, the quenched free entropy density.
In the remainder we tacitly assume that {\bf DEG}, {\bf BAL} and {\bf SYM} hold and reuse the conventions and notions from Section \ref{dd}, Section \ref{mc} and Section \ref{typ_ass}.
For $t\in\mathcal T_n$ and a factor graph $\mcG$ in the support of $\mcRG_\mcT$ the free entropy density of $\mcG$ is $\phi(\mcG)=\frac{1}{n}\ln(\mcZ_{\mcG})$.
Now, depending on our model the quenched free entropy densities given $t$ are $\bar\phi_t=\Erw[\phi(\mcRG_\mcT)]$,
$\phi^*_{t}=\Erw[\phi(\mcRG^*_t(\mcRVA^*))]$ and
$\hat\phi_{t}=\Erw[\phi(\mcRG^*_t(\mcRVAN_t))]$.

From these we obtain the model dependent quenched free entropy densities by averaging over $\bm t_n$, i.e.~$\bar\phi=\Erw[\bar\phi_{\bm t_n}]$, $\phi^*=\Erw[\phi^*_{\bm t_n}]$ and $\hat\phi=\Erw[\hat\phi_{\bm t_n}]$.
As before, the results of this section canonically translate to the factor pruned models as discussed in Section \ref{mc}, combined with the argument in Section \ref{conc_pruned} which ensures that pathological cases can indeed be neglected.
\subsection{Null model}
As opposed to the teacher student scheme concentration of $\bar\phi_{\bm t_n}$ around $\bar\phi$ can be easily obtained.
In the first step we show concentration of $\phi(\mcRG_\mcT)$ around $\bar\phi_t$ for any given $t\in\mathcal T_n^\circ$.
\begin{proposition}\label{conc_std_loc}
There exist constants $c$, $c'\in\mathbb R_{>0}$ such that for all $n\in\mathcal N$, all $t\in\mathcal T_n^\circ$ and all $r\in\mathbb R_{\ge 0}$ we have
\begin{align*}
\pr[|\phi(\mcRG_\mcT)-\bar\phi_t|\ge r]\le c'\exp(-cr^2n)\textrm{.}
\end{align*}
\end{proposition}
This result suggests that for any $t\in\mathcal T_n^\circ$ and $r_n\in\omega(\sqrt{n}^{-1})$ we have $|\phi(\mcRG_\mcT)-\bar\phi_t|<r_n$ with high probability, so the free entropy densities of almost all instances asymptotically coincide with their expectation.
The next result implies that the same is true for the conditional expectations.
\begin{proposition}\label{conc_std_glob}
We have $\bar\phi_t=\bar\phi+o(1)$ uniformly for all $t\in\mathcal T^\circ_n$.
\end{proposition}
Combining Proposition \ref{conc_std_glob} controlling the free entropies globally via the conditional expectations and Proposition \ref{conc_std_loc} controlling the free entropies locally around the conditional expectation gives sufficient control for the arguments in the remainder.
\subsection{Teacher student model}
As indicated in Section \ref{typ_ass} we introduce another conditioning level based on the choice of assignment pairs $(\mcVA,\mcFSHEA)$.
So, for $t\in\mathcal T_n$, $\mcVA\in\mcColSp^{\mcVLSp_n}$ and $\mcFSHEA$ in the support of $\mcRFSHEA_\mcT(\mcVA)$ let $\mcRG^*_{t,\mcVA,\mcFSHEA}=(\mcRG^*_t(\mcVA)|\mcRFSHEA^*_t(\mcVA)=\mcFSHEA)$ be the teacher student model with the assignments on both sides fixed, and notice that the results from Section \ref{Sec_Noela} can be directly applied to this model.
Further, we introduce the corresponding conditional quenched free entropy density $\phi^*_{\mcT,\mcVA,\mcFSHEA}=\Erw[\phi(\mcRG^*_{t,\mcVA,\mcFSHEA})]$.
In the first step we show concentration of $\phi(\mcRG^*_{s})$ around $\phi^*_{s}$ for $s=(t,\mcVA,\mcFSHEA)$ in the support of $\acRS^*_n$ with $\mcT\in\cT^\circ$.
\begin{proposition}\label{conc_tss_loc}
There exist constants $c$, $c'\in\mathbb R_{>0}$ such that for all $n\in\cN$, all $\acS=(\mcT,\mcVA,\mcFSHEA)$ in the support of $\acRS^*_n$ with $t\in\cT^\circ_n$ and all $r\in\mathbb R_{\ge 0}$ we have
\begin{align*}
\pr[|\phi(\mcRG^*_s)-\phi^*_s|\ge r]\le c'\exp(-cr^2n)\textrm{.}
\end{align*}
\end{proposition}
This result suggests that the free entropy densities of almost all instances asymptotically coincide with their expectation.
The next result implies that the same is true for the conditional expectations.
\begin{proposition}\label{conc_tss_glob}
Uniformly for all $s\in\mathcal S^\circ_n$ we have $\phi^*=\phi^*_s+o(1)=\hat\phi+o(1)$ .
\end{proposition}
While Proposition \ref{conc_tss_glob} ensures the equivalence of the quenched free entropy densities, and concentration combined with Proposition \ref{conc_tss_loc}, we will derive significantly stronger exponential tail bounds for the Nishimori model in Section \ref{condp}.
\subsection{Proof strategy}
The following result ensures that it is sufficient to restrict to typical degree sequences $t\in\mathcal T_n^\circ$.
Further, an immediate consequence is that the quenched free entropy densities are bounded.
\begin{lemma}\label{conc_prop_dd_typ}
We have $\Erw[\bar\phi_{\bm t_n}]=\Erw[\bar\phi_{\bm t_n}\vecone\{\bm t_n\in\mathcal T^\circ_n\}]+o(1)$ and the same holds for $\bar\phi$ replaced by $\phi^*$ and $\hat\phi$.
\end{lemma}
\begin{proof}
Using $\varepsilon$ from {\bf SYM} we have uniform bounds for $\phi(\mcG)$ for all $\mcG$ in the support of $\mcRG_t$ given $t=(m,\sd,\sk)\in\cT_n$, namely
\begin{align*}
\ln(q)+\frac{m}{n}\ln(\varepsilon)\le\phi(\mcG)<\ln(q)+\frac{m}{n}\ln(\varepsilon^{-1})\textrm{,}
\end{align*}
so $\phi(\mcG)$ is sublinear in the number of factors.
Hence, any conditional expectation is also sublinear, which completes the proof using Proposition \ref{dd_typ_prop}.
\end{proof}
Hence, we can safely restrict to typical degree sequences for all proofs. Proposition \ref{conc_std_loc} then immediately follows from Azuma's inequality combined with the switching method, discussed in Section \ref{conc_std_loc_proof}.
Proposition \ref{conc_std_glob} follows from a coupling argument that ensures Lipschitz continuity of the conditional expectations, discussed in Section \ref{conc_std_glob_proof}.

For the teacher student models we follow the same strategy on a more granular level. The first result ensures that we can restrict to typical assignments.
\begin{lemma}\label{conc_typ_qfed_prop}
We have
$\Erw\left[\phi(\mcRG^*_{\acRS^*_n})\right]=\Erw\left[\phi^*_{\acRS^*_n}\vecone\{\acRS^*_n\in\acSSpT_n\}\right]+o(1)$
and the same holds for $\acRS^*_n$ replaced by $\acRSN_n$.
\end{lemma}
\begin{proof}
In the proof of Lemma \ref{conc_prop_dd_typ} we observed that $\phi(\mcG)$ is sublinear in the number of factors, hence we can use Proposition \ref{typ_ass_prop}.
\end{proof}
Now, the proof of Proposition \ref{conc_tss_loc} in Section \ref{conc_tss_loc_proof} and the proof of Proposition \ref{conc_tss_glob} in Section \ref{conc_tss_glob_proof} follow the same strategy as their counterparts for the null model, with an additional layer of complexity.
%
%
\subsection{Proof of Proposition \ref{conc_std_loc}}\label{conc_std_loc_proof}
The proof of Proposition \ref{conc_std_loc} is based on Azuma's inequality.
For this purpose fix $\mcT=(m,\sd,\sk)\in\cT^\circ_n$ and consider $\mcG$ in the support $\cG_t$ of $\mcRG_t$ as element of the product space $\mcG\in\prod_{i\in[m]}\cG_{t,i}$ with $\cG_{t,i}=\mcVHESp_t^{k_i}\times\Psi_{k_i}$. This allows to canonically extend the notation for assignments to factor graphs, i.e.~for $i\in[m]$ the coordinate $\mcG_i=((\mcBij^{-1}(\mcFL_i,h))_{h\in[k_i]},\psi_{\mcFL_i})$ encodes the wiring and weight function of the factor $\mcFL_i$.

Now, let $\ell\in[m]$ and $\mcG_{\mathrm{a}}$, $\mcG_{\mathrm{b}}\in\cG_t$ be given such that $\mcG^*=\mcG_{\mathrm{a},[\ell-1]}=\mcG_{\mathrm{b},[\ell-1]}$.
Recall that $\mcRG_\mcT$ is obtained by a uniformly random choice of $\mcRBij_t$ and independent choices of $\vpsi_{\mcG,\mcFL_i}$ for $i\in[m]$. Hence, $\tilde{\mcRG}_r\sim\mcRG_\mcT|\mcRG_{t,[\ell]}=\mcG_{r,[\ell]}$ for $r\in\{\mathrm{a},\mathrm{b}\}$ is obtained by a uniformly random completion of $\mcBij^{-1}_{\mcG_r,[\ell]}$ and independent choices of the remaining weight functions.
This means that we obtain the following canonical coupling of $\tilde{\mcRG}_\mathrm{a}$ and $\tilde{\mcRG}_\mathrm{b}$.
For any instance $\mcG$ from $\tilde{\mcRG}_\mathrm{a}$ obtain $\mcG'=\iota(\mcG)$ by replacing $\psi_{\mcG,\ell}$ with $\psi_{\mcG_\mathrm{b},\ell}$ and successively switching the wires $(\ell,h)$ with $\mcBij_\mcG(\mcBij^{-1}_{\mcG_{\mathrm{b}},\ell}(\ell,h))$ for $h\in[k_\ell]$.
It is obvious from the construction that $\mcG'$ is an instance of $\tilde{\mcRG}_{\mathrm{b}}$, and further that reversing the construction recovers $\mcG$ from $\mcG'$, hence $\iota$ is a bijection. This in turn shows that $\iota(\tilde{\mcRG}_{\mathrm{a}})\sim\tilde{\mcRG}_{\mathrm{b}}$. Further, next to the factor $\mcFL_\ell$ the maximum number of coordinates $\mcG_i$, $i\in[m]\setminus[\ell]$, changed by $\iota$ is upper bounded by the maximum number of rewirings, i.e.~by $k_\ell$.
Using {\bf SYM} and the definition of the free entropy density this gives
\begin{align*}
|\phi(\mcG)-\phi(\mcG')|<\frac{k_{\ell}+1}{n}\ln\left(\varepsilon^{-2}\right)\textrm{.}
\end{align*}
So, under this coupling and using the triangle inequality we have
\begin{align*}
\left|\Erw\left[\phi(\tilde{\mcRG}_{\mathrm{a}})\right]-\Erw\left[\phi(\tilde{\mcRG}_{\mathrm{b}})\right]\right|
=\left|\Erw\left[\phi(\tilde{\mcRG}_{\mathrm{a}})-\phi(\iota(\tilde{\mcRG}_{\mathrm{a}}))\right]\right|
<\frac{k_{\ell}+1}{n}\ln\left(\varepsilon^{-2}\right)\textrm{.}
\end{align*}
Since this bound is uniform in the choice of $\mcG_{\mathrm{b}}$ we obtain the bound
\begin{align*}
\left|\Erw\left[\phi(\mcRG_\mcT)\middle|\mcRG_{t,[\ell]}=\mcG_{\mathrm{a},[\ell]}\right]-\Erw\left[\phi(\mcRG_\mcT)\middle|\mcRG_{t,[\ell-1]}=\mcG_{\mathrm{a},[\ell-1]}\right]\right|
&=\left|\Erw[\gamma(\mcG_{\mathrm{a},\ell})-\gamma(\tilde{\mcRG})]\right|
\le\frac{k_{\ell}+1}{n}\ln\left(\varepsilon^{-2}\right)\textrm{,}\\
\gamma(\mcG)&=\Erw\left[\phi(\mcRG_\mcT)\middle|\mcRG_{t,[\ell-1]}=\mcG_{\mathrm{a},[\ell-1]},\mcRG_{t,\ell}=\mcG\right]\textrm{,}\\
\tilde{\mcRG}&=(\mcRG_{t,\ell}|\mcRG_{t,[\ell-1]}=\mcG_{\mathrm{a},[\ell-1]})\textrm{.}
\end{align*}
Since this bound is uniform in the choice of $\mcG_{\mathrm{a}}$ the corresponding Doob martingale has bounded differences almost surely and Azuma's inequality yields
\begin{align*}
\pr\left[\left|\phi(\mcRG_\mcT)-\bar\phi_t\right|\ge r\right]&\le 2\exp\left(-c_{t}r^2n\right)\textrm{,}\\
c_t&=\frac{1}{2\ln\left(\varepsilon^{-2}\right)^2\frac{m}{n}\Erw[(\bm k_t+1)^2]}
=\frac{1+o(1)}{2\ln\left(\varepsilon^{-2}\right)^2\frac{\bar d}{\bar k}\Erw[(\bm k_t+1)^2]}
\end{align*}
uniformly for all $t\in\cT^\circ_n$. This completes the proof.
%
%
\subsection{Proof of Proposition \ref{conc_std_glob}}\label{conc_std_glob_proof}
While Proposition \ref{conc_std_loc} allows to control the fluctuations of the free entropy density locally, i.e.~for given $t\in\cT^\circ_n$, Proposition \ref{conc_std_glob} allows to control the fluctuations under a variation of the degree sequences.
However, the proof strategy is fairly similar.
Since the setup for the discussion of the teacher student scheme is related but far more involved, we discuss the steps in detail.

First we notice that $\bar\phi_t=\bar\phi_{t'}$ if $t'$ is obtained from $t$ by only relabeling factors and variables.
Hence, the conditional quenched free entropy $\bar\phi_t$ only depends on the absolute degree frequencies on both sides.
Intuitively, this means that for $t$, $t'\in\cT^\circ_n$ we may assume without loss of generality that the degree sequences are sorted such that the difference on both sides is minimized, i.e.~iteratively for increasing $d\in\mathcal D$ we equip $\min(n\pr[\bm d_t=d],n\pr[\bm d_{t'}=d])$ variables with degree $d$ and keep the difference (in any order) at the end. Then we proceed analogously on the factor side.
For transparency, let $n_{\mathrm{g}}\in[n]_0$ denote the number of \emph{good} variables, i.e.~$d_{t,[n_{\mathrm{g}}]}=d_{t',[n_{\mathrm{g}}]}$ by our construction above.
Analogously, we have $m_{\mathrm{g}}\in[\min(m_t,m_{t'})]_0$ \emph{good} factors with $k_{t,[m_{\mathrm{g}}]}=k_{t',[m_{\mathrm{g}}]}$.
The remaining variables $\mathcal I_{\mathrm{b}}=[n]\setminus[n_{\mathrm{g}}]$ variables are flagged as bad, so are the remaining factors $\mathcal A_{\mathrm{b}}=[m_t]\setminus[m_{\mathrm{g}}]$ in $t$ and factors $\mathcal A'_{\mathrm{b}}=[m_{t'}]\setminus[m_{\mathrm{g}}]$ in $t'$.
Finally, assume without loss of generality that the total degree of $t$ is at least the total degree of $t'$, i.e.~$\Erw[\bm d_t]n\ge\Erw[\bm d_{t'}]n$.

Now, we couple $\mcRG_\mcT$ and $\mcRG_{\mcT'}$ by choosing the weights for the factors $\mcFL_i$, $i\in[m_{\mathrm{g}}]$, identically from $\vpsi_{k_{t,i}}$ since $k_{t,i}=k_{t',i}$ and independently for $\mathcal A_{\mathrm{b}}$ and $\mathcal A'_{\mathrm{b}}$.
Further, we draw the bijection $g:\Erw[\bm d_t]n\rightarrow\Erw[\bm d_t]n$ for $\mcRG_\mcT$ uniformly and project it down to a bijection $g':\Erw[\bm d_{t'}]n\rightarrow\Erw[\bm d_{t'}]n$ for $\mcRG_{\mcT'}$ using the switching method, i.e.~by rewiring all positions in $[\Erw[\bm d_{t'}]n]$ pointing to $[\Erw[\bm d_{t}]n]\setminus[\Erw[\bm d_{t'}]n]$ with the positions in $[\Erw[\bm d_{t}]n]\setminus[\Erw[\bm d_{t'}]n]$ pointing to $[\Erw[\bm d_{t'}]n]$ in order of appearance.
This perspective induces a partition of the variable side half-edges $\mcVHESp_t$, namely the good half-edges $\mcVHESp_{\mathrm{g}}$ of the variables $\mcVL_i$, $i\in[n_{\mathrm{g}}]$, the bad half-edges $\mcVHESp_{\mathrm{bc}}$ that $\mcRG_\mcT$ and $\mcRG_{\mcT'}$ have in common with respect to the relative representations above, and the bad half-edges $\mcVHESp_{\mathrm{be}}$ that correspond to $[\Erw[\bm d_{t}]n]\setminus[\Erw[\bm d_{t'}]n]$.
Now, the switching method only affects good factors $\mcFL_i$, $i\in[m_{\mathrm{g}}]$ that have already \emph{turned bad} by the wiring, i.e.~that connect to $\mcVHESp_{t}\setminus\mcVHESp_{\mathrm{g}}$.
In other words, the good factors $\mcFL_i$, $i\in[m_{\mathrm{g}}]$, connecting to $\mcFHESp_{\mathrm{g}}$ are not affected by the switching and are thereby the factors on which we know $\mcRG_\mcT$ and $\mcRG_{\mcT'}$ to coincide under this coupling.

Hence, the maximum number of factors on which $\mcRG_\mcT$ and $\mcRG_{\mcT'}$ differ under this coupling is given by $\max(|\mathcal A_{\mathrm{b}}|,|\mathcal A'_{\mathrm{b}}|)+|\mcVHESp_{t}\setminus\mcVHESp_{\mathrm{g}}|$.
So, in terms of the free entropy density for given $(\mcG,\mcG')$ drawn from the coupling we consider the partition of $[\max(m_t,m_{t'})]$ into the good factors $\mathcal A_{\mathrm{g}}$, the good factors $\mathcal A_{\mathrm{bw}}$ turned bad by the wiring, the bad factors $\mathcal A_{\mathrm{b}}$ given by the difference of $t$ and $t'$, and finally some additional dummy factors $\mathcal A_{\mathrm{d}}$ with constant weights $\psi_{\mcG,\mcFL_i}=1$ for $i\in\mathcal A_{\mathrm{d}}$ in case $m_{t'}>m_t$.
Then we have
\begin{align*}
\phi(\mcG)&<\frac{|\mathcal A_{\mathrm{bw}}\cup\mathcal A_{\mathrm{b}}\cup\mathcal A_{\mathrm{d}}|}{n}\ln(\varepsilon^{-1})+\frac{1}{n}\ln\left(\sum_\mcVA\prod_{i\in\mathcal A_{\mathrm{g}}}\psi_{\mcG,\mcFL_i}(\mcVA_{\partial\mcFL_i})\right)\\
&\le\frac{|\mathcal A_{\mathrm{bw}}\cup\mathcal A_{\mathrm{b}}\cup\mathcal A_{\mathrm{d}}|}{n}\ln\left(\varepsilon^{-2}\right)+\phi(\mcG')
\end{align*}
and the lower bound follows analogously.
Now, notice that
\begin{align*}
|\mathcal A_{\mathrm{b}}|+|\mathcal A'_{\mathrm{b}}|&=\sum_k|m_t\pr[\bm k_t=k]-m_{t'}\pr[\bm k_{t'}=k]|\textrm{, }\\
|\mcVHESp_{t}\setminus\mcVHESp_{\mathrm{g}}|+|\mcVHESp_{t'}\setminus\mcVHESp_{\mathrm{g}}|&=\sum_dd|n\pr[\bm d_t=d]-n\pr[\bm d_{t'}=d]|\textrm{,}
\end{align*}
so with $|\mathcal A_{\mathrm{bw}}|\le|\mathcal H_{\mathrm{v},t}\setminus\mathcal H_{\mathrm{v},\mathrm{g}}|$ (respectively the maximum of the two in general) and $|\mathcal A_{\mathrm{b}}\cup\mathcal A_{\mathrm{d}}|\le\max(|\mathcal A_{\mathrm{b}}|,|\mathcal A'_{\mathrm{b}}|)$ the above gives bounds that match the order (while we still counted a fair amount of factors as being bad although they might connect to the same variables in both models, only that the degrees of the variables differ).

Finally, by the choice of metric for the degree sequences multiple applications of the triangle inequality in order to obtain the distance in terms of the reference distributions yield
\begin{align*}
|\mathcal A_{\mathrm{b}}|+|\mathcal A'_{\mathrm{b}}|&=o(n)\textrm{,}\\
|\mcVHESp_{t}\setminus\mcVHESp_{\mathrm{g}}|+|\mcVHESp_{t'}\setminus\mcVHESp_{\mathrm{g}}|&=o(n)\textrm{,}
\end{align*}
i.e.~bounds that are uniform over any choice of $t$, $t'\in\cT^\circ_n$ and $(\mcG,\mcG')$ from the corresponding coupling, so
$|\bar\phi_t-\bar\phi_{t'}|=o(1)$ uniformly in $t$, $t'$ and hence $|\bar\phi_t-\bar\phi|=o(1)$ uniformly in $t$ (using Proposition \ref{dd_typ_prop}).
%
%
\subsection{Proof of Proposition \ref{conc_tss_loc}}\label{conc_tss_loc_proof}
The proof of this result is very similar to the proof of Proposition \ref{conc_std_loc}, however with an additional layer of complexity due to the assignment pairs.
Recall the distribution of $\mcRG^*_\acS$, $\acS=(t,\mcVA,\mcFSHEA)$ in the support of $\acRS^*_n$, from Section \ref{Sec_Noela} and notice that it is very similar to obtaining $\mcRG_\mcT$, with reweighted distributions for the weight functions and restrictions in the choice of bijections.

The proof of Azuma's inequality for this model is now in almost complete analogy to the proof of Proposition \ref{conc_std_loc}, only that the completions for the bijections have to be chosen separately, while the switching method is not affected (consistency with colors is preserved by switching, since the assignment pair $(\mcVA,\mcFSHEA)$ coincides in both models). Recalling the result this gives
$c'=2$ and
\begin{align*}
c_s&=\frac{1}{2\ln\left(\varepsilon^{-2}\right)^2\frac{m_t}{n}\Erw[(\bm k_t+1)^2]}
\end{align*}
using only {\bf SYM} and for all $s$. Using {\bf DEG} we obtain the uniform bound $c\in\mathbb R_{>0}$ for $s$ with $t\in\cT^\circ_n$ and $n\in\cN$.
%
%
\subsection{Proof of Proposition \ref{conc_tss_glob}}\label{conc_tss_glob_proof}
The result follows from a combination of the concepts for the proof of Proposition \ref{conc_std_glob} and the model introduced in Section \ref{Sec_Noela}.
First, we observe the invariance of $\phi^*_{\acS}$ with respect to a relabeling of variables and factors.
Fix $\acS\in\acSSpT_n$.
Since only the frequencies on both sides are relevant, we may assume that the sequences are sorted as in the Proof of Proposition \ref{conc_std_glob}, with the corresponding partitions into \emph{good} and \emph{bad} factors as well as \emph{good} and \emph{bad} variables.
The model introduced in Section \ref{conc_tss_loc_proof} allows for the same switching strategy, however this time we draw the $q$ bijections separately, and each bijection wiring half-edges of color $\mcCol\in\mcColSp$ for the model with most half-edges of color $\mcCol$ (a quantity that depends on $t$, $t'$, $(\mcVA,\mcFSHEA)$ and $(\mcVA',\mcFSHEA')$ only).
As before, all \emph{good} factors that don't \emph{turn bad} by the wiring of nodes of color $\mcCol\in\mcColSp$ are not affected (regardless of the direction from which model to which model we project), hence the total number $m_{\mathrm{gb}}(\mcCol)$ of \emph{good} factors \emph{turned bad} can still be upper bounded by the maximum $h_{\max,\mathrm{vb}}(\mcCol)=\max(h_{\mathrm{vb}}(\mcCol),h'_{\mathrm{vb}}(\mcCol))$ of the total degrees $h_{\mathrm{vb}}(\mcCol)$, $h'_{\mathrm{vb}}(\mcCol)$ of all \emph{bad} variables of given color $\mcCol$.
With $h_{\max,\mathrm{vb}}(\mcCol)\le h_{\mathrm{vb}}(\mcCol)+h'_{\mathrm{vb}}(\mcCol)$ and summing over all $\mcCol\in\mcColSp$ recovers the upper bound from the proof of Proposition \ref{conc_std_glob}, i.e.~the number of all good factors turned bad is at most the sum $h_{\mathrm{vb}}+h'_{\mathrm{vb}}$ of the total degrees $h_{\mathrm{vb}}$, $h'_{\mathrm{vb}}$ of the \emph{bad} variables in both models. In addition to these we need to consider the \emph{bad} factors, so analogously to the standard model case the number of disagreeing factors can be upper bounded by $h_{\mathrm{vb}}+h'_{\mathrm{vb}}+m_{\mathrm{b}}+m'_{\mathrm{b}}$ with $m_{\mathrm{b}}$, $m'_{\mathrm{b}}$ denoting the numbers of bad factors.
As before, we notice that
\begin{align*}
h_{\mathrm{vb}}+h'_{\mathrm{vb}}=\sum_{\mcVdeg,\mcVHEA}\mcVdeg|n\acVAF{\mcT,\mcVA}(\mcVdeg,\mcVHEA)-n\acVAF{\mcT',\mcVA'}(\mcVdeg,\mcVHEA)|\textrm{, }
m_{\mathrm{b}}+m'_{\mathrm{b}}=\|m_t\acFAF{\mcT,\mcFSHEA}-m_{t'}\acFAF{\mcT',\mcFSHEA'}\|_1\textrm{.}
\end{align*}
Now, by design of the metrics for the degree sequences and assignment sequences, multiple applications of the triangle inequality yield $h_{\mathrm{vb}}+h'_{\mathrm{vb}}=o(n)$, $m_{\mathrm{b}}+m'_{\mathrm{b}}=o(n)$.
This shows that
$\phi^*_{\acS'}=\phi^*_{\acS}+o(1)$ uniformly for all $\acS$, $\acS'\in\acSSpT_n$, and thereby shows $\phi^*=\phi^*_{\acS}+o(1)$ uniformly using Lemma \ref{conc_typ_qfed_prop}.

Mutual contiguity of $\bm s^*_{n}$ and $\hat{\bm s}_{n}$ as discussed in \ref{mc_hence} implies that
$\hat{\bm s}_{n}\in\mathcal S^\circ_{n}$ with high probability
since $\bm s^*_{n}\in\mathcal S^\circ_{n}$ with high probability.
This suggests that
$\hat\phi
=\Erw[\phi^*_{\hat{\bm s}_{n}}\vecone\{\hat{\bm s}_{n}\in\mathcal S^\circ_{n}\}]+o(1)
=\phi^*_s+o(1)$
since $\phi^*_{s'}=\phi^*_{s}+o(1)$ uniformly for all $s$, $s'\in\mathcal S^\circ_n$.
%
%
\subsection{Pruning factors}\label{conc_pruned}
As discussed in Section \ref{mc}, all arguments for fixed sequences $t$ canonically translate to the factor pruned model.
Hence, the only missing argument is that the free entropy of the generalized factor graphs in Section \ref{mc} for arbitrary sequences is still sublinear in the number of factors, i.e.~we only need to establish Lemma \ref{conc_prop_dd_typ}.
But this result is immediate from the definition of the factor graphs, since the number of non-trivial factors, i.e.~factors whose weight functions are not constant $1$, can always be upper bounded by the total number of factors.
\section{Proof of \Prop~\ref{Prop_prune}}\label{prune}

In this section we show Proposition \ref{Prop_prune} and further justify the factor pruning using the results, notions and conventions of Section \ref{conc}.

With $\vt_{\varepsilon,n}$ introduced in Section \ref{mc} and $\vt_{0,n}=\vt_n$ as introduced in Section \ref{dd} let
$\phi^*(\varepsilon,n)=\Erw[\phi(\mcRG^*_{\vt_{\varepsilon,n}}(\mcRVA^*)]$.
We say that $\phi^*$ is asymptotically continuous in $\varepsilon^*_{\mathrm{m}}\in[0,1)$ if for all $\varepsilon\in\mathbb R_{>0}$ there exists $\delta\in\mathbb R_{>0}$ such that for all $\varepsilon_{\mathrm{m}}\in[0,1)\cap\mathcal B_\delta(\varepsilon^*_{\mathrm{m}})$ there exists $n_0\in\mathcal N$ such that $|\phi^*(\varepsilon_{\mathrm{m}},n)-\phi^*(\varepsilon^*_{\mathrm{m}},n)|<\varepsilon$ for all $n\in\mathcal N$ with $n\ge n_0$.
Further, $\phi^*$ is asymptotically continuous in the number of factors if the above holds for all $\varepsilon_{\mathrm{m}}$.
This property ensures that $\phi^*$ can be asymptotically approximated, without assuming that a limit $\lim_{n\rightarrow\infty}\phi^*(\varepsilon,n)$ exists, and without enforcing uniform convergence in that $n_0$ may depend on the choice of the parameter.
\begin{proposition}\label{ef_cont_prop}
The quenched free entropy density $\phi^*$ is asymptotically continuous in the number of factors.
\end{proposition}
\begin{proof}
This result is immediate by combining Proposition \ref{conc_tss_glob} with the coupling in Section \ref{conc_tss_glob_proof} used to obtain Proposition \ref{conc_tss_glob} since we derived bounds in terms of the distance of coloured sequences $\acS$.
\end{proof}

\begin{proof}[Proof of \Prop~\ref{Prop_prune}]
We proceed with the proof for the configuration model and discuss the translation to simple factor graphs at the end.
Proposition \ref{ef_cont_prop} directly translates to degree distributions as follows.
We equip $\cT^*_n$ with the product metric induced by $\Delta$ discussed in Section \ref{dd}, preferably omitting the parts ensuring convergence of the higher moments. Since the underlying assignment distributions $\mu_k$ of the reference distribution (in Section \ref{typ_ass}) given $k$ are invariant to the choice of the degree distribution (analogously on the variable side), Proposition \ref{conc_tss_glob} with the coupling in Section \ref{conc_tss_glob_proof} ensures that $|\phi^*_{1,n}-\phi^*_{2,n}|$ is small for $n$ sufficiently large if $\Delta(t_1,t_2)$ is small, with $\phi^*_{i,n}$ denoting the quenched free entropy density of the teacher student model and $t_i\in\cT^\circ_{i,n}$ denoting typical degree sequences obtained from the degree distributions $(p_{\mathrm{d},i},p_{\mathrm{k},i})$.
This ensures that it is indeed sufficient to work with finitely supported degree distributions in order to approximate the quenched free entropy density (in the limit).

We're left to show that the Bethe functional is also continuous with respect to the degree distributions, then Proposition \ref{Prop_prune} follows from an $\varepsilon$-argument.
For given $d$ in the support of $\vd$, $k_i$ in the support of $\vk$, $h_i\in[k_i]$ and $\psi_i$ in the support of $\vpsi_{k_i}$ for $i\in[d]$, and finally $\mu_{i,j}$ in the support of $\pi\in\PP_*(\mcColSp)$ for $j\in[k_i]$ with $k=(k_i)_i$, $h=(h_i)_i$, $\psi=(\psi_i)_i$ and $\mu=(\mu_{i,j})_{i,j}$ let
\begin{align*}
M_{\mathrm{V}}(d,k,h,\psi,\mu)=\sum_{\mcCol\in\mcColSp}\prod_{i=1}^d\sum_{\tau\in\mcColSp^{k_i}}\vecone\{\tau_{h_i}=\mcCol\}\psi_i(\tau)\prod_{j\in[k_i]\setminus\{h_i\}}\mu_{i,j}(\tau_j)\textrm{.}
\end{align*}
Using the canonical bounds for $\mu$ and {\bf SYM} yields
\begin{flalign*}
q\varepsilon^d\prod_{i=1}^d q^{-(k_i-1)}\le M_{\mathrm{V}}(d,k,h,\psi,\mu)\le q\xi^d\prod_{i=1}^d q^{k_i-1}
\end{flalign*}
uniformly in $\psi$ and $\mu$. On the other hand, with $\pi\in\PP_*(\mcColSp)$ and with respect to Equation \ref{eqBFE} we have
\begin{align*}
\Erw[q^{-1}\xi^{-d}M_{\mathrm{V}}(d,k,(\vh_{k_i,i})_i,(\vpsi_{k_i,i})_i,(\vmu_{i,j,\pi})_{i,j})]=1\textrm{.}
\end{align*}
But now, using the uniform bounds inside the logarithm and the result for the expectation afterwards, the first contribution
\begin{align*}
\cB_1(\pi)=\Erw\left[q^{-1}\xi^{-\vd}\Lambda\left(M_{\mathrm{V}}(\vd,(\hat{\vk}_i)_i,(\vh_{\hat{\vk}_i,i})_i,(\vpsi_{\hat{\vk}_i,i})_i,(\vmu_{i,j,\pi})_{i,j})\right)\right]
\end{align*}
to the Bethe functional can be uniformly bounded by
\begin{align*}
\Erw\left[\ln\left(q\varepsilon^{\vd}\prod_{i=1}^{\vd} q^{-(\hat{\vk}_i-1)}\right)\right]\le\cB_1(\pi)\le\Erw\left[\ln\left(q\xi^{\vd}\prod_{i=1}^{\vd} q^{\hat{\vk}_i-1}\right)\right]\textrm{, }
\end{align*}
and is in particular finite. For the second contribution and fixed $k$ we obtain the uniform bound
\begin{align*}
M_{\mathrm{F}}(k,\vpsi_k,(\vmu_{j,\pi})_j)=\sum_\tau\vpsi_k(\tau)\prod_j\vmu_{j,\pi}(\tau_j)\le q^k\xi
\end{align*}
and $M_{\mathrm{F}}(k,\vpsi_k,(\vmu_{j,\pi})_j)\ge q^{-k}\varepsilon$ inside the logarithm, and the expectation
\begin{align*}
\Erw\left[M_{\mathrm{F}}(k,\vpsi_k,(\vmu_{j,\pi})_j)\right]=\xi\textrm{,}
\end{align*}
hence the second contribution
\begin{align*}
\cB_2(\pi)=\frac{\Erw[\vd]}{\Erw[\vk]}\Erw\left[(\vk-1)\Lambda\left(M_{\mathrm{F}}(\vk,\vpsi_{\vk},(\vmu_{j,\pi})_j)\right)\right]
\end{align*}
can be uniformly bounded by
\begin{align*}
\frac{\Erw[\vd]}{\Erw[\vk]}\Erw\left[(\vk-1)\xi\ln\left(q^{-\vk}\varepsilon\right)\right]\le\cB_2(\pi)\le\frac{\Erw[\vd]}{\Erw[\vk]}\Erw\left[(\vk-1)\xi\ln\left(q^{\vk}\xi\right)\right]\textrm{, }
\end{align*}
so in particular the expectations in the Bethe functional are finite.
However, most importantly the above suggests that the Bethe functional as a function of the degrees $\vd$, $\vk$ is uniformly continuous in the following sense. Let $\vd'$, $\vk'$ be finitely supported degrees such that both $\vd'$, $\vk'$ are close to $\vd$, $\vk$ and $\hat{\vd}'$, $\hat{\vk}'$ are close to $\hat{\vd}$, $\hat{\vk}$ in total variation (which gives bounds on the distance of the first moments), then so are the Bethe functionals uniformly for $\pi\in\PP_*(\mcColSp)$.
The argumentation is similar to the discussion in Section \ref{pp_cont}. The fact that the expectations are finite ensures that we can cut the tails (in $d$, $k=(k_i)_{i\in[d]}$) at arbitrarily small loss, leaving us with a uniform bound for the remaining contributions.
Choosing suitable (finitely supported) distributions $\vd'$, $\vk'$ sufficiently close to $\vd$, $\vk$ then ensures that cutting the tails with respect to $\vd'$, $\vk'$ comes at an arbitrarily small loss and further using the uniform bounds for the remainder we obtain a uniform bound on the distance of the Bethe functionals in terms of the distance of the degree distributions, thereby ensuring uniform continuity.
This immediately translates to $\sup_{\pi\in\PP_*(\mcColSp)}\cB(\pi)$.

With these continuity results we can show that the quenched free entropy density coincides with the supremum of the Bethe functional given that the assertion holds for finitely supported degree distributions.
For any given $\varepsilon$ choose degrees $\vd'$, $\vk'$ with finite support close to $\vd$, $\vk$ in the metric above.
Then the distance of the supremum of the Bethe functional with respect to the two pairs of degree distributions can be bounded by $\varepsilon/3$. Further, for $n$ sufficiently large the quenched free entropy density with respect to $\vd'$, $\vk'$ is at a distance at most $\varepsilon/3$ to the supremum of the Bethe functional with respect to $\vd'$, $\vk'$ since we obtained the results for bounded degrees.
But by the continuity result for the quenched free entropy density above, we know that for $n$ sufficiently large the quenched free entropy densities with respect to $\vd'$, $\vk'$ and with respect to $\vd$, $\vk$ are also at most at a distance $\varepsilon/3$. Taking $\varepsilon$ to $0$ completes the proof.

Since standard arguments, i.e., Section 5 in \cite{Coja_2018} show that there exists a simple $\G$ with the desired degree sequences with positive probability, the proposition readily follows. 
\end{proof}

\section{Proof of \Prop~\ref{Prop_ASS}} \label{Sec_Prop_ASS}

\subsection{Overview}\label{Sec_Prop_ASS_overview}
For a given $\eps>0$ we let $\vm_{\eps,n}$ be a Poisson variable with mean $(1-\eps)\d n/\k$.
Moreover, let $\G_{\eps,n}$ be the random factor graph with variables nodes $x_1,\ldots,x_n$ and factor nodes $a_1,\ldots,a_{\vm_{\eps,n}}$ obtained as follows.
Let
\begin{align*}
\cX&=\bigcup_{i=1}^n\{x_i\}\times[\vd_i],&\cA&=\bigcup_{i=1}^{\vm_{\eps,n}}\{a_i\}\times[\vk_i]
\end{align*}
contain clones of the variable nodes $x_1,\ldots,x_n$ and of the factor nodes, respectively.
Then choose a maximal matching $\vGamma_{\eps,n}$ of the complete bipartite graph on the vertex classes $\cX,\cA$.
For each matching edge we insert the corresponding variable--factor node edge into $\vG_{\eps,n}$.
Finally, for each factor node $a_i$ we choose a weight function $\psi_{a_i}$ independently from the distribution $P$.


Let $\hat\G_{\eps,n},\G^*_{\eps,n}$ be the random factor graph models obtained from $\G_{\eps,n}$ via \eqref{eq_hatG}, \eqref{eq_starG}.
Further, let $\vsigma^*_n:\{x_1,\ldots,x_n\}\to\Omega$ be a uniformly random assignment.
Since $\Erw[\vm_{\eps,n}]<\sd n/\sk-\Omega(n)$, \whp\ some of the variable clones from $\cX$ remain vacant in the random factor graph $\G_{\eps,n}^*$.
Let $\cC^*$ denote the set of all such vacant clones.
As before, we refer to them as the cavities.
Further, let $(\vy_{i,j})_{i,j\geq1}$ be a sequence of uniformly chosen independent cavities.
Also let $\vd_\eps$ be a random variable with distribution $\Bin(\vd,1-\eps)$.
By Remark~\ref{remark_v}, the main step toward the proof of \Prop~\ref{Prop_ASS} is to show the following.

\begin{proposition}\label{Prop_ASS*}
We have
\begin{align*}
\Erw&\brk{\log Z_{\G^*_{\eps,n+1}}}-\Erw\brk{\log Z_{\G^*_{\eps,n}}}\\&
\leq\Erw\brk{q^{-1}\xi^{-\vd_\eps}\bc{\sum_{\sigma\in\Omega}\prod_{i=1}^{\vd_\eps}\PSI_{\hat\vk_i,i}(\sigma,\vsigma_{\vy_{i,2}}^*,\ldots,\vsigma_{\vy_{i,\hat\vk_i}}^*)}
\log\sum_{\sigma\in\Omega}\bck{\prod_{i=1}^{\vd_\eps}\PSI_{\hat\vk_i,i}(\sigma,\vsigma_{\vy_{i,2}},\ldots,\vsigma_{\vy_{i,\hat\vk_i}})}_{\G_{\eps,n}^*}}\\
&\quad-\frac{(1-\eps)\d}{\xi\k}\Erw\brk{(k_{\PSI}-1)\PSI(\vsigma^*_{\vy_{1,1}},\ldots,\vsigma^*_{\vy_{1,k_{\PSI}}})\log\bck{\PSI(\vsigma_{\vy_{1,1}},\ldots,\vsigma_{\vy_{1,k_{\PSI}}})}_{\G_{\eps,n}^*}} +o(1).
\end{align*}
\end{proposition}

To prove \Prop~\ref{Prop_ASS*} we couple the random factor graphs $\G^*_{\eps,n+1}$ and $\G^*_{\eps,n}$.
Specifically, for each $j$ in the support of $\vk$ let $\vM_j$ be a random variable with distribution $\Po((1-\eps)\d\pr\brk{\vk=j}n/\k)$.
Further, let $\vDelta_j$ be a random variable with distribution
\begin{align*}
\vDelta_j\disteq\Po((1-\eps)\d\pr\brk{\vk=j}/\k).
\end{align*}
Additionally, let $\vM_j^+=\vM_j+\vDelta_j$.
Further, let $\vM=(\vM_j)_j,\vM^+=(\vM_j^+)_j$ and let $\vG_{n,\vM}^*$, $\vG_{n,\vM^+}^*$ be the factor graphs obtained as follows.
Choose a random maximal matching $\vGamma_{n,\vM}$ of the complete bipartite graph with vertex classes
\begin{align*}
\cX_n&=\bigcup_{i=1}^n\{x_i\}\times[\vd_i],&\cA_{n,\vM}&=\bigcup_{i\in\supp\vk}\bigcup_{j\in[\vM_i]}\{a_{i,j}\}\times[i].
\end{align*}
Then let $\vG_{n,\vM}$ be the random factor graph with variable nodes $x_1,\ldots,x_n$ and factor nodes $a_{i,j}$, $i\in\supp\vk$, $j\in[\vM_i]$ where each edge of $\vGamma_{n,\vM}$ induces an edge between the corresponding variable and check node.
Additionally, the factor nodes $a_{i,j}$ receive independent weight functions with distribution $P_i$.
Finally, $\vG^*_{n,\vM}$ is the factor graph obtained from $\vG_{n,\vM}$ via \eqref{eq_starG}.
The model $\vG^*_{n,\vM^+}$ is defined analogously.

\begin{lemma}\label{Lemma_MM+}
The random factor graphs $\G_{\eps,n}^*$,$\G_{n,\vM}$ and $\G_{\eps,n+1}^*$, $\G_{n+1,\vM^+}$ are identically distributed.
\end{lemma}
\begin{proof}
This is immediate from the construction.
\end{proof}

Let $\vgamma_i$ be the number of factor nodes of degree $i$ adjacent to $x_{n+1}$ in $\vG^*_{n,\vM^+}$.
Further, let $$\vM_i^-=0\vee(\vM_i-\vgamma_i)$$ and let $\vG_{n,\vM^-}$, $\vG_{n,\vM^-}^*$ be the corresponding factor graphs.
Additionally, let $\fD^-$ be the $\sigma$-algebra generated by $(\vM_i,\vec\gamma_i,\vec\Delta_i)_{i\geq1}$ and $\vsigma^*_{n+1}$ and let $\fM^-$ be the $\sigma$-algebra generated by $\fD^-$ and $\vG_{n,\vM^-}^*$.

To set up the coupling, obtain $\G'$ from $\G_{n,\vM^-}^*$ as follows.
Let $\cC^-$ be the set of cavities of $\G_{n,\vM^-}^*$.
Moreover, for $i\in\supp\vk$ and $j\in[\vM_i-\vM_i^-]$ let $a_{i,j}'$ be a new factor node.
Now, obtain $\G'$ by adding the $a_{i,j}'$ to $\G_{n,\vM^-}^*$ by pairing them to cavities from $\cC^-$ and choosing weight functions such that for any possible result of this experiment we have
\begin{align}
\pr\brk{\G'=g\mid\fM^-}&\propto\prod_{i,j}P_i(\psi_{a_{i,j}'})\psi_{a_{i,j}'}(\vsigma^*). \label{eq_ASS_factor_dist}
\end{align} 

Additionally, let $\G''$ be the random factor graph obtained from $\G_{n,\vM^-}^*$ via the following process.
Add a variable node $x_{n+1}$, factor nodes $a_{i,j}''$ for $i\in\supp\vk$, $j\in[\vM_i^+-\vM_i-\vgamma_i]$ and further factor nodes $a_{i,j}'''$ for $i\in\supp\vk$, $j\in[\vgamma_i]$ with $x_{n+1}\in\partial a_{i,j}'''$ according to the distribution
\begin{align*}
\pr\brk{\G''=g\mid\G_{n,\vM^-}^*,\vsigma_{n+1}^*}&\propto\prod_{i,j}P_i(\psi_{a_{i,j}''})\psi_{a_{i,j}''}(\vsigma^*)\prod_{i,j}P_i(\psi_{a_{i,j}'''})\psi_{a_{i,j}'''}(\vsigma^*)
\end{align*}

\begin{lemma}\label{Lemma_NoelaGraphs}
We have
\begin{align*}
\Erw\brk{\log Z(\G')}&=\Erw\brk{\log Z(\G_{n,\vM}^*)}+o(1),&
\Erw\brk{\log Z(\G'')}&=\Erw\brk{\log Z(\G_{n+1,\vM^+}^*)}+o(1).
\end{align*}
\end{lemma}
\begin{proof}
By construction, $\G'$ is obtained from $\G_{n,\vM^-}^*$ by adding $\sum_{i \in \supp \vk} \vM_i-\vM_i^-$ factor nodes.
Because all degrees are bounded, we have
\begin{align*}
    \Erw \brk{\sum_{i \in \supp \vk} \vM_i-\vM_i^-} = \Theta(1)
\end{align*}
Since a Poisson random variable with bounded expectation is bounded by $O(\log n)$ with probability $1-o(1/n)$, we may assume that the number of factor nodes added from $\G_{n,\vM^-}^*$ to $\G'$ is $O(\log n)$.
Let us add these factor nodes one-by-one.
Then by \Prop~\ref{prop_coupling_int} we can couple $\G'$ and $\G_{n,\vM}^*$ such that
\begin{align*}
    \Pr \brk{\G'=\G_{n,\vM}^*} = 1 - \tilde O(n^{-1})
\end{align*}
whence the first statement of the lemma follows.

Let $$ \cE = \cbc{ \G''-x_{n+1}-\sum_{i,j}a_{i,j}''-\sum_{i,j}a_{i,j}'''=\G_{n+1,\vM^+}^*-\sum_{i,j}a_{i,j}''-\sum_{i,j}a_{i,j}'''  }$$ be the event that on the first $n$ variables the factor graphs $\G''$ and $\G^*_{n+1, \vec M^+}$ coincide. Furthermore, denote by $$ \Delta_s = \abs{\G''-x_{n+1}-\sum_{i,j}a_{i,j}''-\sum_{i,j}a_{i,j}''' \quad \triangle \quad  \G_{n+1,\vM^+}^*-\sum_{i,j}a_{i,j}''-\sum_{i,j}a_{i,j}'''} $$ the amount of edges in which the factor graphs differ (restricted on the first $n$ variables). Then \Prop~\ref{prop_coupling_int} and \Prop~\ref{Prop_oplus} show that $\G''$ and $\G_{n+1,\vM^+}^*$ can be coupled such that 
\begin{align}\label{eqAddVar1}
\pr\brk{ \cE }=1-\tilde O(n^{-1}), \qquad \pr\brk{\Delta_s >\sqrt n\log n}  = \tilde O(n^{-2}).
\end{align} 
Furthermore, comparing the definitions of $\G''$ and $\G_{n+1,\vM^+}^*$, we see that given $\cE$ the factor graphs $\G''$ and $G_{n+1,\vM^+}^*$ satisfy
$$ \dTV \bc{\G''_{ \mid \cE}, G_{n+1,\vM^+ \mid \cE}^* } = \tilde O( n^{-1} ).$$ As all weight functions are strictly positive by assumption, there is a coupling of $\G''$ and $G_{n+1,\vM^+}^*$ such that
\begin{align}
    \label{eqAddVar3} \abs{\Erw \brk{\log Z(\G'')-\log Z(\G_{n+1,\vM^+}^*) \mid \cE}} = o(1).
\end{align}
Additionally, given $\cE_2 = \cbc{ \Delta_s \leq \sqrt n\log n }$, we find
\begin{align}
    \label{eqAddVar4}  \abs{\Erw\brk{ \log Z(\G'')-\log Z(\G_{n+1,\vM^+}^*)|=O(\sqrt n\log n) \mid \cE_2} }.
\end{align}
Since, finally, 
\begin{align}
    \label{eqAddVar5} \abs{\log Z(G'')-\log Z(\G_{n+1,\vM^+}^*)}=O(n)
\end{align} 
deterministically, the second assertion follows from \eqref{eqAddVar1}--\eqref{eqAddVar5}.

\end{proof}

Let $(\vgamma_i')_{i\in\supp\vk}$ be a random vector with distribution
\begin{align*}
\vgamma_i'&=\sum_{h=1}^{\vd_\eps}\vecone\{\hat\vk_h=i\}.
\end{align*}

\begin{lemma}\label{Lemma_NoelaGamma}
We have $\dTV((\vgamma_i)_{\in\supp\vk},(\vgamma_i')_{i\in\supp\vk})=o(1)$.
\end{lemma}
\begin{proof}
Let $\cE$ be the event that the new variable node $x_{n+1}$ is adjacent to particular factor nodes $\alpha_1,\ldots,\alpha_\ell$, ordered according to the clones of $x_{n+1}$ that they connect to.
Let $\kappa_1,\ldots,\kappa_\ell$ be the degrees of $\alpha_1,\ldots,\alpha_\ell$.
Furthermore, let $\G^\star$ be the factor graph obtained from $\G_{n+1,\vM^+}$ by removing $x_{n+1}$ and its adjacent factor nodes.
Finally, let $\cR$ be the event that $\G^\star$ has $(1+o(1))\Delta/q$ cavities with each possible value $\tau\in\Omega$ under $\vsigma^*$.
Then \Prop~\ref{prop_coupling_int_m} implies that
\begin{align}\label{eqLemma_NoelaGamma0}
\pr\brk{\cE}&=\pr\brk{\cE\mid\G^\star\in\cR}+o(1).
\end{align}
To be precise, in order to apply \Prop~\ref{prop_coupling_int_m} we think of $x_{n+1}$ and its adjacent factor nodes $\alpha_1,\ldots,\alpha_\ell$ as a single `super-factor node' with weight function
\begin{align*}
\psi_{x_{n+1},\alpha_1,\ldots,\alpha_\ell}(\sigma)=\sum_{\sigma_{x_{n+1}}\in\Omega}\prod_{i=1}^\ell\psi_{\alpha_i}(\sigma_{x})_{x\in\partial\alpha_\ell}&&(\sigma\in\Omega^{V_n}).
\end{align*}
Furthermore, the random factor graph model $\G^\star$ can be described as follows.
There are $\ell$ fewer factor nodes, and thus \Prop~\ref{prop_coupling_int_m} and {\bf SYM} imply that \whp\
\begin{align}\label{eqLemma_NoelaGamma1}
\frac{\Erw[\psi_{\G^\star}(\vsigma^*)\mid\vsigma^*]}{\Erw[\psi_{\G_{n+1,\vM^+}}(\vsigma^*)\mid\vsigma^*]}=\xi^{-\ell}.
\end{align}
Similarly, 
\begin{align}\label{eqLemma_NoelaGamma2}
\frac{\Erw[\psi_{\G^\star}(\vsigma^*)\mid\vsigma^*]}{\Erw[\psi_{\G_{n+1,\vM^+}}(\vsigma^*)\mid\vsigma^*,\cE]}=\xi^{-\ell}.
\end{align}
Combining \eqref{eqLemma_NoelaGamma1}--\eqref{eqLemma_NoelaGamma2}, we obtain
\begin{align}\label{eqLemma_NoelaGamma4}
\pr\brk{\cE}&\sim\pr\brk{\vd_\eps=\ell}\prod_{h=1}^\ell\pr\brk{\hat\vk_h=\kappa_h}.
\end{align}
Finally, \Lem s \ref{lem_cav1} and \ref{lem_cav2} ensure that \whp there are $(1+o(1))\Delta/q$ cavities of each possible colour $\tau\in\Omega$. Thus, the assertion follows from~\eqref{eqLemma_NoelaGamma4}.
\end{proof}

\begin{lemma}\label{Lemma_G'}
We have
\begin{align*}
\Erw\brk{\log Z_{\G'}-\log Z_{\G_{n,\vM^-}}}&=\frac{(1-\eps)\d}{\xi\k}\Erw\brk{k_{\PSI}\PSI(\vsigma^*_{\vy_{1,1}},\ldots,\vsigma^*_{\vy_{1,k_{\PSI}}})\log\bck{\PSI(\vsigma_{\vy_{1,1}},\ldots,\vsigma_{\vy_{1,k_{\PSI}}})}_{\G_{\eps,n}}} +o(1).
\end{align*}
\end{lemma}
\begin{proof}
Since $\G'$ is obtained from $\G_{n,\vM^-}$ by adding factor nodes $a'_{i,j}$ for $i \in \supp \vk$ and $j \in \brk{\vM_i-\vM_i^-}$ according to \eqref{eq_ASS_factor_dist}, we obtain
\begin{align*}
    \log \frac{Z_{\G'}}{Z_{\G_{n,\vM^-}}} = \log \Bigg \langle \prod_{i \in \supp \vk} \prod_{j \in \brk{\vM_i-\vM_i^-}} \psi_{a'_{i,j}} \bc{\SIGMA(\partial_1 a'_{i,j}), \dots, \SIGMA(\partial_{k_{a'_{i,j}}} a'_{i,j})} \Bigg \rangle_{\G^*_{n,\vM^-}}
\end{align*}
Therefore, with $(\vy_i)_{i\geq1}$ signifying independent uniformly random cavities of $\G_{n,\vM^-}^*$, we obtain
\begin{align}\label{eqLemma_G'1}
\Erw\brk{\log Z_{\G'}-\log Z_{\G_{n,\vM^-}}}&=\frac{(1-\eps)\d}{\xi\k}\Erw\brk{k_{\PSI}\PSI(\vsigma^*_{\vy_{1,1}},\ldots,\vsigma^*_{\vy_{1,k_{\PSI}}})\log\bck{\PSI(\vsigma_{\vy_{1,1}},\ldots,\vsigma_{\vy_{1,k_{\PSI}}})}_{\G_{n,\vM^-}^*}} +o(1).
\end{align}
Since $\G_{n,\vM^-}^*$ and $\G_{\eps,n}^*$ have total variation distance $o(1)$ while the expression inside the expectation is bounded, the assertion follows from \eqref{eqLemma_G'1}.
\end{proof}

\begin{lemma}\label{Lemma_G''}
We have
\begin{align*}
\Erw&\brk{\log Z_{\G''}-\log Z_{\G_{n,\vM^-}}}\\&\quad=\Erw\brk{q^{-1}\xi^{-\vd_\eps}\bc{\sum_{\sigma\in\Omega}\prod_{i=1}^{\vd_\eps}\PSI_{\hat\vk_i,i}(\sigma,\vsigma_{\vy_{i,2}}^*,\ldots,\vsigma_{\vy_{i,\hat\vk_i}}^*)}
\log\sum_{\sigma\in\Omega}\bck{\prod_{i=1}^{\vd_\eps}\PSI_{\hat\vk_i,i}(\sigma,\vsigma_{\vy_{i,2}},\ldots,\vsigma_{\vy_{i,\hat\vk_i}})}_{\G_{\eps,n}}}\\
&\qquad+\frac{(1-\eps)\d}{\xi\k}\Erw\brk{\PSI(\vsigma^*_{\vy_{1,1}},\ldots,\vsigma^*_{\vy_{1,k_{\PSI}}})\log\bck{\PSI(\vsigma_{\vy_{1,1}},\ldots,\vsigma_{\vy_{1,k_{\PSI}}})}_{\G_{\eps,n}}} +o(1).
\end{align*}
\end{lemma}
\begin{proof}
Since $\G''$ is obtained from $\G_{n,\vM^-}$ by adding a variable node $x_{n+1}$ with associated factor nodes $a'''_{i,j}$ for $i \in \supp \vk, j \in [\vgamma_i]$ and further factor nodes $a''_{i,j}$ for $i \in \supp \vk, j \in [\vM^+_i-\vM_i-\vgamma_i]$, we obtain
\begin{align}\label{eqLemma_G''1}
    \log \frac{Z_{\G''}}{Z_{\G_{n,\vM^-}}} &= \log \sum_{\sigma \in \Omega} \Bigg \langle \prod_{i \supp \vk} \prod_{j \in [\vgamma_i]} \psi_{a'''_{i,j}} \bc{\sigma, \SIGMA(\partial_2 a'''_{i,j}), \dots, \SIGMA(\partial_{k_{a'_{i,j}}} a'''_{i,j}} \Bigg \rangle_{\G^*_{n,\vM^-}}\\
    &\qquad \qquad + \log \Bigg \langle \prod_{i \in \supp \vk} \prod_{j \in \brk{\vM^+_i-\vM_i-\vgamma_i}} \psi_{a''_{i,j}} \bc{\SIGMA(\partial_1 a''_{i,j}), \dots, \SIGMA(\partial_{k_{a''_{i,j}}} a''_{i,j}} \Bigg \rangle_{\G^*_{n,\vM^-}}
\end{align}
The assertion follows from \eqref{eqLemma_G''1}, \Lem~\ref{Lemma_NoelaGamma} and the fact that $\G_{n,\vM^-}^*$ and $\G_{\eps,n}^*$ have total variation distance $o(1)$.
\end{proof}

\begin{lemma}\label{Lemma_sigma}
Let $(\vy_i)_{i\geq1}$ be a sequence of uniformly random independent cavities of $\G_{\eps,n}^*$.
For any $\ell\geq1$, $\delta>0$ there exists $\theta$ such that for all functions $f:\Omega^\ell\to[0,1]$ we have 
\begin{align}\label{eqLemma_sigma}
\abs{\Erw\brk{f(\vsigma^*_{\vy_{1,1}},\ldots,\vsigma^*_{\vy_{1,\ell}})\mid\G_{\eps,n}^*}
	-\Erw\brk{\bck{f(\vsigma_{\vy_{1,1}},\ldots,\vsigma_{\vy_{1,\ell}})}\mid\G_{\eps,n}^*} }<\delta.
\end{align}
\end{lemma}
\begin{proof}
Going back to the definitions of $\G^*$ and the Boltzmann distribution, we obtain
\begin{align}\nonumber
\pr\brk{\vsigma^*=\sigma\mid\G_{\eps,n}^*=G}&=\frac{\pr\brk{\G_{\eps,n}^*=G\mid\vsigma^*=\sigma}q^{-n}}{\pr\brk{\G_{\eps,n}^*=G}}=\frac{\psi_G(\sigma)}{q^n\Erw[\psi_{\G_{\eps,n}}(\sigma)]\pr\brk{\G_{\eps,n}^*=G}}\\
&=\frac{\psi_G(\sigma)}{q^n\Erw[\psi_{\G_{\eps,n}}(\sigma)]\sum_{\tau\in\Omega^{V_n}}\psi_G(\tau)/\Erw[Z_{\G_{\eps,n}}]} =\frac{\psi_G(\sigma)}{Z_G}\cdot\frac{\Erw[Z_{\G_{\eps,n}}]}{q^n\Erw[\psi_{\G_{\eps,n}}(\sigma)]}.
\label{eqLemma_sigma1}
\end{align}
There are two cases to consider.
First, if $|\sigma^{-1}(\omega)|=n/q+O(\sqrt n)$, then {\bf BAL} ensures that $q^n\Erw[\psi_{\G_{\eps,n}}(\sigma)]=\Theta(\Erw[Z_{\G_{\eps,n}}])$.
Hence, \eqref{eqLemma_sigma1} shows that for such $\sigma$,
\begin{align} \label{eqLemma_sigma2}	
\pr\brk{\vsigma^*=\sigma\mid\G_{\eps,n}^*=G}=\Theta(\mu_G(\sigma)).
\end{align}
The second case is that $|\sigma^{-1}(\omega)|-n/q\gg\sqrt n$ for some $\omega\in\Omega$.
Then \Prop~\ref{Lemma_genNishi} shows that $$\pr\brk{\vsigma^*=\sigma},\Erw[\mu_{\G_{\eps,n}^*}(\sigma)]=o(1).$$
Thus, we may confine ourselves to the former case and assume that \eqref{eqLemma_sigma2} holds.
In light of \Lem~\ref{Lemma_tpinning} and \Prop~\ref{Lemma_genNishi} we may assume that $\mu_{\G_{\eps,n}^*}$ is $\delta$-symmetric for a small $\delta>0$ (at the expense of increasing $\theta$).
Hence, \eqref{eqLemma_sigma2} implies together with~\cite[\Lem~3.17]{SSBL} that \eqref{eqLemma_sigma} is satisfied.
\end{proof}

\begin{proof}[Proof of \Prop~\ref{Prop_ASS}]
This is an immediate consequence of \Prop~\ref{Prop_ASS*} and \Lem~\ref{Lemma_sigma}.
\end{proof}

\section{Proof of \Prop~\ref{Prop_int}}\label{Sec_Prop_int}
\noindent Throughout this section we assume that {\bf BAL}, {\bf SYM} and {\bf POS} hold.

\subsection{Preliminaries and setup}
The proof of \Prop~\ref{Prop_int} relies on showing that for \emph{any} distribution $\pi\in\PP_*(\Omega)$,
\begin{align}\label{eqInterpolation}
\frac1n\Erw[\ln Z(\hat\G)]&\geq \cB(\pi)
\end{align}
We will show \eqref{eqInterpolation} via the interpolation method.
To be precise, for a given $\pi\in\PP_*(\Omega)$ we will construct a family of random factor graph models parametrised by $t\in[0,1]$.
The proof of \Prop~\ref{Prop_int} is based on two pillars. 
First, it will be easy to see that the free energy of the $t=0$ model is $n\cB(\pi)+o(n)$ and that the $t=1$ model is identical to $\hat\G$.
Second, we will show that the derivative of $\Erw[\ln Z(\hat\G)]/n$ with respect to $t$ is non-negative.
(\ref{eqInterpolation}) readily follows.

The interpolating family is constructed from the generalised model described in \Sec~\ref{Sec_generalised}.
To this end, we introduce the model $\G_{t, \eps, \pi}$ which is constructed as follows. Let
\begin{align*}
    \vm_\eps(t) \sim \Po \bc{(1-\eps)t\d n/\k} \qquad \text{and} \qquad \vm'_\eps(t) \sim \Po \bc{(1-\eps)(1-t)\d n/\k}
\end{align*}
As before, each variable comes with a target degree $d_i \geq 0$ giving rise to a variable degree sequence $\sd$.
Similarly, each of the $\vm_\eps(t)$ factor nodes comes with target degree $\vec k_i \geq 2$, 
while each of the $\vm'_\eps(t)$ factor nodes comes with a target degree of $\vk'_i$, which are independent and distributed as $\vk$.
Let the total number of factor nodes be given by 
$$\vm = \vm_\eps(t) + \sum_{i=1}^{\vm'_\eps(t)} \vk'_i$$
and define the factor degree sequence as $$\sk = (k_i)_{i \in \vm_\eps(t)} \cup (1)_{i \in \vm'_\eps(t), j \in\vk'_i}.$$
Moreover,
let $(\PSI_{i,j}')_{i,j}$ be a sequence of independent random weight functions such that $\PSI_{i,j}'$ has distribution $\PSI_{k_i'}$.
Then with $(\vmu_{i,j,h})_{i,j,h\geq1}$ drawn independently from $\pi$ and $\vh_{i,j}\in[\vk_i']$ drawn independently and uniformly, we let
\begin{align*}
\psi_{b_{i,j}}: \sigma \in \Omega \mapsto \sum_{\tau \in \Omega^{k'_i}} \PSI_{i,j}' \vecone \cbc{\tau_{\vh_{i,j}} = \sigma} \prod_{h \neq j} \vmu_{i,j,h}(\tau_h)
\end{align*}
Finally, let $\G_{t,\eps, \pi}$ be the resulting random factor graph.
In addition, for an integer $T>0$ let $\G_{t, \eps, \pi,T}$ be the random factor graph obtain by adding $\vtheta$ random unary factors that each fix a random variable node to a uniformly random spin chosen from $\Omega$, with $\vtheta\in[T]$ drawn uniformly at random. If the number of factor nodes is not obvious from the context, we will write $\Gmm$ for completeness.
It is straightforward to check the following.

\begin{fact}\label{Fact_NoelasAssumptions}
The $\G_{t,\eps,\pi,T}$ model satisfies the assumptions of \Prop~\ref{prop_coupling_int}.
\end{fact}
Let
$$\Gamma_t=\frac{t\bar d}{\bar k\xi}\Erw\brk{(\vk-1)\Lambda\bc{\sum_{\tau\in\Omega^{\vk}}\PSI_{\vk}(\tau)\prod_{j=1}^{\vk}\MU_j^{(\pi)}(\tau_j)}}.$$
The following proposition, which we prove in \Sec~\ref{Sec_Lemma_interpolation}, shows that the free energy essentially increases with $t$, up to the correction term $\Gamma_t$.

\begin{proposition}\label{Lemma_interpolation}
For every $\eps>0$ there is $T>0$ such that for all large enough $n$ the following is true.
Let
 	$$\phi_T:t\in[0,1]\mapsto(\Erw[\ln Z(\hat\G_{t, \eps,\pi,T})]+\Gamma_t)/n.$$
Then $\phi'_T(t)>-\eps$ for all $t\in[0,1]$.
\end{proposition}

We complement this statement by computing the free energy at `times' $t=0$.

\begin{proposition}\label{Lemma_00}
We have
\begin{align*}
	\frac1n\Erw[\ln Z(\hat\G_{0,0,\pi,0})] &=
	\Erw \brk{\frac{\xi^{-\vd}}{\abs{\Omega}} \Lambda \bc{\sum_{\sigma \in \Omega} \prod_{i=1}^{\vd} \sum_{\tau \in \Omega^{\hat \vk_i}} \vecone \cbc{\tau_{\vec h_i} = \sigma} \PSI_{\hat \vk_i}(\tau) \prod_{j \neq \vec h_i} \MU_{ij} (\tau_j)}}.
\end{align*}
\end{proposition}

\noindent
The proof of \Prop~\ref{Lemma_00} can be found in \Sec~\ref{Sec_Lemma_00}.

\begin{proof}[Proof of \Prop~\ref{Prop_int}]
\Prop~\ref{Lemma_interpolation} implies that
\begin{align}\label{eqProp_int1}
\frac1n\Erw[\log Z(\hat\G_{1,0,\pi_0}]&=O(\eps)+\frac{1}{n}\Erw[\log Z(\hat\G_{1,\eps,\pi,T})]
\geq O(\eps)+\frac1n\Erw[\log Z(\hat\G_{0,\eps,\pi,T})]-\Gamma_1.
\end{align}
Further, \Prop~\ref{Lemma_00} implies together with the fact that all weight functions are strictly positive that
\begin{align}\nonumber
\frac1n\Erw[\log Z(\hat\G_{0,\eps,\pi,T})]&=\frac1n\Erw[\log Z(\hat\G_{0,0,\pi,0})]+O(\eps n)\\&=
	\Erw \brk{\frac{\xi^{-\vd}}{\abs{\Omega}} \Lambda \bc{\sum_{\sigma \in \Omega} \prod_{i=1}^{\vd} \sum_{\tau \in \Omega^{\hat \vk_i}} \vecone \cbc{\tau_{\vec h_i} = \sigma} \PSI_{\hat \vk_i}(\tau) \prod_{j \neq \vec h_i} \MU_{ij} (\tau_j)}}.
\label{eqProp_int2}
\end{align}
Combining \eqref{eqProp_int1} and \eqref{eqProp_int2} completes the proof.
\end{proof}

\subsection{Proof of \Prop~\ref{Lemma_interpolation}}\label{Sec_Lemma_interpolation}
As before let $\vsigma^*\in\Omega^{\{x_1,\ldots,x_n\}}$ be a uniformly random assignment.
Further, let $\fD'$ be the $\sigma$-algebra generated by $(\vd_i,\vk_i,\vk_i')_{i}$.
Let $\G'=\sGmm$ be the random factor graph drawn from the distribution
\begin{align*}
\pr\brk{\G'\in\cE\mid\fD',\vsigma^*}&=\frac{\Erw[\vecone\{\Gmm\in\cE\}\psi_{\Gmm}(\vsigma^*)\mid\fD',\vsigma^*]}{\Erw[\psi_{\Gmm}(\vsigma^*)\mid\fD',\vsigma^*]}.
\end{align*}
We define $\G''=\sGmpm$, $\G'''=\sGmmp$ analogously.
Moreover, let $\cC$ be the set of all variable clones $(x_i,h)$, $h\leq d_i$ that remain unmatched in $\G'$.
Let $(\vy_i)_{i\geq1}$ denote a sequence of independent uniform samples from $\cC$.
We identify the clone $\vy_i$ with its underlying variable node where convenient.
Finally, let $(\vmu_i)_{i\geq1}$ be independent samples from $\pi$.  
The key step towards the proof of \Prop~\ref{Lemma_interpolation} is the derivation of the following formula.

\begin{lemma}\label{Prop_deriv}
Let
	\begin{align*}
	\Xi_{t}&=\Erw \brk{\vec \psi_{\vk} (\SIGMA^{*}(\vy_1), \dots, \SIGMA^{*}(\vy_{\vk})) \log \langle \vec \psi_{\vk} (\SIGMA(y_1), \dots,  \SIGMA(y_{\vk})) \rangle_{\G'}} \\
	& \qquad -\Erw \brk{\sum_{i=1}^{\vk} \sum_{\tau \in \Omega^{\vk}} \vecone \cbc{\tau_j = \SIGMA^{*}(\vy_1)} \PSI_{\vk}(\tau) \prod_{j \neq i} \MU_j(\tau_j) \log \Bigg \langle \sum_{\sigma \in \Omega^{\vk}} \vecone \cbc{\sigma_i = \SIGMA(\vy_1)} \PSI(\sigma) \prod_{j \neq i} \MU_j(\sigma_j) \Bigg \rangle_{\G'}}\\
	&\qquad +\Erw\brk{(\vk-1)\Lambda\bc{\sum_{\tau\in\Omega^{\vk}}\PSI_{\vk}(\tau)\prod_{j=1}^{\vk}\MU_j^{(\pi)}(\tau_j)}}.
				\end{align*}
Then  uniformly for all $t\in(0,1)$ and all $T\geq0$,
	$$\frac{\partial}{\partial t}\phi_T(t)=o(1)+\frac \d{\k\xi}\Xi_{t}.$$
\end{lemma}

The steps to prove \Prop~\ref{Prop_deriv} are the following.
Let
\begin{align*}
	\Delta_t&=\Erw\brk{\ln Z(\sGmpm} -\Erw\brk{\ln Z(\sGmm)}=\Erw\brk{\log Z(\G'')}-\Erw[\log Z(\G')],\\
	\Delta_t'&= \Erw\brk{\ln Z(\sGmmp}- \Erw\brk{\ln Z(\sGmm)}=\Erw\brk{\log Z(\G''')}-\Erw[\log Z(\G')],\\
	\Delta_t''&=\Erw\brk{(\vk-1)\Lambda\bc{\sum_{\tau\in\Omega^{\vk}}\PSI_{\vk}(\tau)\prod_{j=1}^{\vk}\MU_j(\tau_j)}}.
\end{align*}
Because $t$ enters into the definition of the various factor graphs only through the Poisson variables $\vm_\eps(t),\vm_\eps'(t)$, the following claim follows directly from \cite[\Lem~4.2]{CKPZ}.

\begin{claim}[Lemma 4.2 of \cite{CKPZ}]\label{Lemma_PoissonDeriv}
We have $$\frac 1n \frac{\partial}{\partial t}\phi_T(t)=(1-\eps)\frac \d\k\bc{\Delta_t-\Delta_t'+\Delta_t''}.$$
\end{claim}

To calculate $\Delta_t, \Delta'_t$ we continue to denote by $\PSI_k^*$ a weight function distributed as $\PSI_k$, drawn independently of everything else.

\begin{claim}\label{Lemma_Deltat}
We have
\begin{align*}
    \Delta_t = o(1) +  \Erw \brk{ \vec \psi_{\vk} (\SIGMA^{*}(\vy_1), \dots, \SIGMA^{*}(\vy_{\vk})) \log \langle \vec \psi_{\vk} (\SIGMA(\vy_1), \dots,  \SIGMA(\vy_{\vk})) \rangle_{\G'}}/\xi.
\end{align*}
\end{claim}
\begin{proof}
Due to routine concentration arguments we may safely assume that
\begin{align*}
\sum_{i=1}^n\vd_i\geq\sum_{i=1}^{\vm_\eps(t)}\vk_i+\sum_{i=1}^{\vm_\eps'(t)}\vk_i'.
\end{align*}
\Prop~\ref{prop_coupling_int} provides a coupling of $\G',\G''$.
There are three possible scenarios.
\begin{description}
\item[\textbf{Case 1:} $\G'=\G''-a_{\vm_\eps(t)+1}$]
In this case, $ \G''$ can be obtained from $\G'$ by adding a single $\vk_{m_\eps(t)+1}$-ary factor node $\va=a_{\vm_\eps(t)+1}$.
Its weight function and the adjacent variable nodes are drawn from the distribution
    \begin{align} \label{eqPhilippCase0}
    \Pr \brk{\partial \va = (y_1, \dots, y_{\vk_{\vm_{\eps}(t)+1}}), \psi_{\va} = \psi \mid\fD',\vsigma^*}
=\frac{(1+o(1))\pr\brk{\PSI_{\vk_{\vm_\eps(t)+1}}=\psi}\psi(\vsigma^*(y_1),\ldots,\vsigma^*(y_{\vk_{\vm_\eps(t)+1}}))}
		{\Erw[\PSI_{\vk_{\vm_\eps(t)+1}}(\vsigma^*(\vy_1),\ldots,\vsigma^*(\vy_{\vk_{\vm_\eps}(t)+1}))]}.
    \end{align}
with $y_1,\ldots,y_{\vk_{\vm_\eps(t)+1}}\in\cC,\psi\in\Psi$;
the $1+o(1)$ term stems from the fact that the `cavities' where $\va$ attaches should be drawn without replacement.
Furthermore, since with probability $1-\exp(-\Omega(n))$ we have $\sum_{y\in C}\vecone\{\vsigma^*(y)=\tau\}=|\cC|/q+o(n)$ for all $\tau\in\Omega$, the expression \eqref{eqPhilippCase0} simplifies to
\begin{align}
\Pr \brk{\partial \va = (y_1, \dots, y_{\vk_{\vm_{\eps}(t)+1}}), \psi_{\va} = \psi \mid\fD',\vsigma^*}=\frac{1+o(1)}{\xi}\pr\brk{\PSI_{\vk_{\vm_\eps(t)+1}}=\psi}\psi(\vsigma^*(y_1),\ldots,\vsigma^*(y_{\vk_{\vm_\eps(t)+1}})).
\end{align}
Furthermore, the ensuing change in free energy upon adding $\va$ works out to be
\begin{align}\label{eqPhilippCase1}
\log\frac{Z(\G'')}{Z(\G')}&=\log\bck{\psi_{\va}(\vsigma)}_{\G'}.
\end{align}
\item[\textbf{Case 2:} $|\G'\triangle \G''|=O(\sqrt n\log n)$] 
because all weight functions are strictly positive, in this case we obtain
\begin{align}\label{eqPhilippCase2}
\abs{\log Z(\G'')-\log Z(\G')}=O(\sqrt n\log n).
\end{align}
\item[\textbf{Case 3:} cases 1,2 do not occur.]
In this case we have the trivial bound 
\begin{align} \label{eqPhilippCase3}
\log Z(\G'')/Z(\G')=O(n+\vm_\eps).
\end{align}
\end{description}
\Prop~\ref{prop_coupling_int} shows that Case~1 occurs with probability $1-O(1/n)$ and that Case~3 occurs with probability $O(1/n^2)$.
Therefore, \eqref{eqPhilippCase0}--\eqref{eqPhilippCase3} yield
\begin{align*}
\Erw\brk{\log\frac{Z(\G'')}{Z(\G')} \mid\fD',\vsigma^*}&=\frac{1+o(1)}{\xi}
\Erw \brk{ \vec \psi_{\vk} (\SIGMA^{*}(\vy_1), \dots, \SIGMA^{*}(\vy_{\vk})) \log \langle \vec \psi_{\vk} (\SIGMA(\vy_1), \dots,  \SIGMA(\vy_{\vk})) \rangle_{\G'}},
\end{align*}
as claimed.
\end{proof}

\begin{claim}\label{Lemma_Deltat'}
We have
\begin{align*}
	\Delta_t' = o(1) + \Erw \brk{\sum_{i=1}^{\vk} \sum_{\tau \in \Omega^{\vk}} \vecone \cbc{\tau_j = \SIGMA^{*}} \PSI_{\vk}(\tau) \prod_{j \neq i} \MU_j(\tau_j) \log \Bigg \langle \sum_{\sigma \in \Omega^{\vk}} \vecone \cbc{\sigma_i = \SIGMA(x)} \PSI(\sigma) \prod_{j \neq i} \MU_j(\sigma_j) \Bigg \rangle_{\G'}}/\xi.
\end{align*}
\end{claim}
\begin{proof}
We apply \Prop~\ref{prop_coupling_int} as in the proof of the previous proposition to obtain a coupling of $\G'$, $\G'''$.
As in that proof, because all weight functions are strictly positive we just need to consider the case that $\G'$ coincides with the factor graph obtained from $\G'''$ by removing $b_{\vm_{\eps}'(t)+1,1},\ldots,b_{\vm_{\eps}'(t)+1,\vk'_{\vm_{\eps}'(t)+1}}$.
Hence, we may assume that $\G'''$ is obtained from $\G'$ by adding unary factor nodes $\vb_1,\ldots,\vb_{\vk'_{\vm_\eps'(t)+1}}$
defined as follows.
Let $(\vmu''_{i,j})_{i,j\geq1}$ be independent samples from $\pi$ and let $(\vh_i)_{i\geq1}$ be independent and uniform samples from $[\vk'_{\vm_\eps'(t)+1}]$.
To simplify matters, we are going to discretise the continuous distribution on distributions $\pi$.
Then 
\begin{align} \label{eq_Delta2_dist}
    \Pr &\bc{\partial\vb_j=y,\,\psi_{\vb_j}(\nix)=
\sum_{\tau\in\Omega^{\vk'_{\vm_\eps'(t)+1}}}\vecone\{\tau_{\vh_j}=\nix\}\psi(\tau)\prod_{h\neq\vh_j}\mu_{i,j}(\tau_j) \mid\fD',\vsigma^*}\\ 
&= \frac{1+o(1)}{\xi\vk'_{\vm_\eps'(t)+1}}\sum_{\tau \in \Omega^{\vk'_{\vm'_\eps(t)+1}}}
		\vecone \cbc{\tau_i=\SIGMA^{*}(y)} \psi(\tau)\prod_{h \neq i} \mu_{i,j}(\tau_h)\pi(\mu_{i,j}).\nonumber
\end{align}
Let $\SIGMA$ be a sample from $\mu_{\G'}$. Since the factor nodes factorize up to a vanishing error term that is due to some variable nodes having two or more cavities, we have
\begin{align} \label{eq_Delta2_add}
    \log \frac{Z(\G''')}{Z(\G')} = \sum_{j=1}^{\vk'_{\vm_\eps'(t)+1}}\log\bck{\psi_{\vb_j}(\vsigma)}_{\G'} + o(1).
\end{align}
Combining \eqref{eq_Delta2_dist} and \eqref{eq_Delta2_add}, we finally obtain
\begin{align*}
    \Erw \brk{\log \frac{Z(\G''')}{Z(\G')} \mid\fD',\vsigma^*} &= \Erw \Bigg[\sum_{i=1}^{\vk} \sum_{\tau \in \Omega^{\vk}} \vecone \cbc{\tau_{\vh_i} = \SIGMA^{*}(\vy_i)} \PSI_{\vk'_{\vm_\eps'(t)+1}}(\tau) \prod_{j \neq \vh_i} \MU_{i,j}(\tau_j)\\
&\qquad\log \Bigg \langle \sum_{\sigma \in \Omega^{\vk'_{\vm_\eps'(t)+1}}} \psi_{\vk'_{\vm_\eps'(t)+1}}(\sigma)
		 \cbc{\sigma_{\vh_i} = \SIGMA(x)} \PSI(\sigma) \prod_{j \neq \vh_i} \MU_{i,j}(\sigma_j) \Bigg \rangle_{\G'}\Bigg]\bigg/(\xi+o(1)).
\end{align*}
The claim follows.
\end{proof}

\begin{claim}\label{Lemma_Deltat''}
With $\vec\mu_1,\vec\mu_2$ chosen independently from $\pi$ we have

\begin{align*}
\Delta_t''&=\frac{\k\xi}{\d}\frac{\partial}{\partial t}\Gamma_t=
\Erw\brk{(\vk-1)\Lambda\bc{\sum_{\tau\in\Omega^{\vk}}\PSI_{\vk}(\tau)\prod_{j=1}^{\vk}\MU_j^{(\pi)}(\tau_j)}}
\end{align*}
\end{claim}
\begin{proof}
This follows immediately by plugging in the definition of $\Gamma_t$.
\end{proof}

\begin{proof}[Proof of \Lem~\ref{Prop_deriv}]
This lemma follows from Claims~\ref{Lemma_Deltat}, \ref{Lemma_Deltat'} and \ref{Lemma_Deltat''}.
\end{proof}

\begin{proof}[Proof of \Prop~\ref{Lemma_interpolation}]
Let $\rho_{\hat \G_{T,t,\eps}}$ be the empirical distribution of the marginals of $\mu_{\hGi,x}$ defined over the set of cavities, i.e.
\begin{align}
\hat\rho = \frac{1}{\abs{\cC}} \sum_{x \in \cC} \delta_{\mu_{\hGi},x} \quad \in \PP_*(\Omega)
\end{align}
\Lem~\ref{Lemma_tpinning} shows that choosing $T$ sufficiently large, we can ensure that $\mu_{\hGi,x}$ is $\delta$-symmetric for an arbitrarily small $\delta>0$.
Therefore, the Nishimori identity and \Lem~\ref{Prop_deriv} imply that
\begin{align*}
\frac{\partial}{\partial t}\phi_T(t)&=o(1)+\frac{\bar d}{\bar k\xi}\Xi_t=O(\delta)+\frac{\bar d}{\bar k\xi}\Xi_t'\qquad\mbox{where}\\
\Xi_{t}'&= \Erw \brk{\Lambda \bc{\sum_{\tau \in \Omega^{\vk}} \PSI_{\vk}(\tau) \prod_{i=1}^{\vk} \RHO_i(\tau_i)} + (\vk-1) \Lambda \bc{\sum_{\tau \in \Omega^{\vk}} \PSI_{\vk}(\tau) \prod_{i=1}^{\vk} \MU_i(\tau_i)} - \vk \Lambda \bc{\sum_{\tau \in \Omega^{\vk}} \PSI(\tau) \RHO_1(\tau_1) \prod_{i=2}^{\vk} \MU_i(\tau_i)}}.
\end{align*}
Hence, the assertion follows from assumption {\bf POS}.
\end{proof}

\subsection{Proof of \Prop~\ref{Lemma_00}}\label{Sec_Lemma_00}

Because the random graph model is symmetric under permutations of the variable nodes, we can view $\frac 1n \Erw \brk{\log Z(\hat \G_{0,0,\eps})}$ as the contribution to $\Erw \brk{\log Z(\hat \G_{0,0,\eps})}$ of the connected component of $x_1$. The partition function of the component of $x_1$ is nothing but
\begin{align*}
    z = \sum_{\sigma \in \Omega} \prod_{j=1}^{\vd_{x_1}} \psi_{b_{x_1},j} (\sigma)
\end{align*}
By construction at $t=0$, the degree $\vd$ is chosen from $\cD$. On the factor side, the variable is assigned to factor nodes by choosing uniformly at random without replacement among the emanating half-edges of the factor nodes. Moreover, changing the total number of half-edges by a bounded number only changes the probability of selecting factor nodes with specific arities by $O(1/n)$. Thus, the arity of the chosen factor nodes is distributed according to \eqref{eqhatk}.
Hence, we find
\begin{align*}
    \frac 1n \Erw \brk{ \log Z(\hat \G_{0,0,\eps}} = \Erw [z] = \Erw \brk{\frac{\xi^{-\vd}}{\abs{\Omega}} \Lambda \bc{\sum_{\sigma \in \Omega} \prod_{i=1}^{\vd} \sum_{\tau \in \Omega^{\hat \vk_i}} \vecone \cbc{\tau_{\vec h_i} = \sigma} \PSI_{\hat \vk_i}(\tau) \prod_{j \neq \vec h_i} \MU_{ij} (\tau_j)}}
\end{align*}

The overall proposition is immediate from \Prop s~\ref{Lemma_interpolation} and \ref{Lemma_00}.

\section{Applications}\label{Sec_applications}

\subsection{LDGM codes}

We start to show how to apply \Thm~\ref{Thm_main} to derive the statement in \Thm~\ref{Thm_ldgm}. To this end, let $\Omega=\cbc{\pm 1}, \PSI_k = \cbc{\psi_{k,1}, \psi_{k,-1}}$ for all $k \geq 3$ with
\begin{align*}
    \psi_{k,J}(\sigma) = 1+(1-2\eta)J \prod_{i=1}^k \sigma_i
\end{align*}
for all $\sigma \in \Omega^k, J \in \cbc{\pm 1}$. $P_k$ is simply the uniform distribution, i.e. $P_k(\psi_{k,J})=1/2$ for $J \in \cbc{-1,+1}$. Moreover, the distribution on $\Psi_k$ conditioned on the planted configuration for a factor node $a$ with degree $k$ for all $k\in \supp \cK$ is given by
\begin{align*}
    \Pr \brk{\psi_{a} = \psi_{k,J} | \SIGMA_{\partial a} = (\sigma_1 \dots \sigma_k}) = \bc{1+(1-2\eta) J \prod_{i=1}^k \sigma_i}/2
\end{align*}
which yields $1-\eta$ if $\prod_{i=1}^k \sigma_i = 1$ and $\eta$ if $\prod_{i=1}^k \sigma_i = -1$. Furthermore, we have $\xi = \Erw \brk{\abs{\Omega}^{-\vk} \sum_{\tau \in \Omega^{\vk}} \PSI_{\vk}(\tau)} = 1$. Moreover, we find
\begin{align} \label{eq_ldgm_info}
    \Erw \brk{\frac{1}{\abs{\Omega}^{\vk}} \sum_{\tau \in \Omega^{\vk}} \Lambda(\PSI_{\vk}(\tau))} = \brk{\log 2 - H(\eta)}.
\end{align}
Next, we check $\SYM, \BAL, \POS$. $\SYM$ and $\BAL$ are immediate since the function $\sigma \mapsto \Erw \brk{\PSI_{\vk}(\sigma)}$ is constant. For \POS, we employ an argument from \cite[Section 4.4]{CKPZ}. Expanding $\Lambda(\cdot)$ and using Fubini's theorem we obtain
\begin{align*}
    \Erw \brk{\Lambda\bc{\sum_{\tau \in \Omega^{\vk}} \PSI_{\vk}(\tau) \prod_{i=1}^{\vk} \RHO_i(\tau_i)}} 
    &= -1 + \sum_{\ell=2}^{\infty} \frac{\Erw \brk{\bc{1-\sum_{\sigma \in \Omega^{\vk}} \PSI_{\vk}(\sigma) \prod_{i=1}^{\vk} \RHO_i(\sigma_i)}^\ell}}{\ell(\ell-1)} \\
    &= -1 + \sum_{\ell=2}^{\infty} \frac{\Erw \brk{((1-2\eta)\vJ)^\ell} \Erw \brk{(\RHO_1(1) - \RHO_1(-1))^{\ell\vk}}}{\ell(\ell-1)}
\end{align*}
Applying the same procedure to the other two terms of \POS\ and letting $X_\ell = \Erw \brk{(\RHO_1(1) - \RHO_1(-1))^\ell}$ and $Y_\ell = \Erw \brk{(\RHO'_1(1) - \RHO'_1(-1))^\ell}$, we merely need to show that
\begin{align} \label{eq_ldgm_POS}
    \sum_{\ell=2}^{\infty} \frac{1}{\ell(\ell-1)} \Erw \brk{((1-2\eta)\vJ)^\ell} \Erw \brk{X_\ell^{\vk} - \vk X_\ell Y_\ell^{\vk-1} + (\vk-1) Y_\ell^{\vk}} \geq 0
\end{align}
Indeed, if $\ell$ is odd, then $\Erw \brk{((1-2\eta)\vJ)^\ell}=0$ due to the symmetry of $\vJ$. Moreover, for even $\ell$, both $X,Y\geq0$. Thus, \eqref{eq_ldgm_POS} follows from the fact that $X^k-kXY^{k-1}+(k-1)Y^k \geq 0$ for all $X,Y \geq 0$ and $\Erw \brk{((1-2\eta)\vJ)^\ell} \geq 0$ since $\ell$ is even. \Thm~\ref{Thm_main} together with \eqref{eq_ldgm_info} yield
\begin{align*}
    \lim_{n \to \infty} \frac 1n I(\SIGMA^*, \G^*) = (1+\d/\k) \log 2 - H(\eta) - \cB(\eta).
\end{align*}
Finally, we simplify the Bethe function $\cB(\eta)$. To this end, we can map a distribution $\MU^{(\pi')}$ drawn from $\pi' \in \cP_*(\cbc{\pm 1})$ to a distribution $\THETA^{(\rho)}$ drawn from $\rho \in \cP_0([-1,1])$ by
\begin{align*}
    \THETA^{(\rho)} = 2 \MU^{(\pi')}(1) - 1.
\end{align*}
Thus, we can simplify the Bethe functional to
\begin{align*}
    \cB(\pi') &= \Erw \brk{\frac{\xi^{-\vd}}{\abs{\Omega}} \Lambda \bc{\sum_{\sigma \in \Omega} \prod_{i=1}^{\vd} \sum_{\tau \in \Omega^{\hat \vk_i}} \vecone \cbc{\tau_{\vec h_i} = \sigma} \PSI_{\hat \vk_i}(\tau) \prod_{j \neq \vec h_i} \MU_{ij}^{(\pi)} (\tau_j)} - \frac{\d(\vk-1)}{\xi \k} \Lambda \bc{\sum_{\tau \in \Omega^{\vk}} \PSI_{\vk}(\tau) \prod_{j=1}^{\vk} \MU_j^{(\pi)} (\tau_j)}} \\
    &= \Erw \brk{\frac 12 \Lambda \bc{\sum_{\tau \in \cbc{\pm 1}} \prod_{i=1}^{\vd} \bc{ 1+\sum_{\tau \in \cbc{\pm 1}}^{\hat \vk-1} (1-2\eta) \vJ \sigma \prod_{j=1}^{\hat \vk-1} \tau_j \MU^{(\pi')}_{ij} (\tau_j)}} - \frac{\d(\vk-1)}{\k} \Lambda\bc{1+\sum_{\tau \in \cbc{\pm 1}}^{\hat \vk}(1-2\eta) \vJ \prod_{i=1}^{\vk} \tau_j \MU_j^{(\pi')}(\tau_j)}} \\
    &= \Erw \brk{\frac 12 \Lambda \bc{\sum_{\sigma \in \cbc{\pm 1}} \prod_{i=1}^{\vd} \bc{1+(1-2\eta) \sigma \vJ_b \prod_{j=1}^{\hat \vk-1} \THETA^{(\rho)}_{ij}}} - \frac{\d(\vk-1)}{\k} \Lambda \bc{1+(1-2\eta) \vJ \prod_{j=1}^{\vk} \THETA_{1,j}^{(\rho)}}} = \cB_{\mathrm{ldgm}}(\rho,\eta)
\end{align*}
concluding the proof.

\subsection{Stochastic Block Model}

First we need to check that the SBM indeed satisfies the assumptions \SYM, \BAL, \POS, which follows directly from \cite{CKPZ}.

\begin{lemma} \label{lem_sbm_conditions}
The Stochastic Block Model satisfies the assumptions \SYM, \BAL\ and \POS\ for all $q \geq 2, \beta \geq 0$.
\end{lemma}

\begin{proof}
The lemma is an immediate consequence of Lemma 4.3  and 4.5 in \cite{CKPZ} which carry over the Stochastic Block Model defined in \Sec~\ref{Sec_sbm}.
\end{proof}

Now, the proof of \Thm~\ref{Thm_sbm} is reduced to an application of \Thm~\ref{main_kl_divergence} to the stochastic block model.

\subsection{The Potts antiferromagnet on random regular graphs}

Let $\G(n,d)$ denote a random regular graph with $n$ vertices, each with degree $d$.

\begin{theorem} \label{thm_potts}
Let $k=2, d \in \NN \geq 2$ and $m=dn/k$. For $q\geq 2$, and $c \in [0,1]$, let
\begin{align*}
    \cB_{\text{Potts}} (d, q, c) &= \sup_{\pi \in \PP_*([q])} \Erw \brk{ \frac{\Lambda \bc{\sum_{\sigma=1}^q \prod_{i=1}^{d} 1 - c \MU_i^{(\pi)}(\sigma)}}{q(1-c/q)^{d}} - \frac{d \Lambda(1-\sum_{\tau=1}^q c \MU_1^{(\pi)}(\tau) \MU_2^{(\pi)}(\tau)}{2(1-c/q)}}, \\
    \beta_{q, \text{cond}(d)} &= \inf \cbc{\beta >0: \cB_{\text{Potts}}(d, q, 1-e^{-\beta}) > \log q + d \log(1-(1-e^{-\beta})/q)/2}.
\end{align*}
Then we have
\begin{align*}
    &\lim_{n \to \infty} - \frac 1n \Erw \brk{\log Z_{\beta}(\G(n,d))} = -\log q - d \log(1-(1-e^{-\beta})/q)/2 \qquad \qquad \beta < \beta_{q,\text{cond}}(d) \\
    &\lim_{n \to \infty} - \frac 1n \Erw \brk{\log Z_{\beta}(\G(n,d))} < -\log q - d \log(1-(1-e^{-\beta})/q)/2 \qquad \qquad \beta > \beta_{q,\text{cond}}(d)
\end{align*}
\end{theorem}

The key observation towards the proof of \Thm~\ref{thm_potts} is that the Stochastic Block Model is just the planted version of the Potts antiferromagnet.
Indeed, we find 
\begin{align*}
   \Pr [ \G^*_{\text{SBM}} = G \mid \SIGMA^* ] \propto \Pr \brk{ \G = G } \exp \bc{ - \beta \sum_{(v,w) \in E(G)} \vecone \cbc{\SIGMA^*(v) = \SIGMA^*(w)}} \propto \Pr [ \G_{\text{Potts}}(\SIGMA^*) = G]. 
\end{align*}

\begin{proof}[Proof of \Thm~\ref{thm_potts}]
The theorem is an immediate consequence of \Thm~\ref{Thm_sbm} and \Lem~\ref{lem_free_energy_null_star}.
\end{proof}

\subsection{Diluted mixed $k$-spin models}

The proof of \Thm~\ref{thm_mixed_k} is based on \Thm~\ref{Thm_main}. Clearly, the mixed $\vk$-spin model fits the definition of the generalized model underlying \Thm~\ref{Thm_main}. While we have a degree sequence on the factor side, each factor chooses variable nodes uniformly at random without replacement. Thus, up to a smaller-order error that adds $o(1)$ to the free energy, the number of neighbours for a variable node is a Poisson random variable. Of course, one problem is that the number of possible weight functions is infinite. We will tackle this issue in the proof of \Thm~\ref{thm_mixed_k} by introducing a discretised version of $\vJ$ that is cut off at the tails.
Let $p_{k, \vJ, \beta}$ be the law of $\psi_{k,\vJ,\beta}$. Then, fix some $r \in \mathbb{N}$ and define a discretised version of $\vJ$
\begin{align*}
\vJ^{(r)}&:= \sum_{i=0}^{r^2-1} \vecone\cbc{\vJ \in [-r+i/r, -r+(i+1)/r]}\bc{-r + \frac{i}{r}} + \sum_{i=r^2}^{2r^2-1} \vecone\cbc{\vJ \in [-r+i/r, -r+(i+1)/r]}\bc{-r + \frac{i+1}{r}} \\
 & - \vecone\cbc{\vJ < -r}r + \vecone\cbc{\vJ > r}r.
\end{align*}
Note that $r$ in $\vJ^{(r)}$ governs both the value range of the random variable and size of each discretised interval where for $\vJ<0$ the $\vJ^{(r)}$ takes the value of the left interval bound, while for $\vJ>0$ it is the right bound.
By construction, $\vec J^{(r)}$ is symmetric and bounded. Let $p_{k, \vJ, \beta}^{(r)}$ be the law of $\psi_{k,\vJ^{(r)},\beta}$.

\begin{lemma} \label{lem_mixed_conditions}
For all $r \in \mathbb{N}$, $k \geq 2, \d>0, \beta>0$, $p_{k, \vJ, \beta}^{(r)}$ satisfies conditions \SYM, \BAL~ and \POS~.
\end{lemma}

\begin{proof}
Condition \SYM~ is satisfied with $\eps = 1-\tanh(\beta r) > 0$ and $\xi = 1$.
In \BAL~, the function that we need to check for concavity is $\mu \mapsto 1 + \Erw\brk{\tanh(\beta \vJ^{(r)})} \Erw_\mu\brk{\vX} = 1$, as $\vJ^{(r)}$ is distributed as $-\vJ^{(r)}$. Hence, \BAL~ follows.
Finally, for \POS~, we use the expansion $\Lambda(1-x) = - x + \sum_{\ell \geq 2} x^\ell/(\ell(\ell-1))$ and observe that for $j \geq 2$,
\begin{align*}
\bc{1-\sum_{\tau \in \cbc{\pm 1}^k}\psi_{k,\vJ^{(r)},\beta}(\tau)\prod_{i=1}^{k}\vec \mu_{i,\rho}(\tau_i)}^j = \bc{\tanh(\beta \vJ^{(r)})}^j\prod_{i=1}^k\bc{\vec \mu_{i,\rho}(1) - \vec \mu_{i,\rho}(-1)}^j.
\end{align*}
Therefore, by the dominated convergence theorem,
\begin{align*}
\Erw\brk{\Lambda\bc{\sum_{\tau \in \cbc{\pm 1}^k}\psi_{k,\vJ^{(r)},\beta}(\tau)\prod_{i=1}^{k}\vec \mu_{i,\rho}(\tau_i)}} = \sum_{j \geq 2}\Erw\brk{\bc{\tanh(\beta \vJ^{(r)})}^j} \Erw\brk{\bc{\vec \mu_{1,\rho}(1) - \vec \mu_{1,\rho}(-1)}^j}^k/(j(j-1)).
\end{align*}
We apply the same idea to the other two terms from \POS~ and setting $X_j = \Erw\brk{\bc{\vec \mu_{1,\rho}(1) - \vec \mu_{1,\rho}(-1)}^j}$ and $Y_j = \Erw\brk{\bc{\vec \mu_{1,\rho'}(1) - \vec \mu_{1,\rho'}(-1)}^j}$, we arrive at the condition
\begin{align*}
\sum_{j \geq 2}\Erw\brk{\bc{\tanh(\beta \vJ^{(r)})}^j} \bc{X_j^k + (k-1)Y_j^k - kX_jY_j^{k-1}}/(j(j-1)) \geq 0.
\end{align*}
Again, because $\vJ^{(r)}$ is symmetric, $\Erw\brk{\bc{\tanh(\beta \vJ^{(r)})}^j} = 0$ for odd $j$, while $\Erw\brk{\bc{\tanh(\beta \vJ^{(r)})}^j} \geq 0$ for even $j$. The claim follows from the fact that $X^k-kXY^{k-1}+(k-1)Y^k \geq 0$ for all $X,Y \geq 0$.
\end{proof}

\begin{lemma} \label{lem_mixedk_LRC1}
If long-range correlations are absent in $\G$, we have $\lim_{n \to \infty} \Erw \brk{\log Z(\G)}/n=\lim_{n \to \infty} \log \Erw \brk{Z(\G)}/n$
\end{lemma}

\begin{proof}
We readily find that
\begin{align} 
    & \frac{\partial}{\partial d} \frac{1}{n} \Erw \brk{ \log Z(\G)} = \Erw \brk{\log \bc{  1+\tanh(\beta \vJ) \Big \langle \prod_{i=1}^{\vk} \sigma_{\vy_i}  \Big \rangle_G}} \label{eq_kspin_LRC1} \\
    &\qquad \qquad \leq \log \bc{ \Erw \brk{  1+\tanh(\beta \vJ) \Big \langle \prod_{i=1}^{\vk} \sigma_{\vy_i}  \Big \rangle_G}} = \frac{\partial}{\partial d} \frac{1}{n} \log \Erw \brk{ Z(\G)} \nonumber
\end{align}
where the inequality follows by Jensen.
Assume that long-range correlations are absent in $\G$, hence by definition the spins are approximately pairwise independent and by \Lem~\ref{lem_pairwise_symmetry} $k$-wise independent. Therefore, the Jensen gap in \eqref{eq_kspin_LRC1} vanishes.
Finally,
\begin{align*}
    \frac{1}{n} \Erw \brk{ \log Z(\G)} = \int \frac{\partial}{\partial d} \frac{1}{n} \Erw \brk{ \log Z(\G)} \dd d = \int \frac{\partial}{\partial d} \frac{1}{n} \log \Erw \brk{ Z(\G)} \dd d = \frac{1}{n} \log \Erw \brk{Z(\G)}
\end{align*}
whence the lemma follows.
\end{proof}

\begin{claim}\label{claim_kspinsymm1}
If we find for almost all $i, j \in [n]$ that $\langle \sigma_i \sigma_j \rangle = o(1)$, then for all but $o(n)$ coordinates $i \in [n]$ we have $\mu_i(1) = 1/2 + o(1)$.
\end{claim}

\begin{proof}
We prove the claim by using limits, i.e., we associate a function $f_{\sigma}: [0,1] \to \cP(\cbc{-1, 1})$ with $\sigma \in \cbc{-1, 1}^n$ such that $$f_\sigma(x) = \sum_{i=0}^{n-1} \vecone \cbc{ x \in [i/n, (i+1)/n) } \delta_{\sigma_{i}}.$$ (Hence, $f_\sigma \in \cP((\cbc{-1, 1})$ is the atom on ${\sigma_i}$ that represents the assignment $\sigma$ when we \textit{shrink} the coordinates from $[n]$ to $[0,1]$). Coming with this embedding of $(\cbc{-1, 1}^n$ into the space of functions $f: [0,1] \to \cP(\cbc{-1, 1})$, there is an embedding of the corresponding probability measures $\mu \in \cP((\cbc{-1, 1}^n)$ into the space of functions $\hat \mu: [0,1]^2 \to \cP((\cbc{-1, 1})$ by taking well-defined limits. A detailed discussion and formal justification of the procedure is provided by \cite{CutDistanceMax}.  

Hence, we effectively need to prove the following.
Let $F:(s,x)\in[0,1]^2\to[-1,1]$ be a measurable function such that
\begin{align}\label{eqL2_1}
\int_0^1F(s,x)F(s,y)\dd s=0
\end{align}
for almost all $x,y\in[0,1]$.
Then $F(s,x)=0$ almost surely.
To prove this statement think of the integral as an inner product of the vectors $F(\nix,x),F(\nix,y)\in L^2([0,1])$.
Then \eqref{eqL2_1} shows that $(F(\nix,x))_x$ is an orthogonal family.
Since any orthonormal family of the separable Hilbert space $L^2([0,1])$ is countable, this implies that $\{F(\nix,x)/\|F(\nix,x)\|^2:F(\nix,x)\neq0\}$ is countable.
Therefore, unless $F(\nix,x)=0$ for almost all $x$, there exists $x$ with $F(\nix,x)\neq0$ such that the set $\{y\in[0,1]:F(\nix,y)\neq0,F(\nix,y)/\|F(\nix,y)\|_2=F(\nix,x)/\|F(\nix,x)\|_2\}$ has positive measure.
But this contradicts \eqref{eqL2_1}.
\end{proof}

The next lemma follows almost directly from Claim \ref{claim_kspinsymm1} as we find that almost all spins $\sigma_{\vy_1}$ and $\sigma_{\vy_2}$ need to be independent, hence, no long-range correlations are present.

\begin{lemma} \label{lem_mixedk_LRC2}
If we have $\lim_{n \to \infty} \Erw \brk{\log Z(\G)}/n=\lim_{n \to \infty} \log \Erw \brk{Z(\G)}/n$, long-range correlations are absent in $\G$.
\end{lemma}

\begin{proof}
By Jensen's inequality, we have
\begin{align*}
    \lim_{n \to \infty} \frac{\partial}{\partial d} \frac{1}{n} \Erw \brk{ \log Z(\G)} \leq \lim_{n \to \infty} \frac{\partial}{\partial d} \frac{1}{n} \log \Erw \brk{ Z(\G)}
\end{align*}
Since $\lim_{n \to \infty} \Erw \brk{\log Z(\G)}/n=\lim_{n \to \infty} \log \Erw \brk{Z(\G)}/n$ by assumption, we find
\begin{align} \label{eq_kspin_deriv}
    \lim_{n \to \infty} \frac{\partial}{\partial d} \frac{1}{n} \Erw \brk{ \log Z(\G)} = \lim_{n \to \infty} \frac{\partial}{\partial d} \frac{1}{n} \log \Erw \brk{Z(\G)}
\end{align}
Moreover, another application of Jensen's inequality yields
\begin{align}
    \lim_{n \to \infty} \frac{\partial}{\partial d} \frac{1}{n} \Erw \brk{ \log Z(\G)} &= \Erw \brk{\log \bc{  1+\tanh(\beta \vJ) \Big \langle \prod_{i=1}^{\vk} \sigma_{\vy_i}  \Big \rangle_G}} \nonumber \\
    &\leq \log \bc{ \Erw \brk{  1+\tanh(\beta \vJ) \Big \langle \prod_{i=1}^{\vk} \sigma_{\vy_i}  \Big \rangle_G}} = \lim_{n \to \infty} \frac{\partial}{\partial d} \frac{1}{n} \log \Erw \brk{ Z(\G)} \label{eq_kspin_LRC}
\end{align}
By \eqref{eq_kspin_deriv}, we need equality to hold in  \eqref{eq_kspin_LRC}. Since $\Pr \brk{\vk = 2} > \eps$ for some $\eps>0$, this equality needs to hold in particular for $k=2$. By Claim \ref{claim_kspinsymm1}, this implies the absence of long-range correlations in $\G$ closing the proof of the lemma.
\end{proof}

\begin{proof}[Proof of \Thm~\ref{thm_mixed_k}]
By \Lem~\ref{lem_mixed_conditions} and since $\vX$ converges to $\vJ$ in probability as $\eps \to 0$ after taking $n \to \infty$, \Thm~\ref{Thm_main} is applicable to the mixed $\vk$-spin model. Moreover, \Lem s~\ref{lem_mixedk_LRC2} and \ref{lem_mixedk_LRC2} evince that long-range correlations are absent in $\G$ if and only if
\begin{align*}
    \lim_{n \to \infty} \Erw \brk{\log Z(\G)}/n=\lim_{n \to \infty} \log \Erw \brk{Z(\G)}/n
\end{align*}
The theorem readily follows.
\end{proof}

\section{Condensation threshold} \label{klcondp}

In this section we discuss two (asymptotical) quantities considered as functions of the model parameters $q$, $(\vpsi_k)_k$, $\vk$ and $\vd$.
For this purpose let $\mcFEZ_k=\sum_{\mcFHEA\in\mcColSp^k}\Erw[\vpsi_k(\mcFHEA)]$ for $k\in\ZZ_{\ge 0}$.
The \emph{annealed free entropy density} $\phi_{\mathrm{a}}\in\RR$ is given by
\begin{align*}
\phi_{\mathrm{a}}=(1-\d)\ln(q)+\frac{\d}{\k}\mathbb E\left[\ln\left(\mcFEZ_{\vk}\right)\right]\textrm{.}
\end{align*}
Assuming $q$ to be fixed we consider the regimes
\begin{align*}
\mathcal R_{\mathrm{RS}}=\left((\vd,\vk,(\vpsi_k)_k):\mathcal B_{\sup}\le\phi_{\mathrm{a}}\right)\textrm{ and }
\mathcal R_{\mathrm{cond}}=\left\{(\vd,\vk,(\vpsi_k)_k):\mathcal B_{\sup}>\phi_{\mathrm{a}}\right\}\textrm{.}
\end{align*}
The next result is dedicated to the relative entropy of the teacher-student model with respect to the null model.
\begin{theorem}\label{main_kl_divergence}
Assume that {\bf DEG}, {\bf SYM}, {\bf BAL} and {\bf POS} hold. Then we have
\begin{align*}
\lim_{n\rightarrow\infty}\frac{1}{n}\KL{(\mcRVA^*_n,\mcRG^*_n(\mcRVA^*_n))}{(\mcRVAPost_{\mcRG_n},\mcRG_n)}&=0\textrm{, }(\vd,\vk,(\vpsi_k)_k)\in\mathcal R_{\mathrm{RS}}\textrm{,}\\
\lim_{n\rightarrow\infty}\frac{1}{n}\KL{(\mcRVA^*_n,\mcRG^*_n(\mcRVA^*_n))}{(\mcRVAPost_{\mcRG_n},\mcRG_n)}&>0\textrm{, }(\vd,\vk,(\vpsi_k)_k)\in\mathcal R_{\mathrm{cond}}\textrm{.}
\end{align*}
\end{theorem}
The last result establishes that the \emph{quenched free entropy density} and the annealed entropy density coincide exactly in the replica symmetric regime.
\begin{theorem}\label{main_qfed}
Assume that {\bf DEG}, {\bf SYM}, {\bf BAL} and {\bf POS} hold. Then we have
\begin{align*}
\lim_{n\rightarrow\infty}\frac{1}{n}\Erw\left[\ln\left(Z_{\mcRG_n}\right)\right]&=\phi_{\mathrm{a}}\textrm{, }(\vd,\vk,(\vpsi_k)_k)\in\mathcal R_{\mathrm{RS}}\textrm{,}\\
\lim_{n\rightarrow\infty}\frac{1}{n}\Erw\left[\ln\left(Z_{\mcRG_n}\right)\right]&<\phi_{\mathrm{a}}\textrm{, }(\vd,\vk,(\vpsi_k)_k)\in\mathcal R_{\mathrm{cond}}\textrm{.}
\end{align*}
\end{theorem}
In the following we tacitly assume that {\bf DEG}, {\bf SYM}, {\bf BAL} and {\bf POS} are satisfied.
\subsection{Preliminaries}\label{klcondp_prel}
We use the notation from Section \ref{conc}
and further let $\phi_{\mathrm{a}}=\Erw[\phi_{\mathrm{a},\vt_n}]$ denote the annealed free entropy density with $\phi_{\mathrm{a},\mcT}=\frac{1}{n}\ln(\bar Z_\mcT)$, $\bar Z_\mcT=\Erw[Z_{\mcRG_\mcT}]$, denoting the annealed free entropy density for given $\mcT\in\cT_n$, $n\in\cN$.
The first result is a corollary to Proposition \ref{mcp_first_moment}.
\begin{fact}\label{klcondp_prel_afed}
Uniformly in $\mcT\in\cT^\circ_n$ we have $\phi_{\mathrm{a},t}=\phi_{\mathrm{a},\infty}+o(1)$ and further $\phi_{\mathrm{a}}=\phi_{\mathrm{a},\infty}+o(1)$, where
\begin{align*}
\phi_{\mathrm{a},\infty}
=(1-\d)\ln(q)+\frac{\d}{\k}\Erw\left[\ln\left(\mcFEZ_{\vk}\right)\right]\textrm{.}
\end{align*}
\end{fact}
\begin{proof}
Recall from Proposition \ref{mcp_first_moment} that we have $\bar Z_\mcT=(1+o(1))r^*_\mcT q^n\prod_{i\in[m_\mcT]}\xi_{k_{\mcT,i}}$ uniformly in $\mcT\in\cT^\circ_n$ with $r^*_\mcT=\Theta(1)$ uniformly, which yields
\begin{align*}
\phi_{\mathrm{a},\mcT}=\ln(q)+\frac{m_\mcT}{n}\Erw\left[\ln\left(\xi_{\vk_\mcT}\right)\right]+o(1)
\end{align*}
uniformly. Now, notice that $\varepsilon\le\xi_k\le\varepsilon^{-1}$ for all $k$ in the support of $\vk$ using {\bf SYM}, so with these uniform bounds on the expectation and the uniform bounds imposed by $\cT^\circ_n$ we have
\begin{align*}
\phi_{\mathrm{a},\mcT}=\ln(q)+\frac{\d}{\k}\Erw\left[\ln\left(\xi_{\vk}\right)\right]+o(1)
=(1-\d)\ln(q)+\frac{\d}{\k}\Erw\left[\ln\left(\mcFEZ_{\vk}\right)\right]+o(1)\textrm{.}
\end{align*}
Finally, notice that $\phi_{\mathrm{a},\mcT}$, $\mcT\in\cT_n$, is sublinear in the number of factors using {\bf SYM}, so with Proposition \ref{dd_typ_prop} and the uniform convergence given $\mcT\in\cT^\circ_n$ we have $\phi_{\mathrm{a}}=\phi_{\mathrm{a},\mcT}+o(1)=\phi_{\mathrm{a},\infty}+o(1)$.
\end{proof}
The next fact relates the quantities $\hat{\phi}_\mcT$, $\phi_{\mathrm{a},\mcT}$ and $\bar\phi_\mcT$ through the distance of the models $\mcRG^*_\mcT(\mcRVAN_\mcT)$ and $\mcRG_\mcT$.
\begin{fact}\label{klcondp_prel_qfed_inequ}
For all $n\in\cN$ and $\mcT\in\cT_n$ we have
\begin{align*}
\hat\phi_\mcT
=\phi_{\mathrm{a},\mcT}+\frac{1}{n}\KL{\mcRG^*_\mcT(\mcRVAN)}{\mcRG_\mcT}
\ge\phi_{\mathrm{a},\mcT}-\frac{1}{n}\KL{\mcRG_\mcT}{\mcRG^*_\mcT(\mcRVAN)}
=\bar\phi_\mcT\textrm{.}
\end{align*}
\end{fact}
\begin{proof}
Notice that the Radon-Nikodym derivative derivative of $\mcRG^*_\mcT(\mcRVAN)$ with respect to $\mcRG_\mcT$ is $\mcG\mapsto\frac{Z_\mcG}{\bar Z_\mcT}$ which gives the equivalences, while the inequality is obvious due to the non-negativity of the relative entropy.
\end{proof}
\subsection{Proof of Theorem \ref{main_qfed}}\label{condp}
Before we continue, we observe joint concentration given $\mcT\in\cT^\circ_n$.
\begin{lemma}\label{condp_conc}
Jointly in $\acS=(\mcT,\mcVA,\mcFSHEA)\in\acSSpT_n$ we have
$\hat\phi=\phi^*+o(1)=\phi^*_\acS+o(1)=\phi^*_\mcT+o(1)=\hat\phi_\mcT+o(1)$.
\end{lemma}
\begin{proof}
The first two equalities are immediate from Proposition \ref{conc_tss_glob}.
Thanks to Proposition \ref{typ_ass_prop} we have $\phi^*_\mcT=\Erw[\phi^*_{\acRS^*_\mcT}\vecone\{\acRS^*_\mcT\in\acSSpT_n\}]+o(1)$ uniformly in $\mcT\in\cT^\circ_n$,
and further for $(\mcT,\mcVA,\mcFSHEA)\in\acSSpT_n$ with Proposition \ref{conc_tss_glob} and the triangle inequality we have $\phi^*_{\mcT,\mcVA',\mcFSHEA'}=\phi^*_{\mcT,\mcVA,\mcFSHEA}+o(1)$ for any $(\mcT,\mcVA',\mcFSHEA')\in\acSSpT_n$ uniformly.
This yields $\phi^*_\mcT=\phi^*_{\acS}+o(1)$ uniformly in $\acS\in\acSSpT_n$.

Now, for given $\varepsilon\in(0,1)$ use Proposition \ref{mcp_main_prop} to obtain $c$, $r\in\RR_{>0}$ then for any $\mcT\in\cT^\circ_n$ we have
\begin{align*}
\pr[\acRSN_\mcT\not\in\acSSpT_n]
\le\varepsilon+c\pr[\acRS^*_\mcT\not\in\acSSpT_n,\mcRVA^*\in\mathcal E_\mcT]=\varepsilon+o(1)
\end{align*}
uniformly in $\mcT\in\cT^\circ_n$ thanks to Proposition \ref{typ_ass_prop}, i.e.~$\pr[\acRSN_\mcT\not\in\acSSpT_n]=o(1)$ uniformly.
Analogously to the above we obtain $\hat\phi_\mcT=\phi^*_{\mcT,\mcVA,\mcFSHEA}+o(1)$ uniformly which completes the proof.
\end{proof}
First, we derive the following contiguity-like result for the replica symmetric phase.
\begin{lemma}\label{condp_nishnull}
If we have $\bar\phi=\phi_{\mathrm{a}}+o(1)$, then for all $\hat c$, $\hat c'\in\RR_{>0}$ there exist $c$, $c'\in\RR_{>0}$ such that for all $n\in\cN$, all $\mcT\in\cT^\circ_n$ and $\cE\subseteq\mcColSp^{\mcVLSp_n}\times\cG_\mcT$ with $\pr[(\mcRVAN_\mcT,\mcRG^*_\mcT(\mcRVAN_\mcT))\in\cE]\le\hat c'\exp(-\hat cn)$ we have $\pr[(\mcRVAPost_{\mcRG_\mcT},\mcRG_\mcT)\in\cE]\le c'\exp(-cn)$, where $\cG_\mcT$ denotes the support of $\mcRG_\mcT$.
\end{lemma}
\begin{proof}
Using Fact \ref{klcondp_prel_qfed_inequ} and Lemma \ref{condp_conc} we notice that $\bar\phi_\mcT=\phi_{\mathrm{a},\mcT}+o(1)$ jointly in $\mcT\in\cT^\circ_n$.
Now, fix $\hat c$, $\hat c'\in\RR_{>0}$, $n\in\cN$ and an event $\cE$ such that
\begin{align*}
\pr[(\mcRVAPost_{\mcRG^*_\mcT(\mcRVAN_\mcT)},\mcRG^*_\mcT(\mcRVAN_\mcT))\in\cE]\le\hat c'\exp(-\hat cn)\textrm{.}
\end{align*}
For any $\varepsilon\in\RR_{>0}$ and using Proposition \ref{conc_std_loc} we find constants $c$, $c'\in\RR_{>0}$ such that
\begin{align*}
\pr\left[\phi(\mcRG_\mcT)\le\bar\phi_\mcT-\frac{1}{2}\varepsilon\right]\le c'\exp\left(-\frac{c}{4}\varepsilon^2n\right)
\end{align*}
uniformly in $\mcT\in\cT^\circ_n$.
Due to the assumption we have $\bar\phi_\mcT\ge\phi_{\mathrm{a},\mcT}-\frac{1}{2}\varepsilon$ uniformly for all sufficiently large $n\in\cN$.
Combining these gives $\pr\left[\mcRG_\mcT\not\in\cG^\circ_\mcT\right]\le c'\exp\left(-\frac{c}{4}\varepsilon^2n\right)$, $\cG^\circ_\mcT=\{\mcG\in\cG_\mcT:\phi(\mcG)>\phi_{\mathrm{a},\mcT}-\varepsilon\}$, and further
\begin{align*}
\pr[(\mcRVAPost_{\mcRG_{\mcT}},\mcRG_{\mcT})\in\cE]
&\le c'\exp\left(-\frac{c\varepsilon^2}{4}n\right)+\pr\left[(\mcRVAPost_{\mcRG_{\mcT}},\mcRG_{\mcT})\in\cE,\mcRG_{\mcT}\in\cG^\circ_{\mcT}\right]\\
&=c'\exp\left(-\frac{c\varepsilon^2}{4}n\right)+\sum_{\mcVA}\Erw\left[\frac{\psi_{\mcRG_{\mcT}}(\mcVA)}{\exp(n\phi(\mcRG_{\mcT}))}\vecone\{(\mcVA,\mcRG_{\mcT})\in\cE,\mcRG_{\mcT}\in\cG_{\mcT}^{\circ}\}\right]\\
&<c'\exp\left(-\frac{c\varepsilon^2}{4}n\right)+\sum_{\mcVA}\Erw\left[\frac{\psi_{\mcRG_{\mcT}}(\mcVA)}{\exp(n\phi_{\mathrm{a},t}-\varepsilon n)}\vecone\{(\mcVA,\mcRG_{\mcT})\in\cE,\mcRG_{\mcT}\in\cG_{\mcT}^{\circ}\}\right]\\
&=c'\exp\left(-\frac{c\varepsilon^2}{4}n\right)+e^{\varepsilon n}\sum_{\mcVA}\Erw\left[\frac{\psi_{\mcRG_{\mcT}}(\mcVA)}{\bar Z_{\mcT}}\vecone\{(\mcVA,\mcRG_{\mcT})\in\cE,\mcRG_{\mcT}\in\cG_{\mcT}^{\circ}\}\right]\\
&=c'\exp\left(-\frac{c\varepsilon^2}{4}n\right)+e^{\varepsilon n}\Erw\left[\frac{Z_{\mcRG_{\mcT}}}{\bar Z_{\mcT}}\sum_{\mcVA}\mu_{\mcRG_{\mcT}}(\mcVA)\vecone\{(\mcVA,\mcRG_{\mcT})\in\cE,\mcRG_{\mcT}\in\cG_{\mcT}^{\circ}\}\right]\\
&=c'\exp\left(-\frac{c\varepsilon^2}{4}n\right)+e^{\varepsilon n}\pr\left[(\mcRVAPost_{\mcRG^*_{\mcT}(\mcRVAN_{\mcT})},\mcRG^*_{\mcT}(\mcRVAN_{\mcT}))\in\cE,\mcRG^*_{\mcT}(\mcRVAN_{\mcT})\in\cG_{\mcT}^{\circ}\right]\\
&\le c'\exp\left(-\frac{c\varepsilon^2}{4}n\right)+\hat c'\exp\left(\varepsilon n-\hat cn\right)\textrm{.}
\end{align*}
Let $\varepsilon$ be the solution for which the coefficients in the exponents coincide, i.e.~$c_1=\frac{1}{4}c\varepsilon^2=\hat c-\varepsilon\in\RR_{>0}$, then with $c_2=c'+\hat c'$ we have $\pr[(\mcRVAPost_{\mcRG_{\mcT}},\mcRG_{\mcT})\in\cE]<c_2\exp(-c_1n)$.
Recall that the result above holds for all $n\in\cN$ with $n>n_0$ for some suitable $n_0\in\cN$.
Now, redefine $c=c_1$ and let $c'\ge c_2$ be sufficiently large such that $c'\exp(-cn_0)\ge 1$, then the assertion is trivial for all small $n$ and also holds for large $n$.
\end{proof}
Next, we derive a concentration result for the Nishimori quenched free entropy density.
\begin{lemma}\label{condp_nishi_conc}
For all $\varepsilon\in\RR_{>0}$ there exist $c$, $c'\in\RR_{>0}$ such that for all $n\in\cN$ and $\mcT\in\cT^\circ_n$ we have
\begin{flalign*}
\pr\left[\left|\phi(\mcRGN_\mcT)-\hat\phi_\mcT\right|\ge\varepsilon\right]\le c'\exp(-cn)\textrm{.}
\end{flalign*}
\end{lemma}
\begin{proof}
Using Lemma \ref{condp_conc} we obtain uniform bounds for the distance of $\phi^*_\acS$ and $\hat\phi_\mcT$ over any choice of $\acS\in\acSSpT_n$ and $\mcT\in\cT^\circ_n$ for $n$ sufficiently large.
Further, for given $\varepsilon'\in\RR_{>0}$ with Lemma \ref{typ_ass_center} we obtain $\delta$ and exponential bounds for $\acADD(\acAD_{\acRSN_\mcT},\acAD^*)\ge\varepsilon'$ given that the distance of $\mcRCFRN_\mcT$ and $u_{\mcColSp}$ is less than $\delta$.
But Proposition \ref{mcp_tails} exactly provides the corresponding exponential bounds.
Combining these results leaves us with assignment distributions close to the reference distribution, for which the coupling in Section \ref{conc_tss_glob_proof} ensures that the corresponding quenched free entropy densities $\phi^*_\acS$, i.e.~$\acADD(\acAD_\acS,\acAD^*)<\varepsilon'$ and $\mcT\in\cT^\circ_n$ with $\acS=(\mcT,\mcVA,\mcFSHEA)\in\acSSpV_n$, are close to each other and the center.
Finally, Proposition \ref{conc_tss_loc} provides uniform exponential bounds for the distance of the free entropy density to its expectation given $\acS$, which concludes the proof for large $n$. However, choosing $c'\in\RR_{>0}$ sufficiently large ensures that the bound is valid for all $n$.
\end{proof}
\begin{lemma}\label{condp_rs_phase}
We have $\hat\phi=\phi_{\mathrm{a}}+o(1)$ if and only if $\bar\phi=\phi_{\mathrm{a}}+o(1)$.
\end{lemma}
\begin{proof}
With Fact \ref{klcondp_prel_afed} we have $\phi_{\mathrm{a}}=\phi_{\mathrm{a},\mcT}+o(1)$, with Lemma \ref{condp_conc} we have $\hat{\phi}=\hat{\phi}_{\mcT}+o(1)$ and with Proposition \ref{conc_std_glob} we have $\bar\phi=\bar\phi_{\mcT}+o(1)$, all uniformly in $\mcT\in\cT^\circ_n$.
Now, we show that $\hat\phi_{\mcT}=\phi_{\mathrm{a},\mcT}+o(1)$ if and only if $\bar\phi_{\mcT}=\phi_{\mathrm{a},\mcT}+o(1)$ for a fixed sequence $\mcT=\mcT_n\in\cT^\circ_n$, since then the assertion follows by the arguments above.
Let $c$, $c'\in\RR_{>0}$ and $\hat c$, $\hat c'\in\RR_{>0}$ be the constants obtained from Proposition \ref{conc_std_loc} and Proposition \ref{conc_tss_loc} respectively.

First, assume that $\hat\phi_{\mcT}=\phi_{\mathrm{a},\mcT}+o(1)$ holds.
Fix a sequence $\varepsilon_n\in\RR_{>0}$, $n\in\cN$, such that $\varepsilon_n=o(1)$, $\varepsilon_n^2n=\omega(1)$ and $|\hat\phi_{\mcT}-\phi_{\mathrm{a},\mcT}|<\frac{1}{3}\varepsilon_n$.
Use Lemma \ref{condp_conc} to obtain $|\phi^*_{\acS}-\hat{\phi}_{\mcT}|<\frac{1}{3}\varepsilon_n$ for all $\acS=(t_n,\mcVA,\mcFSHEA)\in\acSSpT_n$ and sufficiently large $n\in\cN$.
The probability for the event $\cE_n=\{\mcG\in\cG_\mcT:|\phi(\mcG)-\phi_{\mathrm{a},\mcT}|<\varepsilon_n\}$ can be bounded by
\begin{align*}
\pr[\mcRG^*_{\mcT}(\mcRVAN_{\mcT})\not\in\cE_n]
&\le\pr[\acRSN_\mcT\not\in\acSSpT_n]+\Erw\left[\vecone\{\acRSN_\mcT\in\acSSpT_n\}\pr\left[\mcRG^*_{\acRSN_\mcT}\not\in\cE_n\middle\vert\acRSN_\mcT\right]\right]\\
&\le o(1)+\Erw\left[\vecone\{\acRSN_\mcT\in\acSSpT_n\}\pr\left[\left|\phi(\mcRG^*_{\acRSN_\mcT})-\phi^*_{\acRSN_\mcT}\right|\ge\frac{1}{3}\varepsilon_n\middle\vert\acRSN_\mcT\right]\right]\\
&\le o(1)+\pr[\acRSN_\mcT\in\acSSpT_n]\hat c'\exp\left(-\frac{\hat c}{9}\varepsilon_n^2n\right)=o(1)\textrm{,}
\end{align*}
using $\pr[\acRSN_\mcT\in\acSSpT_n]=1+o(1)$ from the proof of Lemma \ref{condp_conc},
and that the exponent is of order $\omega(1)$ since $\varepsilon_n^2 n=\omega(1)$.
Next, we use $\bm Z^\circ_\mcT=Z_{\mcRG_\mcT}\vecone\{\mcRG_\mcT\in\cE_n\}$ to obtain $\bar Z^\circ_\mcT=\bar Z_\mcT\pr[\mcRG^*_\mcT(\mcRVAN_\mcT)\in\cE_n]\ge\frac{1}{2}\bar Z_\mcT$, where $\bar Z^\circ_\mcT=\Erw[\bm Z^\circ_\mcT]$ and $n$ sufficiently large such that $\pr[\mcRG^*_\mcT(\mcRVAN_\mcT)\in\cE_n]\ge\frac{1}{2}$.
On the other hand, using the definition of $\cE_n$ we have $\Erw[\bm Z^{\circ2}_\mcT]\le\exp(2n(\phi_{\mathrm{a},\mcT}+\varepsilon_n))\pr[\mcRG_\mcT\in\cE_n]\le\exp(2\varepsilon_nn)\bar Z_\mcT^2$, so the Paley-Zygmund inequality yields
\begin{align*}
\pr\left[\bm Z^\circ_\mcT\ge\frac{1}{2}\bar Z^\circ_\mcT\right]\ge\frac{\bar Z^{\circ 2}_\mcT}{4\Erw[\bm Z^{\circ 2}_\mcT]}
\ge\frac{1}{16}\exp\left(-2\varepsilon_nn\right)\textrm{.}
\end{align*}
Since by definition we always have $\bm Z^\circ_\mcT\le Z_{\mcRG_\mcT}$, the event $\bm Z^\circ_\mcT\ge\frac{1}{2}\bar Z^\circ_\mcT$ implies $Z_{\mcRG_\mcT}\ge\frac{1}{4}\bar Z_\mcT$ and hence
\begin{align*}
\pr\left[\mcRG_\mcT\in\cE'_n\right]=\pr\left[Z_{\mcRG_\mcT}\ge\frac{1}{4}\bar Z_\mcT\right]\ge\frac{1}{16}\exp\left(-2\varepsilon_nn\right)\textrm{, }
\cE'_n=\left\{\mcG\in\cG_\mcT:\phi(\mcG)\ge\phi_{\mathrm{a},\mcT}-\frac{\ln(4)}{n}\right\}\textrm{.}
\end{align*}
Fix a sequence $\delta_n\in\RR_{>0}$, $n\in\cN$, with $\delta_n=o(1)$ and $\delta_n^2=\omega(\varepsilon_n)$.
Now we can use Proposition \ref{conc_std_loc} with $\cE''_n=\{\mcG\in\cG_\mcT:|\phi(\mcG)-\bar\phi_\mcT|<\delta_n\}$ to obtain
\begin{align*}
\pr\left[\mcRG_\mcT\in\cE'_n\cap\cE''_n\right]
&\ge\frac{1}{16}\exp\left(-2\varepsilon_nn\right)-c'\exp(-c\delta_n^2n)\\
&=\left(\frac{1}{16}-c'\exp\left(-c\delta_n^2n\left(1-\frac{2\varepsilon_n}{c\delta_n^2}\right)\right)\right)\exp\left(-2\varepsilon_nn\right)
=(1+o(1))\frac{1}{16}\exp\left(-2\varepsilon_nn\right)\textrm{,}
\end{align*}
so in particular we have $\mcRG_\mcT\in\cE'_n\cap\cE''_n$ asymptotically with positive probability, and for all $\mcG\in\cE'_n\cap\cE''_n$ we have $|\phi_{\mathrm{a},\mcT}-\bar\phi_\mcT|\le|\phi_{\mathrm{a},\mcT}-\phi(\mcG)|+|\phi(\mcG)-\bar\phi_\mcT|\le n^{-1}\ln(4)+\delta_n=o(1)$.

Conversely, assume that $\hat\phi_\mcT=\phi_{\mathrm{a},\mcT}+\Omega(1)$, so there exists $\delta\in\RR_{>0}$ such that
$\hat\phi_\mcT\ge\phi_{\mathrm{a},\mcT}+\delta$ for $n$ sufficiently large using Fact \ref{klcondp_prel_qfed_inequ}.
Using Lemma \ref{condp_nishi_conc} yields that $\pr[|\phi(\mcRGN_\mcT)-\hat\phi_\mcT|\ge\delta/2]\le c'\exp(-cn)$, so $\pr[\phi(\mcRGN_\mcT)\le\phi_{\mathrm{a},\mcT}+\delta/2]\le c'\exp(-cn)$.
On the other hand, Fact \ref{klcondp_prel_qfed_inequ} shows that $\bar\phi_\mcT\le\phi_{\mathrm{a}}$, and further Proposition \ref{conc_std_loc} suggests that $\pr[|\phi(\mcRG_\mcT)-\bar\phi_\mcT|\ge\frac{1}{4}\delta]\le c'\exp(-cn)$, so $\pr[\phi(\mcRG_\mcT)\le\phi_{\mathrm{a},\mcT}+\delta/2]\ge 1-c'\exp(-cn)$. So, with Lemma \ref{condp_rs_phase}, contraposition and Fact \ref{klcondp_prel_qfed_inequ} we obtain $\bar\phi_\mcT=\phi_{\mathrm{a},\mcT}-\Omega(1)$.
\end{proof}
Recall that $\hat\phi=\sup_{\pi\in\cP^2_*(\mcColSp)}\cB(\pi)+o(1)$ using all assumptions.
Then with Fact \ref{klcondp_prel_afed} we obtain $\phi_{\mathrm{a}}=\phi_{\mathrm{a},\infty}+o(1)$,
and Lemma \ref{condp_rs_phase} yields $\bar\phi=\phi_{\mathrm{a},\infty}+o(1)$ iff $\sup_{\pi\in\cP^2_*([q])}\cB(\pi)=\phi_{\mathrm{a},\infty}$.
If $\sup_{\pi\in\cP^2_*([q])}\cB(\pi)\neq\phi_{\mathrm{a},\infty}$, then we have $\sup_{\pi\in\cP^2_*([q])}\cB(\pi)>\phi_{\mathrm{a},\infty}$ and $\limsup_{n\rightarrow\infty}\bar\phi<\phi_{\mathrm{a},\infty}$ using Fact \ref{klcondp_prel_qfed_inequ}. This completes the proof of Theorem \ref{main_qfed}.
\subsection{Proof of Theorem \ref{main_kl_divergence}}
Notice that the relative entropy density is given by
\begin{align*}
f(n)=\frac{1}{n}\KL{\mcRG^*_{\vt_n}(\mcRVA^*),\mcRVA^*}{\mcRG_{\vt_n},\mcRVAPost_{\mcRG_{\vt_n}}}
=\frac{1}{n}\Erw\left[\ln\left(r\left(\mcRG^*_{\vt_n}(\mcRVA^*),\mcRVA^*\right)\right)\right]\textrm{, }r(\mcG,\mcVA)=\frac{q^{-n}Z_\mcG}{\bar\psi_{\mcT_\mcG}(\mcVA)}\textrm{, }
\end{align*}
where $r$ denotes the derivative of $(\mcRG^*_{\vt_n}(\mcRVA^*),\mcRVA^*)$ with respect to $(\mcRG_{\vt_n},\mcRVAPost_{\mcRG_{\vt_n}})$. Basic algebra and using $g(\mcT)=n^{-1}\KL{\mcRVA^*}{\mcRVAN_\mcT}$ gives
$f(n)=\phi^*-\phi_{\mathrm{a}}+\Erw[g(\vt_n)]$. Using {\bf SYM} we get
\begin{align*}
q^{-n}\varepsilon^{2m_\mcT}\le\pr[\mcRVAN_\mcT=\mcVA]\le q^{-n}\varepsilon^{-2m_\mcT}\textrm{,}
\end{align*}
so $g(\mcT)$ is sublinear in the number of factors and hence $\Erw[g(\vt_n)]=\Erw[g(\vt_n)\vecone\{\vt_n\in\cT^\circ_n\}]+o(1)$.
Since $g(\mcT)$ coincides with $\delta^*_0(\mcT)$ in the mutual information proof we obtain $\Erw[g(\vt_n)]=o(1)$ using {\bf BAL}. Now, the result is immediate using Theorem \ref{main_qfed}.

\end{document}